\documentclass[
	showpacs,
	aps,
	superscriptaddress,
	notitlepage,
	reprint,
    prx,
    groupedaddress,
	longbibliography
]{revtex4-2}

\usepackage{array}[=2016-10-06] 
\usepackage{graphicx}
\usepackage{subfigure}
\usepackage{amsmath,amssymb,amsthm,mathrsfs,amsfonts,dsfont,amscd}
\usepackage{mathtools}
\usepackage{yfonts}
\usepackage{titletoc}
\usepackage[subfigure]{tocloft}
\usepackage{enumitem}
\usepackage{xcolor}
\usepackage{bbm}
\usepackage{adjustbox}
\usepackage{tensor}
\usepackage{braket}
\usepackage{slashed}
\usepackage{tikz-cd}
\usepackage{circuitikz}
\usepackage{multirow}
\usepackage{csquotes}
\usepackage[colorlinks=true,linkcolor=blue,citecolor=magenta,urlcolor=blue]{hyperref}
\usepackage{url}
\usepackage{tikz}

\newtheorem{m-theorem}{Theorem}

\numberwithin{equation}{section}
\theoremstyle{definition}
\newtheorem{definition}{Definition}[section]
\newtheorem*{example}{Example}
\theoremstyle{plain}
\newtheorem{remark}{Remark}
\newtheorem{claim}{Claim}
\newtheorem{theorem}[definition]{Theorem}
\newtheorem{proposition}[definition]{Proposition}
\newtheorem{lemma}[definition]{Lemma}
\newtheorem{corollary}[definition]{Corollary}

\newtheorem*{conjecture}{Conjecture}

\newcommand{\comments}[1]{}

\DeclareMathOperator{\im}{im}
\DeclareMathOperator{\tr}{Tr}
\DeclareMathOperator{\rk}{rank}

\newcommand{\CZ}{\mathrm{C}Z}

\newcommand{\CCZ}{\mathrm{C}\mathrm{C}Z}

\DeclareMathOperator{\C}{\mathcal{C}}

\DeclareMathOperator{\F}{\mathcal{F}}
\DeclareMathOperator{\HH}{\mathsf{H}}
\DeclareMathOperator{\G}{\mathsf{G}}
\DeclareMathOperator{\x}{\mathrm{x}}

\makeatletter
\def\l@subsubsection#1#2{} 
\makeatother

\begin{document}

\title{Theory of (co)homological invariants on quantum LDPC codes}

\author{Zimu Li}
\affiliation{Yau Mathematical Sciences Center, Tsinghua University, Beijing 100084, China}

\author{Yuguo Shao}
\affiliation{Yau Mathematical Sciences Center, Tsinghua University, Beijing 100084, China}

\author{Fuchuan Wei}
\affiliation{Yau Mathematical Sciences Center, Tsinghua University, Beijing 100084, China}

\author{Yiming Li}
\affiliation{Yau Mathematical Sciences Center, Tsinghua University, Beijing 100084, China}

\author{Zi-Wen Liu}
\affiliation{Yau Mathematical Sciences Center, Tsinghua University, Beijing 100084, China}

\begin{abstract}
We establish a systematic mathematical framework for constructing and computing (co)homological invariants that induce logical gates on quantum low-density parity-check (qLDPC) codes. By synthesizing tools from graph theory, algebraic topology, and number theory, we study constructions ranging from hypergraph product (HGP) codes to sheaf codes.

Using commutative algebra, we generalize the notion of canonical logical representatives from HGP codes to the sheaf code setting, resolving a longstanding challenge in explicitly characterizing sheaf codewords. Building on this foundation, we present the first comprehensive computation of cup products within the intricate framework of sheaf codes. Given Artin's primitive root conjecture which holds under the generalized Riemann hypothesis, we establish a nearly linear family of independent cup product cochains on almost-good qLDPC codes. Under a separate nonvanishing pairing condition, they are expected to yield parallel logical multi-controlled-$Z$ gates.

Moreover, by interpreting sheaf codes as covering spaces of HGP codes, we construct interlaced families through odd-degree lifts that preserve compatible invariant pairings, allowing nontrivial, constant-depth logical gates throughout an infinite family to be certified on a constant-size HGP seed. Treating local code constructions and expansion bounds as external inputs, our framework identifies the compatibility conditions for combining large dimension and distance, local testability, and constant-depth non-Clifford gates, thereby provides a mathematical foundation for separate constructions that realize these properties simultaneously.
\end{abstract}

\maketitle


{\section{Introduction}\label{sec:intro}

The quest for quantum low-density parity-check (qLDPC) codes that simultaneously achieve high encoding rate, large code distance, and desired fault-tolerant logical gates is central to fault-tolerant quantum computation. The construction of qLDPC codes has seen rapid progress in recent years, driven by the synthesis of ideas from classical coding theory~\cite{Tanner1981,SS1996,kalachev2023ProductExpansion,PK2023RobustlyTestable,Golowich_Lin2024,KP2025Extendable}, high-dimensional expanders~\cite{Hastings2021FiberBundle,PK2021,PK2022Good,QuantumTanner2022,Jeronimo2021}, and algebraic topology~\cite{DHLV2022,Breuckmann2024Cups,First2024,Dinur2024sheaf,Panteleev2024,Zhu2025A,Zhu2025B}, which culminated in the recent discovery of good qLDPC codes that have constant encoding rate and linear distance~\cite{PK2022Good,QuantumTanner2022,DHLV2022}. 
Nonetheless, quantum computing requires more than just good codes: universal logical quantum information processing with low overhead is essential, with non-Clifford gates such as 
$T$ and $\CCZ$ incurring dominant costs and hence constituting the key bottleneck. Utilizing the intrinsic code structures to realize encoded logical gates in a fault-tolerant manner is a natural and compelling route toward highly efficient fault tolerance schemes. Yet, while the standard code parameters have been made asymptotically optimal, combining favorable parameters with fault-tolerant non-Clifford logical gates has been a much more formidable problem, with progress typically restricted to subsets of the desired properties despite substantial efforts from both mathematical~\cite{RainbowCode,Wills2024magic,He2025addressable,VirgileCCZ,Golowich_Guruswami2025A,Nguyen2025CCZ,Breuckmann2024Cups,Tiew2026copycup,Lin2024transversal,Golowich_Lin2024,qtannercolor,Li2025Poincare,Barg_2026QRM,Coble_2026Coxeter} and physical perspectives~\cite{Jochym_O_Connor_2021,10.21468/SciPostPhys.14.4.065,Chen_2023,Zhu2023,Zhu2025A,Zhu2025B,zhu2026nonabelianqldpctqftformalism,christos2026nonabelianquantumlowdensityparity}.

Underlying the understanding and developments of qLDPC codes is an elegant and powerful topological correspondence: quantum Calderbank--Shor--Steane (CSS) codes can be understood as cochain complexes, where physical qubits correspond to cells of a certain dimension and stabilizers are given by coboundary operators~\cite{Kitaev_2003,Bombin_2013,BH2013homological}. This viewpoint leads to the study of logical gates or code symmetries via (co)homological invariants. Cup and cap products are fundamental operations in algebraic topology that provide bilinear maps on (co)homology classes. Recently, there have been intensive research efforts on using them to induce logical diagonal gates such as the multi-controlled-$Z$ gates~\cite{10.21468/SciPostPhys.14.4.065,Chen_2023,Breuckmann2024Cups,Tiew2026copycup,Lin2024transversal,Zhu2025A,Zhu2025B,Li2025Poincare} on CSS codes. Crucially, when the underlying complex is sparse, i.e., each cell is incident to a bounded number of higher-dimensional cells, the resulting cup product circuit has constant depth, keeping error propagation within a bounded range and yielding a fault-tolerant implementation~\cite{Gottesman1998,Gottesman2013}. As explained, non-Clifford gates that correspond to high-order invariants hold outstanding significance for universal fault-tolerant quantum computation. They also underlie the magic state distillation scheme~\cite{Bravyi_2012,Krishna_Tillich2019,Wills2024magic,He2025addressable,Nguyen2025CCZ}, where high gate parallelizability is a particularly desirable feature for reducing the overhead.

The sheaf code framework~\cite{PK2022Good,QuantumTanner2022,DHLV2022,VirgileLiftCSS,VirgileLiftCSSHandlebody,Virgilemoderatelengthliftedquantumtanner,First2024,Dinur2024sheaf,Panteleev2024,Li2025Poincare} generalizes the above constructions by assigning local classical codes to cells of a high-dimensional expander, producing cochain complexes with local coefficients. This provides a powerful and flexible mechanism for code construction which encompasses the asymptotically (almost) good qLDPC and quantum locally testable codes (qLTCs) with complexes in higher dimensions~\cite{PK2022Good,DHLV2022,QuantumTanner2022,Dinur2024sheaf}. However, computing logical actions is already subtle for generic hypergraph product (HGP) codes with only scalar local coefficients, whose cochain complexes are constructed from basic tensor products. For sheaf codes, local coefficients and graph lifts introduce further algebraic complexity: the structure of logical representatives becomes opaque, and the computation of cup products, already nontrivial on ordinary complexes, becomes even less tractable. Combining fault-tolerant non-Clifford gates with other good parameters on qLDPC codes therefore requires a new mathematical framework for analyzing logical actions.

The main contribution of this work is a systematic mathematical theory for tackling this problem through explicit calculations: we construct logical representatives for sheaf codes, evaluate their cup products at the cochain level, and separate the algebraic conditions controlling code parameters and nontrivial pairing. 
To be precise, the principal  object of interest in this work is the invariant form:
\begin{align}\label{eq:invariant_form_main}
	T_\xi(x_1,\ldots,x_\rho)
	\coloneq
	\langle x_1\smile_c\cdots\smile_c x_\rho,\xi\rangle.
\end{align}
Here $x_j$ represent logical $X$ operators in different code blocks, $\smile_c$ denotes the cubical cup product that will be defined later, and $\xi$ is a cycle associated with the corresponding tensor product sheaf. We reserve the notation $\smile$ for the ordinary cup product on simplicial complexes. The order $\rho$ determines the multi-controlled-$Z$ type of the logical operation \cite{10.21468/SciPostPhys.14.4.065,Zhu2023,Chen_2023,RainbowCode,Wang_2024,BalancedProduct2020,Lin2024transversal,Golowich_Lin2024,Zhu2025A,Zhu2025B,Li2025Poincare}, while the tensor \emph{subrank} of $T_\xi$ measures the largest diagonal restriction of the logical action \cite{Golowich_Lin2024,Golowich_Guruswami2025A,Golowich_Guruswami2025B}. A large subrank certifies many independent logical gates on suitable logical subspaces, an important feature for magic state distillation~\cite{Nguyen2025CCZ,Wills2024magic,He2025addressable}.

More specifically, our results divide into three complementary levels. First, for HGP codes, we construct canonical logical representatives for which every relevant cup product reduces to a single top-degree cube. This provides a systematic way to compute logical actions on HGP codes and leads to parallel multi-controlled-$Z$ actions.
Second, on the more complicated structure on sheaf codes, we use commutative algebra techniques to recover an analogous pattern through what we call the \emph{polarized logical representatives}. 
{Under a standard number-theoretic condition---Artin's primitive root conjecture---our construction gives $\widetilde{\Theta}(N)$ independent cup product cochains.  Although an explicit cycle $\xi$ in the tensor product sheaf is still absent, this result paves a path toward overcoming the earlier difficulty in analyzing logical actions on sheaf codes.}
Finally, in a more unified language, we interpret sheaf codes as covering spaces of HGP codes, enabling code properties and possibly new logical operations on the easier HGP codes to be pulled back to the sheaf code setting. Odd-degree lifts generate interlaced families from a constant-size HGP seed while preserving compatible invariant pairings on canonically extended logical classes. As a result, a finite calculation on the seed can certify nontrivial logical gates throughout an infinite family. Unlike the cyclic localization above, this induction does not require Artin's conjecture. The desired asymptotic code parameters can be supplied by separate construction theorems, provided that their hypotheses hold for the same local codes and lifts. In presenting all these three results, we treat local code constructions and expansion bounds as external inputs, stating the properties required of them so that results developed separately can be incorporated into the gate theory. This provides a systematic foundation for combining large code dimension $k$, large distance $d$, and many constant-depth multi-controlled-$Z$ gates with qLDPC structure, provided that the required hypotheses hold for compatible choices of complexes and local codes.

The paper is organized as follows. Sec.~\ref{sec:main_results} introduces the minimal technical dictionary and states the main results. Sec.~\ref{sec:main_HGP} gives the explicit  HGP computation: canonical logical representatives reduce the cup product to a single cube, and dual free-variable cycles select logical tuples within assigned directional sectors. This computation is elementary yet enlightening. Then Sec.~\ref{sec:main_sheaf} explains how local codes, lifts, and commutative algebra extend this calculation to sheaf codes. Sec.~\ref{sec:main_induction} presents the induction scheme that transports seed invariants through the interlaced HGP--sheaf family. Sec.~\ref{sec:discussion} concludes with discussion and outlook. 

The Appendix is extensive by design. Our constructions draw on concepts across theoretical computer science, coding theory, high-dimensional expansion, sheaf theory, and algebraic topology, not all of which are standard in the quantum information literature. We comprehensively develop the background needed for the invariant calculations through explicit definitions, low-dimensional examples, and concrete matrix formulas. Results on local code existence and expansion enter as external inputs with explicitly stated hypotheses. Accordingly, the Appendix serves as a self-contained technical reference for readers who wish to examine a particular construction or proof in detail.}


\section{Framework and main results}\label{sec:main_results}

\subsection{Basic notions}

All vector spaces are finite-dimensional over a finite field $\mathbb{F}_q$ of characteristic $2$, where $q=2^s$ and $s$ is a sufficiently large constant. Roughly speaking, a code over $\mathbb F_q$ can be realized by grouping $s$ physical qubits at each site. Characteristic $2$ matches the qubit setting, while a sufficiently large field is needed for product-expanding classical codes in the construction of sheaf codes with (nearly) optimal parameters~\cite{KP2025Extendable,Dinur2024sheaf,LLL2026nontrivial}. A \emph{cochain complex} is a sequence of vector spaces
\begin{equation}
	\hspace{-3mm}
	\begin{tikzcd}
		\cdots \arrow[r, "\delta^{i-1}"] & C^i \arrow[r, "\delta^{i}"] & C^{i+1} \arrow[r, "\delta^{i+1}"] & C^{i+2} \arrow[r, "\delta^{i+2}"] &
		\cdots
	\end{tikzcd}
\end{equation}
with \emph{coboundary operators} satisfying $\delta^i \delta^{i-1} = 0$. The dual sequence is the following chain complex:
\begin{equation}
	\hspace{-3mm}
	\begin{tikzcd}
		\cdots & C_i \arrow[l, "\partial_{i} "'] & C_{i+1} \arrow[l, "\partial_{i+1} "'] & C_{i+2} \arrow[l, "\partial_{i+2} "'] & \cdots \arrow[l, "\partial_{i+3} "']
	\end{tikzcd}
\end{equation} 
Here $C_i \coloneq (C^i)^\ast$, and $\partial_i \coloneq (\delta^{i-1})^\ast$ is the \emph{boundary operator}. After bases are fixed, $\partial_i$ is represented by $(\delta^{i-1})^T$, and we use the induced coordinate identification $C_i \cong C^i$. Elements of $\ker \delta^i$ are the $i$-th \emph{cocycles} and elements of $\im \delta^{i-1}$ are the $i$-th \emph{coboundaries}. Analogously, we define \emph{cycles} and \emph{boundaries} for chain complexes. We abbreviate these complexes as $(C_\bullet,\partial_\bullet)$ and $(C^\bullet,\delta^\bullet)$, respectively. 

A \emph{Calderbank--Shor--Steane (CSS) code} is extracted from three consecutive terms. Physical qudits are indexed by a basis of $C^i$, with $\delta^{i} = H_Z$ and $\delta^{i-1} = H_X^T$ defining the stabilizer generators. The $X$-type and $Z$-type logical operators are represented by elements of the cohomology $H^i(C) = \ker \delta^i / \im \delta^{i-1}$ and the homology $H_i(C) = \ker \partial_i / \im \partial_{i+1}$, respectively. The \emph{code dimension} $k$, which is the number of encoded logical qudits, is given by 
\begin{align}
	\begin{aligned}
		k = & \dim H^i(C) = \dim C^i - \rk \delta^i - \rk \delta^{i-1} \\
		= & \dim C_i - \rk \partial_i - \rk \partial_{i+1} = \dim H_i(C).
	\end{aligned}
\end{align}
The \emph{code distance} is given by $d = \min\{d_X, d_Z\}$, where
\begin{align}
	d_X = \min_{x \in \ker \delta^i \setminus \im \delta^{i-1}} \vert x \vert, \quad 
	d_Z = \min_{z \in \ker \partial_i \setminus \im \partial_{i+1}} \vert z \vert.
\end{align}
The \emph{Hamming weight} $\vert x \vert$ counts the number of nonzero components in $x$. The quantities $d_X$ and $d_Z$ are the cosystolic and systolic distances, respectively.

Cochain complexes can be realized on combinatorial cell complexes. For example, tessellating a torus by squares gives $C^0 \to C^1 \to C^2$ from its vertices, edges, and plaquettes. The coboundary operators are defined by inclusion relations of these geometric ingredients. The resulting cochain complex is exactly that of the toric code. To obtain (almost) good qLDPC codes, one uses substantially more intricate combinatorial complexes, namely \emph{high-dimensional expanders (HDXs)}. A standard realization employs \emph{cubical complexes} built from lifts and products of \emph{expander graphs} \cite{LPS1988,Margulis1988,Alon1986Expander,Bilu2004,ZigZzag2008,Marcus2015I,Marcus2018IV,Hall_2018,Agarwal2016,Jeronimo2021,PK2021,PK2022Good,Dinur2024sheaf}.

To be precise, let $G_0 = (V_0, E_0)$ be an $n$-regular graph. Throughout the paper, $n$ is a sufficiently large constant, whereas $N$ denotes the system size. Let $\HH$ be a finite Abelian group with $|\HH| = l = \exp(O( | V_0 | ))$ elements, where $l$ is assumed to be an odd integer. Let $\gamma: E_0 \rightarrow \HH$ be a map assigning each edge $(u,v)$ a group element $s \in \HH$. Then we have $t$ permutation sets $A_1,\cdots,A_t$ acting on $\HH \times V_0^t \times \{0,1\}^t$, where each $A_j = \{a_j^1,\cdots, a_j^n\}$ and
\begin{align}\label{eq:generator_action_main}
	\begin{aligned}
		& a_j^{\mu} \cdot (h,v_1,...,v_j,...,v_t; b_1,...,b_j,...,b_t) \\
		= & (\gamma_{(v_j, v_j')} \cdot h,v_1,...,v_j',...,v_t; b_1,...,1-b_j,...,b_t),
	\end{aligned}
\end{align} 
where $h$ is the lift coordinate and $b_j$ is the bipartite bit in direction $j$, which removes several orientation ambiguities. This defines the \emph{double cover} of $G_0$ in the $j$-th direction, where $(v_j',1-b_j)$ is the $\mu$-th neighbor of $(v_j,b_j)$ in the double cover. The whole structure defines a $t$-dimensional cubical complex $X$: for $0 \leq p \leq t$, the set of $p$-dimensional cubes of $X$ is defined as follows~\cite{Dinur2024sheaf}:
\begin{align}\label{eq:cubical_complex_main}
	X(p) \coloneqq \bigg \{ & [g;\,(a_j)_{j\in S},\,(b_j)_{j\notin S}] : g\in \HH \times V_0^t,\\
	& S\subseteq [t],\, |S|=p,\,
	a\in\prod_{j\in S}A_j,\,
	b\in\prod_{j\notin S}\{0,1\}\bigg\}. \notag
\end{align}
For simplicity, we suppress the superscripts on $a_j^\mu$. For a cell $\sigma = [g;\,(a_j)_{j\in S},\,(b_j)_{j\notin S}]$, we call $S=\operatorname{type}(\sigma)$ the \emph{type} of $\sigma$. In particular, a $t$-cube has $S=[t]$ and is written as
\begin{align}
	[(h,v_1,\ldots,v_t);a_1,\ldots,a_t].
\end{align}
Since $\vert A_j \vert = n$ for every $j \in [t]$, $\vert X(p) \vert = \binom{t}{p} 2^{t-p} n^p \cdot l \cdot |V_0|^t$. If we put physical qudits on $1$-cubes, the system size $N = t 2^{t-1} n \cdot l \cdot |V_0|^t$. As mentioned, the global structure of $X$ can be rather complicated. Nevertheless, its local building blocks are just cubes, which motivates the term ``cubical complex'' (see Fig.~\ref{fig:3-cube_main} for $t = 3$). Lower-dimensional cubes are contained in higher-dimensional ones by a natural inclusion relation. Secs.~\ref{sec:covering} and~\ref{sec:cubical} give the full definitions and concrete two- and three-dimensional examples.

\begin{figure}[t]
	\centering
	\resizebox{0.9\linewidth}{!}{%
		\begin{tikzpicture}[x=0.75pt,y=0.75pt,yscale=-1,xscale=1]
			\draw  [color={rgb, 255:red, 0; green, 0; blue, 0 }  ,draw opacity=0.5 ] (178,88) -- (255,11) -- (461.33,11) -- (461.33,190.67) -- (384.33,267.67) -- (178,267.67) -- cycle ; \draw  [color={rgb, 255:red, 0; green, 0; blue, 0 }  ,draw opacity=0.5 ] (461.33,11) -- (384.33,88) -- (178,88) ; \draw  [color={rgb, 255:red, 0; green, 0; blue, 0 }  ,draw opacity=0.5 ] (384.33,88) -- (384.33,267.67) ;
			\draw [color={rgb, 255:red, 0; green, 0; blue, 0 }  ,draw opacity=0.5 ] [dash pattern={on 4.5pt off 4.5pt}]  (255,11) -- (255,118.17) -- (255,191.67) ;
			\draw [color={rgb, 255:red, 0; green, 0; blue, 0 }  ,draw opacity=0.5 ] [dash pattern={on 4.5pt off 4.5pt}]  (255,190.67) -- (461.33,190.67) ;
			\draw [color={rgb, 255:red, 0; green, 0; blue, 0 }  ,draw opacity=0.5 ] [dash pattern={on 4.5pt off 4.5pt}]  (180.46,267.67) -- (225.21,222.92) -- (257.46,190.67) ;
			
			\draw (115,263) node [anchor=north west][inner sep=0.75pt]   [align=left] {$\displaystyle [ g;0,0,0]$};
			\draw (400,263) node [anchor=north west][inner sep=0.75pt]   [align=left] {$\displaystyle [ a_{1} \cdot g;1,0,0]$};
			\draw (394,85) node [anchor=north west][inner sep=0.75pt]   [align=left] {$\displaystyle [ a_{1} a_{3} \cdot g;1,0,1]$};
			\draw (473,181) node [anchor=north west][inner sep=0.75pt]   [align=left] {$\displaystyle [ a_{1} a_{2} \cdot g;1,1,0]$};
			\draw (467,2) node [anchor=north west][inner sep=0.75pt]   [align=left] {$\displaystyle [ a_{1} a_{2} a_{3} \cdot g;1,1,1]$};
			\draw (95,85) node [anchor=north west][inner sep=0.75pt]   [align=left] {$\displaystyle [ a_{3} \cdot g;0,0,1]$};
			\draw (180,181) node [anchor=north west][inner sep=0.75pt]   [align=left] {$\displaystyle [ a_{2} \cdot g;0,1,0]$};
			\draw (160,2) node [anchor=north west][inner sep=0.75pt]   [align=left] {$\displaystyle [ a_{2} a_{3} \cdot g;0,1,1]$};
			\draw (235,274) node [anchor=north west][inner sep=0.75pt]   [align=left] {$\displaystyle \textcolor[rgb]{0.82,0.01,0.11}{[ g;a_{1} ,0,0]}$};
			\draw (277,194) node [anchor=north west][inner sep=0.75pt]   [align=left] {$\displaystyle \textcolor[rgb]{0.82,0.01,0.11}{[a_{2} \cdot g;a_{1},1,0]}$};
			\draw (235,70) node [anchor=north west][inner sep=0.75pt]   [align=left] {$\displaystyle \textcolor[rgb]{0.82,0.01,0.11}{ [a_{3} \cdot g;a_{1},0,1]}$};
			\draw (277,-5) node [anchor=north west][inner sep=0.75pt]   [align=left] {$\displaystyle \textcolor[rgb]{0.82,0.01,0.11}{ [a_{2} a_{3} \cdot g;a_{1},1,1]}$};
			\draw (97,139) node [anchor=north west][inner sep=0.75pt]   [align=left] {$\displaystyle \textcolor[rgb]{0.96,0.65,0.14}{[ g;0,0,a_{3}]}$};
			\draw (326,139) node [anchor=north west][inner sep=0.75pt]   [align=left] {$\displaystyle \textcolor[rgb]{0.96,0.65,0.14}{[ a_{1} \cdot g;1,0,a_{3}]}$};
			\draw (205,139) node [anchor=north west][inner sep=0.75pt]   [align=left] {$\displaystyle \textcolor[rgb]{0.96,0.65,0.14}{[ a_{2} \cdot g;0,1,a_{3}]}$};
			\draw (465,139) node [anchor=north west][inner sep=0.75pt]   [align=left] {$\displaystyle \textcolor[rgb]{0.96,0.65,0.14}{[ a_{1} a_{2} \cdot g;1,1,a_{3}]}$};
			\draw (180,219) node [anchor=north west][inner sep=0.75pt]   [align=left] {$\displaystyle \textcolor[rgb]{0.29,0.56,0.89}{[ g;0,a_{2} ,0]}$};
			\draw (432,219) node [anchor=north west][inner sep=0.75pt]   [align=left] {$\displaystyle \textcolor[rgb]{0.29,0.56,0.89}{[ a_{1} \cdot g;1,a_{2} ,0]}$};
			\draw (160,41) node [anchor=north west][inner sep=0.75pt]   [align=left] {$\displaystyle \textcolor[rgb]{0.29,0.56,0.89}{[ a_{3} \cdot g;0,a_{2} ,1]}$};
			\draw (432,41) node [anchor=north west][inner sep=0.75pt]   [align=left] {$\displaystyle \textcolor[rgb]{0.29,0.56,0.89}{[ a_{1} a_{3} \cdot g;1,a_{2} ,1]}$};
		\end{tikzpicture}	
	}
	\caption{A 3-cube $[g;a_1,a_2,a_3]$ with all vertices (0-cubes) and edges (1-cubes) shown. For clarity, the faces (2-cubes) are omitted.}
	\label{fig:3-cube_main}
\end{figure}

High-dimensional expansion controls the global geometry of the complex, while a \emph{cellular sheaf} refines its local structure and is needed to obtain (almost) good qLDPC codes \cite{PK2022Good,DHLV2022,QuantumTanner2022,Dinur2024sheaf}. Here, by a cellular sheaf we mean a covariant coefficient system; henceforth, we refer to it simply as a sheaf. A cellular sheaf $\F$ assigns a vector space $\F(\sigma)$ to every $p$-cube $\sigma$ and a linear map $\F_{\sigma,\tau}:\F(\sigma)\to\F(\tau)$ whenever $\sigma$ is contained in the $(p+1)$-cube $\tau$, denoted by $\sigma \prec \tau$.  Moreover, if $\sigma \prec \tau \prec \pi$, then the linear maps satisfy the compatibility condition
\begin{align}\label{eq:compatible_main}
	\mathcal{F}_{\tau,\pi} \circ \mathcal{F}_{\sigma,\tau} = \mathcal{F}_{\sigma, \pi}. 
\end{align}
Naturally, $\mathcal{F}_{\sigma,\sigma}$ is also required to be the identity map.

Given $X$ and $\mathcal{F}$, the sheaf cochain space is
\begin{align}\label{eq:sheaf_cochain_space_main}
	C^i(X, \mathcal{F}) \coloneq \bigoplus_{\sigma \in X(i)} \F(\sigma).
\end{align}
For any $x \in C^i(X,\mathcal{F})$, the coboundary $\delta^i x \in C^{i+1}(X,\mathcal{F})$ is defined componentwise by
\begin{align}\label{eq:coboundary_operator_main}
	(\delta^i x)(\tau) = \sum_{\substack{\sigma \in X(i) \\ \sigma \prec \tau}} \mathcal{F}_{\sigma,\tau}\bigl(x(\sigma)\bigr), \qquad \tau \in X(i+1).
\end{align}
The coboundary condition $\delta^{i+1} \circ \delta^i = 0$ can be verified by the inclusion relation of cubes in $X$ and Eq.~\eqref{eq:compatible_main}. Further details about realizing $\F$ via \emph{local classical codes} are postponed to Sec.~\ref{sec:local_codes_main}.

For sheaves $\F_1,\ldots,\F_\rho$ and degrees $p_1,\ldots,p_\rho$, the cubical cup product is a map
\begin{align}\label{eq:cup_map_main}
	\prod_{j=1}^{\rho}C^{p_j}(X,\F_j) \rightarrow C^{p_1+\cdots+p_\rho} (X,\F_1\otimes\cdots\otimes\F_\rho)
\end{align}
that satisfies the Leibniz rule \cite{Chen_2023,Wang_2024,Lin2024transversal,Zhu2025A,Li2025Poincare}. Intuitively, this operation combines lower-dimensional cubes into higher-dimensional ones. The tensor product sheaf $\mathcal{F} \otimes \mathcal{G}$ is defined
componentwise by tensoring the corresponding local coefficient spaces and structure maps:
\begin{align}\label{eq:tensor_sheaf_main}
	\mathcal{F}_{\sigma,\tau}  \otimes \mathcal{G}_{\sigma,\tau}: \ \mathcal{F}(\sigma) \otimes \mathcal{G}(\sigma) \rightarrow \mathcal{F}(\tau) \otimes \mathcal{G}(\tau).
\end{align}	
Taking the inner product (pairing) with a cycle
\begin{align}\label{eq:invariant_form_2_main}
	\xi\in C_{p_1 + \cdots + p_\rho}(X,\F_1\otimes\cdots\otimes\F_\rho)
\end{align}
gives the invariant form in Eq.~\eqref{eq:invariant_form_main}, which depends only on the cohomology classes, i.e.,
\begin{align}
	T_\xi(\ldots,x_j,\ldots) = T_\xi(\ldots,x_j + \delta^{p_j-1} y_j,\ldots)
\end{align}
for any $j = 1,\ldots,\rho$ and any $y_j \in C^{p_j-1}(X,\F_j)$.
When the base graph $G_0$ has a constant vertex degree $n$, the coboundary operators have $O(1)$ row and column weights independent of the system size $N$. We call $X$ a sparse cubical complex. As a result, the cochain-level formula $T_\xi$ is a bounded-degree polynomial, so the associated physical circuit has constant depth; see Theorem~\ref{thm:cup_cubical}.

By Eq.~\eqref{eq:invariant_form_2_main}, $T_\xi$ can be regarded as a $\rho$-tensor on cohomology. Its \emph{subrank} is the largest $r$ for which one can find cohomological representatives $L_j^a$ satisfying
\begin{align}\label{eq:main_subrank_main}
	T_\xi(L_1^{i_1},\ldots,L_\rho^{i_\rho}) = \delta_{i_1,\ldots,i_\rho}, \qquad 1\leq i_j \leq r.
\end{align}
The physical circuit then induces $r$ independent logical $\mathrm{C}^{\rho-1}Z$ gates on disjoint tuples. 


\subsection{Main statements}

As mentioned, a CSS code can be extracted from three consecutive terms $C^{i-1} \to C^i \to C^{i+1}$. In the following theorem, we avoid excessive subscripts by taking all codes in degree $1$, although the statement holds in general degrees whenever $p_1+\cdots+p_\rho\leq t$. Good qLDPC code families have been constructed from two-dimensional expanders~\cite{PK2022Good,DHLV2022,QuantumTanner2022}, for which logical $\CZ$ gates were established in \cite{Li2025Poincare}. To obtain non-Clifford multi-controlled-$Z$ gates, we take $t \geq 3$, for which only almost-good qLDPC codes with parameters
\begin{align}
	[\![ N, \Theta(N), \Theta(N/ \mathrm{polylog}\, N)  ]\!],
\end{align}
are currently known~\cite{Dinur2024sheaf} \footnote{We remark that this result was the best known at the time of writing. Recently, good locally testable codes that achieves good parameters and constant soundness were found in different high dimensional cubical complexes. We do not use this improvement.}. Let $\HH = C_l$ be the cyclic group of $l$ elements and let $\x \in C_l$ be any generator of the group.  

\begin{theorem}[Informal]\label{thm:cup_sheaf_main}
	Assume that $s$ is odd. Under Artin's primitive root conjecture, which Hooley proved under GRH~\cite{Hooley1967}, there are infinitely many primes $l$ and cubical complexes $X_l$ obtained from certain $l$-lifts such that
	\begin{enumerate}
		\item There exist $t$ sheaves $\F_1,\ldots,\F_t$ such that for each $j = 1,\ldots,t$, the complexes $\{C^\bullet(X_l,\F_j)\}_l$ form an asymptotic family of almost-good qLDPC codes.
		
		\item For each fixed $l$, there are cocycle representatives of cohomology classes ($X$-type logical representatives)
		\begin{align}
			\{L_1^{i_1}\}_{i_1=1}^l \subset C^1(X_l, \F_1), \ldots, \{L_t^{i_t}\}_{i_t=1}^l \subset C^1(X_l, \F_t), \notag
		\end{align}  
		such that 
        \begin{align}\label{eq:sheaf_cup_thm_main}
        \begin{aligned}
            & L_1^{i_1}\smile_c\cdots\smile_cL_t^{i_t} \\ 
= & \sum_{i=1}^l\delta_{i_1,\ldots,i_t,i}[(\x^i,v_1,\ldots,v_t);a_1,\ldots,a_t].
        \end{aligned}
        \end{align}
where $\x^i\in C_l$ labels the lift, and
$\delta_{i_1,\ldots,i_t,i}=1$ if
$i_1=\cdots=i_t=i$ and is zero otherwise.
Here $v_1,\ldots,v_t$ are vertices of the base graph, and $a_1,\ldots,a_t$ are the permutation actions defined in Eq.~\eqref{eq:cubical_complex_main}.
	\end{enumerate}
\end{theorem}
The formal statement and proof details are given in Sec.~\ref{sec:cup_sheaf}. For generic choices of logical representatives, the cup products can be intricate and may involve many top-dimensional cells. In contrast, the products in Eq.~\eqref{eq:sheaf_cup_thm_main} are either a single $t$-cube or zero, giving a particularly simple form. To achieve this, we use \emph{polarized logical representatives}, defined in Sec.~\ref{sec:polarize_main} specifically for sheaf codes. Moreover, the number-theoretic conjecture~\cite{Artin1966,Hooley1967} supplies special primes for which the group algebra $\mathbb{F}_q[C_l]$ has a particularly simple ideal structure, enabling the required computations.

By Eq.~\eqref{eq:invariant_form_2_main}, if there exists a cycle $\xi$ that is nonzero on these $t$-cubes, we immediately obtain
\begin{align}\label{eq:sheaf_invariant_main}
	T_\xi(L_1^{i_1},\ldots,L_t^{i_t}) = \delta_{i_1,\ldots,i_t}.
\end{align}
Recall that $l = \exp(O( | V_0 | )) = \Theta(N/\mathrm{polylog}\,N)$. Then the subrank is $\Theta(N/\mathrm{polylog}\,N)$, producing the same number of parallel logical $\mathrm{C}^{t-1}Z$ gates. The existence of $\xi \in C_t(X,\F_1\otimes\cdots\otimes\F_t)$ is a logically separate condition on the tensor product sheaf. Establishing the existence of such a $\xi$ for choices of $\F_i$ that yield (almost) good codes remains an open problem. This condition can be verified directly when good code parameters are not required, with an explicit example provided in Sec.~\ref{sec:cycles}.


We next present a general induction principle for realizing higher-order cohomological invariants as constant-depth circuits without additional assumptions or conjectures. Recall that Eq.~\eqref{eq:cubical_complex_main} defines $X$ using a finite group $\HH$. We take $\HH=C_l$ in Theorem~\ref{thm:cup_sheaf_main}, whereas here we take $\HH$ to be trivial:
\begin{align}\label{eq:cubical_complex_HGP_main}
	X'(p) \coloneqq \bigg \{ & [g;\,(a_j)_{j\in S},\,(b_j)_{j\notin S}] : g\in  V_0^t,\\
	& S\subseteq [t],\, |S|=p,\,
	a\in\prod_{j\in S}A_j,\,
	b\in\prod_{j\notin S}\{0,1\}\bigg\}. \notag
\end{align}
Because the coordinates in Eq.~\eqref{eq:cubical_complex_HGP_main} separate across the $t$ graph directions, $X'$ is the $t$-fold Cartesian product of the double cover of $G_0$. As mentioned, the binary bits $b_j$ specifying these double covers make the cubical construction well defined and facilitate the sheaf code computations below. The resulting cochain complex $C^\bullet(X',\F')$ is precisely the \emph{hypergraph product (HGP)} complex and gives rise to HGP codes~\cite{Zemor2014,Leverrier_2015}. Setting $\HH$ to the trivial group produces $X'$, whereas for general $\HH$, forgetting the lift coordinate defines a natural projection from $X$ in Eq.~\eqref{eq:cubical_complex_main} to $X'$ in Eq.~\eqref{eq:cubical_complex_HGP_main}: for every $h\in\HH$,
\begin{align}\label{eq:projection_main}
	\begin{aligned}
		& \sigma = [(h,v_1,\ldots, v_t);\,(a_j)_{j\in S},\,(b_j)_{j\notin S}] \\
		\mapsto & \ \sigma' = [(v_1,\ldots, v_t);\,(a_j)_{j\in S},\,(b_j)_{j\notin S}].
	\end{aligned}
\end{align}
The preimage $\{\sigma\}$ is the \emph{fiber} of $\sigma'$. Accordingly, the sheaf cochain complex of $X$ can be viewed as the \emph{covering space} of the HGP cochain complex of $X'$:
\begin{align}
	C^\bullet(X,\F) \longrightarrow C^\bullet(X',\F'),
\end{align}
where $\F$ is the pullback cellular sheaf of $\F'$ via the covering map: for any $\sigma$ as in Eq.~\eqref{eq:projection_main}, we define $\F(\sigma) \coloneq \F'(\sigma')$. The local linear maps are defined similarly. Accordingly, we identify $\F$ with $\F'$. Conversely, from the perspective of graph-lift theory~\cite{Bilu2004,ZigZzag2008,Marcus2015I,Marcus2018IV,Chandrasekaran2017,Hall_2018,Agarwal2016,Jeronimo2021}, HDXs and good qLDPC codes naturally arise as lifts of base structures such as Cartesian products of graphs and HGP codes \cite{PK2021,PK2022Good,Dinur2024sheaf,VirgileLiftCSS,VirgileLiftCSSHandlebody,Virgilemoderatelengthliftedquantumtanner}. The following theorem is proved in Sec.~\ref{sec:induct_invariant}.

\begin{theorem}[Informal]\label{thm:induct_cup_main}
	Suppose that the base HGP codes obtained from $C^\bullet(X',\F')$ support a nontrivial logical gate induced by a (co)homological invariant. Then the same logical gate is inherited by the covering sheaf codes defined by $C^\bullet(X,\F)$. In particular, when the sheaves are chosen as in Ref.~\cite{KP2025Extendable}, the covering cochain complexes on high-dimensional expanders form a family of (nearly) good qLDPC codes~\cite{Dinur2024sheaf} while preserving the invariant.
\end{theorem}

We highlight two notable points arising from Theorem~\ref{thm:induct_cup_main}. First, the invariant may be much more general than the specific cup product form in Eq.~\eqref{eq:invariant_form_main}; the theorem provides a unified viewpoint in which sheaf codes, including existing good qLDPC codes, are covering spaces of more tractable HGP codes. Second, the presence of a nontrivial invariant in an entire qLDPC family can be certified by testing a constant-size HGP seed. Starting from a constant $X_1'$, an odd constant-size lift in each graph direction produces the next Cartesian-product complex $X_2'$, whereas an exponentially large Abelian lift of $X_1'$ produces the high-dimensional expander $X_1$. This gives the double sequence
\begin{align}\label{eq:induct_lift_main}
	\begin{array}{ccccc}
		X_1' & \xrightarrow{\HH_1'} & X_2' & \xrightarrow{\HH_2'} & \cdots \\
		{\scriptstyle \HH_1}\downarrow && {\scriptstyle \HH_2}\downarrow && \\
		X_1 & \dashrightarrow & X_2 & \dashrightarrow & \cdots
	\end{array}
\end{align}
The upper row consists of HGP complexes, whereas the lower row consists of HDXs. In asymptotic sheaf-code constructions, two consecutive HDXs $X_i$ and $X_{i+1}$ need not have a transparent direct relation because their base graphs and lifting groups may differ. This obstructs a direct inductive transport of logical gates. The auxiliary HGP row in Eq.~\eqref{eq:induct_lift_main} bridges consecutive exponential lifts. Therefore, once the constant-size seed supports the desired logical operation, it is inherited throughout both rows, including the asymptotic code family. This separates the finite calculation from the asymptotic construction.  

As a result of these insights, a thorough analysis of (co)homological invariants of HGP codes can significantly contribute to the study of sheaf codes. Many of the computations in Sec.~\ref{sec:proof_HGP} for HGP codes serve as a basis for Theorem~\ref{thm:cup_sheaf_main}. We now state one representative example; details are given in Sec.~\ref{sec:HGP_scalar}.


\begin{theorem}[Informal]\label{thm:cup_HGP_main}
	Let $X'$ be the $t$-fold Cartesian product of the double cover of $G_0$ with trivial lift. Take $\F'(\sigma') \equiv \mathbb{F}_q$ and every $\F_{\sigma',\tau'}'$ to be the scalar identity. We call $\F'$ the \emph{constant cellular sheaf with scalar coefficients} and write the cochain complex as $C^\bullet(X',\mathbb{F}_q)$. Let
	\begin{align}
		C^0(X',\mathbb{F}_q) \rightarrow C^1(X',\mathbb{F}_q) \rightarrow C^2(X',\mathbb{F}_q)
	\end{align}
	be the extracted CSS code at degree $1$. The $X$-type logical representatives decompose into $t$ groups. For every $t$-tuple containing one representative from each group, there is an addressable cohomological invariant and hence an addressable logical $\mathrm{C}^{t-1}Z$ gate. Moreover, there exists an invariant form of subrank $k/t$, where $k$ is the code dimension.
\end{theorem}

An \emph{addressable logical gate} acts on one prescribed tuple of logical qudits.
The parallel circuits above can be viewed as compositions of addressable gates. For vertex degree $n=2$, Theorem~\ref{thm:cup_HGP_main} reduces to the familiar $t$-dimensional toric code example of cup product gates~\cite{10.21468/SciPostPhys.14.4.065,Zhu2023,Chen_2023,RainbowCode,Wang_2024,BalancedProduct2020,Lin2024transversal,Golowich_Lin2024,Zhu2025A,Zhu2025B,Li2025Poincare}.

Table~\ref{tab:conditions} summarizes the ingredients used in Theorem~\ref{thm:cup_sheaf_main} and the invariant form in Eq.~\eqref{eq:invariant_form_main}. It provides a guide to the key mathematical ingredients for building qLDPC codes with desired properties. The basic requirements are broadly applicable, while some of the advanced requirements could be modified and even replaced for different (co)homological invariants. In Table~\ref{tab:conditions}, $\lambda$ denotes the largest absolute value of a nontrivial adjacency eigenvalue, excluding $n$ and, in the bipartite case, $-n$.

\begin{table*}[t]
	\centering
	\small
	\renewcommand{\arraystretch}{1.3}
	\setlength{\tabcolsep}{7pt}
	\setlength{\extrarowheight}{2pt}
	\resizebox{\textwidth}{!}{%
		\begin{tabular}{|>{\centering\arraybackslash}m{0.15\textwidth}|>{\centering\arraybackslash}m{0.39\textwidth}|>{\centering\arraybackslash}m{0.39\textwidth}|}
			\hline
			& \begin{tabular}[c]{@{}c@{}} \textbf{Graphs / Complexes} \\[-5pt] {\vspace{1pt}} (with fixed vertex degree $n$)\end{tabular} & \begin{tabular}[c]{@{}c@{}} \textbf{Local codes} \\[-5pt] \vspace{1pt} ($[n,k,d]$, with $m \times n$ parity-check matrices) \end{tabular} \\
			\hline
			\textbf{Basic requirements}
			&
			\begin{minipage}[c]{\linewidth}
				\vspace{4pt}
				\begin{itemize}[topsep=0pt,partopsep=0pt,parsep=0pt,itemsep=1pt,leftmargin=*]
					\item \textbf{Cayley graphs and their lifts:}
					for well-defined local orderings of edges compatible with double covers and local codes; see Sec.~\ref{sec:DLV_code}.
					
					\item \textbf{Expander family:}
					the construction requires a sufficiently large spectral gap \(n-\lambda\); a suitable sequence of lifts forms an expander family with this property (see Theorem~\ref{thm:Abelian_lift}).
				\end{itemize}
				\vspace{4pt}
			\end{minipage}
			&
			\begin{minipage}[c]{\linewidth}
				\vspace{4pt}
				\begin{itemize}[topsep=0pt,partopsep=0pt,parsep=0pt,itemsep=1pt,leftmargin=*]
					\item \(\boldsymbol{d > \lambda}\): for the Tanner code to have linear distance; see Theorem~\ref{thm:cLDPC}.
					
					\item \textbf{Full-rank local parity-check matrix and dual distance \(\boldsymbol{d^\perp > \lambda}\):}
					for linear distances of the (co)boundary codes, viewed here as a one-dimensional analogue of Poincar\'e duality; see Theorem~\ref{thm:cLDPC}.
					
					\item \(\boldsymbol{m < n/2}\): for the Tanner code to have linear dimension; see Secs.~\ref{sec:HGP_scalar} and~\ref{sec:HGP_lift}.
				\end{itemize}
				\vspace{4pt}
			\end{minipage}
			\\
			\hline
			
			\textbf{Advanced requirements}
			&
			\begin{minipage}[c]{\linewidth}
				\vspace{4pt}
				\begin{itemize}[topsep=0pt,partopsep=0pt,parsep=0pt,itemsep=1pt,leftmargin=*]
					\item \textbf{Abelian lift:}
					for well-defined three and higher-dimensional cubical complexes; see Sec.~\ref{sec:cubical}.
					
					\item \textbf{Double cover:}
					for well-behaved complexes, codewords, and invariants; e.g., Theorems~\ref{thm:cycle}.
					
					\item \textbf{Shift \(l\)-lift with \(l\) given by Artin's conjecture:}
					to build well-structured cup products via cyclic action on lifts; see Theorems~\ref{thm:basis} and~\ref{thm:cup_sheaf_1}.
				\end{itemize}
				\vspace{4pt}
			\end{minipage}
			&
			\begin{minipage}[c]{\linewidth}
				\vspace{4pt}
				\begin{itemize}[topsep=0pt,partopsep=0pt,parsep=0pt,itemsep=1pt,leftmargin=*]
					\item \textbf{Product-expansion over \(\mathbb{F}_q\):}
					for the sheaf code to have (almost) linear distance.
					
					\item \(\boldsymbol{n/2 < m < n}\): for the sheaf code to have linear dimension; see Sec.~\ref{sec:HGP_like}.
					
					\item \textbf{Nontrivial tensor product sheaf:}
					for the existence of cycles defining the invariant form; see Secs.~\ref{sec:cup} and~\ref{sec:cycles}.
				\end{itemize}
				\vspace{4pt}
			\end{minipage}
			\\
			\hline
		\end{tabular}%
	}
	\caption{Summary of important conditions for building qLDPC codes with desired properties. We divide them into requirements on the global graphs or complexes and on the local codes. The basic requirements are general principles that typically underlie the design of classical and quantum expander codes with favorable properties~\cite{SS1996,Zemor2014,Leverrier_2015,PK2021}. The advanced requirements arise from recent developments in asymptotically (almost) good qLDPC codes and qLTCs, e.g.,~\cite{PK2022Good,QuantumTanner2022,DHLV2022,Dinur2024sheaf,KP2025Extendable}, as well as from our results on (co)homological invariants and constant-depth multi-controlled-$Z$ gates. Depending on the setting, the relevant range of $m$ may vary.}
	\label{tab:conditions}
\end{table*}

\section{Logical action on hypergraph product codes}\label{sec:main_HGP}

HGP codes provide the tractable members of the interlaced double sequence in Eq.~\eqref{eq:induct_lift_main}. Their coboundary operators are sums of ordinary tensor products, so the logical basis, cup products, and pairing cycles can all be written explicitly. HGPs do not give good qLDPC parameters, but they exhibit the exact structure that will later be recovered for sheaf codes.

\subsection{Hypergraph product with scalar coefficients}
For clarity, consider the three-dimensional case: $X' = (G_0 \times K_2)^3$, where $G_0$ is an $n$-regular graph whose double cover $G_0 \times K_2$ is connected, with vertex and edge sets:
\begin{align}
	& \{ [v,b]: v \in V_0, b \in \{0,1\} \}, \\
	& \{[(u,0),(v,1)], [(u,1),(v,0)]: (u,v) \in E_0 \}.
\end{align}
It is straightforward to check that this $X'$ is precisely the three-dimensional specialization of Eq.~\eqref{eq:cubical_complex_HGP_main}. As mentioned after Eq.~\eqref{eq:generator_action_main}, the double cover is useful for defining the cubical complex (see also Fig.~\ref{fig:3-cube_main}) and for carrying out the computation. Let $\mathbb{F}_q^{V_i}$ and $\mathbb{F}_q^{E_i}$ denote the vertex and edge spaces of $G_0 \times K_2$ in direction $i=1,2,3$. For simplicity, we use the same graph in all directions, so $\vert V_i \vert \equiv 2 \vert V_0 \vert$ and $\vert E_i \vert \equiv 2 \vert E_0 \vert$. We also use $V_i, E_i$ to denote the numbers for brevity. Let
\begin{align}
	H_i: \mathbb{F}_q^{V_i} \rightarrow \mathbb{F}_q^{E_i}
\end{align}
be the graph coboundary in direction $i$, i.e., $H_i$ maps a formal vertex $[v_i;b_i]$ to its adjacent edges. The cochain complex $C^\bullet(X',\mathbb{F}_q)$ has the following form:
\begin{align} 
	\resizebox{0.45\textwidth}{!}{$\displaystyle 
		\begin{aligned}
			\mathbb{F}_q^{E_1 \times E_2 \times E_3} \xleftarrow{\delta^2}
			& \mathbb{F}_q^{V_1 \times E_2 \times E_3} \oplus \mathbb{F}_q^{E_1 \times V_2 \times E_3} \oplus \mathbb{F}_q^{E_1 \times E_2 \times V_3} \xleftarrow{\delta^1}   \\
			& \mathbb{F}_q^{E_1 \times V_2 \times V_3} \oplus \mathbb{F}_q^{V_1 \times E_2 \times V_3} \oplus \mathbb{F}_q^{V_1 \times V_2 \times E_3} \xleftarrow{\delta^0}
			\mathbb{F}_q^{V_1 \times V_2 \times V_3}. \notag
		\end{aligned}
		$}
\end{align}
Its middle coboundary has the following matrix form:
\begin{equation}\label{eq:HGP_coboundary_main}
	\resizebox{0.45\textwidth}{!}{$\displaystyle 
		\delta'^1 = \begin{pmatrix}
			0 & I_{V_1} \otimes I_{E_2} \otimes H_3 & I_{V_1} \otimes H_2 \otimes I_{E_3} \\
			I_{E_1} \otimes I_{V_2} \otimes H_3 & 0 & H_1 \otimes I_{V_2} \otimes I_{E_3} \\
			I_{E_1} \otimes H_2 \otimes I_{V_3} & H_1 \otimes I_{E_2} \otimes I_{V_3} & 0
		\end{pmatrix}.
		$}
\end{equation}
The coboundaries $\delta'^0$ and $\delta'^2$ are defined analogously to those in Sec.~\ref{sec:SS_code}.

An enlightening example takes $G_0$ to be a cycle with an odd number of vertices. Then the double cover $G_0 \times K_2$ is a cycle with an even number of vertices and the above cochain complex defines a 3D toric code. We work with graphs whose double covers are connected. Let $K_i \coloneq \ker H_i$ and choose a complement $Q_i$ of $\im H_i$ in $\mathbb{F}_q^{E_i}$, so that $Q_i \cong \mathbb{F}_q^{E_i}/\im H_i$, as defined below. By the Künneth formula~\cite{Hatcher2015AT,Gallier2022}, the degree-one cohomology decomposes as
\begin{align}\label{eq:HGP_H1_main}
		H^1(X',\mathbb{F}_q) \cong & (Q_1\otimes K_2\otimes K_3) \oplus (K_1\otimes Q_2\otimes K_3) \notag \\ 
		&\oplus (K_1\otimes K_2\otimes Q_3).
\end{align}
To define $Q_i$, let $k_i = \dim \ker H_i$ and $k_i^T = \dim \ker H_i^T$. Applying elementary column operations to $H_i$, we obtain its reduced column echelon form:
\begin{align}\label{eq:free_variable_main}
	\begin{pmatrix} I_{E_i - k_i^T} & 0 \\ J_i & 0	\end{pmatrix} \text{ with } J_i \in \mathbb{F}_q^{k_i^T \times (E_i - k_i^T)}.
\end{align}
The row indices of $J_i$ are the \emph{free variables}, in contrast to the \emph{pivot variables} in $I_{E_i-k_i^T}$, when solving $x^T H_i=0$. For $H_1^T$, the standard basis vectors $e_1^{i_1}$, $i_1=1,\ldots,k_1^T$, on the free-variable coordinates give a basis of $Q_1$. We also choose any basis $\{\zeta_1^{j_1}\}_{j_1=1}^{k_1}$ of $K_1=\ker H_1$. Considering the full coboundary operator $\delta^1$, the \emph{canonical $X$-type logical representatives} respecting the decomposition in Eq.~\eqref{eq:HGP_H1_main} are
\begin{align}
	L_1^{i_1,j_2,j_3} & = e_1^{i_1} \otimes \zeta_2^{j_2} \otimes \zeta_3^{j_3}, \label{eq:HGP_basis_1_main}
	\\
	L_2^{j_1,i_2,j_3} & = \zeta_1^{j_1} \otimes e_2^{i_2} \otimes \zeta_3^{j_3}, \label{eq:HGP_basis_2_main}
	\\
	L_3^{j_1,j_2,i_3} & = \zeta_1^{j_1} \otimes \zeta_2^{j_2} \otimes e_3^{i_3}. \label{eq:HGP_basis_3_main}
\end{align}
Canonical $Z$-type logical representatives are defined similarly using homology, but they are not used in the cup products considered here. For a connected double cover, every $K_j = \ker H_j$ is spanned by the all-ones vector supported on all vertex coordinates. Hence $k_j \equiv 1$ and a representative in the first group occupies one chosen edge in direction $1$ and all vertices in directions $2$ and $3$. This is an analogue of a 3D toric code logical membrane aligned with one spatial direction (see also Fig.~\ref{fig:logical_cup}). 

The canonical basis not only spans cohomology but also makes the cup products and logical $\CCZ$ gates completely explicit.
Take three copies of this three-dimensional cochain complex, with each code defined in degree $1$. Write the free-variable edges as $[v_{i_1};a_{i_1}]$, $[v_{i_2};a_{i_2}]$ and $[v_{i_3};a_{i_3}]$. Taking one representative from each directional group for each copy of the code, the cubical cup product rule selects exactly one $3$-cube ($j_1,j_2,j_3 \equiv 1$):
\begin{align}\label{eq:HGP_cup_main}
	L_1^{i_1}\smile_c L_2^{i_2}\smile_c L_3^{i_3}
	=[(v_{i_1},v_{i_2},v_{i_3});a_{i_1},a_{i_2},a_{i_3}].
\end{align}
The coefficient is a unit because all transverse kernel vectors are all-ones vectors. No global manifold orientation is used, and the calculation follows solely from the combinatorics of the cubical complex. 
Pictorially, the cup product in Eq.~\eqref{eq:HGP_cup_main} is supported on the single $3$-cube shown in Fig.~\ref{fig:logical_cup}.
\begin{figure}[t]
	\centering
	\includegraphics[width=0.45\textwidth]{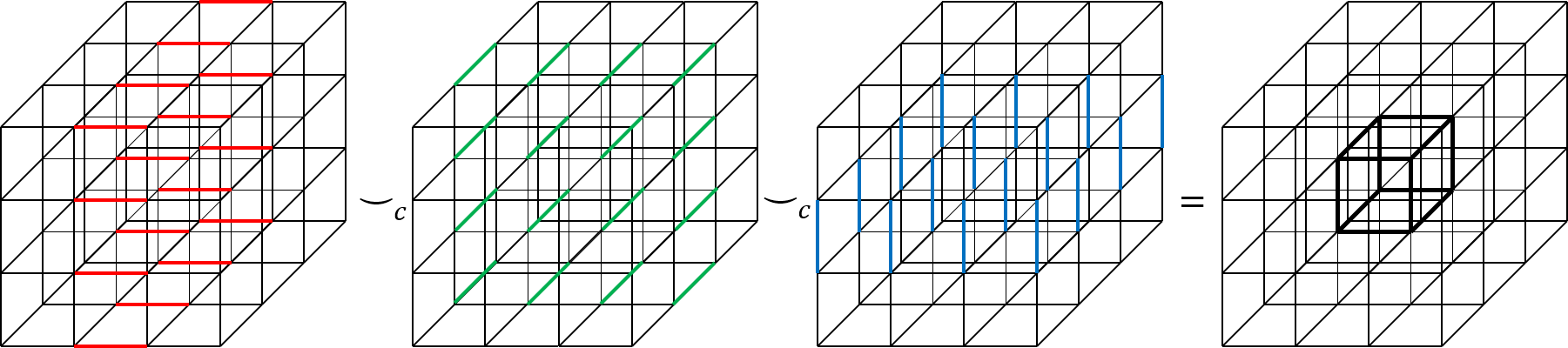}	
	\caption{For the ordered tuple $(L_1^{i_1},L_2^{i_2},L_3^{i_3})$, the cup product is supported on the single $3$-cube determined by the three 1-cubes $[(v_{i_1},v_{i_2},v_{i_3});a_{i_1},0,0]$, $[a_{i_1}\cdot(v_{i_1},v_{i_2},v_{i_3});1,a_{i_2},0]$, and $[a_{i_1}a_{i_2}\cdot(v_{i_1},v_{i_2},v_{i_3});1,1,a_{i_3}]$.}
	\label{fig:logical_cup}
\end{figure} 
One can see that Eq.~\eqref{eq:HGP_cup_main} resembles Eq.~\eqref{eq:sheaf_cup_thm_main}. Eq.~\eqref{eq:HGP_cup_main} is indeed the specialization of Eq.~\eqref{eq:sheaf_cup_thm_main} to a trivial lift and constant sheaves. The rigorous definition of the cup product begins with its conventional definition on a simplicial complex~\cite{Hatcher2015AT,Gallier2022}. Subdivision and pullback extend it to cubical complexes~\cite{Lin2024transversal,Li2025Poincare}. The details, together with explicit two- and three-dimensional computations, are given in Secs.~\ref{sec:cup} and~\ref{sec:example}.

The dual cycle $\xi$ used to form Eq.~\eqref{eq:invariant_form_main} in this HGP case is equally explicit. Let $\eta_{i_1} \in \ker H_1^T$ be the vector whose free-variable coordinate at $i_1$ equals $1$ and whose remaining entries occupy pivot coordinates. Then by definition,
\begin{align}
	\xi_{i_1,i_2,i_3}
	= \eta_{i_1} \otimes \eta_{i_2} \otimes \eta_{i_3} \in \ker\partial_3'.
\end{align}
By Eq.~\eqref{eq:HGP_cup_main},
\begin{align}\label{eq:HGP_addressable_main}
	T_{\xi_{i_1',i_2',i_3'}}(L_1^{i_1},L_2^{i_2},L_3^{i_3}) = \delta_{i_1,i_1'} \delta_{i_2,i_2'} \delta_{i_3,i_3'}.
\end{align}
This proves the addressable part of Theorem~\ref{thm:cup_HGP_main}. Summing the diagonal cycles $\xi \coloneq \sum_i \eta_i \otimes \eta_i \otimes \eta_i$ gives the parallel case $T_{\xi}(L_1^{i_1},L_2^{i_2},L_3^{i_3}) = \delta_{i_1,i_2,i_3}$.

\subsection{Limitation of HGPs with scalar coefficients}

When $n=2$, the graph is a cycle and the construction reduces to the $t$-dimensional toric code: $k=t$, the subrank is one, and the cup product represents its familiar orientation class. When $n\geq3$ and the double covers form a high-girth expander family with girth $\Theta(\log N)$, the scalar HGP code has
\begin{align}
    [\![N,k=\Theta(N^{1/t}),d=\Theta(\log N)]\!],
\end{align}
while the invariant has subrank $k/t$. The general HGP code parameters follow from~\cite{Zemor2014,Zeng_2019}. In the scalar case, the distance equals the girth of the graph factor. The high-girth assumption gives the logarithmic lower bound~\cite{Margulis1988,LPS1988}, while the Moore bound gives the matching logarithmic upper bound~\cite{Alon2002MooreBound}. The logical action is explicit and highly addressable, but the distance is only polylogarithmic. Placing qudits in a higher cochain degree can increase $k$, but it lowers the maximal degree of a cup-product invariant and does not repair the distance. These limitations are summarized precisely in Supplementary Remark~\ref{remark:HGP_1}.

The scalar case therefore supplies the right logical geometry but not the desired code parameters. Equipping the construction with nontrivial sheaves yields quantum expander codes with $d=\Theta(\sqrt N)$~\cite{Zemor2014,Leverrier2023Decoder}, but these codes remain within the HGP framework. Nontrivial sheaves and large lifts should work together to magnify the code dimension and distance to the good parameters regime. 

\section{Logical action on sheaf codes}\label{sec:main_sheaf}

Let $X$ be an Abelian lift of a Cartesian product $X'$, and let $C^\bullet(X,\F)$ be the resulting sheaf complex. With current techniques for constructing cubical complexes, non-Abelian lifts are viable only when $t=2$, whereas we are interested in $t\geq3$.  With a global lift coordinate in Eq.~\eqref{eq:cubical_complex_main}, the simple tensor product formalism on HGP ceases to hold.

\subsection{Local classical codes generating cellular sheaves}\label{sec:local_codes_main}

Before proceeding, we explain how to equip the cochain complex with a nontrivial cellular sheaf. Start with a general (lifted) $t$-dimensional cubical complex $X$. To define $\mathcal{F}$, we need to prepare $t$ local classical codes with local parity-check matrices $h_1 \in \mathbb{F}_q^{m_1 \times n}, \ldots,h_t \in \mathbb{F}_q^{m_t \times n}$. Recall that $n$ is the vertex degree of the base graph $G_0$ in $X$ and $m_1, \ldots,m_t \leq n$. It is also important to assume that each $h_i$ has full rank. Otherwise, global codewords of constant weight can occur.

We assign each $p$-cube $\sigma = [g; (a_i)_{i \in I}, (b_j)_{j \in J}]$ a vector space $\mathcal{F}(\sigma) = \bigotimes_{j \in J} \mathbb{F}_q^{m_j}$. Since $\sigma$ has $\vert J \vert = t - p$ binary bits, we can think of the $(t-p)$-fold tensor product space attached to each of these bits. Note that $\mathcal{F}(\pi) = \mathbb{F}_q$ is a scalar field for a $t$-cube $\pi$. To define the map $\mathcal{F}_{\sigma,\tau}$, we use the three-dimensional case shown in Fig.~\ref{fig:3-cube_main}: 
\begin{enumerate}
	\item For $\theta = [a_3 \cdot g;0,0,1]$ and $\sigma = [g; 0,0,a_3]$, 
	\begin{align}
		\begin{aligned}
			\mathcal{F}_{\theta,\sigma} = & I_{m_1} \otimes I_{m_2} \otimes h_3^T(a_3, \text{-}) \\
			& \in \mathbb{F}_q^{m_1 \times m_1} \otimes \mathbb{F}_q^{m_2 \times m_2} \otimes \mathbb{F}_q^{1 \times m_3},
		\end{aligned}
	\end{align} 
    where the $n$ elements $a_3 \in A_3$ label the columns of $h_3$, and $h_3^T(a_3,\text{-})$ is the corresponding row of $h_3^T$, equivalently the transpose of the column of $h_3$ labeled by $a_3$. Since every base graph vertex has $n$ incident edges, the prescribed local ordering identifies these edges with the $n$ columns of $h_3$.
	
	\item For $\sigma = [g;0,0,a_3]$ and $\tau = [g; a_1,0,a_3]$, 
	\begin{align}
		\mathcal{F}_{\sigma,\tau} =  h_1^T(a_1, \text{-}) \otimes I_{m_2} \in \mathbb{F}_q^{1 \times m_1} \otimes \mathbb{F}_q^{m_2 \times m_2}.
	\end{align} 
	
	\item For $\tau = [g;a_1,0,a_3]$ and $\pi = [g; a_1,a_2,a_3]$, 
	\begin{align}
		\mathcal{F}_{\tau,\pi} = h_2^T(a_2, \text{-}) \in \mathbb{F}_q^{1 \times m_2}.
	\end{align} 
\end{enumerate}
By Definition~\ref{def:presheaf}, $\mathcal{F}_{\sigma,\pi} = \mathcal{F}_{\tau,\pi} \circ \mathcal{F}_{\sigma,\tau}$. In the above example, $\sigma = [g;0,0,a_3]$ and $\tau = [g;a_1,0,a_3]$. Then
\begin{align}
	\begin{aligned}
		\F_{\sigma,\pi} & = \big( h_2^T(a_2, \text{-} ) \big) \circ \big( h_1^T(a_1, \text{-} ) \otimes I_{m_2} \big) \\
		& = h_1^T(a_1, \text{-} ) \otimes h_2^T(a_2, \text{-} ).
	\end{aligned}
\end{align}
One can take $\tau' = [g;0,a_2,a_3] \succ \sigma$. Then $\mathcal{F}_{\tau',\pi} \circ \mathcal{F}_{\sigma,\tau'}$ still equals $\mathcal{F}_{\sigma,\pi}$, so the compatibility condition in Eq.~\eqref{eq:compatible_main} is satisfied. Further subtleties are discussed in Secs.~\ref{sec:SS_code}, \ref{sec:DLV_code}, and~\ref{sec:example}. At the matrix level, $C^\bullet(X,\mathbb{F}_q)$ has scalar coboundary operators recording only the incidence of lower-dimensional cubes in higher-dimensional ones. The sheaf complex $C^\bullet(X,\F)$ replaces each nonzero scalar entry by the appropriate block $\F_{\sigma,\tau}$ (cf. Eq.~\eqref{eq:coboundary_operator_main}). Once the local parity-check matrices are chosen, every such subblock is explicit.

\subsection{Polarized logical representatives}\label{sec:polarize_main}

Lifting further enlarges the coboundary operator. Starting from the HGP coboundary in Eq.~\eqref{eq:HGP_coboundary_main}, the local codes replace each nonzero scalar entry by the corresponding block $\F'_{\sigma',\tau'}$. Recall that Eq.~\eqref{eq:projection_main} defines the projection $X\to X'$ and that $\F$ is the pullback of $\F'$. Under the lift, each block $\F'_{\sigma',\tau'}$ becomes a block-permutation matrix whose nonzero blocks are copies of $\F'_{\sigma',\tau'}$, arranged according to the matching between the corresponding lift fibers. Explicitly, for 1D coboundary operators, we have the following sequence of replacements:
\begin{align}
    H_i \in \mathbb{F}_q^{E_i \times V_i}
    \;\Longrightarrow\;
    H_i \in \mathbb{F}_q^{E_i \times m_i V_i}
    \;\Longrightarrow\;
    \hat{H}_i \in
    \mathbb{F}_q^{lE_i \times lm_i V_i}.
\end{align}
Here, for simplicity, the same symbol $H_i$ denotes the scalar and sheaf coboundaries before lifting.

Let $R=\mathbb{F}_q[\HH]$ be the \emph{group algebra} of the Abelian lifting group. After identifying each lift fiber with the regular $R$-module $R$, a permutation block associated with $h\in\HH$ is represented by multiplication by $h$. Thus, $\hat H_i$ can be regarded as a matrix in $R^{E_i\times m_iV_i}$, and its directional blocks regain an HGP-like tensor form over $R$-modules. For odd $|\HH|$ in characteristic $2$, Maschke's theorem makes $R$ semisimple~\cite{Goodman2009}, and the K\"unneth formula~\cite{Hatcher2015AT,Gallier2022} holds with no torsion term. This algebraic viewpoint restores a tensor product form analogous to the HGP case, now over $R$-modules. When $t = 3$,
\begin{align}
	\resizebox{0.48\textwidth}{!}{$\displaystyle 
		\delta^1 = \begin{pmatrix} 
			0 & I_{V_1} \otimes_R I_{E_2} \otimes_R \hat{H}_3 & I_{V_1} \otimes_R \hat{H}_2 \otimes_R I_{E_3} \\
			I_{E_1} \otimes_R I_{V_2} \otimes_R \hat{H}_3 & 0 & \hat{H}_1 \otimes_R I_{V_2} \otimes_R I_{E_3} \\
			I_{E_1} \otimes_R \hat{H}_2 \otimes_R I_{V_3} & \hat{H}_1 \otimes_R I_{E_2} \otimes_R I_{V_3} & 0 
		\end{pmatrix}.
		$}
\end{align}
Fix an edge type $[v_1;a_1]$ in the first direction and retain only cochains supported on that type. The following \emph{restricted coboundary} tests the remaining directions:
\begin{align}\label{eq:delta_res_main}
	\delta_{\mathrm{res}}^1
	= I_{E_1} \otimes_R
	\begin{pmatrix}
		0\\
		I_{m_2 V_2} \otimes_R \hat H_3\\
		\hat{H}_2 \otimes_R I_{m_3 V_3}
	\end{pmatrix}.
\end{align}
Elements of its kernel are supported on a fixed edge direction and vary over the fiber and the transverse vertices. We call them \emph{polarized}:
\begin{align}\label{eq:sheaf_basis_main}
	[v_1;a_1] \otimes_R \sum_{h,v_2,v_3,b_2,b_3} [(h,v_2,v_3);b_2,b_3].
\end{align}
The local coefficient vectors on the second and third binary bits (see the definition in Sec.~\ref{sec:local_codes_main}) are suppressed here but are retained in the full definition. The explicit tensor product in Eq.~\eqref{eq:sheaf_basis_main} can also be written over $\mathbb{F}_q$ because the local coefficient of $[v_1;a_1]$ is $1 \in \mathbb{F}_q$.

Similarly, polarized logical representatives exist in the other two directions when $t=3$. One important caveat is that, because of the more complicated module structure, the column-reduction and free-variable methods used in the HGP case do not apply directly. Instead, We prove in Appendix Section \ref{sec:module} and \ref{sec:edge_choice} the existence of many independent cohomology classes with polarized representatives. The details is deferred to the Appendix because of its intricate structure: it combines the special cyclic group $C_l$ lifts and the Artin's conjecture. We emphasize that this is a substantive extension of the canonical HGP basis, and the outcome is still concrete enough to evaluate cup products explicitly.

\subsection{Induced vectors, cyclic ideals, and Artin's conjecture}

We now specialize to a cyclic lift via $\HH=C_l$, which is sufficient to produce favorable HDXs in the construction of (almost) good qLDPC codes \cite{Agarwal2016,Jeronimo2021,Dinur2024sheaf}. Let $x$ be any polarized logical representative as in Eq.~\eqref{eq:sheaf_basis_main}. For any fixed base 1-cube $\sigma'=[(v_1,v_2,v_3);(a_1,b_2,b_3)]\in X'$, the restriction of $x$ to the fiber of $\sigma'$ is naturally an element of $\mathbb{F}_q[C_l]\otimes_{\mathbb{F}_q}\mathbb{F}_q^{m_2\times m_3}$. For any fixed $a_2,a_3$, we define the \emph{induced vector} in $\mathbb{F}_q[C_l]$ by the evaluation:
\begin{align}\label{eq:induced_vector_main}
	& \mathrm{v}_{\sigma'}^{a_2,a_3}(h) \coloneq \\
	&  h_2^T(a_2, \text{-}) \otimes h_3^T(a_3, \text{-}) \Big( x ( [(h,v_1,v_2,v_3);(a_1,b_2,b_3)] ) \Big) \notag
\end{align}
for any $h \in C_l$. By basic algebra,
\begin{align}
	R = \mathbb{F}_q[C_l] \cong \mathbb{F}_q[\x]/(1+\x^l),
\end{align}
which is the quotient algebra of $\mathbb{F}_q$-polynomials modulo the relation $\x^l=1$. Given a collection of polarized logicals, if their induced vectors generate the whole group algebra $R$, i.e., if they are able to produce
\begin{align}
	1,\x,\ldots,\x^{l-1} \in R,
\end{align}
then, at the cochain level, there are polarized logical representatives whose supports on the fiber of $\sigma'$ (after being acted on by $ h_2^T(a_2, \text{-}) \otimes h_3^T(a_3, \text{-})$ ) are localized to individual lifts. 

Taking one localized family in every coordinate direction, a single top-degree cube is formed exactly when the polarized logical representatives intersect at the same localized lift:
\begin{align}\label{eq:sheaf_diagonal_main}
	L_1^{i_1}\smile_c\cdots\smile_cL_t^{i_t}=0
	\quad\text{unless}\quad i_1=\cdots=i_t.
\end{align}
This is the structural content of Eq.~\eqref{eq:sheaf_cup_thm_main}. The intuition follows from Fig.~\ref{fig:logical_cup}, although we have to deal with local codes and lifts. The rigorous computations are given in Secs.~\ref{sec:example} and~\ref{sec:cup_sheaf}.

The criterion for induced vectors to generate the whole algebra is that they avoid every maximal ideal in $R$. For general $l$, the factorization of $1+\x^l$ may contain many factors, and a large family of induced vectors can still generate a proper cyclic-code ideal. Artin's conjecture \cite{Artin1966,Hooley1967} supplies infinitely many primes $l$ for which the multiplicative order of $q$ modulo $l$ is $l-1$. For these primes, $1+\x^l = (1 + \x) \cdot (1 + \x + \cdots + \x^{l-1})$ has exactly two irreducible factors over $\mathbb{F}_q$, leaving only the two maximal ideals to analyze. We provide a thorough counting argument in Sec.~\ref{sec:module} to exclude both. Therefore, we obtain $l$ independent cup products over fiber-localized cubes with $l = \exp(O( | V_0 | )) = \Theta(N/\mathrm{polylog}\,N)$.

This is the only step that uses Artin's conjecture, and hence GRH in Theorem \ref{thm:cup_sheaf_main}. Without using these primes $l$, the construction still produces polarized representatives, but the induced vectors may span only proper ideals, turning the question into a difficult multiplicative problem in classical cyclic-code theory. More details are presented in Sec.~\ref{sec:cup_sheaf}.

\subsection{Local codes alignment and symmetrization}

To form the cup products, we use multiple sheaf codes on the same $X$, with distinct sheaves $\F_1,\ldots,\F_t$. In conventional examples like toric codes, the multiple code blocks are always identical. As in Eqs.~\eqref{eq:HGP_basis_1_main}--\eqref{eq:HGP_basis_3_main}, each code block contributes one logical representative in one ``spatial direction'' to generate a nontrivial cup product. Since the local codes of each sheaf code are randomly sampled in Ref.~\cite{KP2025Extendable}, we have to find a way to ``symmetrize'' them among the code blocks in order to generalize the intuition to the polarized logical representatives. Interestingly, we can solve the problem by extremal graph theory. Specifically, we treat classical codes as points in the corresponding Grassmannian. As claimed in Ref.~\cite{KP2025Extendable}, for sufficiently large $q$, one can find desired local codes with high probability. A high probability of finding a desired tuple of points indicates that many points are joined to form a certain hypergraph. Guided by the hypergraph Tur\'{a}n theory~\cite{Turan1941,Erdos1946,Razborov2010,Keevash2011}, this hypergraph must contain some suitable subhypergraph whose associated local codes will be used to generate the sheaves. We postpone the notation-heavy details to Secs.~\ref{sec:module} and~\ref{sec:cup_sheaf}. Combining these ingredients gives Theorem~\ref{thm:cup_sheaf_main}; each step is treated algebraically and accompanied by explicit examples in the Appendix.

Compared with the HGP case, a nontrivial logical invariant still requires
\begin{align}
	\xi\in C_t(X,\F_1\otimes\cdots\otimes\F_t), \qquad \partial_t \xi = 0,
\end{align}
with nonzero value on the fiber cubes. Sec.~\ref{sec:cycles} gives an explicit cycle $\xi$ using sheaves generated by Reed--Solomon codes. To obtain (almost) good qLDPC codes, however, the local codes must also be product-expanding~\cite{PK2023RobustlyTestable,DHLV2022,KP2025Extendable}. 
Simultaneously establishing product expansion of the local codes and the existence of a nontrivial cycle $\xi$ in the associated tensor product sheaf would yield nontrivial multi-controlled-$Z$ gates on (almost) good qLDPC codes.

\section{Code induction scheme}\label{sec:main_induction}

Having presented the HGP and sheaf-code calculations, we now explain the induction scheme based on the covering relation in Eq.~\eqref{eq:projection_main}. In the upper row of Eq.~\eqref{eq:induct_lift_main}, each Cartesian product $X_{i+1}'$ covers the preceding one $X_i'$ because every graph factor is replaced by a lift. Vertically, the high-dimensional expander $X_i$ covers the HGP complex $X_i'$. 
To explain the mechanism, it suffices to consider a single covering $X\to X'$; the same argument is then applied to every horizontal and vertical arrow in Eq.~\eqref{eq:induct_lift_main}.

\subsection{Canonical extension of logical representatives}

The covering space formalism is less stringent than the case in Sec.~\ref{sec:main_sheaf}: we only ask $\HH$ to have an odd number of elements and there are no additional assumptions or conjectures. We work at degree $1$ to simplify the presentation. Let $x'\in C^1(X',\F')$ be a cocycle representing a nontrivial logical $X$ operator. Every $1$-cube $\sigma'\in X'$ has $|\HH|$ lifts $\sigma\in X$. Its \emph{canonical extension} $x\in C^1(X,\F)$ is obtained simply by assigning the same value to every lift:
\begin{align}\label{eq:canonical_extension_main}
	x(\sigma)=x'(\sigma')
	\qquad\text{whenever }\sigma\mapsto\sigma'.
\end{align}
The coefficient spaces on $\sigma$ and $\sigma'$ are naturally identified because the sheaf on $X$ is pulled back from $X'$, as discussed after Eq.~\eqref{eq:projection_main}.

As discussed at the beginning of Sec.~\ref{sec:polarize_main}, the coboundary operators are lifted via permutation subblocks. Consequently, the canonical extension preserves the cocycle condition:
\begin{align}
	\delta' x'=0 \quad\Longrightarrow\quad \delta x=0.
\end{align}
If $x'$ is a nontrivial logical representative, i.e., it is a cocycle but not a coboundary, then the extended cocycle is also nontrivial when $\vert \HH \vert$ is odd. Suppose, to the contrary, that $x=\delta y$ for some $y\in C^0(X,\F)$. For every vertex $\nu'\in X'(0)$, sum the values of $y$ over all vertices $\nu$ above $\nu'$ and denote the resulting $0$-cochain on $X'$ by $y'$. Summing the relation $x=\delta y$ over the fiber of a $1$-cube $\sigma'$ gives
\begin{align}
	\begin{aligned}
		(\delta' y')(\sigma')
		&=\sum_{\sigma\mapsto\sigma'}(\delta y)(\sigma)
		=\sum_{\sigma\mapsto\sigma'}x(\sigma) \notag\\
		&=|\HH|\,x'(\sigma')
		=x'(\sigma').
	\end{aligned}
\end{align}
In the last equality, we use that $\mathbb{F}_q$ has characteristic $2$ and $|\HH|$ is odd. The displayed identity implies $x'=\delta' y'$, contradicting the assumption that $x'$ is not a coboundary. The same argument applies to cycles on chain complexes by transposing the local incidence maps as stated in Sec.~\ref{sec:induct_boundary}.

\subsection{Preservation of the invariant form}

We next explain why the logical gate itself is preserved. For clarity, we just consider three code blocks defined at degree-$1$ of the cochain complexes with cocycles
\begin{align}
	x_i' \in C^1(X',\F_i'), \ i = 1,2,3
\end{align}
and let $x_i$ be the corresponding canonical extensions to $C^1(X,\F_i)$. As in the previous computations, a nonzero cup product of three elementary $1$-cubes is determined entirely by their local arrangement inside a $3$-cube. On the other hand, every lift of that $3$-cube contains a unique compatible lift of the ordered triple of $1$-cubes. Since the local coefficient maps are unchanged under the covering, i.e., $\F_i$ are pullbacks of $\F_i'$, the cup product coefficient on the lifted cube is identical to the coefficient on the base cube. Therefore,
\begin{align}\label{eq:cup_extension_main}
	x_1\smile_c x_2\smile_c x_3
\end{align}
is precisely the canonical extension of $x_1'\smile_c x_2'\smile_c x_3'$.

This local statement is the essential point. Rigorously, it is formalized by the commutative diagrams in Sec.~\ref{sec:induct_invariant}. Similarly, let $\xi'$ be a $3$-cycle on $X'$ such that
\begin{align}
	\left\langle x_1'\smile_c x_2'\smile_c x_3',\ \xi' \right\rangle\neq 0.
\end{align}
As in the cochain case, the canonical extension of $\xi'$ is a $3$-cycle $\xi$ on $X$. By the argument just described, every contribution to the following pairing is repeated once on each sheet of the covering. Since the number $\vert \HH \vert$ of sheets is odd, these repetitions sum to the original value over $\mathbb{F}_q$:
\begin{align}\label{eq:induct_pairing_main}
	\left\langle x_1\smile_c x_2\smile_c x_3, \ \xi \right\rangle
	= \left\langle x_1'\smile_c x_2'\smile_c x_3', \ \xi' \right\rangle \neq 0.
\end{align}
Therefore, the same logical $\CCZ$ action is present on the covering code. The argument is identical for general degrees on a $t$-dimensional complex and hence for logical $\mathrm{C}^{\rho-1}Z$ gates. More generally, this mechanism applies to other (co)homological invariants whenever their defining operations obey the same local compatibility with coverings.

It is also important to mention that the induction scheme transports the logical operation already known on the seed code, but it does not characterize the additional degrees of freedom created by a large lift, e.g., $\widetilde{\Theta}(N)$ independent cup products on sheaf codes. The results developed in the preceding sections address precisely this complementary problem.

\section{Discussion and outlook}\label{sec:discussion}

In this work, we established a framework for analyzing and constructing cohomological invariants and fault-tolerant logical gates (in particular, non-Clifford gates) for qLDPC codes. While our mathematical results have demonstrated the theoretical viability of this approach, many interesting directions remain open for future investigation, ranging from numerical and algorithmic considerations to the exploration of richer algebraic structures.

A central ingredient of our inductive construction is the availability of a suitable sequence of lifts. From a theoretical perspective, known results on expander graphs and lifted codes provide a blueprint for building such sequences with the necessary spectral properties. Efficient algorithms exist for sequential 2-lifts~\cite{Bilu2004,Mohanty2022}, and extensions to 3-lifts are plausible. 
While these algorithmic issues matter for explicit constructions, they are largely orthogonal to the main theoretical contributions of this work. A natural next step is to implement and benchmark the lifting procedures so as to demonstrate the feasibility of constructing explicit code families that admit our logical gates.

Next, according to Ref.~\cite{KP2025Extendable}, two-way product-expanding codes occur with probability $1 - 2n^t 2^{n^t - \log q + 1}$. This requires $q > 2n^t 2^{n^t + 1}$. For $t=3$ and a large fixed $n$, searching for these codes by brute force is extremely difficult as the required field size satisfies $q>2n^3 2^{n^3+1}$. Consequently, the conditions on the local codes and the associated tensor product sheaves may appear stringent at first glance. In practice, however, we expect that simpler or even well-known code constructions will already suffice in many regimes of interest. For instance, as we show in the examples in Sec.~\ref{sec:cycles},
simple HGP codes equipped with Reed--Solomon local codes already support nontrivial cup product gates, although they may not lead to asymptotically good qLDPC codes. Moreover, in the finite regime where the code distance scaling may not be important, it is also worthwhile to consider plain HGP codes without any local codes; for such codes, addressable multi-controlled-$Z$ gates have been explicitly formulated in Theorem~\ref{thm:cup_HGP_main}.
This suggests that our framework is flexible enough to accommodate a remarkably wide spectrum of code constructions, from elementary examples to asymptotically good families. More importantly, given the well-established methods for constructing expanders and HGP codes, we expect that numerical experiments are feasible and fruitful for exploring and testing these codes in the finite regime.

The cup product pairing in Eq.~\eqref{eq:invariant_form_main} is just one instance of a rich family of (co)homological invariants, and there are many other potential candidates for generating non-Clifford logical operations. By changing the ways and orders of taking cup and cap products, several distinct invariant forms have been obtained in Ref.~\cite{Li2025Poincare}. We believe that exploring these alternative invariants is a promising direction. Such generalizations may relax the constraints on the local codes, thereby leading to more streamlined constructions of non-Clifford gates. More broadly, the interplay between algebraic topology and quantum error correction that we have advanced in this work opens up diverse new possibilities for fault-tolerant logical operations. We anticipate that this program will uncover a broader landscape of gates and codes and ultimately combine fault-tolerant non-Clifford gates with asymptotically good qLDPC families. Such advances could lead to more efficient fault-tolerant quantum computation. For example, these codes can achieve constant space and almost-logarithmic time overhead~\cite{Nguyen2025FT}. They may also provide new approaches to the quantum PCP conjecture through fault tolerance~\cite{AharonovAradVidick2013_qPCPSurvey,anshu2024circuit}.


\section*{Acknowledgments} This work is supported in part by NSFC under Grant No.~12475023, Dushi Program, and a startup funding from YMSC.

\bibliography{reference.bib}

\widetext
\appendix
\appendix

\startcontents[sections]

\titlecontents{section}[0pt]
  {\vspace{5mm}}
  {\thecontentslabel\hspace{1em}}
  {}
  {\titlerule*[1pc]{.}\contentspage}

\titlecontents{subsection}[1.5em]
  {}
  {\thecontentslabel\hspace{1em}}
  {}
  {\titlerule*[1pc]{.}\contentspage}

\printcontents[sections]{}{1}{%
  \par\addvspace{0.4cm plus 1ex minus .2ex}%
  {\centering\normalfont\small\bfseries
   APPENDIX CONTENTS\par}
  \vspace{0.1cm}%
}

\section{Preliminaries}\label{sec:pre}

\subsection{CSS codes, invariant polynomials and diagonal logical gates}\label{sec:CSS}

We assume that readers are familiar with the basic notions of stabilizer codes and logical gates (see~\cite{NielsenChuang2010} for details), so this subsection reviews only the concepts used later.

\begin{definition}\label{def:cochain}
	A \emph{cochain complex} is a sequence of vector spaces $C^i$ and \emph{coboundary operators} $\delta^i$ between them, arranged as
	\begin{equation}
		\begin{tikzcd}
			\cdots \arrow[r, "\delta^{i-1}"] &
			C^i \arrow[r, "\delta^{i}"] &
			C^{i+1} \arrow[r, "\delta^{i+1}"] &
			C^{i+2} \arrow[r, "\delta^{i+2}"] &
			\cdots
		\end{tikzcd}
	\end{equation}
		such that $\delta^{i+1}\delta^i=0$ for every $i$. This relation is called the \emph{coboundary condition}. We always assume that $C^i$ is a finite-dimensional vector space over a finite field $\mathbb{F}_q$. Dualizing, we obtain the chain complex
	\begin{equation}
		\begin{tikzcd}
			\cdots & C_i \arrow[l, "\partial_{i} "'] &
			C_{i+1} \arrow[l, "\partial_{i+1} "'] &
			C_{i+2} \arrow[l, "\partial_{i+2} "'] &
			\cdots \arrow[l, "\partial_{i+3} "']
		\end{tikzcd}
	\end{equation} 
		Here $C_i \coloneq (C^i)^\ast$, and $\partial_i \coloneq (\delta^{i-1})^\ast$ is the \emph{boundary operator}. After bases are fixed, $\partial_i$ is represented by $(\delta^{i-1})^T$, and we use the induced coordinate identification $C_i \cong C^i$. Elements of $\ker \delta^i$ are the $i$-th \emph{cocycles} and elements of $\im \delta^{i-1}$ are the $i$-th \emph{coboundaries}. Cycles and boundaries in chain complexes are defined analogously. We also abbreviate a (co)chain complex as $(C_\bullet,\partial_\bullet)$ ($(C^\bullet,\delta^\bullet)$). 
\end{definition}

\begin{definition}\label{def:CSS}
	A \emph{Calderbank--Shor--Steane (CSS) code} is a stabilizer code given by a 2-cochain complex over a finite field $\mathbb{F}_q$: 
	\begin{align}
		C^{i-1} \xrightarrow{\delta^{i-1}} C^i \xrightarrow{\delta^i} C^{i+1}
	\end{align}
		where $\delta^{i} = H_Z$ and $\delta^{i-1} = H_X^T$ define the stabilizer generators. The $X$-type and $Z$-type logical operators are represented by elements of the cohomology $H^i(C) = \ker \delta^i / \im \delta^{i-1}$ and the homology $H_i(C) = \ker \partial_i / \im \partial_{i+1}$, respectively. The \emph{code dimension} $k$, which represents the number of encoded logical qudits, is given by
	\begin{align}
		k = \dim H^i(C) = \dim C^i - \rk \delta^i - \rk \delta^{i-1}
		= \dim C_i - \rk \partial_i - \rk \partial_{i+1} = \dim H_i(C).
	\end{align}
	The \emph{code distance} is given by $d = \min\{d_X, d_Z\}$, where
	\begin{align}
		d_X = \min_{x \in \ker \delta^i \setminus \im \delta^{i-1}} \vert x \vert, \quad 
		d_Z = \min_{z \in \ker \partial_i \setminus \im \partial_{i+1}} \vert z \vert.
	\end{align}
	Here $d_X$ and $d_Z$ are the cosystolic and systolic distances, respectively. 
\end{definition}

A CSS code can always be defined by three consecutive terms of a cochain complex $C^\bullet$. The surrounding terms are also related to important code properties, including soundness and non-Clifford logical gates~\cite{Bombin_2007,Bombin_2013,BK2013,Kubica2015,Kubica2015Unfolding,PK2022Good,DHLV2022,Dinur2024sheaf,First2022,First2024}. In particular, constructing logical multi-controlled-$Z$ gates via cup products requires a sufficiently long cochain complex~\cite{10.21468/SciPostPhys.14.4.065,Chen_2023,Breuckmann2024Cups,Wang_2024,Lin2024transversal,Golowich_Lin2024}. We provide further details in Sec.~\ref{sec:cup} to make the presentation self-contained.

We always work over a finite field $\mathbb{F}_q$ of characteristic $2$. This assumption simplifies the sign conventions in the definition of the (co)boundary operators and many subsequent calculations. In this case, $\mathbb{F}_q\cong\mathbb{F}_{2^s}$ for some positive integer $s$ (see Proposition~\ref{prop:finite_field} below), so the corresponding stabilizer code on qudits with local dimension $q = 2^s$ can be naturally regarded as a code on registers of $s$ qubits. For theoretical purposes, particularly to achieve favorable code parameters, $s$ often needs to be chosen as a sufficiently large constant so that suitable local codes exist (see Secs.~\ref{sec:SS_code} and~\ref{sec:DLV_code})~\cite{kalachev2023ProductExpansion,PK2023RobustlyTestable,KP2025Extendable}. However, one may typically take $s=1$ and recover the ordinary qubit setting.

As a familiar example over $\mathbb F_2$, let $\omega_\rho=e^{2\pi i/2^\rho}\in\mathbb C$ and let $\ket{x}=\ket{x_1,\ldots,x_n}$ denote a computational basis state of an $n$-qubit system. When $\rho=1$, tensor products of Pauli $Z$ gates in the first level of the Clifford hierarchy can be expressed as
\begin{align}
	\sum_{x_i \in \{0,1\}} (-1)^{\sum_i a_i x_i}\ket{x}\bra{x}
	= \sum_{x\in\mathbb{F}_2^n} \omega_1^{F_1(x)}\ket{x}\bra{x},
\end{align}
where $F_1(x)=\sum_i a_i x_i$ is an $\mathbb{F}_2$-valued polynomial.
The values of $F_1$ are naturally identified with their representatives in $\{0,1\}\subset\mathbb{Z}$ when they appear in an exponent, so $\omega_1^{F_1(x)}$ is well defined. When $\rho=2$, products of $\CZ$ gates are
described by
\begin{align}
	\sum_{x_i \in \{0,1\}}(-1)^{\sum_{i<j} a_{ij}x_i x_j} \ket{x}\bra{x}
	= \sum_{x\in\mathbb{F}_2^n}\bigl(\omega_2^2\bigr)^{F_2(x)}\ket{x}\bra{x},
\end{align}
where $F_2(x) = \sum_{i<j} a_{ij} x_i x_j$. Similarly, products of $\CCZ$
gates arise when $\rho=3$:
\begin{align}
	\sum_{x_i \in \{0,1\}}(-1)^{\sum_{i<j<k}a_{ijk}x_ix_jx_k}\ket{x}\bra{x}
	= \sum_{x\in\mathbb{F}_2^n}\bigl(\omega_3^4\bigr)^{F_3(x)}\ket{x}\bra{x},
\end{align}
where $F_3(x) = \sum_{i<j<k} a_{ijk} x_i x_j x_k$.These examples illustrate the correspondence between $(\rho-1)$-controlled-$Z$ phases and degree-$\rho$ Boolean monomials.

For general qudits with the local dimension equal to any prime, we have a complete polynomial description of the diagonal Clifford hierarchy~\cite{PhysRevA.95.012329,Anderson2016}, and similar results extend to prime powers~\cite{He2025addressable} and even arbitrary dimensions~\cite{Feng2024} for certain gates. Our previous discussion concerns the simplest qubit case, which is used to motivate the idea of multi-controlled-$Z$ gates on CSS codes.

To this end, we also review a few simple properties about finite fields. 

\begin{proposition}\label{prop:finite_field}
	Let $\mathbb{F}_q$ be any finite field. Then the following facts hold. 
	\begin{enumerate}
		\item For any $a \in \mathbb{F}_q$, $q \cdot a = 0$ and $a^q = a$.
		
		\item The characteristic of $\mathbb{F}_q$ is a prime number $p$ and $q = p^s$ for some positive integer $s$. Moreover, $\mathbb{F}_q$ is a vector space over $\mathbb{F}_p$ of dimension $s$.
		
		\item Let $1 \in \mathbb{F}_q$ be the multiplicative identity. By adding $1$ to itself up to $p-1$ times, $\{0, 1, 1+1, \ldots,1+ \cdots +1\} \cong \mathbb{F}_p$ is a subfield of $\mathbb{F}_q$. 
		
		\item The \emph{trace function} $\tr_{\mathbb{F}_q/\mathbb{F}_p}: \mathbb{F}_q \to \mathbb{F}_p$ defined by 
		\begin{align}
			\tr_{\mathbb{F}_q/\mathbb{F}_p}(x) = x + x^p + \cdots + x^{p^{s-1}} \in \mathbb{F}_p \subset \mathbb{F}_q
		\end{align}
		is $\mathbb{F}_p$-linear when we view $\mathbb{F}_q$ as an
$s$-dimensional vector space over $\mathbb{F}_p$. It is the trace of the
$\mathbb{F}_p$-linear multiplication map $y\mapsto xy$, and the associated
trace pairing is nondegenerate in the sense that if
$\tr_{\mathbb{F}_q/\mathbb{F}_p}(xy)=0$ for all $x\in\mathbb{F}_q$, then
$y=0$.
	\end{enumerate}
\end{proposition}

\begin{definition}\label{def:k_controlled_Z}
		Let $M$ be a monomial of degree $\rho$ in $\mathbb{F}_q[x_1, \ldots,x_n]$ with $q = p^s$ a prime power. Let $\tr_{\mathbb{F}_q/\mathbb{F}_p}$ be the trace function. For any $n$-qudit system with local dimension $q$, the monomial $M$ defines a \emph{$\rho$-controlled-$Z$ gate $\mathrm{C}^{\rho-1} Z$} by 
	\begin{align}
		U_{M} \coloneq \sum_{x \in \mathbb{F}_q^n} e^{\frac{2\pi i}{p} \tr_{\mathbb{F}_q/\mathbb{F}_p}( M(x) )} \ket{x} \bra{x}.
	\end{align} 
	Any composition of $\mathrm{C}^{\rho-1}Z$ gates can therefore be represented by a homogeneous polynomial of degree $\rho$.
\end{definition}

When $q=p=2$, Definition~\ref{def:k_controlled_Z} reduces to the usual $\mathrm C^{\rho-1}Z$ gates on qubits. For conciseness, we do not formally define the Pauli group and Clifford hierarchy on general qudits; comprehensive details can be found in, e.g.,~\cite{PhysRevA.95.012329,Feng2024}. Definition~\ref{def:k_controlled_Z} is sufficient for our purposes.

\begin{definition}[Invariant form] \label{def:invariant_poly}
	Given a cochain complex $C^\bullet$, a linear function $T: C^i \to \mathbb{F}_q$ is said to be \emph{cohomology invariant} if
	\begin{align}\label{eq:invariant_1}
		T(x_i + y_i) = T(x_i)
	\end{align}  
	for every cocycle $x_i\in\ker\delta^i$ and every coboundary $y_i\in\im\delta^{i-1}$. In other words, the restriction of $T$ to $\ker\delta^i$ descends to a linear functional on the cohomology $H^i(C)$.
\end{definition}

Given $\rho$ cochains, a $\rho$-linear function $T$ is invariant if
\begin{align}\label{eq:invariant_2}
	T(x_{p_1} + y_{p_1}, \ldots,x_{p_\rho} + y_{p_\rho}) = T(x_{p_1}, \ldots,x_{p_\rho})
\end{align}  
for cocycles and coboundaries. Expanding $T$ in coordinates, we can equivalently represent it as a degree-$\rho$ homogeneous polynomial in the components of the vectors $x_{p_1},\ldots,x_{p_\rho}$. Below we use the terms \emph{invariant form} and \emph{invariant polynomial} interchangeably.

By Proposition~\ref{prop:finite_field}, we can also show that Eq.~\eqref{eq:invariant_2} implies the following identity after applying the trace map:
\begin{align}\label{eq:invariant_2_trace}
	\tr_{\mathbb{F}_q/\mathbb{F}_p} \Big(  T(x_{p_1} + y_{p_1}, \ldots,x_{p_\rho} + y_{p_\rho}) \Big) = \tr_{\mathbb{F}_q/\mathbb{F}_p} \Big(  T(x_{p_1}, \ldots,x_{p_\rho}) \Big).
\end{align}
Now consider the $\rho$ CSS codes defined by the following consecutive terms:
\begin{align}
	(C^{p_1 - 1} \to C^{p_1} \to C^{p_1 + 1} ), \cdots\cdots, ( C^{p_\rho - 1} \to C^{p_\rho} \to C^{p_\rho + 1} ).
\end{align}
The invariant polynomial $T$ induces a (composition of) diagonal multi-controlled-$Z$ gate with respect to the codewords:
\begin{align}\label{eq:invariant_poly}
	\begin{aligned}
		U_T \ket{\bar{x}_{p_1}, \ldots,\bar{x}_{p_\rho}}
		= & \frac{1}{\mathcal{N}} \sum_{y_{p_1}, \ldots,y_{p_\rho}} e^{\frac{2\pi i}{p} \tr_{\mathbb{F}_q/\mathbb{F}_p}\!\left(T(x_{p_1} + y_{p_1}, \ldots,x_{p_\rho} + y_{p_\rho})\right)} \ket{x_{p_1} + y_{p_1}, \ldots,x_{p_\rho} + y_{p_\rho}} \\
		= & \frac{1}{\mathcal{N}} e^{\frac{2\pi i}{p} \tr_{\mathbb{F}_q/\mathbb{F}_p}\!\left(T(x_{p_1}, \ldots,x_{p_\rho})\right)} \sum_{y_{p_1}, \ldots,y_{p_\rho}} \ket{x_{p_1} + y_{p_1}, \ldots,x_{p_\rho} + y_{p_\rho}} \\
		= & e^{\frac{2\pi i}{p} \tr_{\mathbb{F}_q/\mathbb{F}_p}\!\left(T(x_{p_1}, \ldots,x_{p_\rho})\right)} \ket{\bar{x}_{p_1}, \ldots,\bar{x}_{p_\rho}},
	\end{aligned}
\end{align}
where $\mathcal{N}$ is the normalization factor for the code states.

For fault tolerance, we seek to implement the logical gate using a constant-depth circuit $U$ composed of bounded-size local gates. For an operator $M$ on an $n$-qudit system, let $\vert M\vert$ denote the number of sites on which $M$ acts nontrivially. Such a circuit enlarges the support by at most a constant factor:
$\vert UMU^\dagger\vert\leq c\vert M\vert$.
Consequently, error propagation remains bounded, thereby ensuring that the logical gate implementation is fault-tolerant. For a logical gate defined by an invariant polynomial, it can be scheduled in constant depth if each variable appears in only a constant number of monomials.

We further draw attention to two refined types of parameters, $n_{\mathrm{C}^{\rho-1} Z}$ and $k_{\mathrm{C}^{\rho-1} Z}$. The quantity $n_{\mathrm{C}^{\rho-1}Z}$ denotes the total number of physical gates used to implement $U_T$ in Eq.~\eqref{eq:invariant_poly}, or equivalently, the number of monomials in the invariant polynomial $T$. Given the constant-depth assumption, each variable from any code block appears in at most $O(1)$ monomials of $T$. Since $U_T$ acts nontrivially on $\Theta(N)$ physical qudits, it follows that $n_{\mathrm{C}^{\rho-1}Z}=\Theta(N)$, where $N$ is the number of physical qudits in a single code block. Since $U_T$ is generally a composition of logical $\mathrm{C}^{\rho-1} Z$ gates, $k_{\mathrm{C}^{\rho-1} Z}$ measures the number of \emph{disjoint} or \emph{parallelizable} logical $\mathrm{C}^{\rho-1} Z$ gates in the composition. Intuitively, $k_{\mathrm{C}^{\rho-1} Z}/n_{\mathrm{C}^{\rho-1} Z}$ characterizes the actual resource overhead in implementing the gate. In particular, the magic state distillation overhead exponent is given by $\gamma=\frac{\log( n_{\mathrm{C}^{\rho-1} Z}/k_{\mathrm{C}^{\rho-1} Z} )}{\log(d)}$.

Mathematically, when $\rho = 3$, we assemble the values $T(x_{p_1},x_{p_2},x_{p_3})$ into a 3-tensor, still denoted by $T$,
\begin{align}
	T \in \mathbb{F}_q^{k_1} \otimes \mathbb{F}_q^{k_2} \otimes \mathbb{F}_q^{k_3},
\end{align}
where $k_i$ is the dimension of the $i$-th code. The number of nontrivial logical actions of $U_T$ is captured by the \emph{subrank} of $T$. Formally speaking, the \emph{rank} of $T$ is the smallest integer $\mathfrak{r}$ such that $T$ can be written as a sum of $\mathfrak{r}$ simple tensors:  
\begin{align}
	T = \sum_{i = 1}^\mathfrak{r} u_i \otimes v_i \otimes w_i.
\end{align}
The \emph{subrank} of $T$ is the largest integer $r$ for which there are linear maps $A_j: \mathbb{F}_q^{k_j} \rightarrow \mathbb{F}_q^r$ satisfying
\begin{align}\label{eq:subrank}
	A_1 \otimes A_2 \otimes A_3 (T) = \sum_{i = 1}^r e_i \otimes e_i \otimes e_i,
\end{align}
where $\{e_i\}_{i=1}^r$ is the standard basis of $\mathbb{F}_q^r$. In other words, we seek local linear maps that transform $T$ into a diagonal tensor. For matrices viewed as $2$-tensors, both rank and subrank equal the ordinary matrix rank. However, for higher-order tensors, the subrank can be strictly smaller than the rank. For example, the tensor $e_1 \otimes e_1 \otimes e_2 + e_1 \otimes e_2 \otimes e_1 + e_2 \otimes e_1 \otimes e_1$ has rank $3$ but subrank $1$. More strikingly, deciding whether a $3$-tensor over a finite field has rank at most a given value is known to be $\mathsf{NP}$-complete~\cite{Hastad1990,Hillar2013}. Although subrank is generally a less stringent notion, we lack useful tools for estimating it except in certain special cases~\cite{Kopparty2023Geometric,Juvekar2023}. To our knowledge, previous work has obtained subrank estimates for the gate tensors arising from certain hypergraph product codes~\cite{Golowich_Lin2024} and algebraic geometry codes~\cite{Wills2024magic,VirgileCCZ,Golowich_Guruswami2025A,Nguyen2025CCZ,He2025addressable}. Here, we make the first such attempt for (almost) good qLDPC codes.


\subsection{Expander graph families and Ramanujan graphs}\label{sec:expander}

Expander graph families constitute one of the central ingredients in constructing cochain complexes that underpin (almost) good qLDPC codes. We provide a brief introduction here and further details in Secs.~\ref{sec:covering} and~\ref{sec:cubical}. To avoid a notational conflict with code distance $d$, we write $\vert V(G)\vert=N$ and denote the vertex degree by $n$. In our discussion below, $n$ is always a constant. Let $G=(V(G),E(G))$ be an undirected simple graph with adjacency matrix $A$.
Let $S \subset V(G)$ with $S^c = V(G) \setminus S$, and let $\partial S = E(S,S^c)$ be the set of edges with one endpoint in $S$ and one endpoint outside of $S$. The \emph{edge expansion constant} of $G$ is defined by
\begin{align}
    h(G) = \min_{0<\vert S\vert\leq\vert V(G)\vert/2} \frac{\vert\partial S\vert}{\vert S\vert}.
\end{align}

\begin{theorem}[Cheeger's inequality]
	Suppose $G$ is a $n$-regular graph with $\vert V(G) \vert = N$. Let $\lambda_1 \geq \lambda_2 \geq \cdots \geq \lambda_N$ be the spectrum of the adjacency matrix $A$. Then 
	\begin{align}
		\frac{n - \lambda_2}{2} \leq h(G) \leq \sqrt{2n(n - \lambda_2)}.
	\end{align}
\end{theorem}

As a reminder, a similar form of Cheeger's inequality holds for general transition matrices of Markov chains~\cite{DiaconisCheeger1991}. For our purposes, we use only the special case in which the transition matrix is $\frac{1}{n}A$ for an $n$-regular graph.

\begin{definition}[Expander graph] \label{def:expander}
	A sequence of $n$-regular graphs $\{G_i\}$ with vertices $\to \infty$ is said to be a family of \emph{edge expander graphs} if there is a constant $c$ such that $c \leq h(G_i)$ for all $i$.
\end{definition}

By Cheeger's inequality, a fixed-degree graph family is expanding if the \emph{spectral gap} $n-\lambda_2$ is bounded below by a positive constant. This spectral gap condition is also central to many distance bounds for good classical and quantum LDPC codes~\cite{SS1996,PK2021,PK2022Good,QuantumTanner2022,DHLV2022,Dinur2024sheaf}. As a counterexample, the family of cycles $\{C_N\}_{N\geq3}$ is not expanding: taking $S$ to consist of $\lfloor N/2\rfloor$ consecutive vertices gives $\vert\partial S\vert=2$ and hence $h(C_N)\leq2/\lfloor N/2\rfloor = O(1/N)$. Equivalently, since $\lambda_2=2\cos(2\pi/N)$, the right-hand side of Cheeger's inequality also shows that $h(C_N)=O(1/N)$ and is therefore not bounded below by a positive constant. It is worth mentioning that there are several definitions of expander graphs. For a more comprehensive exposition of their relationships, see Ref.~\cite{Alon1986Expander}.

\begin{definition}[Ramanujan graph] \label{def:Ramanujan}
	Let $n \geq 2$. An $n$-regular connected graph $G$ with $N$ vertices is called a \emph{Ramanujan graph} if all the eigenvalues $\lambda$ of its adjacency matrix $A$ satisfy either $\lambda = \pm n$ or $\vert \lambda \vert \leq 2 \sqrt{n - 1}$. Eigenvalues not equal to $\pm n$ are said to be \emph{nontrivial} and $[- 2 \sqrt{n - 1},  2 \sqrt{n - 1}]$ is the \emph{Ramanujan interval}.
\end{definition}

Note that when $n = 2$, the Ramanujan bound gives no positive spectral gap since $n - 2\sqrt{n-1} = 0$, so we assume \(n\geq 3\) below. For every fixed $n\geq 3$, the Ramanujan bound implies
\begin{align}
	n - \lambda_2(G) \geq n - 2\sqrt{n-1} > 0.
\end{align}
This lower bound is uniform in the number of vertices. Therefore, by Cheeger's inequality, any family of $n$-regular Ramanujan graphs whose numbers of vertices tend to infinity is an expander family. The converse does not hold, that is, an expander family need not be Ramanujan. Expander graphs whose nontrivial adjacency eigenvalues have absolute values $O(\sqrt{n})$ or $O(\sqrt{n}\,\mathrm{polylog}\,n)$ are called \emph{almost-Ramanujan graphs}.

There is a long history of constructing Ramanujan graphs, e.g., Refs.~\cite{LPS1988,Margulis1988,Marcus2015I,Marcus2018IV,Chandrasekaran2017,Hall_2018,Huang2025Ramanujan}, and expander families, e.g., Refs.~\cite{ZigZzag2002,ZigZzag2008,Bilu2004,Agarwal2016,Mohanty2022,Alon2020,Jeronimo2021}. As another example, the complete graph $K_n$ is an $(n-1)$-regular Ramanujan graph because all of its nontrivial adjacency eigenvalues equal $-1$. The complete bipartite graph $K_{n,n}$ is also an $n$-regular Ramanujan graph because its spectrum consists of $n$, $-n$, and $0$ with multiplicity $2n-2$. However, neither sequence forms an expander family in the sense of Definition~\ref{def:expander}, because the vertex degrees grow with the number of vertices. A growing vertex degree also obstructs an LDPC construction because the resulting boundary operator is not sparse.


\subsection{Cellular sheaves and sheaf codes}\label{sec:sheaf}

We now formally define sheaf codes. To make the presentation most accessible, we primarily use the language and methods of basic linear algebra and combinatorics. Complementary treatments can be found in Refs.~\cite{First2022,First2024,Dinur2024sheaf,Panteleev2024}. Throughout, $\mathcal{F}$ denotes a \emph{cellular sheaf}, by which we mean a covariant coefficient system on the face poset; below we abbreviate it as a sheaf. We use \emph{local coefficients} for its stalks and structure maps, and we keep contravariant coefficient systems distinct.
For a systematic treatment of sheaf codes emphasizing their topological and homological properties, see Ref.~\cite{Li2025Poincare}. 
In Secs.~\ref{sec:covering}, \ref{sec:SS_code}, and~\ref{sec:DLV_code}, we show how familiar structures including lifted graphs, expander codes, and certain good qLDPC codes fit into the sheaf framework.

\begin{definition}[Graded incidence poset]
	A finite \emph{poset} $X$ is a set equipped with a partial order $\preceq$. We write $\sigma\prec\tau$ if $\sigma\preceq\tau$ and $\sigma\neq\tau$, and say that $\tau$ \emph{covers} $\sigma$ if there is no $\eta\in X$ such that $\sigma\prec\eta\prec\tau$. The poset $X$ is called \emph{graded} if it is equipped with a rank function $g:X\rightarrow\mathbb{Z}_{\geq0}$ such that all minimal elements have grade $0$ and
	$g(\tau)=g(\sigma)+1$ whenever $\tau$ covers $\sigma$. A graded poset is called a \emph{graded incidence poset} if, for every $\sigma\prec\pi$ with $g(\pi)=g(\sigma)+2$, the set
	\begin{align}
		\{\tau\in X:\sigma\prec\tau\prec\pi\}
	\end{align}
	has even cardinality. The highest grade of $X$ is called its dimension.
\end{definition}

\begin{example}
	Any family of sets ordered by inclusion is a poset. The face poset of any finite regular CW complex is a graded poset. Any finite simplicial complex, with the empty face omitted and its faces ordered by inclusion, is a graded incidence poset.
\end{example}

Given a finite graded incidence poset $X$, we can establish a combinatorial cochain complex:
\begin{enumerate}
	\item Let $X(i)$ be the collection of elements in $X$ of grade $i$. Viewing these elements as a formal basis, we define the vector space $C^i = C^i(X,\mathbb{F}_q) \coloneq \mathbb{F}_q^{\vert X(i) \vert}$.
	
	\item The coboundary operator $\delta^i: C^i(X,\mathbb{F}_q) \rightarrow C^{i+1}(X,\mathbb{F}_q)$ is defined by 
	\begin{align}\label{eq:poset_coboundary}
		\delta^i (\sigma) = \sum_{\sigma \prec \tau \in X(i+1)} \tau. 
	\end{align}
	The coboundary condition follows from the even incidence property and the fact that $\mathbb{F}_q$ has characteristic $2$.
\end{enumerate}

To further generalize the above construction, we introduce cellular sheaves.

\begin{definition}[Cellular sheaf]\label{def:presheaf}
	Given a poset $X$, a \emph{cellular sheaf} (or \emph{covariant coefficient system}) $\mathcal{F}$ consists of a vector space $\mathcal{F}(\sigma)$ for each $\sigma \in X$ and a linear map $\mathcal{F}_{\sigma,\tau}: \F(\sigma) \rightarrow \F(\tau)$ for each $\sigma \preceq \tau$, subject to
	\begin{align}
		\sigma \preceq \tau \preceq \pi  \implies  \mathcal{F}_{\tau,\pi} \circ \mathcal{F}_{\sigma,\tau} = \mathcal{F}_{\sigma, \pi}. 
	\end{align}
	We also require $\mathcal{F}_{\sigma,\sigma}: \F(\sigma) \rightarrow \F(\sigma)$ to be the identity. Thus $\mathcal{F}$ is a covariant functor $X \rightarrow \mathbf{Vect}_{\mathbb{F}_q}$. A contravariant coefficient system $\mathcal{G}$ instead assigns maps $\mathcal{G}_{\tau,\sigma}: \mathcal{G}(\tau) \rightarrow \mathcal{G}(\sigma)$ and is a functor $X^{\mathrm{op}} \rightarrow \mathbf{Vect}_{\mathbb{F}_q}$; we keep these notions distinct.
\end{definition}

For finite-dimensional vector spaces, dualization reverses arrows. A contravariant coefficient system $\mathcal{G}$ determines a covariant coefficient system $\mathcal{G}^{\vee}$ with
\begin{align}
	\mathcal{G}^{\vee}(\sigma) \coloneq \mathcal{G}(\sigma)^*, \qquad
	\mathcal{G}^{\vee}_{\sigma,\tau} \coloneq \mathcal{G}_{\tau,\sigma}^*: \mathcal{G}(\sigma)^* \longrightarrow \mathcal{G}(\tau)^*.
\end{align}
This dualization is canonical. After bases are fixed, the dual maps are represented by transpose matrices; identifying $\mathcal{G}(\sigma)^*$ with a coordinate copy of $\mathcal{G}(\sigma)$ is basis dependent. We therefore do not identify covariant and contravariant coefficient systems.

Given a graded incidence poset $X$ and a cellular sheaf $\mathcal{F}$, we generalize $C^\bullet(X,\mathbb{F}_q)$ to $C^\bullet(X, \mathcal{F})$ as follows:
\begin{enumerate}
	\item We define 
	\begin{align}
		C^i(X, \mathcal{F}) = \bigoplus_{\sigma \in X(i)} \F(\sigma).
	\end{align}
    For comparison, $C^i(X,\mathbb{F}_q)=\mathbb{F}_q^{\vert X(i)\vert}=\bigoplus_{\sigma\in X(i)}\mathbb{F}_q$. In this sense, $C^\bullet(X,\mathbb{F}_q)$ is defined by the constant scalar coefficient system with stalk $\mathbb{F}_q$ at every $\sigma\in X$ and identity structure maps.

    \item For any $x\in C^i(X,\mathcal{F})$, its components can be read off by taking any $\sigma\in X(i)$ and looking at the local vector $x(\sigma)\in\mathcal{F}(\sigma)$. The coboundary operator
$\delta^i:C^i(X,\mathcal{F})\rightarrow C^{i+1}(X,\mathcal{F})$
is defined componentwise by
\begin{align}
	(\delta^i x)(\tau)
	=
	\sum_{\substack{\sigma\in X(i)\\ \sigma\prec\tau}}
	\mathcal{F}_{\sigma,\tau}\bigl(x(\sigma)\bigr)
\end{align}
for every $\tau\in X(i+1)$.
\end{enumerate}
Concretely, $\delta^i$ is a \emph{block matrix} whose columns are indexed by $\sigma\in X(i)$ and whose rows are indexed by $\tau\in X(i+1)$. We first take the $0$--$1$ matrix whose $(\tau,\sigma)$ entry is nonzero exactly when $\sigma\prec\tau$. By construction, this is the coboundary matrix of $C^\bullet(X,\mathbb{F}_q)$. We then replace each nonzero entry by the matrix representing $\mathcal{F}_{\sigma,\tau}$: $\delta^i = ( \mathcal{F}_{\sigma, \tau})$.

\begin{remark}\label{remark:presheaf}
Intuitively, substituting subblocks is the most flexible way to generalize and manipulate a matrix with scalar entries (cf. Kronecker or tensor products), which is one of the most important insights for adjusting code parameters. However, subblocks cannot be inserted arbitrarily. Some compatibility constraints are necessary to obtain a cochain complex. First, for fixed $\sigma\in X(i)$ and $\pi\in X(i+2)$ with $\sigma\prec\pi$, the $(\pi,\sigma)$-block of the composition is
\begin{align}
	\begin{aligned}
	\bigl(\delta^{i+1}\circ\delta^i\bigr)_{\pi,\sigma} =
	\sum_{\substack{\tau\in X(i+1)\\
	\sigma\prec\tau\prec\pi}} \mathcal{F}_{\tau,\pi}\circ\mathcal{F}_{\sigma,\tau} 
= \sum_{\substack{\tau\in X(i+1) \\
	\sigma\prec\tau\prec\pi}} \mathcal{F}_{\sigma,\pi}.
	\end{aligned}
\end{align}
Second, if $X$ is a graded incidence poset, then the last sum contains an even number of identical terms and therefore vanishes because $\mathbb{F}_q$ has characteristic $2$. This verifies the coboundary condition.

	In the standard formulation, the same identity can be verified on a cochain $x$ as follows:
	\begin{align}
		\begin{aligned}
			\left[ \delta^{i+1} \circ \delta^i (x) \right] (\pi) 
			& = \sum_{\tau \in X(i+1), \tau \prec \pi} \mathcal{F}_{\tau, \pi} \left[ \left( \delta^i (x) \right) (\tau) \right] \\
			& = \sum_{\tau \in X(i+1), \tau \prec \pi} \mathcal{F}_{\tau, \pi} \left[ \sum_{\sigma \in X(i), \sigma \prec \tau} \mathcal{F}_{\sigma, \tau} \left( x(\sigma) \right) \right] \\
			& = \sum_{\tau \in X(i+1), \tau \prec \pi} \sum_{\sigma \in X(i), \sigma \prec \tau} \mathcal{F}_{\sigma, \pi} \left( x(\sigma) \right).
		\end{aligned}
	\end{align} 
    In the last step, we regroup the summation by $\sigma$. The even incidence condition then shows that each term $\mathcal{F}_{\sigma,\pi}\bigl(x(\sigma)\bigr)$ occurs an even number of times, and hence the summation is zero.
\end{remark}

Dualizing the cellular sheaf gives a contravariant coefficient system $\mathcal{F}^{\vee}$ and an intrinsic chain complex with
\begin{align}
	C_i(X, \mathcal{F}^{\vee}) = \bigoplus_{\sigma \in X(i)} \F(\sigma)^*.
\end{align}
After a basis is fixed in every stalk, we use the induced coordinate identifications and, for notational simplicity, write this group as
$C_i(X,\mathcal{F})=\bigoplus_{\sigma\in X(i)}\mathcal{F}(\sigma)$.
Under this convention, the boundary operator is represented by block-matrix transposes and is defined componentwise by
\begin{align}
	\bigl(\partial_{i+1}x\bigr)(\sigma)
	=
	\sum_{\substack{\tau\in X(i+1)\\ \sigma\prec\tau}}
	\mathcal{F}_{\sigma,\tau}^{T}\bigl(x(\tau)\bigr),
\end{align}
or equivalently,
\begin{align}
	\left(\partial_{i+1} (x)\right) (\sigma) = \sum_{\tau \in X(i+1), \tau \succ \sigma} \mathcal{F}_{\sigma, \tau}^T ( x(\tau) ).
\end{align} 

\begin{definition}[(Co)chain map]\label{def:chain_map}
Suppose $C_\bullet$ and $D_\bullet$ are chain complexes with boundary maps $\partial_C$ and $\partial_D$, respectively. We say that
$f_\#=(f_i:C_i\rightarrow D_i)$ is a \emph{chain map} if 
\begin{align}
	f_{i-1}\circ\partial_{C,i} =
	\partial_{D,i}\circ f_i
\end{align}
in every degree. Equivalently, the following diagram commutes:
	\begin{equation}
		\begin{tikzcd}[column sep=1.5cm, row sep=0.7cm]
			\cdots \arrow[r, "\partial_{C,i+1}"] & C_i \arrow[r, "\partial_{C,i}"] \arrow[d,"f_i"] & C_{i-1} \arrow[d, "f_{i-1}"] \arrow[r, "\partial_{C,i-1}"] & \cdots \\
			\cdots \arrow[r, "\partial_{D,i+1}"] & D_i \arrow[r, "\partial_{D,i}"] & D_{i-1} \arrow[r, "\partial_{D,i-1}"] & \cdots
		\end{tikzcd}
	\end{equation}
	We can define \emph{cochain maps} for cochains analogously.
\end{definition}

\begin{definition}[Sheaf code]
	Let $X$ be a $t$-dimensional graded incidence poset with a cellular sheaf $\mathcal{F}$. A \emph{sheaf code} is a CSS code defined by extracting any three consecutive terms from the following cochain complex:
	\begin{equation}
		\begin{tikzcd}
			\cdots \arrow[r] &
			C^i(X, \mathcal{F}) \arrow[r,"\delta^i"] & 
			C^{i+1}(X, \mathcal{F}) \arrow[r,"\delta^{i+1}"] &
			C^{i+2}(X, \mathcal{F}) \arrow[r] & 
			\cdots 
		\end{tikzcd}.
	\end{equation}
\end{definition}

Here we use ``cellular sheaf'' in the combinatorial sense of a covariant coefficient system on the graded incidence poset $X$. Additional topological hypotheses are required for the stronger homological and duality results used below; see~\cite{Li2025Poincare}.


\subsection{Lifted graphs and their expansions}\label{sec:covering}

The framework of sheaf codes reveals how global structures $X$ and local coefficients $\mathcal{F}$ can be combined to construct good codes. In the one-dimensional case of classical codes, it suffices to take $X$ to be an expander graph. In higher-dimensional settings related to quantum codes, graph lifts are essential for constructing $X$. The notions of graph lifting and covering are sometimes used interchangeably in the literature, although several slightly different definitions occur. We summarize the relevant definitions and properties used in constructing expander families. We also show in Remark~\ref{remark:lift} that a lifted graph can be described by a cellular sheaf on the graph incidence poset, with the corresponding chain complex obtained from the duals of its structure maps.

\begin{definition}[Graph covering] \label{def:covering}
	Let $G$ be a graph. Any graph $\hat{G}$ is an \emph{$l$-covering} of $G$ if it is equipped with a map $\pi: V(\hat{G}) \rightarrow V(G)$, called a covering map, such that
	\begin{enumerate}
		\item For any $v \in V(G)$, $\vert \pi^{-1}(v) \vert = l$. The set $\pi^{-1}(v)$ is called the \emph{fiber} of $v$.
		
		\item Every edge $(u',v') \in E(\hat{G})$ projects to an edge $(\pi(u'),\pi(v')) \in E(G)$.
		
		\item For any $(u,v) \in E(G)$, the edges between $\pi^{-1}(u)$ and $\pi^{-1}(v)$ are determined by a one-to-one correspondence between the two fibers.
	\end{enumerate}
\end{definition}

The last property says that there are exactly $l$ edges between the fibers over the endpoints of each edge of $G$. More precisely, after labeling each fiber by $\{1,\ldots,l\}$, every such covering can be described as a lift.

\begin{definition}[Lift] \label{def:lift}
	Let $S_l$ be the symmetric group acting on $\{1,\ldots,l\}$. Given any graph $G$, its \emph{$l$-lift} $\hat{G}$ is defined by the following procedure:
	\begin{enumerate}
		\item We replace each vertex $v$ by the $l$ vertices $(v,i)$ with $i = 1, \ldots,l$.
		
		\item For any edge $e$ connecting $u$ and $v$, we prescribe an \emph{orientation}, say $(u,v)$, and assign a permutation $\gamma_{(u,v)} \in S_l$. The resulting assignment is called a \emph{signing} or \emph{voltage assignment}. We set $\gamma_{(v,u)} = \gamma_{(u,v)}^{-1}$.
		
		\item We define a new graph $\hat{G}$ with vertices $\{(v,i)\}$ by connecting $(u,i)$ and $(v,j)$ whenever $(u,v)$ is an oriented edge in the base graph $G$ and $j = \gamma_{(u,v)}(i)$.
	\end{enumerate}
\end{definition}

The orientation chosen for each edge in $G$ must be used consistently. Indeed, if one starts from the reverse direction and considers $(v,j)$ and $(u,i)$, then the inverse voltage must be used, so $i = \gamma_{(v,u)}(j) = \gamma_{(u,v)}^{-1}(j)$.

\begin{remark}
	Formally, let $\HH$ be the subgroup of $S_l$ generated by all the permutations $\gamma_{(u,v)}$. Then we say that $\hat{G}$ is lifted by $\HH$ via the natural permutation representation inherited from $S_l$, in which each element of $\HH$ is represented by its associated permutation matrix on $l$ sites. A notable example used throughout this paper is the regular cyclic subgroup $C_l = \langle (1\,2\,\cdots\,l) \rangle \subset S_l$. It is Abelian, as required in the cubical-complex construction in Sec.~\ref{sec:cubical}. A lift whose voltages lie in this cyclic subgroup is called a \emph{shift $l$-lift}.
\end{remark}

There is still another related notion of lifted graph that is widely used in the literature, and we also collect it here. Let $\HH$ be a finite (Abelian or non-Abelian) group. We consider the \emph{left regular representation} $\mathcal{L}: \HH \rightarrow \text{GL}( \vert{\HH}\vert, \mathbb{F}_q)$ of $\HH$ on $\mathbb{F}_q^{\vert \HH \vert}$ by the \emph{left translation action}, where the elements of $\HH$ label the standard basis:
\begin{align}
	\mathcal{L}(h) \cdot h' = h \cdot h'.
\end{align}
By definition, the matrix representation of $\mathcal{L}(h)$ is also a permutation matrix on the $\vert{\HH}\vert = l$ elements. This is Cayley's theorem, which allows us to embed any finite group $\HH$ into the symmetric group $S_{l}$. Then we define the \emph{left lift} of $G$ by Definition~\ref{def:lift}, but each $\gamma_{(u,v)}$ is now taken from $\HH$. Moreover, lifted vertices can be easily labeled by $(h,v)$ with $h \in \HH$. Similarly, we can define \emph{right lifts} by using the right translation of $\HH$.

\begin{example}
	Given any graph $G$, we take a 6-lift whose voltages lie in a subgroup $\HH\subset S_6$. If $\HH = C_6$ is the regular cyclic subgroup of $S_6$, the notion of lift in Definition~\ref{def:lift} and the left lift using $C_6$ are identical. If instead $\HH = C_2 \times C_2 \times C_2$ is generated by SWAPs $(1,2)$, $(3,4)$ and $(5,6) \in S_6$, we may either lift $G$ by the defining representation of $S_6$, which gives a 6-fold lift, or lift it by the left action of $\HH$, which gives an 8-fold lift. Despite the superficial discrepancy, Definition~\ref{def:lift} is still general enough. The reason is that, by Cayley's theorem, $\HH \subset S_8$, and the left lift can be interpreted using the defining representation of $S_8$.    
\end{example}

If $G$ is $n$-regular, any lift of $G$ is still $n$-regular. If $G$ happens to be a Cayley graph, then its vertices are elements of some group $\G$, which should not be confused with $\HH$.


\begin{remark}\label{remark:lift}
	Let $A$ be the adjacency matrix of $G$. Then the entry $A(u,v) = 1$ corresponds to the edge $(u,v)$. For the adjacency of the lift $\hat{G}$, we simply replace that entry by the $l \times l$ permutation matrix $\gamma_{(u,v)}^T=\gamma_{(u,v)}^{-1}$. Zero entries are replaced by the $l \times l$ zero matrix. By definition, $A(v,u)$ is replaced by $\gamma_{(u,v)}$. This also guarantees that the lifted adjacency matrix is still symmetric. 
	
	Let $\partial$ be the boundary operator/incidence matrix of $G$: a pair of nonzero entries $\partial(u,e), \partial(v,e)$ indicates that the vertices $u$ and $v$ are bounded through $e$. In defining the lift, we first replace $\partial(u,e)$ by the $l \times l$ identity matrix and regard this as a one-to-one correspondence between the lifted edges over $e$ and the lifted vertices over $u$. Then we need to replace $\partial(v,e)$ by the permutation matrix of $\gamma_{(u,v)}$. We can also reverse the process by identifying $v$ and $e$ at the beginning and then defining $\partial(u,e)$ by $\gamma_{(u,v)}^{-1}$. 
	
	In the language of Sec.~\ref{sec:sheaf}, lifting is a way to redefine the boundary operator of a graph into a block matrix. Equivalently, a lifted graph defines a cellular sheaf on the graph incidence poset; after bases are fixed, its dual structure maps give the block boundary operator $\hat{\partial}: \mathbb{F}_q^{l\vert E(G) \vert} \rightarrow \mathbb{F}_q^{l\vert V(G) \vert}$. 
\end{remark}

\begin{example}\label{example:2-lift}
	Given any graph $G$ with adjacency matrix $A$ and boundary operator $\partial$, taking 2-lifts by $\HH = S_2$ can yield
	\begin{enumerate}
		\item $G \sqcup G$, if $\gamma_{(u,v)}$ is always the identity element $\text{id} \in S_2$. The lift simply copies $G$ twice:
		\begin{align}
			\hat{A} = A \otimes I_2 = \begin{pmatrix} A & 0 \\ 0 & A \end{pmatrix}, \quad
			\hat{\partial} = \begin{pmatrix} \partial & 0 \\ 0 & \partial \end{pmatrix}.
		\end{align}
		It is easy to verify that $\hat{A}$ and $\hat{\partial}$ agree with the block matrices described in Remark~\ref{remark:lift} up to reordering the vertices and edges.
		
		\item $G \otimes K_2$, where $K_2$ is the complete graph on two vertices. This is the standard \emph{double cover} of $G$, and $\gamma_{(u,v)}$ is always the transposition $(1,2) \in S_2$:
		\begin{align}
			\hat{A} = A \otimes \begin{pmatrix} 0 & 1 \\ 1 & 0 \end{pmatrix} = \begin{pmatrix} 0 & A \\ A & 0 \end{pmatrix}.
		\end{align}
		As a caveat, writing $\hat{\partial}$ as an off-diagonal block matrix analogous to $\hat{A}$ is wrong in general and we should follow the method in Remark~\ref{remark:lift}. 
		An important example that we will encounter in Sec.~\ref{sec:cubical} is the double cover of the triangle $K_3$. It is the $6$-cycle $C_6$, not to be confused with the cyclic group denoted by the same symbol:
		\begin{equation}\label{eq:K_3_double}
			K_3 \otimes \begin{pmatrix} 0 & 1 \\ 1 & 0 \end{pmatrix} = C_6 = 	
			\resizebox{0.2\linewidth}{!}{%
				\begin{tikzpicture}[x=0.75pt,y=0.75pt,yscale=-1,xscale=1, baseline=(current bounding box.center)]
					\draw  [line width=1.5pt]  (255.33,87.96) -- (136.33,18) -- (374.33,18) -- cycle ;
					\draw    (136.33,18) -- (136.33,203.04) ;
					\draw    (374.33,18) -- (374.33,203.04) ;
					\draw    (255.33,87.96) -- (255.33,273) ;
					\draw [line width=1.5pt]    (137.33,203.04) -- (256.33,273) ;
					\draw [line width=1.5pt]    (256.33,273) -- (375.33,203.04) ;
					\draw [line width=1.5pt]  [dash pattern={on 5.63pt off 4.5pt}]  (136.33,203.04) -- (374.33,203.04) ;
					\draw [color={rgb, 255:red, 208; green, 2; blue, 27 }  ,draw opacity=1 ]   (136.33,18) -- (255.33,273) ;
					\draw [color={rgb, 255:red, 208; green, 2; blue, 27 }  ,draw opacity=1 ]   (255.33,87.96) -- (136.33,203.04) ;
					\draw [color={rgb, 255:red, 208; green, 2; blue, 27 }  ,draw opacity=1 ]   (255.33,87.96) -- (374.33,203.04) ;
					\draw [color={rgb, 255:red, 208; green, 2; blue, 27 }  ,draw opacity=1 ]   (374.33,18) -- (255.33,273) ;
					\draw [color={rgb, 255:red, 208; green, 2; blue, 27 }  ,draw opacity=1 ] [dash pattern={on 4.5pt off 4.5pt}]  (374.33,18) -- (136.33,203.04) ;
					\draw [color={rgb, 255:red, 208; green, 2; blue, 27 }  ,draw opacity=1 ] [dash pattern={on 4.5pt off 4.5pt}]  (136.33,18) -- (374.33,203.04) ;
				\end{tikzpicture}
			}
		\end{equation}
	\end{enumerate}	
\end{example}


The eigenvalues of a lift inherited from the base graph $G$ are called \emph{old eigenvalues}, and the remaining eigenvalues are called \emph{new eigenvalues}. Let $\rho(G)$ denote the spectral radius of the universal cover of $G$. An $l$-covering of $G$ is called \emph{Ramanujan} if all its new eigenvalues lie in $[-\rho(G),\rho(G)]$. For an $n$-regular graph, $\rho(G)=2\sqrt{n-1}$.

More generally, the voltages $\gamma_{(u,v)}$ may be chosen independently for different edges. Every connected regular bipartite graph $G_0$ has a Ramanujan $2$-covering. In particular, if $G_0$ is itself Ramanujan, then the lifted graph is also Ramanujan. Iterating this construction from a bipartite Ramanujan base graph yields a family of Ramanujan graphs.

\begin{theorem}[\cite{Marcus2015I,Marcus2018IV}]
	There is an infinite sequence of $n$-regular bipartite Ramanujan graphs with $2^i N_0$ vertices for $i \to \infty$ and some constant $N_0$.
\end{theorem}

The above theorem demonstrates the existence of Ramanujan bipartite graphs for any degree. More recently,~\cite{Huang2025Ramanujan} affirms the existence for the non-bipartite case. This remarkable achievement focuses on the distribution of the spectrum of random graphs. Since our inductive construction of codes and logical gates is based on lifted graphs, we discuss additional relevant results below.

\begin{theorem}[\cite{Chandrasekaran2017,Hall_2018}]\label{thm:3-lift} 
	Every bipartite graph $G$ has a Ramanujan $3$- or $4$-covering, where the permutation above every edge is cyclic.
\end{theorem}

\begin{remark}\label{remark:non-Abelian}
	Here is an important caveat for the shift $4$-lift: it cannot be treated as a 2-lift of some 2-lift of $G$, even though they have the same numbers of vertices and edges. Actually, a 2-lift of a 2-lift may stem from a non-Abelian subgroup of $S_4$, but a shift $4$-lift is defined using the cyclic group $C_4$. To see this, let us denote the 4-lifted vertices via $h \in C_4$ by $(h,v)$. A 4-lift creates an edge of  the form
	\begin{align}
		[ (h,v), (\gamma_{(v,v')} \cdot h, v') ]
	\end{align} 
	for some $\gamma_{(v,v')} \in C_4$. On the other hand, let us denote a vertex in a 2-lift of a 2-lift by $(h_2,(h_1,v))$. Such a lift creates an edge of the form
	\begin{align}\label{eq:sequential_2_lifts}
	[ (h_2,(h_1,v)),
	(\gamma_{2,[(h_1,v),(\gamma_{1,(v,v')}\cdot h_1,v')]}
	\cdot h_2,
	(\gamma_{1,(v,v')}\cdot h_1,v')) ].
\end{align}
	where the action of $\gamma_2$ is chosen based on the previous round of lifting. The edge does not look like
	\begin{align}\label{eq:independent_2_lifts}
		[ (h_2,h_1,v), (\gamma_{2,(v,v')} \cdot h_2, \gamma_{1,(v,v')} \cdot h_1,v') ].
	\end{align} 
	Therefore, Eq.~\eqref{eq:sequential_2_lifts} is defined sequentially, but in Eq.~\eqref{eq:independent_2_lifts}, the lifts can be taken in any order. In terms of group actions, Eq.~\eqref{eq:independent_2_lifts} uses $(\gamma_{2,(v,v')}, \gamma_{1,(v,v')} ) \in S_2 \times S_2$, which is not isomorphic to $C_4$ but still Abelian. However, in the sequential lifting, suppose we are given $(v_0,v_0')$ and $(v_1,v_1')$ after the first round of lifting. Then we will have four possible ways to do the second round:
	\begin{align}
		(v_{00},v_{00}'), \ (v_{01},v_{01}'), \ (v_{10},v_{10}'), \ (v_{11},v_{11}'); \\
		(v_{00},v_{00}'), \ (v_{01},v_{01}'), \ (v_{10},v_{11}'), \ (v_{11},v_{10}'); \\
		(v_{00},v_{01}'), \ (v_{01},v_{00}'), \ (v_{10},v_{10}'), \ (v_{11},v_{11}'); \\
		(v_{00},v_{01}'), \ (v_{01},v_{00}'), \ (v_{10},v_{11}'), \ (v_{11},v_{10}').
	\end{align}
	Suppose $(v_0,v_1')$ and $(v_1,v_0')$ are given at first. Then we have four extra ways to do the second:
	\begin{align}
		(v_{00},v_{10}'), \ (v_{01},v_{11}'), \ (v_{10},v_{00}'), \ (v_{11},v_{01}'); \\
		(v_{00},v_{10}'), \ (v_{01},v_{11}'), \ (v_{10},v_{01}'), \ (v_{11},v_{00}'); \\
		(v_{00},v_{11}'), \ (v_{01},v_{10}'), \ (v_{10},v_{00}'), \ (v_{11},v_{01}'); \\
		(v_{00},v_{11}'), \ (v_{01},v_{10}'), \ (v_{10},v_{01}'), \ (v_{11},v_{00}').
	\end{align}
	One can check that some of the permutations in the first case do not commute with some in the second and they form a non-Abelian subgroup in $S_4$. Consequently, rewriting a 2-lift of 2-lift into one group action turns out to be a non-Abelian lift. A similar phenomenon also occurs for other sequential Abelian lifts. Consequently, a sequence of Abelian lifts need not combine into a single Abelian lift and therefore cannot in general replace the Abelian lift required in the construction of a family of high-dimensional cubical complexes (see Sec.~\ref{sec:cubical}).
\end{remark}


Lifts by Abelian groups also have limitations. For any fixed base graph $G$, an Abelian $l$-lift cannot remain even almost Ramanujan for arbitrarily large $l$. In particular, Ref.~\cite{Agarwal2016} gives an exponential upper bound on the lift degree $l$ for which an Abelian lift can be almost Ramanujan. If we instead allow general $S_l$-valued edge labels rather than restricting them to an Abelian subgroup, then the following theorem holds.

\begin{theorem}[\cite{Hall_2018}]\label{thm:l-covering} 
	Every connected bipartite graph has a Ramanujan $l$-covering for every $l$.
\end{theorem}

Recall that constant-degree Ramanujan graph families are spectrally optimal expander families, whereas an expander family need not be Ramanujan. There has been extensive work on explicit expander constructions~\cite{ZigZzag2002,ZigZzag2008,Bilu2004}, including constructions of almost-Ramanujan families~\cite{Agarwal2016,Mohanty2022,Alon2020,Jeronimo2021,Jeronimo2024}.

\begin{theorem}[\cite{Agarwal2016}]\label{thm:Abelian_lift} 
    Let $G$ be an $n$-regular graph on $N$ vertices. Suppose that the absolute values of its nontrivial eigenvalues are bounded by $\lambda$ where $\lambda\geq\sqrt{n}$ and that $2\leq n\leq\sqrt{N/(3\ln N)}$. Let $H$ be a uniformly random shift $l$-lift of $G$ by the cyclic group $C_l$, and let $\lambda_{\mathrm{new}}$ denote the maximum absolute value among the new nontrivial eigenvalues of $H$. Then, with probability at least $1-l e^{-\Omega(N/n^2)}$,
	\begin{align}
		\lambda_{\mathrm{new}}=O(\lambda).
	\end{align}
\end{theorem}

Therefore, if $l \leq 2^{cN/n^2}$ for some constant $c>0$ and the base graph $G$ is almost Ramanujan, then there exists an almost-Ramanujan shift $l$-lift of $G$.


\subsection{Classical and quantum expander codes}\label{sec:SS_code}

To highlight the significance of local coefficients, we now introduce classical Tanner codes and Sipser--Spielman expander codes using cellular sheaves on one-dimensional incidence posets, in a form analogous to graph lifts. They serve as the basis for defining quantum expander codes via hypergraph products~\cite{Zemor2014,Leverrier_2015} and have also inspired recent constructions of (almost) good qLDPC codes~\cite{PK2021,PK2022Good,QuantumTanner2022,DHLV2022}.

\begin{definition}[\cite{Tanner1981}]\label{def:Tanner} 
	Let $G$ be an $n$-regular graph with a \emph{local classical code} $\C$ defined by the parity-check matrix $h: \mathbb{F}_q^n \rightarrow \mathbb{F}_q^m$ for some constant integer $m$. Let $N = \vert E(G) \vert = n \vert V(G) \vert/2$ be the code length. Then the \emph{classical Tanner code} $T(G,\C)$ defined by $G$ and $\C$ is a subspace of $\mathbb{F}_q^N = \mathbb{F}_q^{\vert E(G) \vert}$ constructed as follows.
	\begin{enumerate}
		\item For each $v \in V(G)$, we prescribe an order on the $n$ edges $e$ incident to $v$.
		
		\item Given any $x \in \mathbb{F}_q^N$, let $x_v$ be the local vector with components $x(e)$ for all edges $e$ incident to $v$. These components are arranged according to this edge order. 
	\end{enumerate}
	Then
	\begin{align}
		T(G,\C) \coloneq \{x \in  \mathbb{F}_q^N: x_v \in \ker h \text{ for any } v \in V(G) \}.
	\end{align}
\end{definition}

\begin{remark}\label{remark:SS_code}
We now incorporate Definition~\ref{def:Tanner} into the scheme of sheaf code. As in Remark~\ref{remark:lift}, the parity-check matrix $H$ of the Tanner code $T(G,\C)$ is defined from the boundary operator $\partial:\mathbb{F}_q^{\vert E(G)\vert}\to\mathbb{F}_q^{\vert V(G)\vert}$ of the graph $G$. We replace each nonzero entry $\partial_{ve}$ by the corresponding column of $h$. Recall that $h:\mathbb{F}_q^n\to\mathbb{F}_q^m$ has $n$ columns, which are matched with the $n$ edges incident to $v$ according to the local order at $v$. Consequently, $H\in\mathbb{F}_q^{m\vert V(G)\vert\times\vert E(G)\vert}$.
	For example, for the row of $\partial$ associated with a vertex of a $3$-regular graph, we make the following substitution:
	\begin{align}
		\begin{pmatrix}
			1 & 0 & 0 & 1 & 1
		\end{pmatrix}
		\to 
		\begin{pmatrix}
			h^{(2)} & 0 & 0 & h^{(1)} & h^{(3)}
		\end{pmatrix},
	\end{align}
	where $h = (h^{(1)}, h^{(2)}, h^{(3)})$ is the local parity-check matrix and the local order on edges is $(2,1,3)$. Note that for a graph family with fixed vertex degree, the matrices $H$ are sparse and define a family of classical LDPC codes. As a basic example, take $h=(1\ 1\cdots 1)$, the parity-check matrix of the dual repetition code. Then each $h^{(i)}$ equals $1$, and the undecorated boundary operator $\partial$ is recovered. Let $X$ be a 1D graded incidence poset defined by $X(0) = V(G)$ and $X(1) = E(G)$. The original boundary operator defines the following 1-chain complex with scalar local coefficients:
	\begin{align}
		C_1(X, \mathbb{F}_q) = \mathbb{F}_q^{\vert E(G) \vert} \xrightarrow{\partial} C_0(X, \mathbb{F}_q) = \mathbb{F}_q^{\vert V(G) \vert}.
	\end{align} 
	Let
	\begin{align}
		\mathcal{F}(e) \coloneq \mathbb{F}_q, \text{ for any } e \in E(G), \quad
		\mathcal{F}(v) \coloneq \mathbb{F}_q^m, \text{ for any } v \in V(G).
	\end{align}
	For every $v\prec e$, let $\mathcal{F}_{v,e}$ be the row of $h^T$ corresponding to $e$ in the local edge order at $v$. This yields a nontrivial cellular sheaf $\mathcal{F}$ (although the edge stalks remain one-dimensional) that corresponds to the parity-check matrix $H$ defined by the above block matrix. 
\end{remark}

The \emph{Sipser--Spielman expander codes} are constructed from $T(G,\C)$ using an expander family and a fixed good classical local code. Under mild conditions on the graph expansion and local code parameters, the resulting codes form a good classical LDPC code family~\cite{SS1996} (see also Theorem~\ref{thm:cLDPC}).
For comparison, suppose we define a classical code using only the boundary operator $\partial$. Every cycle defines a codeword because each vertex on the path is incident to two selected edges, so the boundary vanishes in characteristic $2$.
The length of the smallest closed path, called the \emph{girth} of $G$, satisfies the following estimates.
\begin{enumerate}
	\item If $G$ is a cycle, then $\operatorname{girth}(G)=N$, and the associated code has parameters $[N,1,N]$.
	
	\item If $G$ has average degree $\bar{n}\geq 3$, then the Moore bound gives 
    $\operatorname{girth}(G)
\leq 2\log_{\bar n-1}\vert V(G)\vert+2
=O(\log N)$~\cite{Margulis1988,LPS1988,Alon2002MooreBound}.
\end{enumerate} 
Therefore, codes defined solely by graph boundary operators cannot simultaneously have constant rate and linear distance.

\begin{theorem}[\cite{SS1996,Meshulam2018,PK2021}]\label{thm:cLDPC} 
	Let $G$ be a connected $n$-regular graph and let $H$ be the parity-check matrix of $T(G,\C)$. Let $\lambda$ denote the maximum absolute value among the nontrivial adjacency eigenvalues of $G$, excluding the trivial eigenvalue $n$ and also $-n$ when $G$ is bipartite (Definition~\ref{def:Ramanujan}). Suppose the local code $\C=\ker h$ has parameters $[n,k,d_0]$ with $d_0>\lambda$. Then the distance of $\ker H=T(G,\C)$ satisfies
	\begin{align}
		d\geq\frac{d_0}{2n}(d_0-\lambda)\vert V(G)\vert.
	\end{align}
	Furthermore, suppose that $h$ has full row rank and that the dual code $\C^\perp=\im h^T$ has distance at least $d_0$. Then the distance of $\ker H^T$ satisfies
	\begin{align}
		d^T\geq\frac{1}{2n}(d_0-\lambda)\vert V(G)\vert.
	\end{align}
\end{theorem}

Linear lower bounds on the distances of both $\ker H$ and $\ker H^T$ are important for controlling the parameters of the resulting quantum codes (see Eq.~\eqref{eq:HGP_parameters}). At a high level, the two linear distance bounds can be viewed as the simplest one-dimensional manifestation of Poincar\'e duality between sheaf codes and their duals~\cite{Li2025Poincare}.


\medskip
The \emph{quantum expander codes} are defined by taking hypergraph products (HGP) of classical expander codes. 
Hypergraph product codes have attracted extensive recent interest due to their relatively simple construction from classical codes; see, e.g., Refs.~\cite{BH2013homological,Zemor2014,Leverrier_2015,Zeng_2019,Zeng_2020,Chen_2023,Wang_2024,Golowich_Lin2024}. We only outline the three-dimensional case to make a direct comparison with the construction of higher-dimensional expanders in Sec.~\ref{sec:cubical}. Let $H_i^T:\mathbb{F}_q^{N_i}\to\mathbb{F}_q^{M_i}$, $i=1,2,3$, be parity-check maps defining classical codes.
The three-dimensional hypergraph product is the following cochain complex:
\begin{align}\label{eq:3D_HGP}
	\begin{aligned}
		\mathbb{F}_q^{N_1 \times N_2 \times N_3} \xleftarrow{\delta^2}
		& \mathbb{F}_q^{M_1 \times N_2 \times N_3} \oplus \mathbb{F}_q^{N_1 \times M_2 \times N_3} \oplus \mathbb{F}_q^{N_1 \times N_2 \times M_3}   \xleftarrow{\delta^1} \\
		& \mathbb{F}_q^{N_1 \times M_2 \times M_3} \oplus \mathbb{F}_q^{M_1 \times N_2 \times M_3} \oplus \mathbb{F}_q^{M_1 \times M_2 \times N_3} \xleftarrow{\delta^0} 
		\mathbb{F}_q^{M_1 \times M_2 \times M_3}.
	\end{aligned}
\end{align}
With the row and column conventions of Remarks~\ref{remark:lift} and~\ref{remark:SS_code}, the coboundary maps are
\begin{align}\label{eq:HGP_coboundary}
	\delta^1 = \begin{pmatrix} 
		0 & I_{M_1} \otimes I_{N_2} \otimes H_3 & I_{M_1} \otimes H_2 \otimes I_{N_3} \\ 
		I_{N_1} \otimes I_{M_2} \otimes H_3 & 0 & H_1 \otimes I_{M_2} \otimes I_{N_3} \\
		I_{N_1} \otimes H_2 \otimes I_{M_3} & H_1 \otimes I_{N_2} \otimes I_{M_3} & 0 
	\end{pmatrix},\quad
	\delta^0 = \begin{pmatrix} H_1 \otimes I_{M_2} \otimes I_{M_3} \\ I_{M_1} \otimes H_2 \otimes I_{M_3} \\ I_{M_1} \otimes I_{M_2} \otimes H_3	\end{pmatrix}
\end{align}
and $\delta^2 = (H_1 \otimes I_{N_2} \otimes I_{N_3} \quad I_{N_1} \otimes H_2 \otimes I_{N_3} \quad I_{N_1} \otimes I_{N_2} \otimes H_3 )$.  
We extract three consecutive terms from Eq.~\eqref{eq:3D_HGP} to define a CSS code.
For example, let $H_Z = \delta^1$ and $H_X^T = \delta^0$. Its code parameters $[\![N, k, d]\!]$ are~\cite{Zemor2014,Zeng_2019}
\begin{align}\label{eq:HGP_parameters}
	\begin{aligned}
		& N = N_1 M_2 M_3 + M_1 N_2 M_3 + M_1 M_2 N_3, 
		&& k = k_1 k_2 k_3^T + k_1 k_2^T k_3 + k_1^T k_2 k_3, \\
		& d_X = \min\{d_1 d_2, d_1 d_3, d_2 d_3\},
		&& d_Z = \min\{d_1^T, d_2^T, d_3^T \},
	\end{aligned}
\end{align}
where $(k_i,d_i)$ and $(k_i^T,d_i^T)$ are the dimension and distance of $\ker H_i$ and $\ker H_i^T$, respectively. In the expressions for $d_X$ and $d_Z$, the minimum is taken only over terms corresponding to nonzero K\"{u}nneth summands. 
We can take the matrices $H_i^T$ to be parity-check matrices of classical expander codes. As a simple example, if each $H_i^T$ is the boundary operator of a cycle as in Eq.~\eqref{eq:K_3_double}, then the resulting three-dimensional HGP is the 3D toric code. We elucidate how to upgrade an HGP code to a sheaf code in Sec.~\ref{sec:cubical} and explore the logical action of cup product gates on these codes in Secs.~\ref{sec:proof_HGP} and~\ref{sec:proof_sheaf}.


\section{High-dimensional expanders and cup products}\label{sec:HDX_Cup}

From the perspective of Sec.~\ref{sec:sheaf}, lifted graphs and classical Tanner codes are described by cellular sheaves on one-dimensional incidence posets and their associated (co)chain complexes. Following this paradigm, previous work has used higher-dimensional analogues of expander graphs, known as \emph{high-dimensional expanders}~\cite{Kaufman2014,Evra2020,Gotlib_Kaufman2023}, together with local codes to construct (almost) good qLDPC codes. We review several basic notions and observations needed later.

\subsection{Cubical complexes through products and lifts}\label{sec:cubical}

We adopt the procedure in Ref.~\cite{Dinur2024sheaf} to define a \emph{cubical complex} $X$ underlying the cochain complex $C^\bullet(X,\mathcal{F})$ introduced in Sec.~\ref{sec:sheaf}. As a preview, this procedure lifts the Cartesian product of several graphs rather than a single graph. If the lifting group is trivial, the construction reduces to the hypergraph product discussed at the end of Sec.~\ref{sec:SS_code}. More precisely, the construction proceeds as follows:
\begin{enumerate}
	\item Let $G_0 = (V_0, E_0)$ be an $n$-regular graph with $\vert V_0 \vert = n'$. As before, $n$ is always a constant. 
    Let $\HH$ be a finite Abelian group of order $l$ embedded in $S_l$ via its left regular representation, and choose a left $\HH$-lift of $G_0$ as in Definition~\ref{def:lift}. As in the case of classical Tanner code (Definition~\ref{def:Tanner}), we prescribe an order on the edges incident with any vertex in $G_0$. 
    We require this labeling to be consistent: for each $\mu\in[n]$, following the edge labeled $\mu$ defines a permutation of $V_0$, and these permutations form an inverse-closed set. Let $\gamma:E_0 \rightarrow \HH$ assign a group element $\gamma_{(u,v)}\in\HH$ to each edge $(u,v)$ with the prescribed orientation, and set $\gamma_{(v,u)}=\gamma_{(u,v)}^{-1}$. Let $A=(a^1,\ldots,a^n)$ be defined by
	\begin{align}
		a^{\mu} \cdot (h,v) = (\gamma_{(v, v')} \cdot h,v'),
	\end{align} 
	where $(v,v')$ is the $\mu$-th edge incident to $v$. There is one caveat in determining the order, which is elaborated on in Sec.~\ref{sec:DLV_code}.
	
	\item Now we consider the $t$-dimensional Cartesian product: 
	\begin{align}
		(h,v_1, \ldots,v_t; b_1, \ldots,b_t) \in \HH \times V_0 \times \cdots \times V_0 \times \{0,1\} \times \cdots \times \{0,1\},  
	\end{align}
	where $(b_1, \ldots,b_t)$ is a binary string. We now take $t$ copies of $A$, denoted by $A_j=(a_j^1,\ldots,a_j^n)$ for $j=1,\ldots,t$. They act on $h$ and the corresponding component of the Cartesian product as follows:
	\begin{align}\label{eq:generator_action}
		a_j^{\mu} \cdot (h,v_1, \ldots,v_j, \ldots,v_t; b_1, \ldots,b_j, \ldots,b_t) 
		= (\gamma_{(v_j, v_j')} \cdot h,v_1, \ldots,v_j', \ldots,v_t; b_1, \ldots,1-b_j, \ldots,b_t),
	\end{align} 
	where $(v_j,v_j')$ is the $\mu$-th edge incident to $v_j$. If $t = 1$, this simply defines the double cover of the lift of $G_0$.
	
	\item Write $g=(h,v_1,\ldots,v_t)$ and set $\G\coloneq\HH\times V_0^t$. For simplicity, we suppress the superscript $\mu$ on $a_j^\mu$. Then define
    \begin{align}
	X(0)\coloneq
	\bigl\{[g;(b_j)_{j\in[t]}]\colon
	g\in\G,\ (b_j)_{j\in[t]}\in\{0,1\}^t\bigr\}
	=\G\times\{0,1\}^t
    \end{align}
	as the collection of $0$-cubes/vertices. Naturally, the $1$-cubes are edges between $0$-cubes:
\begin{align}
	X(1) \coloneq \Bigl\{[g;a_i,(b_j)_{j\neq i}]:
	g\in\G,\ i\in[t],\ a_i\in A_i,\
	(b_j)_{j\neq i}\in\{0,1\}^{t-1} \Bigr\}.
\end{align}
Here $[g;a_i,(b_j)_{j\neq i}]$ denotes the edge with endpoints $[g;(b_1,\ldots,b_{i-1},0,b_{i+1},\ldots,b_t)]$ and $[a_i\cdot g;(b_1,\ldots,b_{i-1},1,b_{i+1},\ldots,b_t)]$, where $a_i$ acts on $g$ via Eq.~\eqref{eq:generator_action}.
	
	\item In the same spirit, we also have
\begin{align}
	X(2) \coloneq \Big\{
	[g;a_{i_1},a_{i_2},(b_j)_{j\notin\{i_1,i_2\}}]:
	g\in\G,\ i_1,i_2\in[t],\ i_1<i_2,\ 
	a_{i_1}\in A_{i_1},\ a_{i_2}\in A_{i_2},\
	b_j\in\{0,1\} \Big\}.
\end{align}
	More generally, let $S\subseteq[t]$ have cardinality $p$. A $p$-cube in $X(p)$ has the form $[g;(a_i)_{i\in S},(b_j)_{j\in S^c}]$, where $a_i\in A_i$ and $b_j\in\{0,1\}$. By definition, 
	\begin{align}\label{eq:system_size}
		\vert \G \vert = \vert \HH \times V_0 \times \cdots \times V_0 \vert = \vert \HH \vert \cdot \vert V_0 \vert^t.
	\end{align}
    Since $\vert A_i\vert=n$ for every $i\in[t]$, $\vert X(p)\vert=\binom{t}{p}2^{t-p}n^p\vert\G\vert$. If physical qudits are placed on the $1$-cubes, the system size is $N=\vert X(1)\vert=t2^{t-1}n\vert\G\vert$.
\end{enumerate}
The construction above defines a graded incidence poset $X$. Fig.~\ref{fig:$3$-cube} illustrates the inclusion relations within a $3$-cube of a three-dimensional cubical complex. When $\HH$ is trivial, $X$ is the $t$-fold Cartesian product of the double cover of $G_0$, and $C^\bullet(X,\mathbb{F}_q)$ is the corresponding $t$-fold HGP. In particular, when $G_0=K_3$ is the triangle, its double cover is $C_6$ in Eq.~\eqref{eq:K_3_double}. Comparing Eqs.~\eqref{eq:poset_coboundary} and~\eqref{eq:HGP_coboundary} shows that $C^\bullet(X,\mathbb{F}_q)$ produces the toric code.

One can also insert local codes, in which case \(C^\bullet(X,\F)\) is described by the formalism of quantum expander codes but remains an HGP complex. Going beyond ordinary HGP codes to obtain (almost) good codes within this framework requires a nontrivial \(\HH\)-lift. As discussed in Sec.~\ref{sec:covering}, the (co)boundary operator of a lifted graph is obtained from that of the base graph by replacing each nonzero entry with the corresponding permutation matrix. From a more algebraic perspective, the resulting operator uses entries from \(\HH\), or more rigorously, from its group algebra \(\mathbb{F}_q[\HH]\). The above construction of high-dimensional expanders can be recast as a hypergraph product over $\mathbb{F}_q[\HH]$, rather than over $\mathbb{F}_q$. This formulation is known as the \emph{lifted product}~\cite{PK2021,PK2022Good} and is discussed further in Sec.~\ref{sec:HGP_like}. 

\begin{figure}[t]
	\centering
	\resizebox{0.65\linewidth}{!}{%
		\begin{tikzpicture}[x=0.75pt,y=0.75pt,yscale=-1,xscale=1]
			\draw  [color={rgb, 255:red, 0; green, 0; blue, 0 }  ,draw opacity=0.5 ] (178,88) -- (255,11) -- (461.33,11) -- (461.33,190.67) -- (384.33,267.67) -- (178,267.67) -- cycle ; \draw  [color={rgb, 255:red, 0; green, 0; blue, 0 }  ,draw opacity=0.5 ] (461.33,11) -- (384.33,88) -- (178,88) ; \draw  [color={rgb, 255:red, 0; green, 0; blue, 0 }  ,draw opacity=0.5 ] (384.33,88) -- (384.33,267.67) ;
			\draw [color={rgb, 255:red, 0; green, 0; blue, 0 }  ,draw opacity=0.5 ] [dash pattern={on 4.5pt off 4.5pt}]  (255,11) -- (255,118.17) -- (255,191.67) ;
			\draw [color={rgb, 255:red, 0; green, 0; blue, 0 }  ,draw opacity=0.5 ] [dash pattern={on 4.5pt off 4.5pt}]  (255,190.67) -- (461.33,190.67) ;
			\draw [color={rgb, 255:red, 0; green, 0; blue, 0 }  ,draw opacity=0.5 ] [dash pattern={on 4.5pt off 4.5pt}]  (180.46,267.67) -- (225.21,222.92) -- (257.46,190.67) ;
			
			\draw (115,263) node [anchor=north west][inner sep=0.75pt]   [align=left] {$\displaystyle [ g;0,0,0]$};
			\draw (400,263) node [anchor=north west][inner sep=0.75pt]   [align=left] {$\displaystyle [ a_{1} \cdot g;1,0,0]$};
			\draw (394,85) node [anchor=north west][inner sep=0.75pt]   [align=left] {$\displaystyle [ a_{1} a_{3} \cdot g;1,0,1]$};
			\draw (473,181) node [anchor=north west][inner sep=0.75pt]   [align=left] {$\displaystyle [ a_{1} a_{2} \cdot g;1,1,0]$};
			\draw (467,2) node [anchor=north west][inner sep=0.75pt]   [align=left] {$\displaystyle [ a_{1} a_{2} a_{3} \cdot g;1,1,1]$};
			\draw (95,85) node [anchor=north west][inner sep=0.75pt]   [align=left] {$\displaystyle [ a_{3} \cdot g;0,0,1]$};
			\draw (180,181) node [anchor=north west][inner sep=0.75pt]   [align=left] {$\displaystyle [ a_{2} \cdot g;0,1,0]$};
			\draw (160,2) node [anchor=north west][inner sep=0.75pt]   [align=left] {$\displaystyle [ a_{2} a_{3} \cdot g;0,1,1]$};
			\draw (235,274) node [anchor=north west][inner sep=0.75pt]   [align=left] {$\displaystyle \textcolor[rgb]{0.82,0.01,0.11}{[ g;a_{1} ,0,0]}$};
			\draw (277,194) node [anchor=north west][inner sep=0.75pt]   [align=left] {$\displaystyle \textcolor[rgb]{0.82,0.01,0.11}{[a_{2} \cdot g;a_{1},1,0]}$};
			\draw (235,70) node [anchor=north west][inner sep=0.75pt]   [align=left] {$\displaystyle \textcolor[rgb]{0.82,0.01,0.11}{ [a_{3} \cdot g;a_{1},0,1]}$};
			\draw (277,-5) node [anchor=north west][inner sep=0.75pt]   [align=left] {$\displaystyle \textcolor[rgb]{0.82,0.01,0.11}{ [a_{2} a_{3} \cdot g;a_{1},1,1]}$};
			\draw (97,139) node [anchor=north west][inner sep=0.75pt]   [align=left] {$\displaystyle \textcolor[rgb]{0.96,0.65,0.14}{[ g;0,0,a_{3}]}$};
			\draw (326,139) node [anchor=north west][inner sep=0.75pt]   [align=left] {$\displaystyle \textcolor[rgb]{0.96,0.65,0.14}{[ a_{1} \cdot g;1,0,a_{3}]}$};
			\draw (205,139) node [anchor=north west][inner sep=0.75pt]   [align=left] {$\displaystyle \textcolor[rgb]{0.96,0.65,0.14}{[ a_{2} \cdot g;0,1,a_{3}]}$};
			\draw (465,139) node [anchor=north west][inner sep=0.75pt]   [align=left] {$\displaystyle \textcolor[rgb]{0.96,0.65,0.14}{[ a_{1} a_{2} \cdot g;1,1,a_{3}]}$};
			\draw (180,219) node [anchor=north west][inner sep=0.75pt]   [align=left] {$\displaystyle \textcolor[rgb]{0.29,0.56,0.89}{[ g;0,a_{2} ,0]}$};
			\draw (432,219) node [anchor=north west][inner sep=0.75pt]   [align=left] {$\displaystyle \textcolor[rgb]{0.29,0.56,0.89}{[ a_{1} \cdot g;1,a_{2} ,0]}$};
			\draw (160,41) node [anchor=north west][inner sep=0.75pt]   [align=left] {$\displaystyle \textcolor[rgb]{0.29,0.56,0.89}{[ a_{3} \cdot g;0,a_{2} ,1]}$};
			\draw (432,41) node [anchor=north west][inner sep=0.75pt]   [align=left] {$\displaystyle \textcolor[rgb]{0.29,0.56,0.89}{[ a_{1} a_{3} \cdot g;1,a_{2} ,1]}$};
		\end{tikzpicture}	
	}
    \caption{A $3$-cube $[g;a_1,a_2,a_3]$ with all its vertices ($0$-cube) and edges ($1$-cube). The faces ($2$-cube) are omitted for clarity.}
	\label{fig:$3$-cube}
\end{figure}

To formally define the inclusion relation needed to establish the high-dimensional (co)boundary operators, we can think of a $p$-cube $[g; (a_i)_{i \in S}, (b_j)_{j \in S^c} ]$ as a function $f$ with
\begin{align}
	f_0 = g, \quad  f_i = \begin{cases} a_i, & i \in S, \\ b_i, & i \in S^c. \end{cases}
\end{align}
If $f'$ is a $(p-1)$-cube, then it is contained in $f$ if the following conditions hold.
\begin{enumerate}
	\item[(a)] There exists some $j_0$ such that $f_{j_0}' \in \{0,1\}$ while $f_{j_0} \in A_{j_0}$.
	
	\item[(b)] $f_i' = f_i$ for all $i \neq 0, j_0$.
	
	\item[(c)] Finally,
	\begin{align}
		f_0' = f_{j_0}^{f_{j_0}'} \cdot f_0 = \begin{cases}	f_0, & f_{j_0}' = 0, \\ f_{j_0} \cdot f_0, & f_{j_0}' = 1. \end{cases}
	\end{align}
\end{enumerate} 
Then higher-order inclusion relations can be obtained recursively from this rule. 

We would like to clarify several subtle points concerning the cubical complex.

\begin{itemize}
\item \emph{The actions associated with different directions must commute.}
For the construction above with $t\geq3$, this commutativity is ensured by taking the lifting group $\HH$ to be Abelian. For example, a $2$-cube $[g;a_1,a_2,0]$ at the bottom of the $3$-cube in Fig.~\ref{fig:$3$-cube} has four vertices:
	\begin{align}
		[g; 0,0,0], \quad [a_1 \cdot g; 1,0,0 ], \quad [a_2 \cdot g; 0,1,0 ], \quad [a_1 a_2 \cdot g; 1,1,0 ] = [a_2 a_1 \cdot g; 1,1,0 ]
	\end{align}
	where the last equality expresses the required commutativity of the two directional actions. By Eq.~\eqref{eq:generator_action}, this equality holds when $\HH$ is Abelian. The special case $t=2$ where commuting left and right actions allow $\HH$ to be non-Abelian is discussed in Remark~\ref{remark:disconnect}. 
    \item \emph{The binary string $(b_j)$ is necessary.} Even without $(b_j)$, we can still define the product as a general HGP, but omitting it causes several problems. At the level of notation, the binary string makes a $1$-cube such as $[g;a_1,0,0]$ unambiguous. Without it, both $[g;a_1]$ and $[a_1\cdot g;a_1^{-1}]$ where $a_1^{-1}$ reverses the direction of the edge and acts as the inverse of $a_1$ in $\HH$ represent the same edge. For the analysis, if $g\neq g'$ have the same binary string $(b_j)_{j\in[t]}$, then no $t$-cube contains both of them (cf. Lemma~\ref{lemma:smallest_cube}). This property is also used to prove the local minimality of certain codewords when bounding the code distance~\cite{QuantumTanner2022,DHLV2022}.
    
    \item \emph{The one-dimensional structures of $X$ underlie its expansion properties.}
    For a fixed $j$, using only $a_j \in A_j$, Eq.~\eqref{eq:generator_action} defines the double cover of the $\HH$-lift of a single $G_0$. 
    
    Counting the choices of $v_i$ and $b_i$ with $i\neq j$, we find $1/r\coloneq(2\vert V_0\vert)^{t-1}$ copies of the double cover. Let $\hat{A}_j$ be the associated adjacency matrix. After reordering the vertices, it can be block diagonalized into $1/r$ equal-sized subblocks, each with $2l\vert V_0\vert=r\vert X(0)\vert$ vertices. If the double cover is connected, these subblocks are precisely its connected components.  Otherwise, each subblock may decompose further. Applying the expander mixing lemma to these subblocks is an intermediate step in evaluating the code distance in Refs.~\cite{QuantumTanner2022,DHLV2022,Dinur2024sheaf}. The number $1/r$ of subblocks in $\hat{A}_j$ gives rise to the prefactor $r$ in the distance bound.
\end{itemize}

\begin{remark}\label{remark:disconnect}
In the above construction of the cubical complex, $t\geq3$ and $n$ are fixed constants, while $n'$, $l$, $N$, and $1/r$ are allowed to grow. In Ref.~\cite{Dinur2024sheaf}, the code distance is proved to be $\Omega(N/\mathrm{polylog}\,N)$, which yields an almost-good qLDPC code family. This is achieved by taking an exponentially large Abelian group $\HH$ with $l=\vert\HH\vert=\exp(\Theta(n'))$~\cite{Agarwal2016,Jeronimo2021}. As a result,
$N=\exp(\Theta(n'))n'^t$ and
$r=\Theta((\log N)^{-(t-1)})$.
    Interestingly, when $t = 2$, there is the so-called \emph{left-right-Cayley complex} defined by the left and right actions of $\HH$:
	\begin{align}
		a_1 \cdot (h,v_1,v_2) \coloneq (\gamma_{(v_1,v_1')} \cdot h, v_1',v_2), \quad
		a_2 \cdot (h,v_1,v_2) \coloneq (h \cdot \gamma_{(v_2,v_2')}, v_1,v_2').
	\end{align}
    Then suitable non-Abelian lifts can be used in the left-right construction above to obtain two-dimensional expanders. Theorem~\ref{thm:l-covering} provides a general existence result for Ramanujan coverings without restricting the lifting permutations to an Abelian subgroup. This avoids the inverse-polylogarithmic loss associated with the prefactor $r$ and gives rise to good qLDPC codes~\cite{PK2022Good,QuantumTanner2022,DHLV2022}.
\end{remark}

However, constructing non-Clifford gates using cup products requires at least a three-dimensional structure (see Theorems~\ref{thm:cup_simplical} and~\ref{thm:cup_cubical}), for which only almost-good qLDPC codes are currently known, to the best of our knowledge.


\subsection{Local codes on cubical complexes}\label{sec:DLV_code}

The Dinur--Lin--Vidick code proposed in Ref.~\cite{Dinur2024sheaf} generalizes both the Sipser--Spielman construction, viewed as a 1-cochain complex with one local code, and recent good qLDPC constructions, viewed as 2-cochain complexes with two local codes~\cite{PK2022Good,QuantumTanner2022,DHLV2022}, to higher dimensions. It yields almost-good qLDPC codes for $t\geq3$ and almost-good qLTCs for $t\geq4$. The construction equips a $t$-dimensional cubical complex $X$ with $t$ local codes. We outline the construction, focusing primarily on how to read off its (co)boundary operators in order to establish logical gates and facilitate numerical computations.

We begin with a $t$-dimensional cubical complex $X$ defined in Sec.~\ref{sec:cubical} and determine its (co)boundary operators from the  inclusion relations illustrated in Fig.~\ref{fig:$3$-cube}. This gives rise to the following cochain complex with scalar coefficients:
\begin{equation}
	\begin{tikzcd}
		C^0(X,\mathbb{F}_q) \arrow[r, "\delta^{0}"] &
		C^1(X,\mathbb{F}_q) \arrow[r, "\delta^{1}"] &
		\cdots\cdots \arrow[r, "\delta^{t-2}"] &
		C^{t-1}(X,\mathbb{F}_q) \arrow[r, "\delta^{t-1}"] &
		C^{t}(X,\mathbb{F}_q).
	\end{tikzcd}
\end{equation}
Using the recipe in Sec.~\ref{sec:sheaf}, we replace the scalar entries in the coboundary operators with $\mathcal{F}_{\sigma,\tau}$. To define $\mathcal{F}$, we choose $t$ local codes $h_1 \in \mathbb{F}_q^{m_1 \times n}, \ldots,h_t \in \mathbb{F}_q^{m_t \times n}$. Here, $n$ is the vertex degree of the base graph $G_0$ underlying $X$, and $m_1, \ldots,m_t \leq n$.
It is also important to require each $h_i$ to have full rank. Together with suitable topological properties of $X$, this full-rank condition is among the hypotheses of the Poincar\'e duality theorem~\cite{Li2025Poincare}, which is used to bound the code distance~\cite{Dinur2024sheaf}. 

Returning to the definition of $\mathcal{F}$, we assign each $p$-cube $\sigma = [g; (a_i)_{i \in I}, (b_j)_{j \in J}]$ a vector space $\mathcal{F}(\sigma) = \bigotimes_{j \in J} \mathbb{F}_q^{m_j}$. Since $\sigma$ has $\vert J \vert = t - p$ binary bits, $\mathcal{F}(\sigma)$ is a $(t-p)$-fold tensor product space attached to these coordinates. Note that $\mathcal{F}(\pi) = \mathbb{F}_q$ is a scalar field when $\pi$ is a $t$-cube. To define the map $\mathcal{F}_{\sigma,\tau}$, we illustrate the three-dimensional case as in Fig.~\ref{fig:$3$-cube}:
\begin{enumerate}
	\item For $\theta = [a_3 \cdot g;0,0,1]$ and $\sigma = [g; 0,0,a_3]$, 
	\begin{align}
		\mathcal{F}_{\theta,\sigma} = I_{m_1} \otimes I_{m_2} \otimes h_3^T(a_3, \text{-}) \in \mathbb{F}_q^{m_1 \times m_1} \otimes \mathbb{F}_q^{m_2 \times m_2} \otimes \mathbb{F}_q^{1 \times m_3},
	\end{align} 
	where the $n$ elements $a_3 \in A_3$ label the columns of $h_3$, and $h_3^T(a_3,\text{-})$ denotes the transpose of the column labeled by $a_3$.
	This construction is the higher-dimensional analogue of the classical Tanner code construction in Remark~\ref{remark:SS_code}.
	
	\item For $\sigma = [g;0,0,a_3]$ and $\tau = [g; a_1,0,a_3]$, 
	\begin{align}
		\mathcal{F}_{\sigma,\tau} =  h_1^T(a_1, \text{-}) \otimes I_{m_2} \in \mathbb{F}_q^{1 \times m_1} \otimes \mathbb{F}_q^{m_2 \times m_2}.
	\end{align} 
	It is important to note that, analogously to the previous case, the local morphism from $[a_1 \cdot g;1,0,a_3]$ to $\tau$ is identical to $\mathcal{F}_{\sigma,\tau}$, rather than being defined using $h_1^T(a_1^{-1}, \text{-})$, although $a_1^{-1}$ is well defined on Cayley graphs. This convention makes it convenient to read off higher-dimensional cubes, since $a_1^{-1}$ does not appear in $\tau$. Moreover, for this definition to be well defined, we must ensure that there do not exist distinct vertices $u,v \in V(G_0)$ and a vertex $w\in V(G_0)$ such that the edges $(u,w)$ and $(v,w)$ carry the same label $\mu \in [n]$. If such vertices existed, we could let $g = (h,u,v_2,v_3)$ and $g' = (h,v,v_2,v_3)$, which would give
	\begin{align}
		[a_1^\mu \cdot g;1,0,a_3]
		&= [(\gamma_{(u,w)} \cdot h, w,v_2,v_3); 1,0,a_3], \\
		[a_1^\mu \cdot g';1,0,a_3]
		&= [(\gamma_{(v,w)} \cdot h, w,v_2,v_3); 1,0,a_3].
	\end{align}
	If $\gamma_{(u,w)} \cdot h = \gamma_{(v,w)} \cdot h$, for example when $\HH$ is trivial, the local morphisms from $[(\gamma_{(u,w)} \cdot h, w,v_2,v_3); 1,0,a_3]$ to higher-dimensional cubes would involve repeated rows of $h_1^T$. Such a collision cannot occur when $G_0$ is a Cayley graph or, more generally, a Schreier graph, because each label determines a unique predecessor of every vertex.
	
	\item For $\tau = [g;a_1,0,a_3]$ and $\pi = [g; a_1,a_2,a_3]$, 
	\begin{align}
		\mathcal{F}_{\tau,\pi} = h_2^T(a_2, \text{-}) \in \mathbb{F}_q^{1 \times m_2}.
	\end{align} 
\end{enumerate}
By Definition~\ref{def:presheaf}, $\mathcal{F}_{\sigma,\pi} = \mathcal{F}_{\tau,\pi} \circ \mathcal{F}_{\sigma,\tau}$. In the above example, $\sigma = [g;0,0,a_3]$ and $\tau = [g;a_1,0,a_3]$. Then
\begin{align}
	\F_{\sigma,\pi} = \big( h_2^T(a_2, \text{-} ) \big) \circ \big( h_1^T(a_1, \text{-} ) \otimes I_{m_2} \big) = h_1^T(a_1, \text{-} ) \otimes h_2^T(a_2, \text{-} ).
\end{align}
One can take $\tau' = [g;0,a_2,a_3] \succ \sigma$. Direct composition gives $\mathcal{F}_{\tau',\pi} \circ \mathcal{F}_{\sigma,\tau'} = \mathcal{F}_{\sigma,\pi}$, so the comparability conditions in Definition~\ref{def:presheaf} hold. Higher-dimensional cases are defined accordingly (see also Sec.~\ref{sec:example}).


\subsection{Cup products and cap products}\label{sec:cup}

In algebraic topology, the cup and cap products are defined on simplicial complexes~\cite{Hatcher2015AT,Gallier2022}. We summarize some of their basic properties below and then elucidate how to extend these notions to cubical complexes based on previous work such as Refs.~\cite{Chen_2023,Lin2024transversal,Li2025Poincare}.

\begin{definition}[Cup product] \label{def:cup}
	Let $X$ be a simplicial complex. For cochain complexes $C^\bullet(X,\mathcal F)$ and $C^\bullet(X,\mathcal G)$, the \emph{cup product} is the following map:
	\begin{align}
		\smile: C^{p_1}(X,\mathcal{F}) \times C^{p_2}(X,\mathcal{G}) \rightarrow C^{p_1 + p_2}(X,\mathcal{F} \otimes \mathcal{G})
	\end{align}
	To define this map, let $\sigma$ be a $(p_1 + p_2)$-simplex with vertices ordered as $v_0, v_1, \ldots,v_{p_1 + p_2}$ and let ${}_{p_1} \sigma, \sigma_{p_2}$ be the $p_1$- and $p_2$-simplices defined by the first $p_1 + 1$ vertices $v_0,v_1, \ldots,v_{p_1}$ and the last $p_2 + 1$ vertices $v_{p_1},v_{p_1 + 1}, \ldots,v_{p_1 + p_2}$, respectively. Then given any $x_1 \in C^{p_1}(X,\mathcal{F})$ and $x_2 \in C^{p_2}(X,\mathcal{G})$,
	\begin{align}
		(x_1 \smile x_2)(\sigma) = \mathcal{F}_{{}_{p_1}\sigma, \sigma} ( x_1( {}_{p_1}\sigma ) ) \otimes \mathcal{G}_{\sigma_{p_2}, \sigma}( x_2( \sigma_{p_2} ) ).   
	\end{align}
\end{definition}

\begin{proposition}\label{prop:cup}
	The following properties hold for the cup product:
	\begin{enumerate}
		\item It satisfies the Leibniz rule:
		\begin{align}
			\delta(x_1 \smile x_2) = (\delta x_1) \smile x_2 + x_1 \smile (\delta x_2). 
		\end{align}
		
		\item It is associative:
		\begin{align}
			(x_1 \smile x_2) \smile x_3 = x_1 \smile (x_2 \smile x_3)   
		\end{align}
		for arbitrary $x_1 \in C^{p_1}(X, \mathcal{F})$, $x_2 \in C^{p_2}(X, \mathcal{G})$ and $x_3 \in C^{p_3}(X, \mathcal{H})$.
	\end{enumerate}
\end{proposition}

\begin{definition}[Cap product] \label{def:cap}
	Let $X$ be a simplicial complex. The \emph{cap product}
	\begin{align}
		\frown: C^{p_1}(X,\mathcal{F}) \times C_{p_1 + p_2}(X,\mathcal{F}) \rightarrow C_{p_2}(X,\mathbb{F}_q)
	\end{align}
	is defined as follows. For any $x \in C^{p_1}(X,\mathcal{F})$ and $\xi \in C_{p_1 + p_2}(X,\mathcal{F})$, we set
	\begin{align}
		x \frown \xi = \sum_{\sigma \in X(p_1 + p_2)} \Big\langle \mathcal{F}_{{}_{p_1}\sigma, \sigma} (x({}_{p_1}\sigma)),  \xi(\sigma) \Big\rangle \sigma_{p_2}.
	\end{align}
\end{definition}

\begin{proposition}\label{prop:cap_Leibniz}
	The cap product satisfies the Leibniz rule:
	\begin{align}
		\partial(x \frown \xi) = (\delta x) \frown \xi + x \frown (\partial \xi),
	\end{align}
	where $\xi \in C_{p_1 + p_2 + 1}(X, \mathcal{F})$.
\end{proposition}

We now consider the special case of the cap product when $p_2 = 0$. In this case, $\mathcal{F}_{{}_{p_1}\sigma, \sigma}$ in Definition~\ref{def:cap} is simply the identity matrix and $x \frown \xi \in C_0(X,\mathbb{F}_q)$. We can further sum over the components of $x \frown \xi$ to obtain a scalar in $\mathbb{F}_q$. Formally, it is called a \emph{pairing} and is defined by  
\begin{align}\label{eq:paring}
	\langle x, \ \xi \rangle \coloneq \sum_{\sigma \in X(p_1)} \left\langle x(\sigma), \ \xi(\sigma) \right\rangle.
\end{align}
Suppose $\xi \in C_{p_1}(X,\mathcal{F})$ is a cycle. For any $y \in C^{p_1 - 1}(X,\mathcal{F})$, we have $\langle \delta y, \xi \rangle = \langle y, \partial \xi \rangle = 0$, which validates the following theorem.

\begin{theorem}[\cite{Li2025Poincare}]\label{thm:cup_simplical}
	If $p_1 + \cdots + p_\rho \leq t$ with $t = \dim X$ and $\xi \in C_{p_1 + \cdots + p_\rho}(X,\mathcal{F}^{\otimes \rho})$ is a cycle, then
	\begin{align}\label{eq:circuit1}
		T_\xi(x_{p_1}, \ldots,x_{p_\rho}) \coloneq \langle x_{p_1} \smile \cdots \smile x_{p_\rho}, \ \xi \rangle
	\end{align}
	defines a cohomological invariant from $H^{p_1}(X,\mathcal{F}) \times \cdots \times H^{p_\rho}(X,\mathcal{F})$ to $\mathbb{F}_q$. In particular, if the coboundary operators of $X$ are sparse, that is, each simplex is contained in a constant number of higher-dimensional simplices, with $t$ being a constant, the corresponding invariant polynomial defines a constant-depth multi-controlled-$Z$ gate in the sense of Definition~\ref{def:invariant_poly}.
\end{theorem}


The fundamental idea behind extending cup and cap products to a cubical complex $X$ is to subdivide $X$ into a simplicial complex $\tilde{X}$. The products are then defined on $\tilde{X}$ and pulled back to $X$. Intuitively, drawing a diagonal between two opposite vertices of a square divides it into two triangles (Fig.~\ref{fig:$2$-cube}). We make this construction precise in arbitrary dimension following the methods of Refs.~\cite{FreedmanHastings2021,Portnoy2023,Lin2024transversal,Li2025Poincare}. Concrete examples in 2D and 3D are also given in Sec.~\ref{sec:example}.

\begin{definition}[Cellular map]
	Let $f: X \rightarrow Y$ be a continuous map between two CW complexes $X,Y$. It is called a \emph{cellular map} if it maps the $p$-skeleton $X^p$ of $X$ to that of $Y$ for every $p \geq 0$, i.e., $f(X^p) \subset Y^p$. 
\end{definition}

The notion of a CW complex is fundamental in the study of algebraic topology~\cite{Hatcher2015AT,Gallier2022}. In our setting, one of $X$ and $Y$ is a cubical complex and the other is its subdivision, so we do not require the full formal definition. A $p$-skeleton is the union of all cubes or simplices of dimension at most $p$. Although we do not explicitly specify the underlying topologies, continuity is understood with respect to the natural CW topologies, and the explicit constructions in Sec.~\ref{sec:example} show that the cellular maps in Eqs.~\eqref{eq:S_ast} and~\eqref{eq:A_ast} are well defined.

Let $X$ be a $t$-dimensional cubical complex. After choosing a subdivision, let $\tilde{X}$ denote the resulting simplicial complex. Let $S: X \rightarrow \tilde{X}$ be the identity map on the underlying topological space. By definition, $S$ is a cellular map. 
Let $S(\sigma)$ be the collection of simplices in $\tilde{X}$ contained in the image of $\sigma$. Then we can define $S_p: C_p(X,\mathbb{F}_q) \rightarrow C_p(\tilde{X},\mathbb{F}_q)$ by
\begin{align}
	S_p( \sigma ) = \sum_{\tilde{\sigma} \in S(\sigma) \cap \tilde{X}(p) } \tilde{\sigma}.
\end{align}
The definition of $S_p$ on subdivided cells shows that $S_p$ commutes with the boundary operator (see Sec.~\ref{sec:example}). Its transpose $S^p = (S_p)^T: C^p(\tilde{X},\mathbb{F}_q) \to C^p(X,\mathbb{F}_q)$
defines a cochain map
\begin{align}\label{eq:S_ast_F}
	S^p( \tilde{\sigma} ) = \sum_{\sigma \in X(p), \tilde{\sigma} \in S(\sigma)} \sigma.
\end{align}
When each cube is subdivided into a constant number of simplices, the rows of matrix representation of $S^p$ have constant weights. Each column corresponds to a $p$-simplex contained in a unique minimal cube of dimension at least $p$ (see Fig.~\ref{fig:$2$-cube}). Therefore, if $\dim X=t$ is constant, $S^p$ is sparse. This sparsity is important for establishing that the cup product gate in Theorem~\ref{thm:cup_cubical} has constant depth. Collectively, we denote these families of (co)chain maps by $S_\#$ and $S^\#$. When the degree is not specified, the membership conditions $\tilde{\sigma}\in\tilde{X}(p)$ and $\sigma\in X(p)$ are omitted. Given any $\tilde{\sigma} \in \tilde{X}(i)$, let $\tau_{\tilde{\sigma}}$ be the minimal cube in $X$ that contains $\tilde{\sigma}$. Note that $\dim \tau_{\tilde{\sigma}} \geq i$ because $\tilde{\sigma}$ may lie in the subdivision of an $i$-cube or of a higher-dimensional cube. Because distinct cubes of the same dimension intersect, if at all, in a lower-dimensional cube, the minimal cube $\tau_{\tilde{\sigma}}$ is unique. A more formal proof is given after Lemma~\ref{lemma:smallest_cube}.

In addition to $S$, we define another cellular map $A:\tilde{X}\rightarrow X$ such that, for every $\tilde{\sigma}$ and every $\sigma\in A(\tilde{\sigma})$, $\sigma\subset\tau_{\tilde{\sigma}}$. Analogously to $S(\sigma)$, $A(\tilde{\sigma})$ is the collection of all cubes appearing in the image of $\tilde{\sigma}$. Intuitively, $A$ maps $\tilde{\sigma}$ to cubes that remain close to it. We call $A$ a \emph{cellular approximation map}. The induced maps $A_\#$ and $A^\#$ are defined by
\begin{align}\label{eq:A_ast_F}
	A_p(\tilde{\sigma}) = \sum_{\sigma\in A(\tilde{\sigma})\cap X(p)} \sigma, \quad
	A^p(\sigma) = \sum_{\substack{\tilde{\sigma}\in\tilde{X}(p)\\ \sigma\in A(\tilde{\sigma})}} \tilde{\sigma}.
\end{align}
To verify the sparsity of $A^p$, observe that $\tilde{\sigma}$ can appear in $A^p(\sigma)$ only if $\sigma\subset\tau_{\tilde{\sigma}}$. If the coboundary operator of $X$ is sparse and $\dim X=t$ is constant, the matrix representing $A^p$ has constant-weight columns. Moreover, for each fixed $\tilde{\sigma}$, the cube $\tau_{\tilde{\sigma}}$ contains only a constant number of cubes, so this matrix also has constant-weight rows.

If $A$ were also the identity map, it would not be a cellular approximation map in the example above. Indeed, it would not even be a cellular map because it would map the diagonal edge (a $1$-simplex) into the interior of the square, which is not in $X^1$. We illustrate a suitable definition of $A$ in Figs.~\ref{fig:$2$-cube} and~\ref{fig:$3$-cube_cup}.


We now outline the method for equipping $A^\#$ and $S^\#$ with local coefficients $\F$, in order to define cup products on $C^\bullet(X,\F)$. We construct a cellular sheaf $\tilde{\mathcal{F}}$ together with the following cochain maps, which we continue to denote by $S^\#$ and $A^\#$:
\begin{align}
	S^\#: C^\bullet(\tilde{X},\tilde{\mathcal{F}}) \rightarrow C^\bullet(X,\mathcal{F}), \quad
	A^\#: C^\bullet(X,\mathcal{F}) \rightarrow C^\bullet(\tilde{X},\tilde{\mathcal{F}}). 
\end{align}
We emphasize that the maps $S^\#$ and $A^\#$ preserve cochain degree. This notation abbreviates the corresponding degreewise cochain maps from Definition~\ref{def:chain_map}.
The following results can be verified per degree, while it is standard to study the complex as a whole in algebraic topology. 
Taking tensor products of two local coefficient systems likewise defines the following cochain map (see Eq.~\eqref{eq:S_ast_tensor}):
\begin{align}
	S^\#: C^\bullet(\tilde{X},\tilde{\mathcal{F}} \otimes \tilde{\mathcal{G}} ) \rightarrow C^\bullet(X,\mathcal{F} \otimes \mathcal{G} ).
\end{align}
Then we obtain the cup product on the cubical complex:
\begin{align}\label{eq:cup_cell}
	\begin{aligned}
		\smile_c: C^{p_1}(X,\mathcal{F}) \times C^{p_2}(X,\mathcal{G}) \xrightarrow{A^\# \times A^\#} 
		& C^{p_1}(\tilde{X}, \tilde{\mathcal{F}}) \times C^{p_2}(\tilde{X}, \tilde{\mathcal{G}} ) \\
		& \xrightarrow{\smile}   
		C^{p_1 + p_2}(\tilde{X}, \tilde{\mathcal{F}} \otimes \tilde{\mathcal{G}}) \xrightarrow{S^\#}
		C^{p_1 + p_2}(X,\mathcal{F} \otimes \mathcal{G}).
	\end{aligned}
\end{align}
The Leibniz rule follows because $S^\#$ and $A^\#$ commute with the corresponding coboundary operators:
\begin{align}
	\begin{aligned}
		\delta(\alpha \smile_c \beta) 
		= & \delta S^\# ( (A^\# \alpha) \smile (A^\# \beta) )
		= S^\# \delta ( (A^\# \alpha) \smile (A^\# \beta) ) \\
		= & S^\# ( \delta (A^\# \alpha) \smile (A^\# \beta) +  (A^\# \alpha) \smile \delta (A^\# \beta) ) 
		= \delta \alpha \smile_c \beta + \alpha \smile_c \delta \beta.
	\end{aligned} 
\end{align}
The cap product is defined as follows:
\begin{align}
	\begin{aligned}
		\frown_c: C^{p_1}(X,\mathcal{F}) \times C_{p_1 + p_2}(X,\mathcal{F}) \xrightarrow{A^\# \times S_\#} 
		& C^{p_1}(\tilde{X}, \tilde{\mathcal{F}} ) \times C_{p_1 + p_2}(\tilde{X}, \tilde{\mathcal{F}} ) \\
		& \xrightarrow{\frown}   
		C_{p_2}(\tilde{X}, \mathbb{F}_q) \xrightarrow{A_\#}
		C_{p_2}(X,\mathbb{F}_q).
	\end{aligned}
\end{align}
It satisfies the Leibniz rule:
\begin{align}
	\begin{aligned}
		\partial(\alpha \frown_c x) 
		= & \partial A_\#( (A^\# \alpha) \frown (S_\# x) ) 
		= A_\# \partial( (A^\# \alpha) \frown (S_\# x) ) \\ 
		= & A_\# ( \delta (A^\# \alpha) \frown (S_\# x) +  (A^\# \alpha) \frown \partial (S_\# x)) 
		= \delta \alpha \frown_c x + \alpha \frown_c \partial x.
	\end{aligned}
\end{align}

It is reasonable to expect that, as long as the cubical complex is sparse, i.e., its coboundary operators are sparse, and a suitable subdivision is chosen, $S_\#$, $S^\#$, $A_\#$, $A^\#$, as well as the coboundary operators of $\tilde{X}$, are sparse, and the following theorem holds as a counterpart to Theorem~\ref{thm:cup_simplical}.

\begin{theorem}[\cite{Li2025Poincare}]\label{thm:cup_cubical}
	Let $X$ be a cubical complex. If $p_1 + \cdots + p_\rho \leq t$ with $t = \dim X$ and $\xi \in C_{p_1 + \cdots + p_\rho}(X,\mathcal{F}^{\otimes \rho})$ is a cycle, then
	\begin{align}\label{eq:circuit2}
		T_\xi (x_{p_1}, \ldots,x_{p_\rho}) \coloneq \langle x_{p_1} \smile_c \cdots \smile_c x_{p_\rho}, \ \xi \rangle
	\end{align}
	defines a cohomological invariant from $H^{p_1}(X,\mathcal{F}) \times \cdots \times H^{p_\rho}(X,\mathcal{F})$ to $\mathbb{F}_q$. With the above assumptions, the corresponding invariant polynomial defines a constant-depth multi-controlled-$Z$ gate in the sense of Definition~\ref{def:invariant_poly}.
\end{theorem}


We now give the details needed to define $\tilde{\F}$ and the (co)chain maps and to justify the subsequent analysis. Since every $\tilde{\sigma}$ corresponds to a unique $\tau_{\tilde{\sigma}}$, we simply assign the local coefficient space $\F(\tau_{\tilde{\sigma}})$ to $\tilde{\sigma}$. Moreover, $\tilde{\sigma} \subset \tilde{\tau} \implies \tau_{\tilde{\sigma}} \subset \tau_{\tilde{\tau}}$ and thus we define
\begin{align}
	\tilde{\F}_{\tilde{\sigma}, \tilde{\tau}} \coloneq \F_{\tau_{\tilde{\sigma}}, \tau_{\tilde{\tau}}}: \F(\tau_{\tilde{\sigma}}) \rightarrow  \F(\tau_{\tilde{\tau}}).
\end{align}
This defines $\tilde{\F}$ for $\tilde{X}$, and its functoriality condition in Definition~\ref{def:presheaf} is inherited from $\F$. Formally, the assignment $\mathcal{M}: \tilde{X} \to X$, $\mathcal{M}(\tilde{\sigma}) = \tau_{\tilde{\sigma}}$, is order preserving, and the definition above identifies $\tilde{\F}$ with the pullback $\mathcal{M}^*\F$. With these identifications, we write
\begin{align}\label{eq:natural}
	\mathcal{M}_{\tilde{\sigma}, \tau_{\tilde{\sigma}} } = \mathrm{id}: \tilde{\F}(\tilde{\sigma}) \to \F(\tau_{\tilde{\sigma}}). 
\end{align} 
Different simplices may be contained in the same minimal cube, but $\mathcal{M}_{\tilde{\sigma},\tau_{\tilde{\sigma}}}$ always acts as the identity on the corresponding local vectors. Then one can easily construct $C^\bullet(\tilde{X}, \tilde{\F})$ and define $S^\#: C^\bullet(\tilde{X},\tilde{\F}) \rightarrow C^\bullet(X,\mathcal{F})$ by
\begin{align}\label{eq:S_ast}
	S^\# ( \tilde{x}(\tilde{\sigma}) \tilde{\sigma}) \coloneq \sum_{\substack{\sigma\\ \tilde{\sigma}\in S(\sigma)}} \Big[ \mathcal{F}_{\tau_{\tilde{\sigma}},\sigma} \mathcal{M}_{\tilde{\sigma}, \tau_{\tilde{\sigma}} } \tilde{x}(\tilde{\sigma}) \Big] \sigma, 
\end{align}
where $\tilde{x}(\tilde{\sigma})\tilde{\sigma}$ denotes the cochain supported on $\tilde{\sigma}$ with local vector $\tilde{x}(\tilde{\sigma})$. On the other hand,
\begin{align}\label{eq:A_ast}
	A^\# (x(\sigma) \sigma) \coloneq \sum_{\substack{\tilde{\sigma}\\ \sigma\in A(\tilde{\sigma})}} \Big[ \mathcal{M}_{\tilde{\sigma}, \tau_{\tilde{\sigma}} }^{-1}  \mathcal{F}_{\sigma,\tau_{\tilde{\sigma}}} x(\sigma) \Big] \tilde{\sigma}.
\end{align}
For a different cellular sheaf $\mathcal{G}$, let $\mathcal{N}_{\tilde{\sigma}, \tau_{\tilde{\sigma}} }$ be defined by the identity matrix on the corresponding local coefficient spaces as in Eq.~\eqref{eq:natural}. Then
\begin{align}\label{eq:S_ast_tensor}
	S^\# ( \tilde{x}(\tilde{\sigma}) \tilde{\sigma}) \coloneq \sum_{\substack{\sigma\\ \tilde{\sigma}\in S(\sigma)}} \Big[ \Big( \mathcal{F}_{\tau_{\tilde{\sigma}},\sigma} \mathcal{M}_{\tilde{\sigma}, \tau_{\tilde{\sigma}} }  \otimes  \mathcal{G}_{\tau_{\tilde{\sigma}},\sigma} \mathcal{N}_{\tilde{\sigma}, \tau_{\tilde{\sigma}} }  \Big) \tilde{x}(\tilde{\sigma}) \Big] \sigma.
\end{align}
For a more comprehensive definition of $S^\#$ and $A^\#$ via pullback cellular sheaves, we refer the reader to~\cite{Li2025Poincare}. For explicit computations, given $x_1(\sigma) \sigma$ and $x_2(\tau)\tau$, their cup product is obtained by the following steps:
\begin{enumerate}
	\item We find $\sigma \smile_c \tau = \omega$ by Eqs.~\eqref{eq:S_ast_F} and~\eqref{eq:A_ast_F} without local coefficients. We illustrate 2D and 3D examples in Sec.~\ref{sec:example} (see also Lemma~\ref{lemma:cup_induct_1}).
	
	\item Applying Eqs.~\eqref{eq:cup_cell}, \eqref{eq:A_ast}, and~\eqref{eq:S_ast_tensor} to the local vectors, we obtain
	\begin{align}\label{eq:cup_local_vector}
	\begin{aligned}
	& \Big[ \Big( \mathcal{F}_{\tau_{\tilde{\omega}},\omega} \mathcal{M}_{\tilde{\omega},\tau_{\tilde{\omega}}} \Big) \otimes
	\Big( \mathcal{G}_{\tau_{\tilde{\omega}},\omega}
	\mathcal{N}_{\tilde{\omega},\tau_{\tilde{\omega}}} \Big)
	\Big]\Big(\tilde{\mathcal{F}}_{\tilde{\sigma},\tilde{\omega}} \otimes
	\tilde{\mathcal{G}}_{\tilde{\tau},\tilde{\omega}} \Big)
	\Big[ \Big( \mathcal{M}_{\tilde{\sigma},\tau_{\tilde{\sigma}}}^{-1}\mathcal{F}_{\sigma,\tau_{\tilde{\sigma}}} x_1(\sigma) \Big) \otimes
	\Big( \mathcal{N}_{\tilde{\tau},\tau_{\tilde{\tau}}}^{-1}\mathcal{G}_{\tau,\tau_{\tilde{\tau}}}x_2(\tau) \Big) \Big] \\
	= {} & \Big( \mathcal{F}_{\sigma,\omega}x_1(\sigma) \Big) \otimes \Big( \mathcal{G}_{\tau,\omega}x_2(\tau) \Big).
	\end{aligned}
	\end{align}
\end{enumerate}
The calculation of consecutive cup products obeys the same principle. 


\subsection{Explicit computations in 2D and 3D}\label{sec:example}

We provide explicit formulas for computing cup products on 2D and 3D cubical complexes.

\begin{figure}[t]
	\centering
	\resizebox{0.96\linewidth}{!}{%
		\begin{tikzpicture}[x=0.75pt,y=0.75pt,yscale=-1,xscale=1]
			\draw  [color={rgb, 255:red, 0; green, 0; blue, 0 }  ,draw opacity=0.5 ] (49,37) -- (244.33,37) -- (244.33,232.33) -- (49,232.33) -- cycle ;
			\draw  [color={rgb, 255:red, 0; green, 0; blue, 0 }  ,draw opacity=0.5 ] (339,37) -- (534.33,37) -- (534.33,232.33) -- (339,232.33) -- cycle ;
			\draw  [color={rgb, 255:red, 0; green, 0; blue, 0 }  ,draw opacity=0.5 ] (629,37) -- (824.33,37) -- (824.33,232.33) -- (629,232.33) -- cycle ;
			\draw    (253,135) -- (325,135) ;
			\draw [shift={(327,135)}, rotate = 180] [color={rgb, 255:red, 0; green, 0; blue, 0 }  ][line width=0.75]    (10.93,-3.29) .. controls (6.95,-1.4) and (3.31,-0.3) .. (0,0) .. controls (3.31,0.3) and (6.95,1.4) .. (10.93,3.29)   ;
			\draw    (543,135) -- (615,135) ;
			\draw [shift={(617,135)}, rotate = 180] [color={rgb, 255:red, 0; green, 0; blue, 0 }  ][line width=0.75]    (10.93,-3.29) .. controls (6.95,-1.4) and (3.31,-0.3) .. (0,0) .. controls (3.31,0.3) and (6.95,1.4) .. (10.93,3.29)   ;
			\draw [color={rgb, 255:red, 245; green, 166; blue, 35 }  ,draw opacity=1 ]   (339,232.33) -- (534.33,37) ;
			\draw [color={rgb, 255:red, 245; green, 166; blue, 35 }  ,draw opacity=1 ] [dash pattern={on 4.5pt off 4.5pt}]  (629,232.33) .. controls (740,215.83) and (785,177.83) .. (824.33,37) ;
			\draw  [color={rgb, 255:red, 245; green, 166; blue, 35 }  ,draw opacity=1 ] (747.17,153.53) -- (772.6,178.73) -- (777.65,173.63) -- (784.49,200.64) -- (757.43,194.04) -- (762.48,188.94) -- (737.06,163.73) -- cycle ;
			
			\draw (25,236) node [anchor=north west][inner sep=0.75pt]   [align=left] {$\displaystyle [ g;0,0]$};
			\draw (197,236) node [anchor=north west][inner sep=0.75pt]   [align=left] {$\displaystyle [ a_{1} \cdot g;1,0]$};
			\draw (186,17) node [anchor=north west][inner sep=0.75pt]   [align=left] {$\displaystyle [ a_{1} a_{2} \cdot g;1,1]$};
			\draw (25,17) node [anchor=north west][inner sep=0.75pt]   [align=left] {$\displaystyle [ a_{2} \cdot g;0,1]$};
			\draw (315,236) node [anchor=north west][inner sep=0.75pt]   [align=left] {$\displaystyle [ g;0,0]$};
			\draw (487,236) node [anchor=north west][inner sep=0.75pt]   [align=left] {$\displaystyle [ a_{1} \cdot g;1,0]$};
			\draw (476,17) node [anchor=north west][inner sep=0.75pt]   [align=left] {$\displaystyle [ a_{1} a_{2} \cdot g;1,1]$};
			\draw (315,17) node [anchor=north west][inner sep=0.75pt]   [align=left] {$\displaystyle [ a_{2} \cdot g;0,1]$};
			\draw (605,236) node [anchor=north west][inner sep=0.75pt]   [align=left] {$\displaystyle [ g;0,0]$};
			\draw (777,236) node [anchor=north west][inner sep=0.75pt]   [align=left] {$\displaystyle [ a_{1} \cdot g;1,0]$};
			\draw (766,17) node [anchor=north west][inner sep=0.75pt]   [align=left] {$\displaystyle [ a_{1} a_{2} \cdot g;1,1]$};
			\draw (604,17) node [anchor=north west][inner sep=0.75pt]   [align=left] {$\displaystyle [ a_{2} \cdot g;0,1]$};
			\draw (275,113) node [anchor=north west][inner sep=0.75pt]   [align=left] {$\displaystyle S$};
			\draw (575,113) node [anchor=north west][inner sep=0.75pt]   [align=left] {$\displaystyle A$};		
		\end{tikzpicture}
	}
	\caption{The maps $S$ and $A$ on the subdivision of a $2$-cube $[g;a_1,a_2]$.}
	\label{fig:$2$-cube}
\end{figure}

\begin{example}\label{example:2D}
	Let $X$ be a 2D cubical complex generated by $\{\G, A_1, A_2\}$ as in Sec.~\ref{sec:cubical}. We triangulate any $2$-cube $[g;a_1,a_2]$ by adding a segment from $[g;0,0]$ to $[a_1 a_2 \cdot g; 1,1]$. This is different from the barycentric subdivision, but the construction extends directly to higher dimensions. The two subdivided 2D simplices are specified by their vertices:
	\begin{align}\label{eq:subdivision_2D}
		\Big[ [g; 0,0], [a_1 \cdot g; 1,0], [a_1 a_2 \cdot g; 1,1] \Big] , \quad 
		\Big[ [g; 0,0], [a_2 \cdot g; 0,1], [a_1 a_2 \cdot g; 1,1] \Big] 
	\end{align}
    In the present case, except for the edge joining $[g;0,0]$ and $[a_1a_2\cdot g;1,1]$, all lower-dimensional simplices are already contained in the original 2D cubical complex, and we use the same notation for them.
  
    The cellular approximation $A$ is defined by collapsing the 2D simplex on the left-hand side of Eq.~\eqref{eq:subdivision_2D} to the $1$-cube $[g;a_1,0]$ and $[a_1\cdot g;1,a_2]$. To ensure continuity, the other $2$-simplex is stretched to cover the entire $2$-cube; see Fig.~\ref{fig:$2$-cube}.
	The map $A$ is well defined over the entire subdivision $\tilde{X}$ because the collapsing and stretching maps are defined locally within each $2$-cube. Higher-dimensional cases require a few additional details, which we provide below. By definition,
	\begin{align}
		& A_\# [g; a_1,0] = [g; a_1,0], \\
		& A_\# [a_1 \cdot g; 1, a_2] = [a_1 \cdot g; 1, a_2], \\
        & A_\# \Big[ [g;0,0], [a_1 a_2 \cdot g; 1,1] \Big] = [g;a_1,0]+[a_1\cdot g;1,a_2].
	\end{align}
    The actions on the other simplices are defined analogously, and direct computation shows that $A_\#$ commutes with the boundary map. As noted above Eq.~\eqref{eq:cup_cell}, $A_\#$ and $A^\#$ are (co)chain maps that include mappings on every degree of the complex. Below we restrict these maps to a fixed degree. Alternatively, one may treat all degrees simultaneously, with the (co)boundary operators acting on the full complexes.
		
	We now calculate the cup product of $[g;a_1,0]$ and $[a_1\cdot g;1,a_2]$. By Eq.~\eqref{eq:A_ast_F},
	\begin{align}\label{eq:A_ast_example}
		\begin{aligned}
			A^\# [g; a_1,0] = A_\#^T [g; a_1,0] 
			= [g; a_1,0] + \Big[ [g;0,0], [a_1 a_2 \cdot g; 1,1] \Big] + \cdots.
		\end{aligned}
	\end{align}
	Here we note that $[g; a_1,0]$ can be contained in different $2$-cube for different choices of $a_2$, but they are irrelevant here because we only take cup products of $1$-simplices.
    We have a similar expansion for $A^\# [a_1 \cdot g; 1, a_2]$ and thus
	\begin{align}
		( A^1 [g; a_1,0] ) \smile ( A^1 [a_1 \cdot g; 1,a_2] ) = \Big[ [g; 0,0], [a_1 \cdot g; 1,0], [a_1 a_2 \cdot g; 1,1] \Big].
	\end{align} 
	As a reminder, by Definition~\ref{def:cup}, we have to order the vertices of the $2$-simplex before taking the cup product, and the natural order is from left to right. As another useful example,
	\begin{align}
		( A^1 [g; a_1,0] ) \smile ( A^1 [g; a_1,1] ) 
		= ( [g; a_1,0] +  \Big[ [g;0,0], [a_1 a_2 \cdot g; 1,1] \Big] ) \smile [g; a_1,1] = 0
	\end{align}
	because the cup product between the diagonal segment and $[g; a_1,1]$ is zero because they do not satisfy the required ordering.
	
	By Eq.~\eqref{eq:S_ast_F}, 
	\begin{align}
		& S^2 \Big[ [g; 0,0], [a_1 \cdot g; 1,0], [a_1 a_2 \cdot g; 1,1] \Big] = S_2^T \Big[ [g; 0,0], [a_1 \cdot g; 1,0], [a_1 a_2 \cdot g; 1,1] \Big] = [g;a_1,a_2] \\
		& \implies [g; a_1,0] \smile_c [a_1 \cdot g; 1,a_2] = [g;a_1,a_2].
	\end{align}
	By the same method, we have
	\begin{align}
		[g; 0,a_2] \smile_c [a_2 \cdot g; a_1,1] = [g;a_1,a_2].
	\end{align} 
	These results are also consistent with Lemma~\ref{lemma:cup_induct_1}. We show the computation with local coefficients in the next example.
\end{example}

\begin{figure}[t]
	\centering
	\resizebox{0.5\linewidth}{!}{%
		\begin{tikzpicture}[x=0.75pt,y=0.75pt,yscale=-1,xscale=1]
			\draw  [color={rgb, 255:red, 0; green, 0; blue, 0 }  ,draw opacity=1 ] (178,88) -- (255,11) -- (461.33,11) -- (461.33,190.67) -- (384.33,267.67) -- (178,267.67) -- cycle ; \draw  [color={rgb, 255:red, 0; green, 0; blue, 0 }  ,draw opacity=1 ] (461.33,11) -- (384.33,88) -- (178,88) ; \draw  [color={rgb, 255:red, 0; green, 0; blue, 0 }  ,draw opacity=1 ] (384.33,88) -- (384.33,267.67) ;
			\draw [color={rgb, 255:red, 0; green, 0; blue, 0 }  ,draw opacity=0.5 ] [dash pattern={on 4.5pt off 4.5pt}]  (255,11) -- (255,118.17) -- (255,191.67) ;
			\draw [color={rgb, 255:red, 0; green, 0; blue, 0 }  ,draw opacity=0.5 ] [dash pattern={on 4.5pt off 4.5pt}]  (255,190.67) -- (461.33,190.67) ;
			\draw [color={rgb, 255:red, 0; green, 0; blue, 0 }  ,draw opacity=0.5 ] [dash pattern={on 4.5pt off 4.5pt}]  (180.46,267.67) -- (225.21,222.92) -- (257.46,190.67) ;
			\draw [color={rgb, 255:red, 0; green, 0; blue, 0 }  ,draw opacity=0.6 ] [dash pattern={on 4.5pt off 4.5pt}]  (178,267.67) .. controls (179.72,218.86) and (212.09,123.05) .. (264,77.17) .. controls (315.91,31.29) and (428.65,11.67) .. (461.33,11) ;
			\draw [color={rgb, 255:red, 0; green, 0; blue, 0 }  ,draw opacity=0.6 ] [dash pattern={on 4.5pt off 4.5pt}]  (384.33,267.67) .. controls (388,184.17) and (421,58.17) .. (461.33,11) ;
			\draw [color={rgb, 255:red, 0; green, 0; blue, 0 }  ,draw opacity=0.6 ] [dash pattern={on 4.5pt off 4.5pt}]  (180.46,267.67) .. controls (219,230.17) and (396,190.17) .. (461.33,190.67) ;
			\draw [color={rgb, 255:red, 208; green, 2; blue, 27 }  ,draw opacity=1 ]   (384.33,267.67) -- (178,267.67) ;
			\draw [color={rgb, 255:red, 74; green, 144; blue, 226 }  ,draw opacity=1 ]   (384.33,267.67) -- (461.33,190.67) ;
			\draw [color={rgb, 255:red, 245; green, 166; blue, 35 }  ,draw opacity=1 ]   (461.33,190.67) -- (461.33,11) ;
			\draw  [color={rgb, 255:red, 0; green, 0; blue, 0 }  ,draw opacity=0.3 ][line width=1.5]  (274.99,231.03) -- (275.38,205.99) -- (300.32,208.36) -- (293.99,214.03) -- (312.4,234.6) -- (299.73,245.93) -- (281.32,225.37) -- cycle ;
			\draw  [color={rgb, 255:red, 0; green, 0; blue, 0 }  ,draw opacity=0.3 ][line width=1.5]  (399.99,118.03) -- (400.38,92.99) -- (425.32,95.36) -- (418.99,101.03) -- (437.4,121.6) -- (424.73,132.93) -- (406.32,112.37) -- cycle ;
			\draw  [color={rgb, 255:red, 0; green, 0; blue, 0 }  ,draw opacity=0.3 ][line width=1.5]  (278.99,102.03) -- (279.38,76.99) -- (304.32,79.36) -- (297.99,85.03) -- (316.4,105.6) -- (303.73,116.93) -- (285.32,96.37) -- cycle ;
			
			\draw (235,274) node [anchor=north west][inner sep=0.75pt]   [align=left] {$\displaystyle \textcolor[rgb]{0.82,0.01,0.11}{ [g;a_{1},0,0] }$};
			\draw (465,139) node [anchor=north west][inner sep=0.75pt]   [align=left] {$\displaystyle \textcolor[rgb]{0.96,0.65,0.14}{ [a_1 a_2 \cdot g;1,1,a_3] }$};
			\draw (432,219) node [anchor=north west][inner sep=0.75pt]   [align=left] {$\displaystyle \textcolor[rgb]{0.29,0.56,0.89}{ [a_{1} \cdot g;1,a_2,0] }$};
		\end{tikzpicture}
	}
	\caption{The cup product in a $3$-cube along the path through $[g;a_{1},0,0]$, $[a_{1} \cdot g;1,a_2,0]$ and $[a_1 a_2 \cdot g;1,1,a_3]$.}
	\label{fig:$3$-cube_cup}
\end{figure}

\begin{example}\label{example:3D}
	We now analyze the 3D case. A $3$-cube $\pi = [g;a_1,a_2,a_3]$ is subdivided into $3! = 6$ $3$-simplices in the following way: we consider all length-$3$ paths defined by consecutive edges of the $3$-cube that connect $[g;0,0,0]$ and $[a_1 a_2 a_3 \cdot g; 1,1,1]$. 
    There are six paths in total and one of them passes through the following vertices, which span a $3$-simplex:
	\begin{align}
		\tilde{\pi}_{\mathrm{id}} = \Big[ [g;0,0,0],  
		[a_1 \cdot g; 1,0,0], 
		[a_1 a_2 \cdot g; 1,1,0], 
		[a_1 a_2 a_3 \cdot g; 1,1,1] \Big].
	\end{align}
	As in the 2D case, the order on these vertices is canonically defined by the way that we list them. Put in another way, $\tilde{\pi}_{\mathrm{id}}$ is determined by acting $a_1,a_2$ and $a_3$ on $[g;0,0,0]$ in order. We can permute the actions, which yields the remaining five $3$-simplices $\tilde{\pi}_{\alpha}$, where $\alpha \in S_3$ is a permutation of three indices. Geometrically, when we take three $1$-cube (edges) from $\pi$, their cup product is nontrivial and equals $\pi$ if and only if they are ordered to form one of the six paths. 
	
	To verify this intuitive phenomenon, we define $A$ by stretching $\tilde{\pi}_{\mathrm{id}}$ to cover the $3$-cube $\pi$ while collapsing all other simplices to guarantee the continuity. Moreover, to ensure the global continuity of $A$ on $\tilde{X}$, we set the rule that for each individual $3$-cube, the stretching and collapsing are defined in the same way here, i.e., we stretch each $\tilde{\pi}_{\mathrm{id}}$ and collapse $\tilde{\pi}_{\alpha}$ in a consistent way among the $3$-cube. During the process of subdivision, any $2$-cube, such as $[a_1 \cdot g; 1,a_2,a_3] \prec [g;a_1,a_2,a_3]$, is subdivided automatically as in Fig.~\ref{fig:$3$-cube_cup}. Since that $2$-cube may be contained in more than one $3$-cube, we need to check the consistency of subdivisions as well as the stretching and collapsing. Fortunately, this consistency is guaranteed by the rule and the fact that the $2$-cube is uniquely denoted by $[a_1\cdot g;1,a_2,a_3]$, independently of the containing $3$-cube.
	
	By definition, $A_3 \tilde{\pi}_{\mathrm{id}} = [g;a_1,a_2,a_3]$. All other $3$-simplices are collapsed to the boundary of the $3$-cube, so $A_3\tilde{\pi}_\alpha=0$ for every $\alpha\in S_3\setminus\{\mathrm{id}\}$.
    We also have
	\begin{align}
		A_1 & [g;a_1,0,0] = [g;a_1,0,0].
	\end{align}
	The same identity holds for every other edge that is both a $1$-simplex and a $1$-cube. Notably,
	\begin{align}
		\begin{aligned}
			A_1 \Big[ [g;0,0,0], [a_1 a_2 a_3 \cdot g;1,1,1] \Big] 
			= & \Big[ [g;0,0,0], [a_3 \cdot g; 0,0,1] \Big] 
			+ \Big[ [a_3 \cdot g; 0,0,1], [a_2 a_3 \cdot g; 0,1,1] \Big] \\
			& + \Big[ [a_2 a_3 \cdot g; 0,1,1], [a_1 a_2 a_3 \cdot g; 1,1,1] \Big] \\
			= & [g; 0,0,a_3] + [a_3 \cdot g; 0,a_2,1] + [a_2 a_3 \cdot g; a_1,1,1],
		\end{aligned}
	\end{align}
	which reflects the stretching prescribed by $A$.
	
	We now elucidate how to calculate the cup product along a path of three consecutive edges connecting $[g;0,0,0]$ and $[a_1 a_2 a_3 \cdot g; 1,1,1]$ (see Fig.~\ref{fig:$3$-cube_cup}): 
	\begin{align}
		A^1 & [g;a_1,0,0] = [g;a_1,0,0], \label{eq:3D_1} \\
		A^1 & [a_1 \cdot g;1,a_2,0] = [a_1 \cdot g;1,a_2,0],  \\
		A^1 & [a_1 a_2 \cdot g;1,1,a_3] = [a_1 a_2 \cdot g;1,1,a_3].      
	\end{align}
	Note that $A^1 [g;a_1,0,0]$ does not contain the diagonal edge $\Big[ [g;0,0,0], [a_1 a_2 \cdot g; 1,1,0] \Big]$ as a term because $A$ collapses this diagonal in the opposite direction. The same reasoning applies to the other two equations above. As in the 2D case of Example~\ref{example:2D},
	\begin{align}
		[a_1 \cdot g;1,a_2,0] \smile_c [a_1 a_2 \cdot g;1,1,a_3] = [a_1 \cdot g;1,a_2,a_3].
	\end{align}
	Moreover,
	\begin{align}
		& A^2 [a_1 \cdot g;1,a_2,a_3] = \Big[ [a_1 \cdot g;1,0,0], [a_1 a_2 \cdot g;1,1,0], [a_1 a_2 a_3 \cdot g; 1,1,1] \Big], \label{eq:3D_2} \\
		\implies & 
		S^3( (A^1 [g;a_1,0,0]) \smile (A^2 [a_1 \cdot g;1,a_2,a_3] )) = [g;a_1,a_2,a_3].  \label{eq:3D_3}
	\end{align}
	
	As another typical example, we compute $[g;a_1,0,0] \smile_c [a_1 \cdot g;1,0,a_3] \smile_c [a_1 a_3 \cdot g;1,a_2,1]$. The expansion of $A^1 [a_1 a_3 \cdot g; 1,a_2,1]$ includes $[a_1 a_3 \cdot g; 1,a_2,1]$ and the diagonal of $[a_1 \cdot g;1,a_2,a_3]$, but as before
	\begin{align}
		[a_1 \cdot g;1,0,a_3] \smile_c [a_1 a_3 \cdot g;1,a_2,1] = [a_1 \cdot g;1,a_2,a_3].
	\end{align}
	Repeating the calculation in Eqs.~\eqref{eq:3D_2} and~\eqref{eq:3D_3}, we find that the cup product equals $[g;a_1,a_2,a_3]$. The same conclusion holds for each of the remaining four paths.
	
	To compute the local coefficients, let us list all ingredients: we begin with three logical representatives $x_i \in C^1(X, \mathcal{F})$ for which $x_i(\sigma_j) \in \mathcal{F}(\sigma_j)$. The cup products are computed by the procedure above, considering each of the six paths within any $3$-cube $\pi$. The local vectors are processed as in Eq.~\eqref{eq:cup_local_vector}:
	\begin{align}
		( \F_{\sigma_1,\pi} x_1(\sigma_1) ) \otimes ( \F_{\sigma_2,\pi} x_2(\sigma_2) ) \otimes ( \F_{\sigma_3,\pi} x_3(\sigma_3) ),  
	\end{align}
	where $\F_{\sigma,\pi}$ is the local map from the $1$-cube $\sigma$ to the $3$-cube $\pi$. Returning to the construction in Sec.~\ref{sec:DLV_code}, suppose $\sigma = [g;0,a_2,0]$. Then in Sec.~\ref{sec:DLV_code},
	\begin{align}
		\F_{\sigma,\pi} = h_1^T(a_1, \text{-} ) \otimes h_3^T(a_3, \text{-} ): \mathbb{F}_q^{m_1} \otimes \mathbb{F}_q^{m_3} \to \mathbb{F}_q.
	\end{align}  
	
	On the other hand, we need to find $\xi \in C_3(X,\F^{\otimes 3})$ such that $\xi \in \ker \partial_3$. Here, $\F^{\otimes3}$ is defined by taking the threefold tensor product of the corresponding local coefficient spaces and morphisms. Therefore,
	\begin{align}\label{eq:partial_3}
		\partial_3: C_3(X,\F^{\otimes 3}) \rightarrow C_2(X,\F^{\otimes 3}) 
	\end{align}
	is defined by replacing the entries of the boundary operator of $X$ by $(\F_{\tau,\pi}^T)^{\otimes 3}$. Depending on the form of $\tau$, it is one of the
    \begin{align}
	\Big( h_1(\text{-}, a_1) \Big)^{\otimes 3} \in \mathbb{F}_q^{m_1^3 \times 1}, \quad \Big( h_2(\text{-}, a_2) \Big)^{\otimes 3} \in \mathbb{F}_q^{m_2^3 \times 1}, \quad \Big( h_3(\text{-}, a_3) \Big)^{\otimes 3} \in  \mathbb{F}_q^{m_3^3 \times 1}.
	\end{align}
    That is, $\dim C_3(X,\F^{\otimes 3})=\dim C_3(X,\F)$, while for a $2$-cube whose remaining binary coordinate is $i$, the local coefficient dimension increases from $m_i$ to $m_i^3$. The general notion of cup products over different cellular sheaves is also used in Sec.~\ref{sec:proof_sheaf}, but the computational rule is the same. 
\end{example}


\section{Induction scheme for qLDPC codes}\label{sec:inductive}

However, in existing constructions based on exponential Abelian lifts, two consecutive complexes $X_i$ and $X_{i+1}$ need not have a manifest structural relationship, because their base graphs and lifting groups may differ.
 This would prevent inductive proofs of code properties. As the first main result in this paper, we propose a simple inductive construction of the code family and prove that as long as the initial code admits a nontrivial cup product gate, all subsequent codes in the family inherit that gate. The code family always begins with a constant-size code, so the condition can be tested by numerical computations.

\subsection{Inductive lifting sequences}\label{sec:induct_boundary}

As a first attempt to define the code family, we might try to apply \emph{sequential lifting} to a fixed base graph $G_0$ as in Theorem~\ref{thm:3-lift}: for example, let $\HH = C_3$. We define an expander family by lifting $G_0$ sequentially:
\begin{align}
	\left( C_3 \times \cdots \times \left( C_3 \times ( V_0 \times \{0,1\})^t \right) \right).
\end{align}
However, as discussed in Remark~\ref{remark:non-Abelian}, this may be equivalent to $\HH'\times(V_0\times\{0,1\})^t$ with $\HH'$ non-Abelian and hence need not satisfy the Abelian-group requirement of the construction in Sec.~\ref{sec:cubical} when $t\geq3$.

Alternatively, we define graphs and cubical complexes by the following steps.
\begin{enumerate}
    \item We initialize the procedure with a Ramanujan or almost-Ramanujan base graph $G_0$ with well-defined local orderings on its vertices (see Sec.~\ref{sec:DLV_code}). Then we define $X_1'(0) = V_0^t \times \{0,1\}^t$ and build the cubical complex $X_1'$.
	
	\item We take an exponentially large lift of $G_0$ with lifting group $\HH_1$ and define another cubical complex $X_1$ with $X_1(0) = \HH_1 \times V_0^t \times \{0,1\}^t$.
	
	\item We take another lift of $G_0$ with lifting group $\HH_1'$. The group $\HH_1'$ need be neither Abelian nor exponentially large. Then we define the cubical complex $X_2'$ with $X_2'(0) = (\HH_1' \times V_0 )^t \times \{0,1\}^t$. 
	
	\item We repeat the above procedure with $\HH_1' \times G_0$. We define $X_2$ using an Abelian lifting group $\HH_2$ that is exponentially large relative to $\HH_1'\times G_0$, and define $X_3'$ using another lifting group $\HH_2'$. 
\end{enumerate}

\begin{definition}[Inductive lifting sequence] \label{def:induct_lift}
	The following (double) sequence $\{X_i', X_i\}$ obtained from the whole process is called an \emph{inductive lifting sequence} of $t$-dimensional cubical complexes:
	\begin{equation}\label{eq:induct_lift}
		\begin{tikzcd}
			X_1' \arrow[r, "\HH_1' "] \arrow[d, "\HH_1"] & X_2' \arrow[r, "\HH_2' "] \arrow[d, "\HH_2"] & \cdots \arrow[r, "\HH_{i-1}'"] & X_i' \arrow[r, "\HH_i' "] \arrow[d, "\HH_i"] & \cdots \\
			X_1 \arrow[r, dashed] & X_2 \arrow[r, dashed] & \cdots \arrow[r, dashed] & X_i \arrow[r, dashed] & \cdots
		\end{tikzcd}
	\end{equation} 
	The second line in Eq.~\eqref{eq:induct_lift} is called an \emph{exponential lifting sequence} because each $X_i$ is obtained from $X_i'$ by an exponentially large lift. The first line is actually a family of ordinary HGPs. If $\HH_i'=\HH'$ for all $i$ where $\HH'$ has constant order, we call the first row a \emph{sequential lifting sequence}.
\end{definition}

Throughout the sequence, local edge orderings are inherited under all horizontal and vertical covering maps, and the same coordinatewise subdivision rule is used on every cubical complex, so that the subdivisions are compatible with these coverings.

Note that all base graphs in $\{X_i',X_i\}$ have a fixed node degree $n$. As a result, with any collection of local codes $h_1 \in \mathbb{F}_q^{m_1 \times n}, \ldots,h_t \in \mathbb{F}_q^{m_t \times n}$, we can define a cochain complex on every $X_i'$ and $X_i$. In summary, given a base graph $G_0$, groups $\HH_i, \HH_i'$ and local codes $h_1, \ldots,h_t$, we can construct an inductive lifting sequence of $t$-dimensional cochain complexes $\{ C^\bullet(X_i',\mathcal{F}_i'), C^\bullet(X_i,\mathcal{F}_i) \}$.

Our purpose is to determine the relationship among the (co)boundary operators of these cubical complexes. This relationship will facilitate the inductive construction of nontrivial cup product gates. To simplify notation as much as possible, we suppress the subscript and only work out the details for one vertical lift $X'\xrightarrow{\HH}X$:
\begin{itemize}
	\item For $X$, we denote any cube by $[g; (a_i), (b_j)]$ with $g = (h, v_1, \ldots,v_t)$.
	
	\item For $X'$, we have $[g'; (a_i'), (b_j)]$ with $g' = (v_1, \ldots,v_t)$. 
	As a reminder, $a_i'$ and $a_i$ have different meanings, but their label sets are in one-to-one correspondence, as in the definition of the cubical complex in Sec.~\ref{sec:cubical}.
\end{itemize}
Let $\delta'^p:C^p(X',\mathbb F_q)\to C^{p+1}(X',\mathbb F_q)$ be the coboundary operator. For some $\sigma' \in X'(p)$ and $\tau' \in X'(p+1)$, if $\sigma' \prec \tau'$, then the entry $\delta'^{p}(\tau',\sigma') = 1$. By the above discussion, each $\sigma'$ has $\vert{\HH}\vert$ lifts $\sigma \in X(p)$, obtained by inserting $h \in \HH$ into $g'$. Now $\sigma' \prec \tau'$ indicates that $\sigma \prec \tau$ for some $\sigma, \tau$ in the corresponding fibers. The exact inclusion is revealed by the action $a_i$ on $h$. Accordingly, the corresponding coboundary matrix $\delta^p$ for $X$ is obtained from $\delta'^p$ by replacing each nonzero scalar entry $\delta'^p(\tau',\sigma')$ with the corresponding $|\HH|\times|\HH|$ permutation matrix. For example, consider the 3D case: if $[a_3' \cdot g';a_1',a_2',1] \prec [g';a_1',a_2',a_3']$, then
\begin{align}
	[a_3 \cdot g; a_1,a_2,1] \prec [g; a_1,a_2,a_3] \text{ with } g = (h,v_1,v_2,v_3), \ a_3 \cdot g = (\gamma_{(v_3,v_3')} \cdot h, v_1,v_2,v_3'). 
\end{align}
Then the scalar entry is substituted by $\gamma_{(v_3,v_3')}$, which is a permutation matrix representing the action of $\gamma_{(v_3,v_3')} \in \HH$. This is the high-dimensional analogue of the graph lifts in Remark~\ref{remark:lift}. From a topological perspective, the coboundary operators of $X'$ lift to the coboundary operators of its covering space $X$.

\begin{lemma}\label{lemma:sheaf}
	Let $\sigma' \prec \tau'$ in $X'$, and let $\sigma \prec \tau$ be corresponding lifts in $X$. For fixed local codes $h_1,\ldots,h_t$, there are natural identifications
	\begin{align}
		\mathcal{F}(\sigma') \cong \mathcal{F}(\sigma),
		\qquad
		\mathcal{F}(\tau') \cong \mathcal{F}(\tau),
	\end{align}
	under which $\mathcal{F}_{\sigma',\tau'}$ and $\mathcal{F}_{\sigma,\tau}$ have identical matrix representations.
\end{lemma}
Lemma~\ref{lemma:sheaf} follows directly from the definition of $\mathcal{F}$ in Sec.~\ref{sec:DLV_code}. All lifts of a given $\sigma'\in X'(p)$ are assigned the same local coefficient space. Within each permutation matrix subblock of $\delta^p$, every nonzero scalar entry is replaced by $\mathcal{F}_{\sigma,\tau}\equiv\mathcal{F}_{\sigma',\tau'}$. This simple observation underlies the results below. Moreover, when the local codes are fixed, we use the notation $\mathcal{F}$ to denote the cellular sheaves for all complexes from an inductive lifting sequence. 

\begin{lemma}\label{lemma:cochain_1}
	The correspondence between $\sigma'$ and $\{\sigma\}$ defines an embedding $E: \mathbb{F}_q \to \mathbb{F}_q^{\vert \HH \vert}$ and a projection $P = E^T: \mathbb{F}_q^{\vert \HH \vert} \to \mathbb{F}_q$ represented by the matrix
	\begin{align}\label{eq:projector}
		P = \begin{pmatrix} 1 & 1 & \cdots & 1	\end{pmatrix}. 
	\end{align}
	They further induce the following cochain maps (see Definition~\ref{def:chain_map}):
	\begin{equation}
		\begin{tikzcd}[column sep=1.5cm, row sep=0.7cm]
			\cdots \arrow[r, "\delta'^{p-1}"] & 
			C^p(X',\F) \arrow[r, "\delta'^p"] \arrow[d,"E^p"] & 
			C^{p+1}(X',\F) \arrow[d, "E^{p+1}"] \arrow[r, "\delta'^{p+1}"] 
			& \cdots \\
			\cdots \arrow[r, "\delta^{p-1}"] & 
			C^p(X,\F) \arrow[r, "\delta^{p}"] & 
			C^{p+1}(X,\F) \arrow[r, "\delta^{p+1}"] & \cdots
		\end{tikzcd}
	\end{equation}
	and
	\begin{equation}
		\begin{tikzcd}[column sep=1.5cm, row sep=0.7cm]
			\cdots \arrow[r, "\delta'^{p-1}"] & 
			C^p(X',\F) \arrow[r, "\delta'^{p}"] & 
			C^{p+1}(X',\F) \arrow[r, "\delta'^{p+1}"] 
			& \cdots \\
			\cdots \arrow[r, "\delta^{p-1}"] & 
			C^p(X,\F) \arrow[r, "\delta^{p}"] \arrow[u,"P^{p}"] & 
			C^{p+1}(X,\F) \arrow[r, "\delta^{p+1}"] \arrow[u, "P^{p+1}"] & \cdots
		\end{tikzcd}
	\end{equation}
\end{lemma}
\begin{proof}
    With trivial local coefficients, $E^p$ is defined by rewriting the identity matrix on $C^p(X',\mathbb{F}_q)$ as a block-diagonal matrix whose diagonal blocks are equal to $E$. With nontrivial local coefficients, Lemma~\ref{lemma:sheaf} shows that each nonzero scalar entry in this block matrix is replaced by the identity matrix acting on $\mathcal{F}(\sigma')$ for the corresponding $\sigma' \in X'(p)$. The map $P^\#$ is defined analogously. By construction, both diagrams commute.
\end{proof}

\begin{remark}\label{remark:sparse}
	Because $E^p$ duplicates each component, $E^p$ is injective on each $C^p(X',\F)$. We prove in Lemmas~\ref{lemma:codeword_1} and~\ref{lemma:codeword_2} that it induces an injective map $H^p(X',\F)\to H^p(X,\F)$ when $\vert\HH\vert$ is odd. For a nontrivial lift, $E^p$ is not surjective because its image consists only of cochains that are constant on each lift fiber. As the lift size grows, $C^\bullet(X,\F)$ may have large cohomology groups. On the other hand, the map $P^p$ is surjective and, under the same condition, induces a surjective homomorphism on cohomology.
	
	More importantly, neither $E^p$ nor $P^p$ is sparse because they contain subblocks with $\vert \HH \vert$ nonzero elements in columns or rows, and as explained in Remark~\ref{remark:disconnect}, the lift size grows with the system size. Additionally, assume that we find an invariant polynomial $T'(x_{p_1}', \ldots,x_{p_\rho}')$ that represents a constant-depth multi-controlled-$Z$ circuit as in Eq.~\eqref{eq:invariant_poly} in $C^\bullet(X',\F)$. Using the cochain map, $T'(P^{p_1} (x_{p_1} ), \ldots,P^{p_\rho} (x_{p_\rho}) )$ is also invariant in $C^\bullet(X,\F)$. However, since the cochain maps are not sparse, their composition with $T'$ cannot produce a constant-depth circuit as mentioned at the end of Sec.~\ref{sec:CSS}. 
\end{remark}

Lemma~\ref{lemma:cochain_1} does not by itself construct constant-depth circuits, but it motivates the following lemmas, which are indispensable for establishing the results in Sec.~\ref{sec:induct_invariant}.

\begin{lemma}\label{lemma:codeword_1}
	With trivial coefficients, let $x'\in\ker\delta'^p$ and set $x=E^p x'$. The map $E^p$ replaces each component $x'(\sigma')$ by $x'(\sigma')$ times the all-ones vector on the lift fiber of $\sigma'$. If $\lvert\HH\rvert$ is odd, then
	\begin{align}
		x'\in\ker\delta'^p\setminus\im\delta'^{p-1}
		\quad \Longrightarrow \quad x \in \ker\delta^p \setminus \im \delta^{p-1}.
	\end{align}
	In this case, $P^\#\circ E^\#$ induces the identity on cohomology.
\end{lemma}
\begin{proof}
	By Lemma~\ref{lemma:cochain_1}, $x$ is always in $\ker \delta^{p}$. Suppose $x'$ is a cocycle but not a coboundary. Assume $x$ is a coboundary, i.e., $x=\delta^{p-1}y$ for some $y\in C^{p-1}(X,\mathbb{F}_q)$. Since $P^\#$ is a cochain map and $\mathbb{F}_q$ is of characteristic $2$, when $\lvert\HH\rvert$ is odd, we have $P^pE^p=\lvert\HH\rvert I=I$. Consequently,
	\begin{align}
		x' = P^p E^p x' = P^p x = P^p\delta^{p-1}y = \delta'^{p-1} P^{p-1} y \in \im \delta'^{p-1},
	\end{align}
	which contradicts the assumption. The identity $P^\#\circ E^\#=I$ at the cochain level also proves the final statement.
\end{proof}

When $\vert \HH \vert$ is even, $P^p E^p x' = 0$ and Lemma~\ref{lemma:codeword_1} is false. The following $[\![4,1,2]\!]$ code gives an explicit counterexample:
\begin{align}   
	H_Z = \begin{pmatrix} 1 & 1 & 1 & 1	\end{pmatrix}, \quad
	H_X = \begin{pmatrix} 1 & 0 & 1 & 0 \\ 0 & 1 & 0 & 1 \\ 1 & 1 & 1 & 1	\end{pmatrix}.
\end{align}
The third row in $H_X$ is used to make the rows linearly dependent. One $X$-codeword is $x' = (1,0,0,1)$. Suppose we take a 2-lift with $\HH = S_2$:
\begin{align}
	H_Z = \left(\begin{array}{cc|cc|cc|cc}
		1 & 0 & 1 & 0 & 1 & 0 & 1 & 0 \\
		0 & 1 & 0 & 1 & 0 & 1 & 0 & 1 
	\end{array} \right), \quad
	H_X = \left(\begin{array}{cc|cc|cc|cc}
		1 & 0 & 0 & 0 & 1 & 0 & 0 & 0 \\
		0 & 1 & 0 & 0 & 0 & 1 & 0 & 0 \\
		\hline
		0 & 0 & 1 & 0 & 0 & 0 & 1 & 0 \\
		0 & 0 & 0 & 1 & 0 & 0 & 0 & 1 \\
		\hline
		1 & 0 & 0 & 1 & 0 & 1 & 1 & 0 \\
		0 & 1 & 1 & 0 & 1 & 0 & 0 & 1 
	\end{array} \right).
\end{align}
The extended vector $x = (1,1,0,0,0,0,1,1)$ lies in the span of the first, third, and sixth rows of $H_X$. 

\begin{lemma}\label{lemma:codeword_2}
	Lemma~\ref{lemma:codeword_1} holds when the trivial coefficient is replaced by a cellular sheaf $\mathcal{F}$ generated by local codes on the (co)chain complexes.
\end{lemma}
\begin{proof}
	By Lemma~\ref{lemma:sheaf}, for $x'\in\ker\delta'^p$, set $x=E^p x'$ by duplicating each $x'(\sigma')\in\mathcal{F}(\sigma')$ across the $\vert\HH\vert$ lifts of $\sigma'$. If $\vert\HH\vert$ is odd, then $P^p x=x'$, so the proof of Lemma~\ref{lemma:codeword_1} applies. The chain complex case is analogous because the transpose of a permutation matrix is still a permutation matrix. We also need to transpose $\mathcal{F}_{\sigma',\tau'}$ but this does not affect the argument.
\end{proof}

\begin{definition}[Canonical extension]
	For convenience, $E^p x'$ is called the \emph{canonical extension} of $x'$, distinguished by the fact that it assigns $x'(\sigma')$ to every cube in the lift fiber of $\sigma'$.
\end{definition}


The first row in Eq.~\eqref{eq:induct_lift} consists of Cartesian products and $X_{i+1}'$ is a covering space of $X_i'$ with a particularly simple structure. For example, let $X' = X_i'$ and $X'' = X_{i+1}'$ with the horizontal lift $X' \xrightarrow{\HH'} X''$. As before, we denote $\F'$ and $\F''$ by $\F$ for simplicity. By construction, $C^\bullet(X',\F)$ is an HGP, and $C^\bullet(X'',\F)$ is the HGP built from the double cover of an $\HH'$-lift on $G_0$. For the notation below, fix an enumeration $\HH'=\{h_1,\ldots,h_{\vert\HH'\vert}\}$.
\begin{itemize}
	\item For $X'$, we have $[g'; (a_i'), (b_j)]$ with $g' = (v_1, \ldots,v_t)$ as before.
	
	\item For $X''$, we use the notation $[g''; (a_i''), (b_j)]$ with $g'' = ( (h_{r_1}, v_1), \ldots,( h_{r_t}, v_t) )$ and $r_k \in [ \vert \HH' \vert ]$. As before, $a_i'$ and $a_i''$ label the $n$ incident edge directions at each vertex of the corresponding base graphs and there is a canonical one-to-one correspondence between them.
\end{itemize}
Again, let $\delta'^p:C^p(X',\mathbb{F}_q)\to C^{p+1}(X',\mathbb{F}_q)$ be the coboundary operator and suppose $\delta'^p(\tau',\sigma')=1$. We now associate each $\sigma'$ with $\vert \HH' \vert^t$ cubes $\sigma'' \in X''(p)$. As before, $\sigma' \prec \tau'$ indicates that $\sigma'' \prec \tau''$ for some lifts. The coboundary operator $\delta''^{p}$ for $X''$ is defined by substituting a $\vert \HH' \vert^t \times \vert \HH' \vert^t$ permutation matrix into each nonzero scalar entry of $\delta'^{p}$. Each permutation matrix also has its own substructure: it is of the form
\begin{align}
	I_{\vert \HH' \vert} \otimes \cdots \otimes I_{\vert \HH' \vert} \otimes P_j \otimes I_{\vert \HH' \vert} \otimes \cdots \otimes I_{\vert \HH' \vert},
\end{align}
where, depending on $b_j = 0$ or $1$, $P_j$ is either the identity matrix or some permutation matrix.

The same argument yields analogues of Lemmas~\ref{lemma:cochain_1}, \ref{lemma:codeword_1}, and~\ref{lemma:codeword_2}.

\begin{lemma}\label{lemma:cochain_2}
The correspondence between each $\sigma'$ and its lift fiber
$\{\sigma_{\boldsymbol r}''\}_{\boldsymbol r\in[\vert\HH'\vert]^t}$
defines an embedding $E': \mathbb{F}_q \to \mathbb{F}_q^{\vert \HH' \vert^t}$ and a projection (covering map) $P' = E'^T: \mathbb{F}_q^{\vert \HH' \vert^t} \to \mathbb{F}_q$ as in Eq.~\eqref{eq:projector}, but now with $\vert \HH' \vert^t$ entries. They further induce the following cochain maps:
	\begin{equation}
		\begin{tikzcd}[column sep=1.5cm, row sep=0.7cm]
			\cdots \arrow[r, "\delta'^{p-1}"] & 
			C^p(X',\F) \arrow[r, "\delta'^{p}"] \arrow[d,"E'^{p}"] & 
			C^{p+1}(X',\F) \arrow[d, "E'^{p+1}"] \arrow[r, "\delta'^{p+1}"] 
			& \cdots \\
			\cdots \arrow[r, "\delta''^{p-1}"] & 
			C^p(X'',\F) \arrow[r, "\delta''^{p}"] & 
			C^{p+1}(X'',\F) \arrow[r, "\delta''^{p+1}"] & \cdots
		\end{tikzcd}
	\end{equation}
	and
	\begin{equation}
		\begin{tikzcd}[column sep=1.5cm, row sep=0.7cm]
			\cdots \arrow[r, "\delta'^{p-1}"] & 
			C^p(X',\F) \arrow[r, "\delta'^{p}"] & 
			C^{p+1}(X',\F) \arrow[r, "\delta'^{p+1}"] 
			& \cdots \\
			\cdots \arrow[r, "\delta''^{p-1}"] & 
			C^p(X'',\F) \arrow[r, "\delta''^{p}"] \arrow[u,"P'^{p}"] & 
			C^{p+1}(X'',\F) \arrow[r, "\delta''^{p+1}"] \arrow[u, "P'^{p+1}"] & \cdots
		\end{tikzcd}
	\end{equation}
\end{lemma}

If $\HH'$ has a constant number of elements, then $P'^\#$ is sparse and may induce a constant-depth circuit (cf. Remark~\ref{remark:sparse}). However, without exponential lifting, the ordinary HGP cannot produce good code parameters (see Eq.~\eqref{eq:HGP_parameters}).

\begin{lemma}\label{lemma:codeword_3}
	Let $x' \in \ker \delta'^{p}$ with the canonical extension $x'' \coloneq E'^p x' \in \ker \delta''^{p}$. If $\HH'$ has an odd number of elements, $x' \in \ker \delta'^{p} \setminus \im \delta'^{p-1}$ implies that $x'' \in \ker \delta''^{p} \setminus \im \delta''^{p-1}$. Analogous results hold for the boundary operators.
\end{lemma}

Assembling these results together and applying them to an inductive lifting sequence $\bigl\{(C^\bullet(X_i',\mathcal F),C^\bullet(X_i,\mathcal F))\bigr\}_i$, we can push forward any codeword from the initial code to all subsequent cochain complexes.


\subsection{Preservation of (co)homological invariants}\label{sec:induct_invariant}

We now prove that if an iterated cup product on the initial complex has a nonzero pairing with a cycle, then this invariant is preserved throughout the inductive lifting sequence. In particular, the associated cup product gate is inherited. Given any cube $\sigma \in X$, let $X_{\geq \sigma}(p)$ be the collection of all $p$-cubes in $X$ that contain $\sigma$. Then the following useful lemma holds for cubical complexes defined by lifts and products of graphs.

\begin{lemma}\label{lemma:smallest_cube}
	If $\sigma,\tau\in X$ satisfy $X_{\geq \sigma}(t) \cap X_{\geq \tau}(t) \neq \emptyset$, then there is a unique smallest cube $\omega$ containing both of them.
\end{lemma}

As a straightforward example, we consider
\begin{align}
	[g; a_1,0,0,1], \ [a_2 a_3 \cdot g; a_1,1,1,1] \prec [g; a_1,a_2',a_3',1].  
\end{align}
The only requirement is that $a_2' a_3' \cdot g = a_2 a_3 \cdot g$. Let $g = (h,v_1,v_2,v_3,v_4)$. With respect to the product structure, we must have $a_2' = a_2$ and $a_3' = a_3$.

\begin{proof}[Proof of Lemma~\ref{lemma:smallest_cube}]
	Let $J_\sigma = [t] \setminus I_\sigma$ and $J_\tau = [t] \setminus I_\tau$. We set
	\begin{align}
		\sigma = [g_\sigma; (a_{i_\sigma})_{I_\sigma}, (b_{j_\sigma})_{J_\sigma}];  \ g_\sigma = (h_\sigma, v_1, \ldots,v_t), \\
		\tau = [g_\tau; (a_{i_\tau})_{{I_\tau}}, (b_{j_\tau})_{{J_\tau}} ]; \ g_\tau = (h_\tau, u_1, \ldots,u_t). 
	\end{align}
	Let $\pi = [g_\pi; (a_i)_{[t]}]$ be any $t$-cube from $X_{\geq \sigma}(t) \cap X_{\geq \tau}(t)$. Then we must have
	\begin{align}
		(a_{i_\sigma})_{ I_\sigma} = (a_i) \vert_{I_\sigma}, \quad  (a_{i_\tau})_{I_\tau} =  (a_i) \vert_{ I_\tau }.
	\end{align}
	In particular, on the intersection,
	\begin{align}
		(a_{i_\sigma}) \vert_{ I_\sigma \cap I_\tau } = (a_i) \vert_{ I_\sigma \cap I_\tau } = (a_{i_\tau}) \vert_{ I_\sigma \cap I_\tau }.
	\end{align}
	As a result, there is no ambiguity in defining a cube $\omega$ with $(a_i)$ for $i \in I_\sigma \cup I_\tau$. For $j \in [t] \setminus ( I_\sigma \cup I_\tau )$, if the $j$-th binary bits of $\sigma$ and $\tau$ are identical, we assign the bit to $\omega$. Otherwise, we assign $a_j$ from $\pi$ to $\omega$. 
    Finally, $g_\omega$ is obtained from $g_\pi$ by the action of the appropriate $a_i$, as determined by the binary string of $\omega$.

    At this point, $g_\omega$ and some of the labels $(a_{i_\omega})$ may appear to depend on the choice of $\pi$, but the product structure of $X$ determines them uniquely. Because $\sigma$ and $\tau$ can be embedded into a common $t$-cube, for each $i\in[t]$, exactly one of the following cases occurs:
	\begin{enumerate}
		\item If the $i$-th components of both $\sigma$ and $\tau$ are actions, then they must be identical. For simplicity, we write $\sigma_i = \tau_i = a_i$. This also implies that $v_i = u_i$. 

        \item Suppose $\sigma_i=a_i$ and $\tau_i$ is a bit. If $\tau_i=0$, then $v_i=u_i$. If $\tau_i=1$, then $a_i\cdot v_i=u_i$.
		
		\item Suppose $\sigma_i$ is a bit while $\tau_i = a_i$. This is the converse to the above case.
		
		\item Suppose both $\sigma_i$ and $\tau_i$ are bits. If they are identical, then $v_i = u_i$. If not, there must be a unique $a_i$ such that either $a_i \cdot v_i = u_i$ or $a_i \cdot u_i = v_i$, depending on the bit values.
	\end{enumerate}  
    The $i$-th vertex coordinate of $g_\omega$ is $v_i$ when $\sigma_i=a_i$ or $\sigma_i=0$, and is $u_i$ when $\sigma_i=1$. This covers all the cases above. To determine the group element $h_\omega\in\HH$, define
	\begin{align}
		J_1 &\coloneq \{j\in J_\sigma:\sigma_j=1\}, \\
		J_2 &\coloneq \{j\in J_\tau:\tau_j=1\}.
	\end{align}
	Then $J_{\sigma,\tau}\coloneq J_1\cap J_2$ consists of all indices for which $\sigma_j=\tau_j=1$. The element $h_\omega$ is the unique group element such that
	\begin{align}
		\prod_{j \in J_1 \setminus J_{\sigma,\tau}} a_j \cdot h_{\omega} = h_\sigma.
	\end{align}
	Here $a_j$ is uniquely defined according to the last case above. Moreover,
	\begin{align}
		\prod_{j \in J_2 \setminus J_{\sigma,\tau}} a_j \cdot h_{\omega} = h_\tau
	\end{align}
	must hold automatically by the assumption $X_{\geq \sigma}(t) \cap X_{\geq \tau}(t) \neq \emptyset$.
\end{proof}

We have another simple but important result for computing cup products. 

\begin{lemma}\label{lemma:cup_induct_1}
	Given $\sigma \in X(p_1)$ and $\tau \in X(p_2)$ in a $t$-dimensional cubical complex $X$, let $\omega$ be the smallest cube containing both $\sigma$ and $\tau$ in Lemma~\ref{lemma:smallest_cube}. If $\omega$ does not exist, then $\sigma \smile_c \tau = 0$. Otherwise, $\sigma \smile_c \tau$ is either zero or $\omega$, depending on the subdivisions.  
\end{lemma}
\begin{proof}
	We first deal with the case over $\mathbb{F}_q$. The approximation map $A$ sends any simplex $\tilde{\sigma}$ to cubes in $\tau_{\tilde{\sigma}}$, where $\tau_{\tilde{\sigma}}$ is the smallest cube that contains $\tilde{\sigma}$. By the previous discussion, $\tau_{\tilde{\sigma}}$ is unique and any cube containing the simplex $\tilde{\sigma}$ must also contain the cube $\tau_{\tilde{\sigma}}$. Then
	\begin{align}\label{eq:cup_induct_1}
		A^{p_1} (\sigma) = \sum_{\sigma \subset \tau_{\tilde{\sigma}_i} } \tilde{\sigma}_i, \quad
		A^{p_2} (\tau) = \sum_{\tau \subset \tau_{\tilde{\tau}_j} } \tilde{\tau}_j,
	\end{align}
	where $\tilde{\sigma}_i$, $\tilde{\tau}_j$ are chosen with dimensions $p_1$ and $p_2$, respectively. 
    We now take the cup product of the simplices $\tilde\sigma_i$ and $\tilde\tau_j$.
    If $\tilde{\sigma}_i \smile \tilde{\tau}_j$ is nontrivial, then we can find the smallest cube, denoted by $\pi$ for conciseness, such that 
	\begin{align}
		\pi \supset \tilde{\sigma}_i \smile \tilde{\tau}_j \supset \tilde{\sigma}_i, \tilde{\tau}_j.
	\end{align}
    If $\dim \pi > p_1 + p_2$, then no $(p_1 + p_2)$-cube contains $\tilde{\sigma}_i \smile \tilde{\tau}_j$. Hence, by the definition of $S^{p_1 + p_2}$, $S^{p_1 + p_2}(\tilde{\sigma}_i \smile \tilde{\tau}_j) = 0$. In other words, the cup product can be nontrivial only when $\pi$ is a $(p_1 + p_2)$-cube. Since
	\begin{align}
		\tau_{\tilde{\sigma}_i} \supset \sigma, \ \tau_{\tilde{\tau}_j} \supset \tau \text{ and } \pi \supset \tau_{\tilde{\sigma}_i}, \tau_{\tilde{\tau}_j},
	\end{align} 
	the $(p_1 + p_2)$-cube $\pi$ contains both $\sigma$ and $\tau$. The dimension constraint implies that 
	\begin{align}
		\pi = \omega, \ \tau_{\tilde{\sigma}_i} = \sigma, \ \tau_{\tilde{\tau}_j} = \tau \text{ and } S^{p_1 + p_2} (\tilde{\sigma}_i \smile \tilde{\tau}_j) = \omega.
	\end{align} 
	Then we consider all possible cup products from Eq.~\eqref{eq:cup_induct_1}. 
    Since $\mathbb F_q$ has characteristic $2$, after summing, the result is either $0$ or $\omega$. Even in the presence of nontrivial local coefficients from $\mathcal F$, Eq.~\eqref{eq:cup_local_vector} and the examples in Sec.~\ref{sec:example} show that we can compute $\sigma\smile_c\tau$ over $\mathbb F_q$ as above and then account for the actions on the local vectors.
\end{proof}

As a reminder, the proof of Lemma~\ref{lemma:cup_induct_1} uses the fact that a nontrivial cup product requires $\sigma\cap\tau$ to be a $0$-cube, as illustrated in Sec.~\ref{sec:example}. If the intersection is empty or has positive dimension, the cup product is zero. This is consistent with the simplicial case, where a nontrivial cup product of a $p_1$-simplex and a $p_2$-simplex has dimension $p_1+p_2$. If they intersect at a $\geq 1$-dimensional simplex, we cannot have a $(p_1 + p_2)$ dimensional cup product by Definition~\ref{def:cup}.

Let $\{C^\bullet(X_i', \F), C^\bullet(X_i, \F)\}$ be any inductive lifting sequence with cup product defined by Eq.~\eqref{eq:cup_cell} on each of them. We first study the case $\mathcal F=\mathbb F_q$ and then treat general coefficients. For clarity, we consider a triple $(X', X, X'')$ as in Sec.~\ref{sec:induct_boundary}. Choose $\sigma'$ and $\tau'$ with $p_1+p_2\leq t$, and write them as follows:
\begin{align}
	\sigma' = [g_\sigma'; (a_{i_\sigma}'), (b_{j_\sigma})] \in X'(p_1); \ g_\sigma' = (v_1, \ldots,v_t), \\
	\tau' = [g_\tau'; (a_{i_\tau}'), (b_{j_\tau})] \in X'(p_2); \ g_\tau' = (u_1, \ldots,u_t). 
\end{align}
As in Sec.~\ref{sec:induct_boundary}, let $\{\sigma_r\}_{r\in[\vert\HH\vert]}, \{\tau_s\}_{s\in[\vert\HH\vert]}$
be the corresponding cubes in $X$. They have the following full expressions:
\begin{align}
	\sigma_r & = [g_{\sigma,r}; (a_{i_\sigma}), (b_{j_\sigma})] \in X(p_1); \ g_{\sigma,r} = (h_r, v_1, \ldots,v_t), \\
	\tau_s & = [g_{\tau,s}; (a_{i_\tau}), (b_{j_\tau})] \in X(p_2); \ g_{\tau,s} = (h_s, u_1, \ldots,u_t). 
\end{align}
Similarly, we have $\{\sigma_{\boldsymbol{r}}'' \}_{\boldsymbol{r} \in [\vert \HH' \vert]^t}$, $\{\tau_{\boldsymbol{s}}'' \}_{\boldsymbol{s} \in [\vert \HH' \vert]^t}$ such that
\begin{align}
	\sigma_{\boldsymbol{r}}'' & = [g_{\sigma,\boldsymbol{r}}; (a_{i_\sigma}''), (b_{j_\sigma})] \in X''(p_1); \ g_{\sigma, \boldsymbol{r}} = ((h_{r_1}, v_1), \ldots,(h_{r_t}, v_t)), \\
	\tau_{\boldsymbol{s}}'' & = [g_{\tau,\boldsymbol{s}}; (a_{i_\tau}''), (b_{j_\tau})] \in X''(p_2); \ g_{\tau,\boldsymbol{s}} = ( (h_{s_1}, u_1), \ldots,(h_{s_t}, u_t) ).
\end{align}

Suppose the cup product $\sigma' \smile_c \tau'$ is nontrivial. By Lemma~\ref{lemma:cup_induct_1}, it equals the minimal cube $\pi'$ encompassing both $\sigma'$ and $\tau'$. The cube $\pi'$ is lifted into $\vert \HH \vert$ cubes in $X$. Each such lifted cube is obtained from a unique pair $(\sigma_r,\tau_s)$. Since $A^\#, S^\#$ and the cup product are uniquely determined by the local geometry  signatures $g$, $(a_i)$ and $(b_j)$ of the involved cubes (see Sec.~\ref{sec:example}), $\sigma_r \smile_c \tau_s$ must equal the cube from which they are decomposed. Therefore, there are exactly $\vert \HH \vert$ pairs $(r,s) \in [\vert \HH \vert] \times [\vert \HH \vert]$ such that $\sigma_r \smile_c \tau_s$ is nontrivial. These facts verify the following lemma.

\begin{lemma}\label{lemma:cup_induct_2}
	The cup product is compatible with the cochain map $E^\#$ defined in Lemma~\ref{lemma:cochain_1} in the following sense:
	\begin{equation}
		\begin{tikzcd}[column sep=1.5cm, row sep=0.7cm]
			C^{p_1}(X',\mathbb{F}_q) \times C^{p_2}(X',\mathbb{F}_q) \arrow[r, "\smile_c"] \arrow[d,"E^{p_1} \times E^{p_2}"] & 
			C^{p_1 + p_2}(X',\mathbb{F}_q) \arrow[d, "E^{p_1 + p_2}"]  \\
			C^{p_1}(X,\mathbb{F}_q) \times C^{p_2}(X,\mathbb{F}_q) \arrow[r, "\smile_c"] & 
			C^{p_1 + p_2}(X,\mathbb{F}_q)
		\end{tikzcd}
	\end{equation}
	Specifically, if $\sigma'\smile_c\tau'$ is nontrivial, then there are exactly $\vert \HH \vert$ pairs of $(r,s) \in [\vert \HH \vert] \times [\vert \HH \vert]$ such that $\sigma_r \smile_c \tau_s$ is nontrivial. The same argument applies when $X$ is replaced by $X''$, yielding exactly $\vert\HH'\vert^t$ nontrivial pairs.
\end{lemma}

However, the following diagram does not commute:
\begin{equation}
	\begin{tikzcd}[column sep=1.5cm, row sep=0.7cm]
		C^{p_1}(X',\mathbb{F}_q) \times C^{p_2}(X',\mathbb{F}_q) \arrow[r, "\smile_c"] & 
		C^{p_1 + p_2}(X',\mathbb{F}_q)\\
		C^{p_1}(X,\mathbb{F}_q) \times C^{p_2}(X,\mathbb{F}_q) \arrow[r, "\smile_c"] \arrow[u,"P^{p_1} \times P^{p_2}"] & 
		C^{p_1 + p_2}(X,\mathbb{F}_q) \arrow[u, "P^{p_1 + p_2}"]  
	\end{tikzcd}.
\end{equation}
This is because for any nontrivial cup product $\sigma' \smile_c \tau'$ in $X'$, we can choose lifts $\sigma_{r}$ and $\tau_{s}$ at different levels such that $\sigma_{r} \smile_c \tau_{s} = 0$. Then
\begin{align}
	P^{p_1+p_2}( \sigma_r \smile_c \tau_s ) = 0 \neq P^{p_1}(\sigma_r) \smile_c P^{p_2}(\tau_s) = \sigma' \smile_c \tau'.
\end{align}
To summarize, the cup product is compatible with the lift $E$ but not with the covering $P$. 

\medskip
Suppose that $\rho$ cochain complexes $C^\bullet(X,\F_1), \ldots,C^\bullet(X,\F_\rho)$ are based on the same cubical complex $X$ but are possibly equipped with different local coefficients. Let $x_i \in C^{p_i}(X, \F_i)$ be cocycles such that $p_1 + \cdots + p_\rho \leq t$ and let $\xi \in C_{p_1 + \cdots + p_\rho}(X, \F_1 \otimes \cdots \otimes \F_\rho)$ be a cycle. Here $\F_1\otimes\cdots\otimes\F_\rho$ is defined componentwise by tensoring the corresponding local coefficient spaces and morphisms:
\begin{align}
	\F_{1,\sigma,\tau} \otimes \cdots \otimes \F_{\rho, \sigma,\tau}: \F_1(\sigma) \otimes \cdots \otimes \F_\rho(\sigma)  \rightarrow \F_1(\tau) \otimes \cdots \otimes\F_\rho(\tau).
\end{align}
In many cases, $\F_i \equiv \F$, although distinct cellular sheaves are used in Sec.~\ref{sec:proof_sheaf}. We adopt the most general expressions here.

\begin{lemma}\label{lemma:cup_induct_3}
	Lemma~\ref{lemma:cup_induct_2} holds for multiple cup products in the presence of cellular sheaves $\F_1, \ldots,\F_\rho$ generated by classical codes:
	\begin{equation}
		\begin{tikzcd}[column sep=1.5cm, row sep=0.7cm]
			C^{p_1}(X,\F_1) \times \cdots \times C^{p_\rho}(X,\F_\rho) \arrow[r, "\smile_c"] & 
			C^{p_1 + \cdots + p_\rho}(X,\F_1 \otimes \cdots \otimes \F_\rho) \\
			C^{p_1}(X',\F_1) \times \cdots \times C^{p_\rho}(X',\F_\rho) \arrow[r, "\smile_c"] \arrow[d,"E'^{p_1} \times \cdots \times E'^{p_\rho}"] \arrow[u,"E^{p_1} \times \cdots \times E^{p_\rho}"]  & 
			C^{p_1 + \cdots + p_\rho}(X',\F_1 \otimes \cdots \otimes \F_\rho) \arrow[d, "E'^{p_1 + \cdots + p_\rho}"] \arrow[u, "E^{p_1 + \cdots + p_\rho}"]  \\
			C^{p_1}(X'',\F_1) \times \cdots \times C^{p_\rho}(X'',\F_\rho) \arrow[r, "\smile_c"] & 
			C^{p_1 + \cdots + p_\rho}(X'',\F_1 \otimes \cdots \otimes \F_\rho) 
		\end{tikzcd}
	\end{equation}
\end{lemma}
\begin{proof}
	Computing cup products with nontrivial local coefficients amounts to computing $\sigma' \smile_c \tau'$ as before and then acting on the local vectors. The corresponding formula is given in Eq.~\eqref{eq:cup_local_vector}. By Lemma~\ref{lemma:sheaf}, for every incidence $\sigma' \prec \tau'$ and its lifts, the structure maps of the cellular sheaves $\F_1,\ldots,\F_\rho$ on $X_i'$, $X_i$, and $X_{i+1}'$ are represented by the same submatrices. Since the canonical extensions defined by $E^\#$ duplicate the local vectors, the same computation applies to all three complexes, proving the claim.
\end{proof}

For cocycles $x_1'$ and $x_2'$, their cup product is evaluated on $\pi'$ as
\begin{align}
	(x_1' \smile_c x_2')(\pi') = \sum_{\sigma' \smile_c \tau' = \pi'} x_1'(\sigma') \smile_c x_2'(\tau').
\end{align}
Lemma~\ref{lemma:cup_induct_3} therefore applies directly to the canonical extensions of $x_1'$ and $x_2'$, and the same argument extends to multiple cup products. This observation leads to Theorem~\ref{thm:induct_cup}. However, the conclusion need not hold for arbitrary lifts $x_j$ satisfying only $P^{p_j}x_j=x_j'$.

We now collect all the assumptions before stating our first main result. Let $G_0$ be an $n$-regular graph with a fixed number $n'$ of vertices. Let $\HH_i$ be finite Abelian groups and $\HH_i'$ arbitrary finite groups, all of odd order. Let $\F_1,\ldots,\F_\rho$ be generated by classical local codes. Using these structures, we define the inductive lifting sequence
\begin{align}
	\bigl\{ \bigl( C^\bullet(X_i',\F_j), C^\bullet(X_i,\F_j) \bigr) \bigr\}_i, \quad j\in[\rho].
\end{align} 

\begin{theorem}\label{thm:induct_cup}
	If the initial cochain complexes $C^\bullet(X', \F_1), \ldots,C^\bullet(X', \F_\rho)$ support a nontrivial cup product gate of degree $\rho$, then the gate is inherited by all the subsequent cochain complexes. Specifically, let $x_1' \in C^{p_1}(X',\F_1)$, $\ldots$, $x_\rho' \in C^{p_\rho}(X', \F_\rho)$ be cocycles such that $p_1 + \cdots + p_\rho \leq t$ and 
	\begin{align}
		\langle x_1' \smile_c \cdots \smile_c x_\rho', \ \xi' \rangle \neq 0 \in \mathbb{F}_q,
	\end{align}
	for some cycle $\xi' \in C_{p_1 + \cdots + p_\rho}(X', \F_1 \otimes \cdots \otimes \F_\rho )$. Define the canonical extensions on $X$ and $X''$, respectively, by
    \begin{align}
    	x_1&=E^{p_1}x_1',\quad \ldots,\quad x_\rho=E^{p_\rho}x_\rho',\\
    	x_1''&=E'^{p_1}x_1',\quad \ldots,\quad x_\rho''=E'^{p_\rho}x_\rho'.
    \end{align}
	Let $r=p_1+\cdots+p_\rho$ and set $\xi=P_r\xi'$ and $\xi''=P_r'\xi'$, where $P_r=(P^r)^T$ and $P_r'=(P'^r)^T$ are the degree-$r$ components of the chain maps dual to $P^\#$ and $P'^\#$, respectively. Under the coordinate identifications in Definition~\ref{def:cochain}, they are represented by the same matrices as $E^r$ and $E'^r$. Since they are chain maps, $\xi$ and $\xi''$ are cycles. Then
	\begin{align}
		\langle x_1 \smile_c \cdots \smile_c x_\rho, \ \xi \rangle 
		= \langle x_1' \smile_c \cdots \smile_c x_\rho', \ \xi' \rangle 
		= \langle x_1'' \smile_c \cdots \smile_c x_\rho'', \ \xi'' \rangle \neq 0 \in \mathbb{F}_q.
	\end{align}
	By induction, the nontrivial invariant form persists throughout the sequence.
\end{theorem}

\begin{corollary}\label{coro:induct_cup}
	Under the same assumptions as in Theorem~\ref{thm:induct_cup}, in an inductive sequence, we further assume that
	\begin{align}
		\vert \HH_i \vert = \exp\left( \Theta( \prod_{j \leq i - 1} \vert \HH_j' \vert ) \right)
	\end{align}
	and all lifted graphs are almost Ramanujan. Let local codes be two-way product-expanding local codes with linear distance as in Ref.~\cite{KP2025Extendable}. If $C^\bullet(X_1', \F_1), \ldots,C^\bullet(X_1', \F_\rho)$ admit a nontrivial cup product gate of degree $\rho$, then the exponential lifting sequences produce families of almost-good qLDPC codes by Ref.~\cite{Dinur2024sheaf}, while Theorem~\ref{thm:induct_cup} shows that they inherit nontrivial cup product gates of the same degree.
\end{corollary}

The same result extends to qLTCs for which the cochain complexes have length at least $5$, or equivalently, $\dim X_i = \dim X_i' \geq 4$, with the physical qudits placed on cochains of degree at least $2$.

\begin{remark}
Theorem~\ref{thm:induct_cup} and Corollary~\ref{coro:induct_cup} extend to any invariant form satisfying the compatibility condition in Lemma~\ref{lemma:cup_induct_3}. The invariant form in Theorem~\ref{thm:cup_cubical} defined by pairing cycles and cocycles is one such example. It is of considerable interest to investigate the situation of general invariant forms.
\end{remark}


\section{Logical action on hypergraph product codes}\label{sec:proof_HGP}

With a nontrivial invariant form, the next step is to explore its logical action. This task turns out to be far more complicated, and we treat HGP codes and sheaf codes separately. As highlighted in Secs.~\ref{sec:cubical} and~\ref{sec:induct_boundary}, given an inductive lifting sequence $\{C^\bullet(X_i', \F), C^\bullet(X_i, \F) \}$, the sequential lifting sequence is actually a family of HGP codes, whereas the exponential lifting sequence gives sheaf codes on high-dimensional expanders. We study HGP codes in this section and postpone sheaf codes to the next.

\subsection{Fundamental difficulties in computing the cup products}\label{sec:difficulties}

By Theorem~\ref{thm:cup_cubical}, 
\begin{align}
	T_\xi (x_1, \ldots,x_\rho) = \langle x_1 \smile_c \cdots \smile_c x_\rho, \xi \rangle
\end{align}
is invariant under the deformation of coboundaries. The following simple condition determines the existence of $\xi$ for which $T_\xi$ is nontrivial.

\begin{proposition}\label{prop:cup_nontrivial}
	There exists a cycle $\xi$ such that $T_\xi$ is nontrivial if and only if the cup product induces a nontrivial multilinear map on the cohomology groups:
	\begin{align}\label{eq:cup_cohomology}
		\text{--}\smile_c \text{--} \cdots \text{--} \smile_c \text{--} : H^{p_1}(X, \F_1) \times \cdots \times H^{p_\rho}(X, \F_\rho) \rightarrow H^{p_1 + \cdots + p_\rho}(X, \F_1 \otimes \cdots \otimes\F_\rho ).
	\end{align}
\end{proposition} 
\begin{proof}
	Given any cocycles $x_{p_1}, \ldots,x_{p_\rho}$ that represent nontrivial cohomology classes, the Leibniz rule indicates that $x_{p_1} \smile_c \cdots \smile_c x_{p_\rho}$ is a cocycle. It may or may not be a coboundary. If it is a coboundary, it lies in the column span of the corresponding coboundary operator. Equivalently, after transposition, it lies in the row span of the boundary operator $\partial_{p_1 + \cdots + p_\rho}$. Consequently, any cycle $\xi$ must be perpendicular to that cup product. If the cup product does not lie in the row span, linear duality yields a cycle $\xi$ for which the associated invariant form is nontrivial.
\end{proof}

From this perspective, what is proved in Theorem~\ref{thm:induct_cup} is that any nontrivial cup product of classes from $H^\bullet(X_1', \F_1), \ldots, H^\bullet(X_1', \F_\rho)$ propagates to the remaining terms in the sequence. Moreover, by basic linear algebra, the multilinear map in Eq.~\eqref{eq:cup_cohomology} uniquely induces a linear map on the tensor product:
\begin{align}
	F_{\text{cup}}: H^{p_1}(X, \F_1) \otimes \cdots \otimes H^{p_\rho}(X, \F_\rho) \rightarrow H^{p_1 + \cdots + p_\rho}(X, \F_1 \otimes \cdots \otimes\F_\rho ).
\end{align}
The rank of $F_{\text{cup}}$, denoted by $\rk F_{\text{cup}}$, provides partial information about the possible number of nontrivial invariant forms. We denote by $M$ a matrix whose rows are cocycle representatives of a basis of $\im F_{\text{cup}}$. By construction, $M$ has $\rk F_{\text{cup}}$ rows.

For simplicity, we denote by $\partial$ the boundary operator acting on $(p_1 + \cdots + p_\rho)$-chains. Stacking $M$ above $\partial$, define
\begin{align}\label{eq:M}
	\tilde{M} = \begin{pmatrix} M \\ \partial \end{pmatrix}.
\end{align}
By construction, $\ker \tilde{M} = \ker \partial \cap \ker M \subseteq \ker \partial$, and any $\xi \in \ker \partial \setminus \ker \tilde{M}$ pairs nontrivially with at least one row of $M$. Because the rows of $M$ represent linearly independent cohomology classes, their restrictions to $\ker \partial$ are linearly independent. Hence, $\dim \ker \partial - \dim \ker \tilde{M} = \rk F_{\text{cup}}$. Extending any basis of $\ker \tilde{M}$ to a basis of $\ker \partial$ yields $\rk F_{\text{cup}}$ additional cycles $\{\xi_i\}_{i = 1}^{\rk F_{\text{cup}}}$. After a suitable change of basis, we have
\begin{align}\label{eq:rank}
	M \begin{pmatrix} \xi_1 & \cdots & \xi_{\rk F_{\text{cup}}} \end{pmatrix} = I_{\rk F_{\text{cup}}}.
\end{align} 

This suggests that one may choose $\rk F_{\text{cup}}$ cycles giving rise to linearly independent invariant forms. However, this description is not yet complete and has three major limitations:
\begin{itemize}
	\item Computing $\rk F_{\text{cup}}$ can be difficult for a complicated combinatorial complex such as $X$ equipped with local coefficients.
	
	\item Since $F_{\text{cup}}$ is defined on general tensor products of cohomology groups, a single row of $M$ may correspond to a linear combination of cup products rather than to an individual cup product.
	
	\item Even if each row corresponds to an individual cup product, there is no guarantee that the constituent cocycles represent distinct cohomology classes or admit disjoint logical representatives. This makes it difficult to determine the subrank introduced in Sec.~\ref{sec:CSS}.
\end{itemize}
We resolve these problems using explicit logical representatives for HGP codes and generalize the method to sheaf codes in Sec.~\ref{sec:proof_sheaf}.


\subsection{Hypergraph product with scalar coefficients}\label{sec:HGP_scalar}

To simplify the discussion and notation, we focus on the case of $\CCZ$ gates on three-dimensional HGP codes. We first consider scalar coefficients in $\mathbb{F}_q$ and then local coefficients in $\mathcal{F}$. 

Let $C^\bullet(X',\mathbb{F}_q)$ be the cochain complex of the three-dimensional cubical complex $X' = (G_0 \times K_2)^3$, where $G_0$ is an $n$-regular graph and the lift is trivial, $\HH = \{e\}$. Let $\mathbb{F}_q^{V_i}$ and $\mathbb{F}_q^{E_i}$ denote the coordinate spaces on the vertices and edges, respectively, of $G_0 \times K_2$. The subscript $i = 1,2,3$ labels the direction. For simplicity, we use the same graph in all three directions, so
\begin{align}
	|V_i| \equiv 2 |V_0|, \quad |E_i| \equiv 2 |E_0| = \frac{n}{2} 2 |V_0|,  \ i = 1,2,3,
\end{align}
where the factor of $2$ accounts for the double cover. Here and below, $V_i$ and $E_i$ also denote their respective cardinalities when used as matrix dimensions.
Then $H_i:\mathbb{F}_q^{V_i}\to\mathbb{F}_q^{E_i}$ is the graph coboundary operator, and the cochain complex is therefore the hypergraph product of 1-cochain complexes (cf. Eq.~\eqref{eq:3D_HGP}):
\begin{align}\label{eq:3D_HGP_scalar}
	\begin{aligned}
		\mathbb{F}_q^{E_1 \times E_2 \times E_3} \xleftarrow{\delta^2}
		&\ \mathbb{F}_q^{V_1 \times E_2 \times E_3} \oplus \mathbb{F}_q^{E_1 \times V_2 \times E_3} \oplus
		\mathbb{F}_q^{E_1 \times E_2 \times V_3} \xleftarrow{\delta^1} \\
		&\ \mathbb{F}_q^{E_1 \times V_2 \times V_3} \oplus \mathbb{F}_q^{V_1 \times E_2 \times V_3} \oplus \mathbb{F}_q^{V_1 \times V_2 \times E_3}
		\xleftarrow{\delta^0} \mathbb{F}_q^{V_1 \times V_2 \times V_3},
	\end{aligned}
\end{align}
where
\begin{align}\label{eq:HGP_coboundary_scalar}
	\delta^1 & =
	\begin{pmatrix}
		0 & I_{V_1}\otimes I_{E_2}\otimes H_3 & I_{V_1}\otimes H_2\otimes I_{E_3} \\
		I_{E_1}\otimes I_{V_2}\otimes H_3 & 0 & H_1\otimes I_{V_2}\otimes I_{E_3} \\
		I_{E_1}\otimes H_2\otimes I_{V_3} & H_1\otimes I_{E_2}\otimes I_{V_3} & 0
	\end{pmatrix}, \qquad
	\delta^0 =
	\begin{pmatrix}
		H_1\otimes I_{V_2}\otimes I_{V_3} \\ I_{V_1}\otimes H_2\otimes I_{V_3} \\ I_{V_1}\otimes I_{V_2}\otimes H_3
	\end{pmatrix},
\end{align}
and
\begin{align}
	\delta^2 =
	\begin{pmatrix}
		H_1\otimes I_{E_2}\otimes I_{E_3} & I_{E_1}\otimes H_2\otimes I_{E_3} & I_{E_1}\otimes I_{E_2}\otimes H_3
	\end{pmatrix}.
\end{align}
The primes are omitted for conciseness (cf. Definition~\ref{def:induct_lift}).

We now introduce notation for the logical representatives. Let $\{e_1^{i_1}\}_{i_1=1}^{k_1^T}$ be the standard basis vectors associated with the \emph{free variables} in the linear system $x^T H_1=0$. After ordering the edge coordinates so that the pivot variables come first and applying elementary column operations to $H_1$, we obtain the following reduced column-echelon form:
\begin{align}\label{eq:free_variable}
	\begin{pmatrix}
		I_{E_1-k_1^T} & 0 \\ J_1 & 0
	\end{pmatrix}, \qquad J_1 \in \mathbb{F}_q^{k_1^T\times(E_1-k_1^T)}.
\end{align}
The coordinates indexed by the rows of $J_1$ are the free variables, whereas those indexed by the rows of $I_{E_1-k_1^T}$ are the pivot variables. Thus, $e_1^{i_1}\in\mathbb{F}_q^{E_1}$, and these vectors span a complement of $\im H_1$. The free-coordinate positions are also referred to as an \emph{information set} or a set of \emph{punctures} in the study of classical and HGP codes~\cite{PK2023RobustlyTestable,Burton2022,Fu2025nogo}. We define $\{e_2^{i_2}\}$ and $\{e_3^{i_3}\}$ similarly.

When $G_0\times K_2$ is connected, its graph coboundary operator $H_i$ has rank $V_i-1$. Thus, the number of free variables, equivalently the number of associated standard basis vectors, is
\begin{align}\label{eq:kappa^T_1}
	k_i^T = \dim\ker H_i^T = E_i-\rk H_i = E_i - V_i + 1 = \left(\frac{1}{2}n-1\right)V_i+1.
\end{align}
If $G_0 \times K_2$ is disconnected, then $\dim \ker H_i > 1$, which can increase the number of logical qudits (see Eq.~\eqref{eq:HGP_parameters}) but reduce the code distance. We henceforth assume that $G_0 \times K_2$ is connected. Since column reduction is not unique, the free variables may occupy different coordinate positions, but their number is always $k_i^T$.

On the other hand, let $\{\zeta_i^{j_i}\}$ be any basis of $\ker H_i$. Then a canonical choice of the $X$ logical representatives of the three-dimensional HGP code can be
\begin{align}\label{eq:HGP_basis_1}
	e_1^{i_1} \otimes \zeta_2^{j_2} \otimes \zeta_3^{j_3}, \quad
	\zeta_1^{j_1} \otimes e_2^{i_2} \otimes \zeta_3^{j_3}, \quad
	\zeta_1^{j_1} \otimes \zeta_2^{j_2} \otimes e_3^{i_3}. 
\end{align}
By definition, $k_i  = \dim \ker H_i = 1$ and hence $\zeta_i^{j_i}$ is unique (up to scalars), as the all-ones vector in $\mathbb{F}_q^{V_i}$ involving all vertices in the graph. There are $k_1^T + k_2^T + k_3^T$ independent logical basis elements, which is consistent with the code dimension in Eq.~\eqref{eq:HGP_parameters}. Equation~\eqref{eq:HGP_basis_1} is simply the graph-based HGP case. More advanced forms of the logical representatives will be gradually introduced.

These canonical logical representatives help determine the logical action of the invariant forms. To apply the cup product formulas in Secs.~\ref{sec:cup} and~\ref{sec:example}, we first express the logical representatives in Eq.~\eqref{eq:HGP_basis_1} as cochains on the cubical complex. A standard basis vector $e_1^{i_1}\in\mathbb F_q^{E_1}$ is represented by a single edge $[v_1;a_1]$ in the first component of $X'$. By contrast, $\zeta_2^{j_2}$ is represented by vertices $[v_2;b_2]$ with $b_2=0,1$. That is, we can write, for example,
\begin{align}\label{eq:HGP_basis_2}
	e_1^{i_1}\otimes\zeta_2^{j_2}\otimes\zeta_3^{j_3} = \sum_{v_2,v_3,b_2,b_3} [v_1;a_1]\otimes[v_2;b_2]\otimes[v_3;b_3] 
	= \sum_{v_2,v_3,b_2,b_3} [g=(v_1,v_2,v_3);a_1,b_2,b_3].
\end{align}
Here $[v_1;a_1]$ is fixed, while $v_2,v_3,b_2,b_3$ range over the support of $\zeta_2^{j_2}\otimes\zeta_3^{j_3}$. The other two types of $X$-type logical representatives can be represented in the same way. Since each $\zeta_i^{j_i}$ is currently the all-ones vector, we suppress the indices and simply write $\zeta$.

Let
\begin{align}\label{eq:HGP_basis_3}
	L^{i_1} \coloneq e_1^{i_1} \otimes \zeta \otimes \zeta, \quad
	L^{i_2} \coloneq \zeta \otimes e_2^{i_2} \otimes \zeta, \quad
	L^{i_3} \coloneq \zeta \otimes \zeta \otimes e_3^{i_3}.
\end{align}
A crucial observation is that, since $e_1^{i_1}$, $e_2^{i_2}$, and $e_3^{i_3}$ are represented by fixed edges in their respective directions, the cubical cup product rule from Sec.~\ref{sec:example} gives
\begin{align}\label{eq:HGP_cup}
	\begin{aligned}
		L^{i_1} \smile_c L^{i_2} \smile_c L^{i_3}
		={} & \Big( L^{i_1}\big([(v_{i_1},v_{i_2},v_{i_3});a_{i_1},0,0]\big) \\ 
		& \quad \cdot L^{i_2}\big([(a_{i_1}\cdot v_{i_1},v_{i_2},v_{i_3}); 1,a_{i_2},0]\big) \\
		& \quad \cdot L^{i_3}\big([(a_{i_1}\cdot v_{i_1}, a_{i_2}\cdot v_{i_2},v_{i_3}); 1,1,a_{i_3}]\big) \Big) \\
		& \quad \cdot [(v_{i_1},v_{i_2},v_{i_3}); a_{i_1},a_{i_2},a_{i_3}].
	\end{aligned}
\end{align}
Since $\zeta$ is the all-ones vector, the coefficient above equals $1$. Therefore, the cup product is nontrivial and is supported only on the $3$-cube $[(v_{i_1},v_{i_2},v_{i_3});a_{i_1},a_{i_2},a_{i_3}]$. By the same reasoning,
\begin{align}
	\begin{aligned}
		& L^{i_1} \smile_c L^{i_3} \smile_c L^{i_2}, \quad
		L^{i_2} \smile_c L^{i_3} \smile_c L^{i_1}, \quad
		L^{i_2} \smile_c L^{i_1} \smile_c L^{i_3}, \\
		& L^{i_3} \smile_c L^{i_1} \smile_c L^{i_2}, \quad
		L^{i_3} \smile_c L^{i_2} \smile_c L^{i_1}
	\end{aligned}
\end{align}
are all nontrivial.

\begin{example}\label{example:toric}
	An enlightening example is the 3D toric code, where $G_0 \times K_2$ is a cycle as in \eqref{eq:K_3_double} and the graph coboundary operators have $k_i = k_i^T \equiv 1$. Let $e = [v;a]$ be the free variable of $H_i^T$. Then $L^{1}, L^{2}, L^{3}$ are uniquely determined and correspond to three spatial directions. By the previous result,
	\begin{align}
		L^1 \smile_c L^2 \smile_c L^3 = [(v,v,v); a,a,a]
	\end{align}
	and the preceding identities indicate that there still $5$ nontrivial cup products. However, if one counts all of them, each logical representative from each code block will be used twice. A genuine counting on the number of disjoint logical operation and subrank reads distinct tuples of logical representatives, which is $3$ here.
\end{example}

To complete the construction of the invariant form, we now construct $\xi$. By Eq.~\eqref{eq:HGP_coboundary}, 
\begin{align}\label{eq:xi_1}
	\xi \in \ker \partial_3 = \ker \begin{pmatrix} \partial \otimes I_E \otimes I_E \\ I_E \otimes \partial \otimes I_E \\ I_E \otimes I_E \otimes \partial \end{pmatrix} = \ker \partial \otimes \ker \partial \otimes \ker \partial \subset C_3(X',\mathbb{F}_q),
\end{align}
where $\partial=H_i^T$; here the $H_i$ are identical, and $i$ labels only the coordinate direction. The reduced column-echelon form of $H_1$ satisfies
\begin{align}
	\begin{pmatrix} J_1 & I_{k_1^T} \end{pmatrix} \cdot \begin{pmatrix} I_{E_1 - k_1^T} & 0 \\ J_1 & 0	\end{pmatrix} = 0.
\end{align}
Namely, the rows of $( J_1 \quad I_{k_1^T} )$ span $\ker \partial$ and we denote them by $\eta_{i_1}$. Let $\eta_{i_2}$ and $\eta_{i_3}$ be defined similarly. Then
\begin{align}\label{eq:eta}
	\eta_{i_1} \otimes \eta_{i_2} \otimes \eta_{i_3} = [v_{i_1}; a_{i_1}] \otimes [v_{i_2}; a_{i_2}] \otimes [v_{i_3}; a_{i_3}] + \cdots \in \ker \partial_3,
\end{align}
where all other terms in the expansion involve at least one pivot variable. Since the cup products $L^{i_1} \smile_c L^{i_2} \smile_c L^{i_3}$ are supported only on cubes formed by free variables, 
\begin{align}
	\langle L^{i_1} \smile_c L^{i_2} \smile_c L^{i_3}, \ \eta_{i_1'} \otimes \eta_{i_2'} \otimes \eta_{i_3'} \rangle = \delta_{i_1,i_1'} \delta_{i_2,i_2'} \delta_{i_3,i_3'}.
\end{align}
The associated invariant form yields a so-called \emph{addressable gate} that acts only on a single triple of logical qudits.

Alternatively, let $\xi \coloneq \sum_{i} \eta_{i} \otimes \eta_{i} \otimes \eta_{i}$ be summed over all free variables. Then the invariant form
\begin{align}\label{eq:diagonalized}
	T_\xi(  L^{i_1}, L^{i_2}, L^{i_3}) = (\delta_{i_1,i_2,i_3})
\end{align}
is a strictly ``diagonalized'' 3-tensor with subrank equal to $k_i^T \equiv \left( \frac{1}{2} n - 1 \right) V + 1$. 

\begin{theorem}\label{thm:cup_HGP_1}
	Let $X'$ be the $t$-fold Cartesian product of a connected double cover of an $n$-regular connected graph with $V_0$ vertices, and let $C^\bullet(X',\mathbb{F}_q)$ be its cochain complex. Consider the CSS code defined by
	\begin{align}
		C^0(X',\mathbb{F}_q) \rightarrow C^1(X',\mathbb{F}_q) \rightarrow C^2(X',\mathbb{F}_q),
	\end{align}
	with physical qudits placed on the $1$-cubes. Set $V=2V_0$. The blocklength and code dimension are
	\begin{align}
		N &= \dim C^1(X',\mathbb{F}_q) = \frac{t}{2}nV^t, \\
		k &= t\left[ \left(\frac{n}{2}-1\right)V+1 \right].
	\end{align}
	We partition the canonical $X$-type logical representatives into $t$ groups as in Eq.~\eqref{eq:HGP_basis_3}, and write $L^{i_j}$ for a representative indexed by $i_j$ in the $j$-th group.
	
	For every permutation $\alpha\in S_t$ and every index tuple $(i_1',\ldots,i_t')$, there exists an addressable invariant form $T_{\alpha,i_1',\ldots,i_t'}$ such that, for all choices of $(i_1,\ldots,i_t)$,
	\begin{align}
		T_{\alpha,i_1',\ldots,i_t'} \big( L^{i_{\alpha(1)}}, \ldots, L^{i_{\alpha(t)}} \big) = \delta_{i_1,i_1'} \cdots \delta_{i_t,i_t'}.
	\end{align}
	Moreover, there exists an invariant form $T_\xi$ satisfying
	\begin{align}
		T_\xi\big( L^{i_1},\ldots, L^{i_t} \big) = \delta_{i_1,\ldots,i_t},
	\end{align}
	where $\delta_{i_1,\ldots,i_t}=1$ if $i_1=\cdots=i_t$ and is zero otherwise. The subrank is $k = \Theta(N^{1/t})$.
\end{theorem}

\begin{remark}\label{remark:HGP_1}
	As in Example~\ref{example:toric}, when $n=2$ and the connected double cover is the cycle $C_V$, Theorem~\ref{thm:cup_HGP_1} reduces to the $t$-dimensional toric code with parameters
	\begin{align}
		[\![N,k=t,d=(N/t)^{1/t}]\!],
	\end{align}
	and the subrank of $T_\xi$ is $3$.
	
	For $n\geq 3$, suppose that the double covers form a high-girth expander family with girth $\Theta(\log V) = \Theta(\log N)$, which is also the distance of the classical code $\ker H_i^T$ (see Sec.~\ref{sec:SS_code}). By contrast, $\ker H_i$ is one-dimensional and spanned by the all-ones vector, and has distance $V=\Theta(N^{1/t})$. Equation~\eqref{eq:HGP_parameters} therefore gives the HGP code parameters
	\begin{align}
		[\![N, k=\Theta(N^{1/t}), d=\Theta(\log N)]\!],
	\end{align}
	and the subrank of $T_\xi$ is $k$. Compared with the toric code family, this construction has a larger code dimension, but its distance remains only logarithmic. One can further increase the code dimension by placing physical qudits on $p$-cubes and considering
	\begin{align}
		C^{p-1}(X',\mathbb{F}_q) \rightarrow C^p(X',\mathbb{F}_q) \rightarrow C^{p+1}(X',\mathbb{F}_q).
	\end{align}
	The resulting blocklength and code dimension are
	\begin{align}
		N_p &= \binom{t}{p} \left(\frac{n}{2}\right)^p V^t, \\
		k_p &= \binom{t}{p} \left[ \left(\frac{n}{2}-1\right)V+1 \right]^p.
	\end{align}
	However, an invariant polynomial constructed from cup products can have degree at most $\lfloor t/p\rfloor$, and this modification does not improve the polylogarithmic distance scaling.
\end{remark}

To address this dilemma, we next incorporate local codes into the construction. Remark~\ref{remark:HGP_2} also discusses additional strategies from prior work for improving the code parameters.


\subsection{Hypergraph product with local codes}\label{sec:HGP_local}

We now consider $C^\bullet(X', \F)$. We will also encounter tensor products of cellular sheaves in this case, but it suffices to take all of their factors to be $\F$ in this subsection. With local coefficients, $H_i^T \in \mathbb{F}_q^{m_i V_i \times E_i}$. They are no longer simple graph boundary operators but instead take the form of the Tanner--Sipser--Spielman parity-check matrices described in Sec.~\ref{sec:SS_code}. Then
\begin{align}\label{eq:kappa^T_2}
	k_i^T = \dim\ker H_i^T = E_i - \rk H_i \geq E_i - m_i V_i = \left( \frac{1}{2}n - m_i \right) V_i,
\end{align} 
where, for now, we use a local parity-check matrix $h_i \in \mathbb{F}_q^{m_i\times n}$ with $m_i < \frac{1}{2}n$. Different choices will be considered in later sections. The canonical logical representatives remain of the form given in Eq.~\eqref{eq:HGP_basis_1}, but $\zeta_i^{j_i}$ need to be redefined to incorporate the local coefficients. As described in Remark~\ref{remark:SS_code} and Sec.~\ref{sec:DLV_code}, the local coefficients assigned to edges are trivial. Thus, for instance, $e_1^{i_1}$ is still represented by an edge $[v_1;a_1]$ corresponding to one of the $k_1^T$ free variables.

Each row of $H_1$ indexed by an edge $[v_1;a_1]$ is of the form
\begin{align}\label{eq:H_1}
	\begin{pmatrix}
		\cdots & h_1^T(a_1, \text{-}) & \cdots &  h_1^T(a_1, \text{-}) & \cdots
	\end{pmatrix}
\end{align}
where $h_1^T(a_1,\text{-})$ acts on the local coefficient spaces of both $(v_1,0)$ and $(a_1\cdot v_1,1)$ (see Sec.~\ref{sec:DLV_code}). Fix $z\in\mathbb{F}_q^{m_1}$ and define
\begin{align}
	\zeta_1^{j_1}[v_1,b_1]\equiv z
\end{align}
for every $[v_1,b_1]$. Because $\zeta_1^{j_1}$ assigns the same vector $z$ to both endpoints of every edge, the two contributions in each row of $H_1$ cancel in characteristic $2$. Hence, $\zeta_1^{j_1}\in\ker H_1$. We define $\zeta_2^{j_2}$ and $\zeta_3^{j_3}$ in the same way.

We now choose $z_1,z_2,z_3$ to ensure that the cup product is nontrivial. Specifically, given three free variables $e_1^{i_1}=[v_{i_1};a_{i_1}]$, $e_2^{i_2}=[v_{i_2};a_{i_2}]$, and $e_3^{i_3}=[v_{i_3};a_{i_3}]$, we choose $z_j\in\mathbb{F}_q^{m_j}$ such that
\begin{align}\label{eq:z_i}
	h_1^T(a_{i_1}, \text{-}) z_1 = h_2^T(a_{i_2}, \text{-}) z_2 = h_3^T(a_{i_3}, \text{-}) z_3 = 1 \in \mathbb{F}_q.
\end{align}
For each $j=1,2,3$, we use $\zeta_j^{i_j}$ to denote the constant element of $\ker H_j$ determined by $z_j$ and associated with the free variable $e_j^{i_j}$. Thus, the superscript $i_j$ records the free variable label rather than a generic basis index of $\ker H_j$.

Unlike the construction in Eq.~\eqref{eq:HGP_basis_3} in which all $\zeta$ are the same, we define
\begin{align}\label{eq:HGP_basis_F}
	L_{1}^{i_1,i_2,i_3} \coloneq e_1^{i_1} \otimes \zeta_2^{i_2} \otimes \zeta_3^{i_3}, \quad
	L_{2}^{i_1,i_2,i_3} \coloneq \zeta_1^{i_1} \otimes e_2^{i_2} \otimes \zeta_3^{i_3}, \quad
	L_{3}^{i_1,i_2,i_3} \coloneq \zeta_1^{i_1} \otimes \zeta_2^{i_2} \otimes e_3^{i_3}.
\end{align}
For all vertices, these cochains satisfy
\begin{align}\label{eq:zeta_i}
	\zeta_1^{i_1} [v_1, b_1] \equiv z_1, \quad \zeta_2^{i_2} [v_2, b_2] \equiv z_2, \quad
	\zeta_3^{i_3} [v_3, b_3] \equiv z_3.
\end{align}
The cubical cup product rule then shows that
\begin{align}\label{eq:HGP_cup_F}
	L_{1}^{i_1,i_2,i_3} \smile_c L_{2}^{i_1,i_2,i_3} \smile_c L_{3}^{i_1,i_2,i_3}
\end{align}
yields the unique $3$-cube $[ (v_{i_1}, v_{i_2}, v_{i_3}); a_{i_1}, a_{i_2}, a_{i_3} ]$ with local coefficient
\begin{align}
	\begin{aligned}
		& \Big[ \big( h_2^T(a_{i_2},\text{-}) \otimes h_3^T(a_{i_3},\text{-}) \big)
		\big( \zeta_2^{i_2}[v_{i_2},0] \otimes \zeta_3^{i_3}[v_{i_3},0] \big) \Big] \\
		& \quad\cdot \Big[ \big( h_1^T(a_{i_1},\text{-}) \otimes h_3^T(a_{i_3},\text{-}) \big)
		\big( \zeta_1^{i_1}[a_{i_1}\cdot v_{i_1},1] \otimes \zeta_3^{i_3}[v_{i_3},0] \big) \Big] \\
		& \quad\cdot \Big[ \big( h_1^T(a_{i_1},\text{-}) \otimes h_2^T(a_{i_2},\text{-}) \big)
		\big( \zeta_1^{i_1}[a_{i_1}\cdot v_{i_1},1] \otimes \zeta_2^{i_2}[a_{i_2}\cdot v_{i_2},1] \big) \Big] =1.
	\end{aligned}
\end{align}
The final equality follows from Eqs.~\eqref{eq:H_1}, \eqref{eq:z_i}, and~\eqref{eq:zeta_i}.

    
The most delicate point when dealing with $\F$ is the construction of the pairing cycle $\xi$. Analogously to Eq.~\eqref{eq:xi_1}, the relevant top-degree boundary operator is
\begin{align}
	\partial_3 \coloneq
	\begin{pmatrix}
		\mathring{H}_1^T \otimes I_E \otimes I_E \\
		I_E \otimes \mathring{H}_2^T \otimes I_E \\
		I_E \otimes I_E \otimes \mathring{H}_3^T
	\end{pmatrix}.
\end{align}
Hence, the pairing cycle must satisfy
\begin{align}\label{eq:xi_2}
	\xi \in \ker \partial_3 = \ker \mathring{H}_1^T \otimes \ker \mathring{H}_2^T \otimes \ker \mathring{H}_3^T 
	\subset C_3(X',\F^{\otimes 3}),
\end{align}
where $\mathring{H}_i$ is defined on the graph boundary operator but its constituent blocks are $(h_i(\text{-}, a_i))^{\otimes 3}$ from $\F^{\otimes 3}$ (see also Example~\ref{example:3D}). Intuitively, $H_i^T \in \mathbb{F}_q^{m_i V_i \times E_i}$ enlarges each row in $\partial$ by local codes, and they are further amplified in $\mathring{H}_i^T \in \mathbb{F}_q^{m_i^3 V_i \times E_i}$. Consequently, the free variables of $H_i^T$ may no longer be free for $\mathring{H}_i^T$. From the size of $\mathring{H}_i^T$, we see that its kernel may be trivial if $m_i/n$ is too large, in which case no nonzero pairing cycle $\xi$ exists. One may require $\frac{1}{2}n \geq m_i^3$, which implies
\begin{align}
	m_i \leq \Big( \frac{1}{2}n \Big)^{1/3} \implies \text{ the dimension of the local code } \geq n - \Big( \frac{1}{2}n \Big)^{1/3}. 
\end{align} 
However, by the Singleton bound for classical codes, this severely limits the distance of the local code $\ker h_i$, namely, $d(\ker h_i)\leq m_i+1$. Consequently, the condition in Theorem~\ref{thm:cLDPC} cannot be fulfilled in this regime, so $H_i^T$ is not guaranteed to define a good classical code and the code parameters in Remark~\ref{remark:HGP_1} may not be improved.


\medskip
We address this problem in Sec.~\ref{sec:cycles} with explicit examples and numerical computations. For now, we only study the properties of invariant forms when $\xi$ exists. If $\ker \mathring{H}_i^T$ is nontrivial, then it also has free variables in the sense of Eq.~\eqref{eq:free_variable}. Using these free variables, define
\begin{align}
	\eta_{i_1} \otimes \eta_{i_2} \otimes \eta_{i_3} = [v_{i_1}; a_{i_1}] \otimes [v_{i_2}; a_{i_2}] \otimes [v_{i_3}; a_{i_3}] + \cdots \in \ker \mathring{H}_1^T \otimes \ker \mathring{H}_2^T \otimes \ker \mathring{H}_3^T
\end{align}
in analogy with Eq.~\eqref{eq:eta} over $\mathbb{F}_q$. We define the $X$-type logical representatives using these free variables as in Eq.~\eqref{eq:HGP_basis_F} and, by construction,
\begin{align}
	\langle L_{1}^{i_1,i_2,i_3} \smile_c L_{2}^{i_1,i_2,i_3} \smile_c L_{3}^{i_1,i_2,i_3}, \ \eta_{i_1} \otimes \eta_{i_2} \otimes \eta_{i_3} \rangle \neq 0.
\end{align}
Although $[v_{i_1};a_{i_1}]$, $[v_{i_2};a_{i_2}]$, and $[v_{i_3};a_{i_3}]$ may not be free variables of $H_1^T$, $H_2^T$, and $H_3^T$, respectively, the relations $\zeta_j^{i_j}\in\ker H_j$ imply that each $L_j^{i_1,i_2,i_3}$ lies in $\ker\delta^1$. Their inequivalence is established by the following pairing argument. To be precise, let $\{[v_{i_1'}; a_{i_1'}]\}$ be any collection of free variables of $\mathring{H}_1^T$. We define 
\begin{align}
	L_{1}^{i_1',i_2',i_3'} = [v_{i_1'}; a_{i_1'}] \otimes \zeta_2^{i_2'} \otimes \zeta_3^{i_3'}
\end{align}
with arbitrary $\zeta_2^{i_2'}\in\ker H_2$ and $\zeta_3^{i_3'}\in\ker H_3$. Regardless of whether $L_{1}^{i_1',i_2',i_3'} \smile_c L_{2}^{i_1,i_2,i_3} \smile_c L_{3}^{i_1,i_2,i_3}$ is zero,
\begin{align}
	\langle L_{1}^{i_1',i_2',i_3'} \smile_c L_{2}^{i_1,i_2,i_3} \smile_c L_{3}^{i_1,i_2,i_3}, \ \eta_{i_1} \otimes \eta_{i_2} \otimes \eta_{i_3} \rangle = 0.
\end{align}
If their linear combination is equivalent to $L_{1}^{i_1,i_2,i_3}$, then
\begin{align}
	\Big\langle \Big( L_{1}^{i_1,i_2,i_3} + \sum_{i_1'} L_{1}^{i_1',i_2',i_3'} \Big) \smile_c L_{2}^{i_1,i_2,i_3} \smile_c L_{3}^{i_1,i_2,i_3}, \ \eta_{i_1} \otimes \eta_{i_2} \otimes \eta_{i_3} \Big\rangle = 1
\end{align}
contradicts the fact that 
\begin{align}
	& \Big \langle \delta^0 y \smile_c L_{2}^{i_1,i_2,i_3} \smile_c L_{3}^{i_1,i_2,i_3}, \ \eta_{i_1} \otimes \eta_{i_2} \otimes \eta_{i_3} \Big\rangle
	= \Big\langle \delta^2 \Big(y \smile_c L_{2}^{i_1,i_2,i_3} \smile_c L_{3}^{i_1,i_2,i_3} \Big), \ \eta_{i_1} \otimes \eta_{i_2} \otimes \eta_{i_3} \Big\rangle \notag \\
	= & \Big\langle y \smile_c L_{2}^{i_1,i_2,i_3} \smile_c L_{3}^{i_1,i_2,i_3} , \ \partial_3 ( \eta_{i_1} \otimes \eta_{i_2} \otimes \eta_{i_3}) \Big\rangle = 0.
\end{align}

Conclusively, let $G_0$ be an $n$-regular expander graph and let $h_i \in \mathbb{F}_q^{m_i \times n}$ be local codes such that $m_i < \frac{1}{2}n$ and Theorem~\ref{thm:cLDPC} holds. Let $C^\bullet(X',\F)$ be the cochain complex in which $X'$ is the $t$-fold Cartesian product of the double cover of $G_0$ and $\F$ is generated by the $h_i$. We extract the HGP code 
\begin{align}\label{eq:HGP_t}
	C^0(X',\F) \rightarrow C^1(X',\F) \rightarrow C^2(X',\F),
\end{align} 
which has parameters
\begin{align}
	[\![ N, k = \Omega(N^{1/t}), d = \Theta(N^{1/t}) ]\!].
\end{align} 
Let $\kappa_i^T = \dim \ker \mathring{H}_i^T$, which is also the number of free variables. We prepare $t$ groups of $X$-type logical representatives as in Eq.~\eqref{eq:HGP_basis_F} 
\begin{align}\label{eq:HGP_logical_local_codes}
	\{L_{1}^{i_1, \ldots,i_t}\}, \cdots\cdots, \{	L_{t}^{i_1', \ldots,i_t'}\}, \quad i_1,\ldots, i_1' \in [\kappa_1^T], \cdots\cdots, i_t,\ldots, i_t' \in [\kappa_t^T]	
\end{align}
Then the following theorem holds. 

\begin{theorem}\label{thm:cup_HGP_2}
	For any $t$-tuple $L_{1}^{i_1, \ldots,i_t}, \ldots,L_{t}^{i_1', \ldots,i_t'}$ of logical representatives from each group, there is an addressable invariant form such that
	\begin{align}
		T_{i_1'', \ldots,i_t''} (L_{1}^{i_1, \ldots,i_t}, \ldots,L_{t}^{i_1', \ldots,i_t'} ) = \delta_{i_1,\ldots,i_1',i_1''} \cdots \delta_{i_t,\ldots,i_t',i_t''}.
	\end{align}
	where $i_1'' \in [\kappa_1^T], \cdots, i_t'' \in [\kappa_t^T]$ also index the free variables. Moreover, there exists one invariant form $T_\xi$ such that
	\begin{align}
		T_\xi ( L_{1}^{i_1, \ldots,i_t}, \ldots,L_{t}^{i_1', \ldots,i_t'} ) = \delta_{i_1,\ldots,i_1'} \cdots \delta_{i_t,\ldots,i_t'}.
	\end{align}
	Its subrank is lower-bounded by $\min\{ \kappa_1^T, \ldots,\kappa_t^T \}$.
\end{theorem}

Note that in defining $T_\xi$, we have to align the free variables of $\mathring{H}_1^T, \ldots,\mathring{H}_t^T$. Alternatively, if $h_1=\cdots=h_t$, then these matrices are identical.

\begin{remark}\label{remark:HGP_2}
	We note that the above HGP code encodes $k = \Omega(N^{1/t})$ logical qudits. In the worst case, it may have the same scaling as the previous case without local codes. This is due to the fact that a linear lower bound on $k_i = \dim \ker H_i$ is missing (see Eq.~\eqref{eq:HGP_parameters}). Although $H_i^T$ are the parity-check matrices of good classical codes such that $k_i^T, d_i^T$ and even $d_i$ are proportional to $V$~\cite{SS1996,PK2021}, a general lower bound on $k_i$ is still unknown. There is one method proposed in Ref.~\cite{Zemor2014} that can achieve constant coding rate: for example, we take products between $H_1:\mathbb F_q^{M_1}\rightarrow\mathbb F_q^{N_1}$ and $H_2^T:\mathbb F_q^{N_2}\rightarrow\mathbb F_q^{M_2}$.
    Then the parameters become
	\begin{align}
		[\![ \Theta(N^2), k_1 k_2 + k_1^T k_2^T = \Theta(N^2), \Theta(N) ]\!].
	\end{align}
    However, this choice creates a conceptual difficulty in defining cup products: the physical qudits are associated with both $0$- and $2$-cubes, whereas the stabilizer checks are associated with $1$-cubes. In addition, there are hyperbolic codes that admit constant coding rates and cup product gates, while the distance has a logarithmic scaling~\cite{Zhu2023,RainbowCode,Breuckmann2024Cups}. 
	
	There is another different framework using simplicial complexes~\cite{Zhu2025A,Zhu2025B}: to create HGP codes with constant coding rates and larger distances, symmetrized parity-check matrices $\bar{H}_i \coloneq H_i H_i^T$, which achieve optimal $\bar{k}_i,\bar{k}_i^T,\bar{d}_i,\bar{d}_i^T$, are used. The manifest structure of cubical complexes would be lost, but the code-to-manifold mapping~\cite{FreedmanHastings2021} can be applied, and then the triangulation on manifolds yields simplicial complexes which still support cup products.
\end{remark}


\subsection{Hypergraph product with lifts}\label{sec:HGP_lift}

In Theorem~\ref{thm:cup_HGP_2}, the subrank of $T_\xi$ is bounded by the unknown number of free variables of $\mathring{H}_i^T$, but we can resolve this issue with more structured cup products. To be precise, recall the cup products defined in Eq.~\eqref{eq:HGP_cup}. Suppose we vary the third free variable $i_3'$ but fix $i_1,i_2$. Then
\begin{align}\label{eq:HGP_cup'}
	L^{i_1} \smile_c L^{i_2} \smile_c L^{i_3'} = [ (v_{i_1}, v_{i_2}, v_{i_3'}); a_{i_1}, a_{i_2}, a_{i_3'} ].
\end{align}
On the other hand, $\xi$ is chosen specifically by the free variables to ensure that $T_\xi(L^{i_1},L^{i_2},L^{i_3'}) = 0$. This support structure allows the subrank of $T_\xi$ to be counted directly, whereas estimating the subrank is difficult for  general tensors. This is also the case for cup products in Eq.~\eqref{eq:HGP_cup_F}, and both Theorem~\ref{thm:cup_HGP_1} and~\ref{thm:cup_HGP_2} rely on appropriate choices of $\xi$ in Eqs.~\eqref{eq:eta} and~\eqref{eq:xi_2}, respectively. 

We seek cup products satisfying
\begin{align}\label{eq:HGP_cup_L}
	L_{1}^{j_1} \smile_c L_{2}^{j_2} \smile_c L_{3}^{j_3} = 0,
\end{align}
whenever the three indices are not all equal. This reduces the dependence on the choice of $\xi$. We utilize lifts to achieve this. The idea can also be extended to (almost) good qLDPC code families and we elaborate on the details in Sec.~\ref{sec:proof_sheaf}. 

Before proceeding, we give further discussions on the choice of local codes:
\begin{enumerate}
	\item By Theorem~\ref{thm:cLDPC}, the parity-check matrix $h_i$ should be of full rank and the associated code as well as the dual code are required to have good distances, in order to ensure that $d_i^T$ and $d_i$ are linear. This improves the HGP code distance relative to the version without local codes in Remark~\ref{remark:HGP_1}.
	
	\item Since $H_i^T \in \mathbb{F}_q^{m_i V \times E}$, it is conventional to take $m_i < \frac{1}{2} n$ so that $k_i^T = \dim \ker H_i^T$ will scale with the system size. As mentioned after Eq.~\eqref{eq:xi_2}, $m_i$ cannot be too small. Otherwise, several code parameter bounds forbid the local codes to have good distances. This range of $m_i$ is adopted in Sec.~\ref{sec:HGP_local} and the example in Sec.~\ref{sec:cycles}.
	
	\item Instead, we consider $\frac{1}{2} n < m_i < n$. Then $k_i = \dim \ker H_i$ begins to increase with the system size. The Gilbert--Varshamov bound confirms that the local codes with the desired (dual) distances still exist. Moreover, since $m_i < n$, a nonzero local codeword in $\ker h_i$ extends, through the double cover construction used in the proof of Theorem~\ref{thm:cycle}, to a nonzero vector in $\ker H_i^T$, so $k_i^T>0$. Then the corresponding HGP code in Eq.~\eqref{eq:HGP_t} now has parameters
	\begin{align}
		[\![ N, k = \Omega(N^{(t-1)/t}), d = \Theta(N^{1/t}) ]\!]
	\end{align} 
	because $k_i^T>0$ and each $k_j=\Omega(V)$ for $j\neq i$, although it remains unclear whether $k_i^T$ itself scales with the system size. Admittedly, a large $m_i$ makes it more difficult to find $\xi$, but, when combined with a large lift, it increases the coding rate and, as we will show, facilitates the construction of well-behaved cup products. This range of $m_i$ is adopted in this subsection and more specific choices of $m_i$ are discussed in Sec.~\ref{sec:HGP_like}.
\end{enumerate}

For concreteness, we now work in the three-dimensional case. Let $\hat{G}_0$ be an exponentially large lift of $G_0$ by a group $\HH$ such that the lift is still an expander graph. We still use the HGP $\hat{X}'$ defined by the $3$-fold Cartesian product of the double cover of $\hat{G}_0$. Let $C^\bullet(\hat{X}',\F)$ be the cochain complex. Let $\vert{\HH}\vert = l$. By definition, the coboundary operator $\hat{\delta}: \mathbb{F}_q^{l \cdot V} \to \mathbb{F}_q^{l \cdot E}$ of the double cover of $\hat{G}_0$ is defined by substituting some $l \times l$ permutation matrix for each nonzero entry $\delta(\tau,\sigma)$. In the same spirit as in Remark~\ref{remark:lift}, for each edge $\tau = [v;a]$ with endpoints $[v;0]$ and $[a \cdot v; 1]$, 
\begin{align}
	\text{ if } \delta(\tau,[v;0]) \mapsto I_l, 
	\text{ then }  \delta(\tau,[a \cdot v;1]) \mapsto \gamma_{(v, a \cdot v)}
\end{align}
where $\gamma_{(v, a \cdot v)}$ is the representation of the group element of $\HH$ acting on $(v, a \cdot v)$. As a reminder, we take the double cover of the lift in Sec.~\ref{sec:cubical}, but do not lift the double cover. Then the local codes are imposed to define $\hat{H}_j$. For instance, each row of $\hat{H}_1$ indexed by any lifted edge is of the form 
\begin{align}\label{eq:row}
	\begin{pmatrix}
		\cdots & h_1^T(a_1, \text{-}) & \cdots &  h_1^T(a_1, \text{-}) & \cdots
	\end{pmatrix},
\end{align}
which resembles a corresponding row of $H_1$ in Eq.~\eqref{eq:H_1} but with added zeros.

The lift of $\delta^1$ from Eq.~\eqref{eq:HGP_coboundary} becomes
\begin{align}\label{eq:HGP_coboundary_lift}
	\hat{\delta}^1 =
	\begin{pmatrix}
		0 & I_{l \cdot m_1 V_1} \otimes I_{l \cdot E_2} \otimes \hat{H}_3 & I_{l \cdot m_1 V_1} \otimes \hat{H}_2 \otimes I_{l \cdot E_3} \\
		I_{l \cdot E_1} \otimes I_{l \cdot m_2 V_2}\otimes \hat{H}_3 & 0 & \hat{H}_1 \otimes I_{l\cdot m_2 V_2} \otimes I_{l \cdot E_3} \\
		I_{l \cdot E_1}\otimes\hat H_2 \otimes I_{l \cdot m_3 V_3} & \hat{H}_1 \otimes I_{l \cdot E_2} \otimes I_{l \cdot m_3 V_3} & 0
	\end{pmatrix}.
\end{align}
The definitions of the canonical $X$ logical representatives remain the same as in Sec.~\ref{sec:HGP_local}, but we have to explicitly include the lifts. For instance, a free variable of $\hat{H}_1^T$ is of the form $[(s_1,v_1);a_1]$ for some $s_1\in\HH$, and a logical representative can be
\begin{align}\label{eq:HGP_basis_L}
	[(s_1,v_1);a_1] \otimes \sum_{s_2,v_2,b_2} \zeta_2([(s_2,v_2);b_2])\,[(s_2,v_2);b_2] \otimes \sum_{s_3,v_3,b_3} \zeta_3([(s_3,v_3);b_3])\,[(s_3,v_3);b_3],
\end{align}
where the last two factors, together with their local vector coefficients, come from $\ker\hat{H}_2$ and $\ker\hat{H}_3$, respectively.

An advantage of exponential lifting is that $\vert{\HH}\vert = l$ dominates in the system scaling. Since we assume $\frac{1}{2} n < m_i < n$, $\dim \ker \hat{H}_i \geq \Big( m_i - \frac{1}{2}n \Big) l \cdot V_i$ also scales with $l$. This dimension bound suggests seeking $v_i$ and $a_i$ for which there are $\Theta(l)$ independent solutions to $\hat{H}_i x=0$ such that the vectors obtained by applying $h_i^T(a_i,\text{-})$ to their components at $[(s,v_i);0]$ and collecting the results over $s\in\HH$ remain independent. The following theorem establishes a stronger version of this statement conditioned on GRH.

\begin{theorem}\label{thm:cup_HGP_3}
    Assume the generalized Riemann hypothesis and that $q=2^s$ with $s$ odd. There are infinitely many primes $l$ such that for each lift of size $l$, there are $\zeta_i^j\in\ker\hat{H}_i$, with $j=1,\ldots,l$, such that, for a constant fraction of the choices of $v_i$ and $a_i$, the vectors 
    \begin{align}
	\left( h_i^T(a_i,\text{-}) \Big(\zeta_i^j([(s,v_i);0])\Big) \right)_{s\in\HH}, \qquad j=1,\ldots,l, 
    \end{align}
    constitute the standard basis of $\mathbb F_q^{\HH}\cong\mathbb F_q^l$.
\end{theorem}

The proof follows from Sec.~\ref{sec:cup_sheaf} for general sheaf codes. We only discuss the consequence here. Fix an ordering $\HH=\{s_1,\ldots,s_l\}$. For each $i$, we reindex the $\zeta_i^j$ so that the induced vector of $\zeta_i^j$ is supported at $s_j$ for every $j\in[l]$. Define
\begin{align}
	& L_{1}^j \coloneq [(s_j,v_1);a_1] \otimes \zeta_2^j \otimes \zeta_3^j, \\
	& L_{2}^j \coloneq \zeta_1^j \otimes [(s_j,v_2);a_2] \otimes \zeta_3^j, \\
	& L_{3}^j \coloneq \zeta_1^j \otimes \zeta_2^j \otimes [(s_j,v_3);a_3].
\end{align}
Then $L_{1}^j \smile_c L_{2}^j \smile_c L_{3}^j$ yields the unique $3$-cube $[((s_j,v_1),(s_j,v_2),(s_j,v_3));a_1,a_2,a_3]$ with local coefficient
\begin{align}
	\begin{aligned}
		& \Big[ \big(h_2^T(a_2,\text{-})\otimes h_3^T(a_3,\text{-})\big)
		\Big( \zeta_2^j([(s_j,v_2);0]) \otimes \zeta_3^j([(s_j,v_3);0]) \Big)
		\Big] \\
		& \quad\cdot \Big[ \big(h_1^T(a_1,\text{-})\otimes h_3^T(a_3,\text{-})\big) \Big( \zeta_1^j([a_1\cdot(s_j,v_1);1])
		\otimes \zeta_3^j([(s_j,v_3);0]) \Big) \Big] \\
		& \quad\cdot \Big[ \big(h_1^T(a_1, \text{-}) \otimes h_2^T(a_2, \text{-}) \big) \Big( \zeta_1^j([a_1 \cdot (s_j,v_1); 1]) \otimes \zeta_2^j([a_2 \cdot (s_j,v_2); 1]) \Big) \Big] = 1,
	\end{aligned}
\end{align}
by Eq.~\eqref{eq:row}. In contrast to Eq.~\eqref{eq:HGP_cup}, $L_{1}^{j_1} \smile_c L_{2}^{j_2} \smile_c L_{3}^{j_3} = 0$ unless $j_1 = j_2 = j_3$. The remaining problem is that, when $m_i > \frac{n}{2}$, it is unclear how to find $\xi$. We discuss this issue further in Secs.~\ref{sec:cup_sheaf} and~\ref{sec:cycles}.


\section{Logical action on sheaf codes}\label{sec:proof_sheaf}

We now analyze the cup products on general sheaf codes. The ultimate goal is to establish Theorem~\ref{thm:cup_HGP_3} on high-dimensional expanders with local coefficients.

\subsection{Polarized logical representatives}\label{sec:HGP_like}

To determine the logical action of invariant forms on general sheaf codes, we adapt the idea of canonical logical representatives for HGP codes from Sec.~\ref{sec:proof_HGP}. To this end, we reformulate the sheaf coboundary operators in a form analogous to Eq.~\eqref{eq:HGP_coboundary}. For simplicity, we continue to work in the three-dimensional case, with the cochain complex
\begin{align}
	C^0(X,\F) \xrightarrow{\delta^0} C^1(X,\F) \xrightarrow{\delta^1} C^2(X,\F) \xrightarrow{\delta^2} C^3(X,\F),
\end{align} 
where $X$ is the three-dimensional expander obtained from an $\HH$-lift and double cover of the Cartesian product of $G_0$. Without the lift, $\delta^1$ is given by
\begin{align}\label{eq:HGP_coboundary'}
	\begin{pmatrix} 
		0 & I_{V} \otimes I_{E} \otimes H_3 & I_{V} \otimes H_2 \otimes I_{E} \\ 
		I_{E} \otimes I_{V} \otimes  H_3 & 0 &  H_1 \otimes I_{V} \otimes I_{E} \\
		I_{E} \otimes H_2 \otimes I_{V} & H_1 \otimes I_{E} \otimes I_{V} & 0 
	\end{pmatrix}.
\end{align}
Since $V_i \equiv 2V_0$ and $E_i \equiv 2E_0$ for all $i$, we drop the subscripts in the formulas that follow. After taking the lift, we explained in Sec.~\ref{sec:induct_boundary} that $\delta^1$ can be upgraded by rewriting each of its subblock matrices. For instance, we rewrite $I_{E} \otimes I_{V} \otimes H_3$ by lifting $H_3$ to $\hat{H}_3$, as in Eq.~\eqref{eq:HGP_coboundary_lift}. Keeping $I_E$ unchanged, we obtain $I_{E} \otimes I_{m_2 V} \otimes \hat{H}_3$. To verify that this definition is correct, consider any $1$-cube in $X(1)$:
\begin{align}\label{eq:lift_cube}
	[g = (h,v_1,v_2,v_3);a_1,b_2,b_3] = [h] \otimes [v_1;a_1] \otimes [v_2;b_2] \otimes [v_3;b_3] = [v_1;a_1] \otimes [v_2;b_2] \otimes [h,v_3;b_3], 
\end{align}
which is mapped to $[g'; a_1,b_2,a_3]$ with
\begin{align}
	g' = (a_3^{-1})^{b_3} \cdot g = ( (a_3^{-1})^{b_3} \cdot h, v_1, v_2, (a_3^{-1})^{b_3} \cdot v_3 ), \quad (a_3^{-1})^{b_3} = \begin{cases}
\mathrm{id}, & b_3 = 0, \\
a_3^{-1}, & b_3 = 1.
\end{cases}
\end{align}
and $a_3^{-1} \cdot h = \gamma_{(v_3, a_3 \cdot v_3)}^{-1} \cdot h$.
After this mapping, $[v_1;a_1] \otimes [v_2;b_2]$ is unchanged, while the action on $[h,v_3;b_3]$ and on its local vectors is exactly described by $\hat{H}_3$.

In the same way, we redefine $I_{E} \otimes H_2 \otimes I_{V}$ as $I_{E} \otimes \hat{H}_2 \otimes I_{m_3 V}$, and similarly for all other subblocks. However, incorporating all these blocks into a single matrix with consistent indices requires a nontrivial reorganization of rows and columns. For instance, $I_{E} \otimes \hat{H}_2 \otimes I_{m_3 V}$ uses the column indices
\begin{align}
	[g = (h,v_1,v_2,v_3);a_1,b_2,b_3] = [h] \otimes [v_1;a_1] \otimes [v_2;b_2] \otimes [v_3;b_3] = [v_1;a_1] \otimes [h,v_2;b_2] \otimes [v_3;b_3],
\end{align}  
which are different from those in Eq.~\eqref{eq:lift_cube}. This ambiguity is most cleanly resolved by using a more algebraic notation. Let $\mathbb{F}_q[\HH]$ denote the group algebra of $\HH$, consisting of all formal $\mathbb{F}_q$-linear combinations of elements of $\HH$. As a vector space it is isomorphic to $\mathbb{F}_q^{l}$, and it also carries the multiplicative structure induced by $\HH$. As above, a lifted coboundary operator is obtained by replacing each nonzero entry $\delta^\bullet(\tau',\sigma')$ by the corresponding permutation matrix $\gamma_{\sigma',\tau'}$. Equivalently, the morphisms $\mathcal{F}_{\sigma',\tau'}$ are replaced by $\gamma_{\sigma',\tau'} \otimes \mathcal{F}_{\sigma',\tau'}$. Since $\gamma_{\sigma',\tau'} \in \HH \subset \mathbb{F}_q[\HH]$, we can regard $\delta^\bullet$ as an operator with entries in $R = \mathbb{F}_q[\HH]$, acting on $C^\bullet(X,\F)$ as an $R$-module rather than only as an $\mathbb{F}_q$-vector space. When $\HH$ is Abelian, as in the cyclic case below, $R$ is commutative. For convenience, we always work over $\mathbb{F}_q$ but the $R$-module structure will also be used.

The expected logical representatives have the form
\begin{align}\label{eq:sheaf_basis}
	[v_1;a_1] \otimes \zeta_{23} = [v_1;a_1] \otimes \sum_{h,v_2,v_3,b_2,b_3} [(h,v_2,v_3);b_2,b_3],
\end{align}
in analogy with Eqs.~\eqref{eq:HGP_basis_1}, \eqref{eq:HGP_basis_F}, and~\eqref{eq:HGP_basis_L} for HGP codes. Note that we hide the local coefficient vectors on $0$-cube for simplicity. The most accessible way to find such representatives is to start from the Cartesian product $X'$, whose lift is $X$. Specifically, let $L_{1}^{i_1,i_2,i_3} = e_1^{i_1} \otimes \zeta_2^{i_2} \otimes \zeta_3^{i_3}$ be a logical representative of the code defined by $X'$ in Eq.~\eqref{eq:HGP_basis_F}. When the lift size $l$ is odd, its canonical extension obtained from Lemma~\ref{lemma:codeword_2} is a valid codeword and has the form Eq.~\eqref{eq:sheaf_basis}. Then we make the following definition.

\begin{definition}[All-ones lift] \label{def:all_ones}
	Recall that $\zeta_2^{i_2}$ and $\zeta_3^{i_3}$ are defined by the all-ones vector in $\mathbb{F}_q^V$ with fixed local vectors. Therefore, we call such canonical extensions \emph{all-ones lifts}.
\end{definition}

The number of all-ones lifts is bounded by the size of the base graph, and hence is exponentially smaller than the system size. Therefore, we need to search for additional logical representatives. The main obstacle is that the block matrix of the \emph{restricted coboundary operator}
\begin{align}\label{eq:delta_R}
	\delta_{\mathrm{res}}^1 = \begin{pmatrix} 
		0 \\ 
		I_{E} \otimes_R I_{m_2 V} \otimes_R \hat{H}_3 \\
		I_{E} \otimes_R \hat{H}_2 \otimes_R I_{m_3 V} 
	\end{pmatrix} \text{ over } R, \text{ or }
	\begin{pmatrix} 
		0 \\ 
		I_{E} \otimes I_{m_2 V} \otimes \hat{H}_3 \\
		I_{E} \otimes \hat{H}_2 \otimes I_{m_3 V} 
	\end{pmatrix} =
	I_{E} \otimes \begin{pmatrix} 
		0 \\ 
		I_{m_2 V} \otimes \hat{H}_3 \\
		\hat{H}_2 \otimes I_{m_3 V} 
	\end{pmatrix}
\end{align}
over $\mathbb{F}_q$ with reorganized entries has size
\begin{align}\label{eq:delta_size}
	\vert{\HH}\vert \cdot E ( m_2 V E + E V m_3) \times \vert{\HH}\vert \cdot E m_2 m_3 V^2 = l \cdot \frac{1}{2} n (m_2 + m_3) E V^2 \times l \cdot m_2 m_3 E V^2.
\end{align}  
Then $l \cdot V^3 = \vert X(0) \vert$. Even if we require $\frac{1}{2}n \leq m_i \leq n$, Eq.~\eqref{eq:delta_size} shows that counting rows alone is not enough to obtain a valid lower bound on $\dim \ker \delta_{\mathrm{res}}^1$. Suppose $\F$ is generated by local codes with the following numbers of independent local checks:
\begin{align}\label{eq:local_rank}
	m_1 = \nu' n, \quad m_j = (1 - \nu) n, \ j \neq 1.
\end{align} 
and suppose $\nu$ and $\nu'$ are sufficiently small constants. Then we can prove the following lemma with a favorable lower bound on the $\dim \ker \delta_{\mathrm{res}}^1$.

\begin{lemma}\label{lemma:rank}
	Let $X$ be a $t$-dimensional cubical complex. Let $\F$ be generated local codes that satisfy \eqref{eq:local_rank}. Then the associated restricted coboundary operator $\delta_{\mathrm{res}}^1$ satisfies
	\begin{align}\label{eq:delta_R_ker}
		\dim \ker\delta_{\mathrm{res}}^1 \geq 
		\Big[ & \frac{1}{t-1} \Big( (1-\nu)^{t-1} 
		+ \frac{t-2}{2^{t-1}} + \frac{(t-2)(t-1)}{2^{t-3}} \nu' \nu \Big) \\
		& - \frac{t-2}{t-1} \Big( (t-1) \nu' (1-\nu)^{t-2} + \frac{1}{2^{t-2}} (\nu' + (t-1)\nu ) + \frac{(t-1)(t-2)}{2^{t-3}} \nu' \nu^2 \Big)
		\Big] \frac{1}{2} n^t l \cdot V^t. \notag
	\end{align}
	In particular, for sufficiently small $\nu'$ and $\nu$, $\dim \ker\delta_{\mathrm{res}}^1$ has the same order as $n^t l \cdot V^t$.
\end{lemma}
\begin{proof}
	Let $\F_T$ be the modified sheaf obtained from $\F$ by replacing $h_j$ with $h_j^\perp$ for every $j\in T$, and let $C^\bullet(X,\F_T,T)$ be the auxiliary cochain complex defined in Ref.~\cite[Sec.~4.1]{Dinur2024sheaf}. Then $\dim H^i(X,\F) = \dim H^i(X,\F_T,T)$. We refer to that work for the full definitions. For a fixed $T \subset [t]$, it is shown that
	\begin{align}
		D_i'=\dim C^i(X,\F_T,T) = N \sum_{S \subset[t], \vert S \triangle T \vert = i } n^{\vert S \vert} 2^{t - \vert S \vert} \prod_{j \notin S} m_j',
	\end{align}
	where $N = \frac{1}{2^t} l \cdot V^t$ and $S \triangle T$ denotes the symmetric difference. In our case $T = \{2, \ldots,t\}$, $m_1' = \nu' n$, while the complementary parameters in the directions $2 \leq j \leq t$ are governed by $\nu n$. Therefore,
	\begin{itemize}
		\item For $D_0'$, we have $S = T = \{2, \ldots,t\}$ and hence
		\begin{align}
			D_0' = N n^{t-1} 2 \nu' n = \frac{1}{2^{t-1}} \nu' n^t l \cdot V^t.
		\end{align}
		
		\item For $D_1'$, $S$ is either $[t]$ or $\{2,\ldots,\hat{i}, \ldots,t\}$ with one $i$ removed. Therefore,
		\begin{align}
			D_1' = N n^t + (t-1) N n^{t-2} 2^2 \nu' \nu n^2
			= \left( \frac{1}{2^t} + \frac{1}{2^{t-2}} (t-1) \nu' \nu \right) n^t l \cdot V^t.
		\end{align}
		
		\item For $D_2'$, $S$ is either $\{1\} \cup \{2,\ldots,\hat{i}, \ldots,t\}$ or $\{2,\ldots,\hat{i}, \ldots,\hat{j}, \ldots,t\}$. Therefore,
		\begin{align}
			D_2' = \left( (t-1) 2\nu + \binom{t-1}{2} 2^3 \nu' \nu^2 \right) N n^t 
			= \left( \frac{t-1}{2^{t-1}} \nu + \frac{1}{2^{t-3}} \frac{(t-1)(t-2)}{2} \nu' \nu^2 \right) n^t l \cdot V^t.
		\end{align}
	\end{itemize}
	As a result, $\dim H^1(X,\F) \geq D_1' - D_2' - D_0'$ and $\dim \ker \delta^1 = \dim H^1(X,\F) + \dim \im \delta^0$ imply
	\begin{align}\label{eq:upper_bound}
		\begin{aligned}
			\rk \delta^1 = & \dim C^1(X,\F) - \dim \ker \delta^1 \leq \dim C^1(X,\F) - \dim H^1(X,\F)  \\
			\leq & \frac{1}{2} \Big[ (t-1) \nu' (1-\nu)^{t-2} + (1-\nu)^{t-1}  \\
			& - \frac{1}{2^{t-1}} - \frac{t-1}{2^{t-3}} \nu' \nu 
			+  \frac{1}{2^{t-2}} (\nu' + (t-1)\nu ) + \frac{(t-1)(t-2)}{2^{t-3}} \nu' \nu^2
			\Big] n^t l \cdot V^t.
		\end{aligned}
	\end{align}
	The lower bounds above are nontrivial only when their right-hand sides are positive. When $\nu' = \nu = 0$, the lower bound on $\dim H^1(X,\F)$ is $\frac{1}{2^t} n^t l \cdot V^t$, and hence remains positive for sufficiently small $\nu'$ and $\nu$. Similarly, the upper bound in Eq.~\eqref{eq:upper_bound} equals $\frac{1}{2} (1 - \frac{1}{2^{t-1}}) n^t l \cdot V^t$ when $\nu' = \nu = 0$.
	
	On the other hand, the restricted coboundary operator in Eq.~\eqref{eq:delta_R} has one of the following two numbers of columns:
	\begin{align}
		\frac{1}{2} \nu' (1-\nu)^{t-2} n^t l \cdot V^t,
		\text{ or } \frac{1}{2} (1-\nu)^{t-1} n^t l \cdot V^t. 
	\end{align}
	Its rank is also upper bounded by $\rk \delta^1$. The difference between $\frac{1}{2} \nu' (1-\nu)^{t-2} n^t l \cdot V^t$ and Eq.~\eqref{eq:upper_bound} is always negative.
	However, for sufficiently small $\nu'$ and $\nu$, the following quantity is positive:
	\begin{align}
		\frac{1}{2} \Big[
		\frac{1}{2^{t-1}} + \frac{t-1}{2^{t-3}} \nu'\nu - (t-1) \nu' (1-\nu)^{t-2} - \frac{1}{2^{t-2}} (\nu' + (t-1)\nu ) - \frac{(t-1)(t-2)}{2^{t-3}} \nu' \nu^2 \Big] n^t l \cdot V^t.
	\end{align} 
	In other words, for $j \geq 2$, when the $\hat{H}_j$ are defined by local codes with $(1-\nu) n$ checks, the corresponding $\delta_{\mathrm{res}}^1$ has a nontrivial kernel whose dimension grows with the size of $\HH$.
	
	The lower bound in Eq.~\eqref{eq:delta_R_ker} is more precise, and we now derive it in the three-dimensional case. Let us divide the independent rows of $\delta_{\mathrm{res}}^1$ into two groups, coming from the rows of $I_{E} \otimes I_{m_2 V} \otimes \hat{H}_3$ and $I_{E} \otimes \hat{H}_2 \otimes I_{m_3 V}$:
	\begin{align}
		\rk \delta_R^1 = r_3 + r_2.
	\end{align}
	In general, $r_3 \leq \rk I_{E} \otimes I_{m_2 V} \otimes \hat{H}_3$ and $r_2 \leq \rk I_{E} \otimes \hat{H}_2 \otimes I_{m_3 V}$. These $r_3$ independent rows also appear in the restriction of $\delta^1$ to the top middle block:
	\begin{align}
		\begin{pmatrix} 
			I_{E} \otimes I_{m_2 V} \otimes \hat{H}_3 \\
			0 \\
			\hat{H}_1 \otimes I_{E} \otimes I_{m_3 V} 
		\end{pmatrix}.
	\end{align}
	All the corresponding row vectors in $\delta^1$ are independent, and hence $r_3 + r_3 + r_2 \leq \rk \delta^1$. Since $\rk \delta^1 = \dim C^1(X,\F) - \dim H^1(X,\F) - \dim\im \delta^0$ and $\dim\im \delta^0 \geq r_2, r_3$, we have
	\begin{align}
		3 r_3 + r_2 \leq \dim C^1(X,\F) - \dim H^1(X,\F). 
	\end{align}
	Without loss of generality, assume $r_3 \geq r_2$. Then a basic linear-programming estimate gives
	\begin{align}
		& \rk \delta_{\mathrm{res}}^1 = r_3 + r_2 \leq \frac{1}{2} ( \dim C^1(X,\F) - \dim H^1(X,\F)) \\
		\implies & \dim\ker \delta_{\mathrm{res}}^1 
		\geq  \frac{1}{2} (1-\nu)^2 n^2 l \cdot V^3 - \frac{1}{2} ( \dim C^1(X,\F) - \dim H^1(X,\F)) 
	\end{align}
	whose generalization is Eq.~\eqref{eq:delta_R_ker}.
\end{proof}

Suppose $R = \mathbb{F}_q[\HH]$ is semisimple; for example, this holds when $q$ is even and $\HH = C_l$ with $l$ odd by Maschke's theorem~\cite{Goodman2009}. Then the K\"{u}nneth formula \cite{Hatcher2015AT,Gallier2022} over $R$ has vanishing $\mathrm{Tor}$-terms and gives
\begin{align}\label{eq:Kunneth}
	H^1(X,\F) \cong \bigoplus_{i = 1}^t \left( Q_i \otimes_R \bigotimes_{j \neq i} \ker \hat{H}_j  \right) 
\end{align}
where $Q_i \cong R^E / \im \hat{H}_i$ can be identified with a submodule in $R^E$. These submodules may not be free, but returning to $\mathbb{F}_q$, each cohomological representative is spanned by elements of the form in Eq.~\eqref{eq:sheaf_basis}. Consequently, such elements form a complete logical basis for the cohomology. Such a representative is said to be \emph{polarized in the first direction}. The same argument applies to other coordinate directions and, more generally, to higher-degree polarized representatives. The only ambiguity happens when we try to write down the representatives explicitly: one would like to reuse the method from the HGP case. However, ordinary Gaussian elimination and the associated free-variable argument do not directly apply over general $R$ when $\hat{H}_1$ has been lifted. Besides, we have to show that they are not coboundaries. All these issues are resolved in Section \ref{sec:module} and \ref{sec:edge_choice}.

Polarized $Z$-logical representatives can be obtained by the Poincar\'e duality of the dual sheaf codes. It is proved in Ref.~\cite{Li2025Poincare} that cap product with a nontrivial top-degree homology class gives an isomorphism between degree-$i$ cohomology and degree-$(t-i)$ homology. Therefore, once a polarized basis of $X$-logical cohomology classes has been constructed, its cap-product image gives a corresponding polarized basis of $Z$-logical homology classes. Polarized logical representatives are expected to have broad applications, including the construction of cup products. Several additional steps are required before polarized logical representatives can be used to construct cup products. The main difficulty lies in controlling the behavior of the local coefficients. The next three subsections are devoted to this analysis.


\subsection{Cyclic action on lifts and Artin's conjecture}\label{sec:group}

Proving Theorem~\ref{thm:cup_HGP_3} and generalizing the results to polarized logical representatives for sheaf codes hinge on a deeper understanding of the group action of $\HH$, which we develop using techniques from algebra and number theory.

Let $x$ be any cocycle obtained from $\ker \delta_{\mathrm{res}}^1$. In the three-dimensional case, the value of $x$ on a basis element has the form
\begin{align}
	x ( [v_1;a_1] \otimes [(h,v_2,v_3);(b_2,b_3)] ) \in \mathbb{F}_q^{m_2} \otimes \mathbb{F}_q^{m_3}.
\end{align}
For simplicity, we continue to work with three-dimensional expanders and cochain complexes. All the following results extend straightforwardly to higher dimensions. We refer to such elements as cocycles for now and prove in Secs.~\ref{sec:module} and~\ref{sec:edge_choice} that they are nontrivial logical representatives. By Eq.~\eqref{eq:sheaf_basis}, $[v_1;a_1]$ can be taken arbitrarily, so we hide it for brevity. As in the analysis of cup products for HGP codes, the values of $b_2, b_3$ are also insignificant because $x \in \ker \delta_{\mathrm{res}}^1$ implies
\begin{align}
	& h_2^T(a_2, \text{-}) \otimes I_{m_3} \Big( x ( [(h,v_2,v_3);(0,0)] ) \Big) =
	h_2^T(a_2, \text{-}) \otimes I_{m_3} \Big(  x (  [a_2 \cdot (h,v_2,v_3);(1,0)] ) \Big), \label{eq:delta_condition1} \\
	& h_2^T(a_2, \text{-}) \otimes I_{m_3} \Big( x ( [(h,v_2,v_3);(0,1)] ) \Big) =
	h_2^T(a_2, \text{-}) \otimes I_{m_3} \Big(  x (  [a_2 \cdot (h,v_2,v_3);(1,1)] ) \Big), \label{eq:delta_condition2} \\
	& I_{m_2} \otimes h_3^T(a_3, \text{-}) \Big( x (  [(h,v_2,v_3);(0,0)] ) \Big) =
	I_{m_2} \otimes h_3^T(a_3, \text{-}) \Big( x ( [a_3 \cdot (h,v_2,v_3);(0,1)] ) \Big), \label{eq:delta_condition3} \\
	& I_{m_2} \otimes h_3^T(a_3, \text{-}) \Big( x (  [(h,v_2,v_3);(1,0)] ) \Big) =
	I_{m_2} \otimes h_3^T(a_3, \text{-}) \Big( x ( [a_3 \cdot (h,v_2,v_3);(1,1)] ) \Big). \label{eq:delta_condition4}
\end{align}
Evaluating both local vectors then gives, for example,
\begin{align}
	h_2^T(a_2, \text{-}) \otimes h_3^T(a_3, \text{-}) \Big( x ( [(h,v_2,v_3);(0,0)] ) \Big) =
	h_2^T(a_2, \text{-}) \otimes h_3^T(a_3, \text{-}) \Big(  x (  [a_2 \cdot (h,v_2,v_3);(1,0)] ) \Big).
\end{align}
These equations are still useful later. We highlight in Sec.~\ref{sec:cubical} that $\HH$ has to be Abelian for the present construction of high-dimensional expanders. Here is another advantage when $\HH$ is Abelian. 

\begin{lemma}\label{lemma:group_action}
	Assume that $\HH$ is Abelian. For any group element $h' \in \HH$, we shift the local vector on $[(h,v_2,v_3);(b_2,b_3)]$ to $[(h' \cdot h,v_2,v_3);(b_2,b_3)]$ and define
	\begin{align}
		h' \cdot x ( [(h,v_2,v_3);(b_2,b_3)] ) \coloneq x( [(h'^{-1} \cdot h,v_2,v_3);(b_2,b_3)] ).
	\end{align}   
	Then $h' \cdot x$ is still a cocycle. 
\end{lemma}
\begin{proof}
	The statement follows directly from Eqs.~\eqref{eq:delta_condition1}--\eqref{eq:delta_condition4}. For instance,
	\begin{align}
		& I_{m_2} \otimes h_3^T(a_3, \text{-}) \Big( h' \cdot x (  [(h,v_2,v_3);(1,0)] ) \Big)  \\
		= & I_{m_2} \otimes h_3^T(a_3, \text{-}) \Big( x (  [(h'^{-1} \cdot h,v_2,v_3);(1,0)] ) \Big)
		= I_{m_2} \otimes h_3^T(a_3, \text{-}) \Big( x ( [a_3 \cdot (h'^{-1} \cdot h,v_2,v_3);(1,1)] ) \Big) \notag \\
		= & I_{m_2} \otimes h_3^T(a_3, \text{-}) \Big( x ( [(a_3 \cdot h'^{-1} \cdot h,v_2,a_3 \cdot v_3);(1,1)] ) \Big) \notag \\
		= & I_{m_2} \otimes h_3^T(a_3, \text{-}) \Big( x ( [(h'^{-1} \cdot a_3 \cdot h,v_2,a_3 \cdot v_3);(1,1)] ) \Big)
		= I_{m_2} \otimes h_3^T(a_3, \text{-}) \Big( h' \cdot x ( [a_3 \cdot (h,v_2, v_3);(1,1)] ) \Big). \notag
	\end{align}
    This finishes the proof.
\end{proof}

Note that $x$ and $h' \cdot x$ could be identical. For example, when $x$ is just an all-ones lift in Definition~\ref{def:all_ones}, then $x \equiv h' \cdot x$ for any $h' \in \HH$. To avoid possible degeneracy, we put more effort into studying the group action. 

\begin{definition}[Induced vector] \label{def:induced_vector}
	For any fixed $1$-cube $\sigma' = [(v_1,v_2,v_3);(a_1,b_2,b_3)] \in X'$, the restriction of $x$ to the lifts of $\sigma'$ is an element of $\mathbb{F}_q[\HH] \otimes \mathbb{F}_q^{m_2} \otimes \mathbb{F}_q^{m_3}$. For fixed $a_2,a_3$, define the \emph{induced vector} in $\mathbb{F}_q[\HH]$ by
	\begin{align}\label{eq:induced_vector}
		\mathrm{v}_{\sigma'}^{a_2,a_3}(h) \coloneq h_2^T(a_2, \text{-}) \otimes h_3^T(a_3, \text{-}) \Big( x ( [(h,v_1,v_2,v_3);(a_1,b_2,b_3)] ) \Big), \quad h \in \HH.
	\end{align}
\end{definition}

By Theorem~\ref{thm:Abelian_lift}, we set $\HH = C_l$ as the cyclic group with $l$ elements. Let $\x$ be a generator of $C_l$. Then $\mathbb{F}_q[C_l] \cong \mathbb{F}_q[\x]/(\x^l - 1)$, i.e., the algebra of polynomials over $\mathbb{F}_q$ modulo the relation $\x^l = 1$. Let $\mathrm{v} \in \mathbb{F}_q[C_l]$ be a nonzero vector. The independence of the vectors $\mathrm{v}, \x \cdot \mathrm{v}, \cdots,  \x^{l-1} \cdot \mathrm{v}$ is determined by whether or not there exist coefficients $c_i \in \mathbb{F}_q$, not all zero, such that
\begin{align}
	\sum_{i=0}^{l-1} c_i (\x^{i} \cdot \mathrm{v}) = 0 \quad \Longleftrightarrow \quad \Big( \sum_{i=0}^{l-1} c_i \x^{i} \Big) \cdot \mathrm{v} = 0.
\end{align}
Mathematically, if $\mathrm{v}$ is not a zero divisor, then no such nontrivial $c_i$ exists. In particular, if $\mathrm{v}$ is a unit, then the actions of $\x$ on $\mathrm{v}$ generate a complete basis for $\mathbb{F}_q[C_l]$. For example, $\mathrm{v} = \x$ is a unit, but both $1 + \x$ and $1 + \x + \cdots + \x^{l-1}$ are zero divisors. 

In the actual setting, with respect to fixed $\sigma'$ and fixed $a_2,a_3$, suppose the cocycles $x_i$ induce a collection of vectors $\{\mathrm{v}_i\} \subset \mathbb{F}_q[C_l]$. By a standard result from commutative algebra, if the ideal generated by $\{\mathrm{v}_i\}$ is not contained in any maximal ideal of $\mathbb{F}_q[C_l]$, then it is the whole ring. This enables us to recover a complete basis for $\mathbb{F}_q[C_l]$. Let $f_{i,j}(\x)$ be polynomials in $\mathbb{F}_q[C_l] \cong \mathbb{F}_q[\x]/(\x^l - 1)$ such that 
\begin{align}
	\sum_i f_{i,j}(\x) \cdot \mathrm{v}_i = \x^j \in \mathbb{F}_q[C_l].
\end{align}
That is, $\x^j$ is one of the standard basis vectors. Accordingly, we have
\begin{align}\label{eq:poly_logical}
	\sum_i f_{i,j}(\x) \cdot x_i \in \ker \delta_{\mathrm{res}}^1.
\end{align}
By Lemma~\ref{lemma:group_action} and Eq.~\eqref{eq:induced_vector}, for any $h' \in C_l$, with $a_2,a_3$  fixed,
\begin{align}
	& h_2^T(a_2, \text{-}) \otimes h_3^T(a_3, \text{-}) ( h' \cdot x_i \vert_{\sigma'} ) = h' \cdot ( h_2^T(a_2, \text{-}) \otimes h_3^T(a_3, \text{-}) x_i \vert_{\sigma'} ) = h' \cdot \mathrm{v}_i \\
	\implies & h_2^T(a_2, \text{-}) \otimes h_3^T(a_3, \text{-}) \Big( \sum_i f_{i,j}(\x) \cdot x_i \vert_{\sigma'} \Big)  = \sum_i f_{i,j}(\x) \cdot \mathrm{v}_i.
\end{align}   
The collection of cocycles $\{\sum_i f_{i,j}(\x)\cdot x_i\}$ is linearly independent, and we show in Corollary~\ref{coro:inequivalent} that, under its hypotheses, these cocycles represent linearly independent cohomology classes. Their linear independence follows because their evaluations under $h_2^T(a_2,\text{-})\otimes h_3^T(a_3,\text{-})$ form the standard basis of $\mathbb{F}_q[C_l]$. Consequently, we have the following lemma.
\begin{lemma}\label{lemma:basis}
	Let $x_i$ be cocycles in $\ker \delta_{\mathrm{res}}^1$. Suppose there is a $1$-cube $\sigma'=$ $[(v_1,v_2,v_3);$ $(a_1,b_2,b_3)]$ with some $a_2,a_3$ such that the ideal generated by the induced vectors $\mathrm{v}_i$ is not contained in any maximal ideal of $\mathbb{F}_q[C_l]$. Then there exist $l$ independent cocycles $L^j$ such that
	\begin{align}
		h_2^T(a_2, \text{-}) \otimes h_3^T(a_3, \text{-}) ( L^j \vert_{\sigma'} ) = \x^j.
	\end{align}
\end{lemma}

Under simple conditions on the underlying expander graphs and local codes, we prove in Corollary~\ref{coro:inequivalent} that the $l$ independent cocycles are nontrivial and correspond to $l$ independent cohomology classes.
From now on, let $l$ be an odd prime.
In general, $\mathbb{F}_q[C_l]$ admits several maximal ideals generated by irreducible factors of $\x^l - 1$, including $\x + 1$. The others are generated by irreducible factors of $1 + \x + \cdots + \x^{l-1}$. Writing $R=\mathbb{F}_q[C_l]\cong \mathbb{F}_q[\x]/(\x^l-1)$, a basic property of quotient rings says that when $f(\x)$ is irreducible,
\begin{align}
	R/\langle f(\x) \rangle \cong \mathbb{F}_q[\x]/(f(\x)) \cong \mathbb{F}_{q^{\deg f}} \cong \mathbb{F}_q^{\deg f}
\end{align} 
as $\mathbb{F}_q$-vector spaces. Hence the kernel of the projection $\pi: R \to R/\langle f(\x) \rangle$ is $\langle f(\x) \rangle$. The rank--nullity theorem gives
\begin{align}\label{eq:ideal_dim}
	\dim \langle f(\x) \rangle = \dim R - \dim R/\langle f(\x) \rangle = l - \deg f.
\end{align}
Except for the factor $\x+1$, the degree of every irreducible factor $f$ satisfies $q^{\deg f} \equiv 1 \mod l$, since $\deg f$ is the size of the corresponding orbit under the Frobenius automorphism~\cite{Aluffi2009}. We also note that the \emph{augmentation ideal} $\langle \x+1 \rangle$ is just the kernel of the evaluation map at $\x = 1$ and thus the cyclotomic polynomial $1 + \x + \cdots + \x^{l-1} \notin \langle \x+1 \rangle$ when $l$ is odd. Recall that any all-ones lift, with a proper choice of $a_2,a_3$, induces $1 + \x + \cdots + \x^{l-1}$. Our task of meeting the condition of Lemma~\ref{lemma:basis} reduces to ruling out all other maximal ideals. We provide some examples to motivate the results.

\begin{example}
	Let $q = 2$ and $l = 5$. Then the smallest solution to $2^{\deg f} \equiv 1 \mod 5$ is $\deg f = 4$. Besides $\langle \x+1 \rangle$, the only other maximal ideal is $\langle \x + 1 + \x^2 + \x^3 + \x^4 \rangle$ of dimension $1$. In this case, finding one cocycle distinct from the all-ones lifts is sufficient.
	
	Let $q = 2$ and $l = 7$. Then the smallest solution to $2^{\deg f} \equiv 1 \mod 7$ is $\deg f = 3$ and we have 
	\begin{align}
		\x + 1 + \x^2 + \x^3 + \x^4 + \x^5 + \x^6 = (1 + \x + \x^3) (1 + \x^2 + \x^3).
	\end{align}
	This factorization corresponds to two more maximal ideals of dimension $4$.
\end{example}

By Eq.~\eqref{eq:ideal_dim}, if every maximal ideal other than $\langle 1+\x \rangle$ has dimension strictly smaller than $\dim_{\mathbb{F}_q}\operatorname{span}_{\mathbb{F}_q}\{\mathrm{v}_i\}$, then this span cannot be contained in any such ideal and Lemma~\ref{lemma:basis} applies.
For instance, one may seek primes $l$ for which $l-\deg f=o(l)$, or equivalently $\deg f=l-o(l)$, in the distribution of prime numbers. Remarkably, this line of inquiry is connected to several deep developments in number theory, among which the following celebrated conjecture is particularly relevant:
\begin{conjecture}[Artin~\cite{Artin1966}]
	If $a \in \mathbb{Z}$ is neither $\pm 1$ nor a perfect square, then there are infinitely many primes $p$ for which $a$ is a primitive root modulo $p$.
\end{conjecture}

In our case, $a = q = 2^s$ and $p = l$. Here we require $s$ to be odd so that $q$ cannot be a perfect square. The statement that $q$ is a primitive root modulo $l$ means that $q^{l-1} \equiv 1 \mod l$ and that $l-1$ is the smallest positive exponent with this property, i.e., $\text{ord}_l(q) = l - 1$. As a result, $\deg f = l - 1$ and the only maximal ideal different from $\langle 1 + \x \rangle$ is $\langle 1 + \x + \cdots + \x^{l-1} \rangle$. In this case, the all-ones lift already rules out containment in $\langle 1+\x \rangle$, so the condition in Lemma~\ref{lemma:basis} reduces to ruling out containment in $\langle 1+\x+\cdots+\x^{l-1} \rangle$. Hooley proved Artin's primitive root conjecture under the relevant GRH hypotheses~\cite{Hooley1967}. There are several unconditional results. For example, in Refs.~\cite{Gupta1984,Hearth-Brown1986}, Artin's conjecture is proved for all primes with at most two exceptions. However, $a = q$ in our case is always a power of $2$. Weaker versions of the conjecture are also known when $\deg f = \Omega(l^r)$ for $r$ around $1/2$~\cite{Agrawal2021,Jarviniemi2021}.

\begin{theorem}\label{thm:basis}
    For $q=2^s$ with $s$ odd, the GRH hypotheses in Hooley's theorem~\cite{Hooley1967} imply that there exist infinitely many primes $l$ for which $q$ is a primitive root modulo $l$. These primes can be used to form high-dimensional expanders to which Lemma~\ref{lemma:basis} applies.
\end{theorem}

It remains to choose a suitable $\sigma'$ and suitable $a_2$ and $a_3$, at which point Lemma~\ref{lemma:basis} applies. This is addressed in the next subsection. Without GRH or Artin's conjecture, the above construction may yield only a basis for a proper ideal of $\mathbb{F}_q[C_l]$, rather than for the full group algebra, making it unfavorable for constructing independent cup products.
More details are given at the end of Sec.~\ref{sec:cup_sheaf}.


\subsection{Proof of Theorem~\ref{thm:basis}}\label{sec:module}

We use additional techniques from commutative algebra to determine $\sigma'$, $a_2$, and $a_3$. Let $R = \mathbb{F}_q[C_l]$. As mentioned in Sec.~\ref{sec:HGP_like}, $\delta_{\mathrm{res}}^1$ has a more compact form as an operator on the $R$-modules in Eq.~\eqref{eq:delta_R}:
\begin{align}
	\delta_{\mathrm{res}}^1 = \begin{pmatrix} 
		0 \\ 
		I_{E} \otimes_R I_{m_2 V} \otimes_R \hat{H}_3 \\
		I_{E} \otimes_R \hat{H}_2 \otimes_R I_{m_3 V} 
	\end{pmatrix} = I_{E} \otimes_R \begin{pmatrix} 
		0 \\ 
		I_{m_2 V} \otimes_R \hat{H}_3 \\
		\hat{H}_2 \otimes_R I_{m_3 V} 
	\end{pmatrix}.
\end{align}
Expanding the $R$-module solutions in the standard $\mathbb{F}_q$-basis of $R$ shows that $\ker \delta_{\mathrm{res}}^1$ over $\mathbb{F}_q$ is obtained from $\ker \delta_{\mathrm{res}}^1$ over $R$. For example, Lemma~\ref{lemma:group_action} shows that the $C_l$-action preserves $\ker \delta_{\mathrm{res}}^1$. In module-theoretic language, these actions amount simply to multiplication by ``scalars'' in $R$.

Lemma~\ref{lemma:rank} demonstrates that $\ker\delta_{\mathrm{res}}^1$ has $\Theta(lV^3)$ independent solutions over $\mathbb{F}_q$. Since $R^E$ is free of rank $E=\Theta(V)$ over $R$, the intersection below has $\Theta(lV^2)$ independent solutions over $\mathbb{F}_q$:
\begin{align}
	\ker (I_{m_2 V} \otimes_R \hat{H}_3 ) \cap \ker ( \hat{H}_2 \otimes_R I_{m_3 V} ) = ( R^{m_2 V} \otimes_R \ker \hat{H}_3 ) \cap (\ker \hat{H}_2 \otimes_R R^{m_3 V} ).
\end{align}
For comparison, let $V' \subset V$ and $W' \subset W$ be subspaces. Then
\begin{align}
	(V' \otimes W) \cap (V \otimes W') = V' \otimes W'.
\end{align}
This intersection property is implicitly used when solving the canonical logical representatives for HGP codes, e.g., in Eq.~\eqref{eq:HGP_basis_1}, but it does not hold for general modules. Fortunately, when $l$ is odd and $q$ is even, $\mathbb{F}_q[C_l]$ is semisimple; equivalently, its representations over $\mathbb{F}_q$ are completely reducible~\cite{Goodman2009}. Consequently, all of its modules are projective and hence flat. We recall the following standard consequences~\cite{Bourbaki,Aluffi2009,Bosch2022}.

\begin{definition}[Flat module]
	An $R$-module $N$ is \emph{flat} if, for every injective homomorphism $M' \to M$, the induced map $M' \otimes_R N \rightarrow M \otimes_R N$ is injective.
\end{definition}

\begin{proposition}\label{prop:exact}
	Suppose that $N$ is flat. If $0 \to M' \to M \to M'' \to 0$ is exact, then
	\begin{equation}
		\begin{tikzcd}
			0 \arrow[r] & M' \otimes_R N \arrow[r] & M \otimes_R N \arrow[r] & M'' \otimes_R N \arrow[r] & 0
		\end{tikzcd}
	\end{equation}
	is exact.
\end{proposition}

\begin{proposition}\label{prop:intersect}
	Suppose that every $R$-module is flat. If $M' \subset M$ and $N' \subset N$ are $R$-submodules, then $(M' \otimes_R N) \cap (M \otimes_R N') = M' \otimes_R N'$.
\end{proposition}
\begin{proof}
	We consider the following commutative diagram:
	\begin{equation}\label{eq:cd_flat}
		\begin{tikzcd}
			0 \arrow[r] & M' \otimes_R N' \arrow[r,"\phi_1"] \arrow[d,"f_1"] & M \otimes_R N' \arrow[r,"\phi_2"] \arrow[d,"f_2"] & M/M' \otimes_R N' \arrow[d,"f_3"] \arrow[r] & 0  \\ 		
			0 \arrow[r] & M' \otimes_R N \arrow[r,"\psi_1"] & M \otimes_R N \arrow[r,"\psi_2"] & M/M' \otimes_R N \arrow[r] & 0  
		\end{tikzcd}
	\end{equation}
	The canonical quotient sequence
	\begin{equation}
		\begin{tikzcd}
			0 \arrow[r] & M' \arrow[r] & M \arrow[r] & M/M' \arrow[r] & 0 
		\end{tikzcd}
	\end{equation}
	is exact, so Proposition~\ref{prop:exact} implies that both rows of Eq.~\eqref{eq:cd_flat} are exact. Since $M/M'$ is flat, $f_3$ is injective. By the snake lemma, $\im (f_2 \circ \phi_1) = \im f_2 \cap \im \psi_1$. Exactness gives
	\begin{align}
		\im \phi_1 \cong M' \otimes_R N', \quad \im \psi_1 \cong M' \otimes_R N.
	\end{align} 
	Since $M$ is flat, $f_2$ is injective, so we may identify $\im f_2$ with $M \otimes_R N'$. This proves the proposition.
\end{proof}

Since $[v_1;a_1]$ can be chosen arbitrarily in the first direction and every $R$-module is flat, Proposition~\ref{prop:intersect} implies the following lemma.

\begin{lemma}\label{lemma:product}
	The kernel of $\delta_{\mathrm{res}}^1$ is
	\begin{align}
		\ker\delta_{\mathrm{res}}^1=
		R^E\otimes_R
		\left(\ker\hat H_2\otimes_R\ker\hat H_3\right).
	\end{align}
\end{lemma}

Before proceeding, we remind the reader that Proposition~\ref{prop:intersect} and Lemma~\ref{lemma:product} cannot replace Lemma~\ref{lemma:rank}, for two reasons. First, the rank-nullity theorem and related properties generally fail for modules. Despite the bounds $\frac{1}{2}n < m_i < n$, the number of rows alone does not determine $\dim_R \ker \hat{H}_i$. The second reason is that $\dim_R \ker \hat{H}_i$ may not be well-defined when $\ker \hat{H}_i$ is not free. 
Nevertheless, once polarized logical representatives are known to exist,
Lemma~\ref{lemma:product} allows us to decompose them as $R$-tensors in
$R^E\otimes_R\ker\hat H_2\otimes_R\ker\hat H_3$,
thereby facilitating the following analysis.


\medskip

Suppose that $\frac{1}{2}n<m_2<n$. Since $\hat{H}_2$ has
$lm_2V$ columns and at most $\frac{1}{2}nlV$ rows, rank-nullity gives
$\dim\ker\hat{H}_2=\Theta(lV)$ for fixed $n$ and $m_2$. We arrange a basis of these solutions as the rows of a matrix and, after permuting the columns if necessary, apply elementary row operations to obtain the reduced row-echelon form $(I\mid M)$. The scalar columns are grouped by the indices $[(h,v_2);b_2]$, with $m_2$ local coordinates attached to each $b_2$. For each $[v_2;b_2]$, we group the basis rows according to the pivot columns indexed by its lifts $[(h,v_2);b_2]$.
There are $\Theta(V)$ such groups, and each group contains at most $m_2l=O(l)$ pivot rows. Since the total number of pivot rows is $\Theta(lV)$, an averaging argument yields the following claim.

\begin{claim}\label{claim_1}
	For a constant fraction of the vertices $[v_2;b_2]$, the associated group contains $\Theta(l)$ basis rows.
\end{claim} 

For any such $[v_2;b_2]$, we restrict the associated basis rows to the components corresponding to the lifts of this vertex (cf. Eq.~\eqref{eq:induced_vector}). Ideally, these restricted vectors would form the following block matrix:
\begin{align}\label{eq:restricted_block}
	\left(
	\begin{array}{ c c c | c c c }
		\ast  &  &  & \ast  & \cdots  & \ast \\
		& \ddots  &  & \vdots  & \ddots  & \vdots \\
		&  & \ast  & \ast  & \cdots  & \ast 
	\end{array} \right),
\end{align}
where each $\ast$ stands for a nonzero local vector in $\mathbb{F}_q^{m_2}$. The actual situation, however, is more complicated: row operations need not eliminate all nonpivot components of the local vectors. A schematic form of the resulting matrix is
\begin{align}\label{eq:restricted_block_2}
	\left(
	\begin{array}{ c | c | c | c c c }
		(\boldsymbol{1},0,1,0)^T  & \cdots &  (1,0,0,0)^T & \ast  & \cdots  & \ast \\
		\vdots & \ddots  & \vdots  & \vdots  & \ddots  & \vdots \\
		(0,1,1,1)^T & \cdots & (0,1,0,\boldsymbol{1})^T   & \ast  & \cdots  & \ast 
	\end{array} \right),
\end{align}
where the boldface $\boldsymbol{1}$ entries are the pivot variables, while the remaining nonpivot entries of the local vectors may be nonzero. For a given $[(h,v_2);b_2]$, two or more pivot variables may occur in different components of the local vectors, corresponding to additional rows. For each lift, there are at most $m_2$ such rows. Since $m_2$ is constant, discarding the additional rows loses at most a constant factor, so the resulting matrix still has $\Theta(l)$ rows.

To define the induced vectors in Eq.~\eqref{eq:induced_vector}, we apply $h_2^T(a_2, \text{-})$ to local vectors. However, an arbitrary choice of $a_2$ does not guarantee that the evaluated matrix in Eq.~\eqref{eq:restricted_block_2} retains $\Theta(l)$ independent rows. To address this issue, we consider all choices of $a_2$, which is equivalent to applying $h_2^T$ to the local vectors. The resulting vectors lie in $\im h_2^T$, namely, the dual code of $\ker h_2$. Under these evaluations, Eq.~\eqref{eq:restricted_block_2} becomes
\begin{align}\label{eq:restricted_block_3}
	\left(
	\begin{array}{ c c c | c c c }
		c_1  &  \cdots  &  d_1  & e_1  & \cdots  & f_1 \\
		\vdots & \ddots  &  \vdots & \vdots  & \ddots  & \vdots \\
		c_{l'}  & \cdots & d_{l'}  & e_{l'}  & \cdots  & f_{l'}
	\end{array} \right),
\end{align}
where $c_i, \ldots, f_i \in \im h_2^T$, and we regard them as row vectors for simplicity.
The resulting matrix has $l' = \Theta(l)$ rows and $n \cdot l$ columns. Since $h_2^T$ has full column rank, applying it to each local vector block is injective and therefore preserves the row rank of Eq.~\eqref{eq:restricted_block_2}. Hence, the matrix in Eq.~\eqref{eq:restricted_block_3} also has row rank
$l'=\Theta(l)$. This yields the following claim.

\begin{claim}\label{claim_2}
	At least a $d_0/n$ fraction of the choices of $a_2$ have the property that evaluation by $h_2^T(a_2,\text{-})$ produces $\Theta(l)$ linearly independent vectors from Eq.~\eqref{eq:restricted_block_2}.
\end{claim}

The reason is also simple. Partition Eq.~\eqref{eq:restricted_block_3} into $n$ column subblocks $A^{(i)}$, each corresponding to a choice $a_2^{(i)}$. By basic linear algebra, the rank of Eq.~\eqref{eq:restricted_block_3} is upper bounded by $\sum_{i=1}^n\rk A^{(i)}$. Since $n = O(1)$, at least one subblock, say $A^{(1)}$, must have rank $\Theta(l)$. 

More generally, after removing any $j<d_0$ selected subblocks, the matrix formed by the remaining subblocks must still have full row rank $l'$. Otherwise, there would be a nontrivial linear combination of the rows that vanishes on all remaining subblocks. Applied to the full matrix in Eq.~\eqref{eq:restricted_block_3}, the same linear combination would be supported on at most $j$ choices of $a_2$. Because the full matrix has row rank $l'$, this combination is nonzero. In at least one lift, it would therefore give a nonzero codeword in $\im h_2^T$ of weight at most $j<d_0$, contradicting the minimum distance $d_0$ of the dual local code (see Theorem~\ref{thm:cLDPC}). Since the remaining matrix has row rank $l'$ and consists of $n-j=O(1)$ subblocks, at least one remaining subblock has rank at least $l'/(n-j)=\Theta(l)$. Iterating this argument yields $A^{(1)},\ldots,A^{(d_0)}$ with rank $\Theta(l)$. Therefore, there are at least $d_0$ choices of $a_2$ for which the action of $h_2^T(a_2,\text{-})$ produces $\Theta(l)$ linearly independent induced vectors from Eq.~\eqref{eq:restricted_block_2}.


\medskip
The above solutions to $\hat{H}_2\psi=0$ over $\mathbb{F}_q$ can be reinterpreted over $R$. Their tensor products with the solutions to $\hat{H}_3\phi=0$ over $R$, together with any $[v_1;a_1]$, give rise to cocycles in $\ker\delta_{\mathrm{res}}^1$. Accordingly, the $R$-valued cochain $\psi\otimes_R\phi$ satisfies
\begin{align}\label{eq:logical_R}
	(\psi\otimes_R\phi)\bigl([v_2;b_2]\otimes[v_3;b_3]\bigr)=\psi([v_2;b_2])\otimes_R\phi([v_3;b_3]).
\end{align}
Since $h_2^T(a_2, \text{-})$ and $h_3^T(a_3, \text{-})$ have scalar entries, they can naturally be viewed as $R$-actions. Thus,
\begin{align}\label{eq:eva_R}
	\begin{aligned}
		& h_2^T(a_2, \text{-}) \otimes_R h_3^T(a_3, \text{-}) \Big( \psi ([v_2;b_2]) \otimes_R \phi ([v_3;b_3]) \Big) \\
		= & \Big( h_2^T(a_2, \text{-}) \psi ([v_2;b_2]) \Big) \otimes_R \Big( h_3^T(a_3, \text{-}) \phi ([v_3;b_3]) \Big). 
	\end{aligned}
\end{align} 
The second line of Eq.~\eqref{eq:eva_R} is simply the product of two elements of $R$. Let $I_2, I_3 \subset \mathbb{F}_q[C_l]$ be the ideals generated by the induced vectors of $\psi,\phi$, respectively. Then, by Eq.~\eqref{eq:eva_R}, the induced vectors of the cocycles in Eq.~\eqref{eq:logical_R} span $I_2 \cdot I_3$.

\begin{remark}
	Without Lemma~\ref{lemma:product}, we can apply the row-reduction and related arguments above directly to cocycles in $\ker \delta_{\mathrm{res}}^1$. This would yield a constant fraction of choices of $[(v_2,v_3);(b_2,b_3)]$ and $(a_2,a_3)$ for which the corresponding evaluations generate an ideal of dimension $\Theta(l)$. The only drawback is that $[v_2;b_2]$, $[v_3;b_3]$, $a_2$, and $a_3$ cannot be selected independently, which would complicate the construction of cup products.
\end{remark}


Assuming GRH and Artin's conjecture, $I_2 = I_3 = \mathbb{F}_q[C_l]$. More generally, if $I_2$ and $I_3$ are \emph{coprime}, then B\'ezout's identity gives $I_2 + I_3 = \langle 1 \rangle$, and hence $I_2 \cdot I_3 = I_2 \cap I_3$. Without the primitive root condition, the worst case occurs when $I_2$ and $I_3$ are generated by complementary divisors of $\x^l-1$: then $I_2 \cdot I_3 = \{0\}$ even though each ideal has $\mathbb{F}_q$-dimension $\Theta(l)$. There are still various possibilities depending on the specific ideals and structure of the algebra~\cite{Aluffi2009,HuffmanPless2003}.

\begin{claim}\label{claim_3}
    For this choice of $l$ and any fixed $[v_1;a_1]$, there is a constant fraction of choices of $[v_2;b_2]$, $[v_3;b_3]$, $a_2$, and $a_3$ for which there are exactly $l$ linearly independent cocycles $L_{1}^i$ whose induced vectors, obtained via $h_2^T(a_2, \text{-}) \otimes h_3^T(a_3, \text{-})$, form the standard basis of $\mathbb{F}_q[C_l] \cong \mathbb{F}_q^l$.
\end{claim}

\begin{corollary}\label{coro:inequivalent}
	Suppose the local parity-check matrices satisfy Eq.~\eqref{eq:local_rank}.  For this choice of $l$, there exist $1$-cubes $[v_1;a_1]$ such that linearly independent cocycles
	\begin{align}
		[v_1;a_1] \otimes \sum_{h,v_2,v_3,b_2,b_3} [(h,v_2,v_3);b_2,b_3] \in \ker \delta_{\mathrm{res}}^1
	\end{align}
	are nontrivial polarized logical representatives of independent cohomology classes.
\end{corollary}
\begin{proof}
    For simplicity, we again restrict the proof to the three-dimensional case, in which Eq.~\eqref{eq:local_rank} specializes to
    \begin{align}\label{eq:local_rank_3D}
		m_1 = \nu' n, \quad \ m_2 = m_3 = (1-\nu) n
	\end{align}
	for sufficiently small $\nu'$ and $\nu \in (0,1)$. Since the difference between two polarized representatives supported on the same $1$-cube $[v_1;a_1]$ is again supported on that cube, it suffices to show that no nonzero linear combination of them lies in $\im \delta^0$, where
	\begin{align}\label{eq:delta^0}
		\delta^0 = \begin{pmatrix} \hat{H}_1 \otimes_R I_{m_2 V} \otimes_R I_{m_3 V} \\ I_{m_1 V} \otimes_R \hat{H}_2 \otimes_R I_{m_3 V} \\ I_{m_1 V} \otimes_R I_{m_2 V} \otimes_R \hat{H}_3 \end{pmatrix}
	\end{align}
	Recall that the subblocks in $\delta^0$ can be explicitly reinterpreted over $\mathbb{F}_q$ when the entries are reindexed as presented in Section \ref{sec:HGP_like}. Let $x \in \im \delta^0$. If it has nontrivial support arising from either of the bottom two subblocks of $\delta^0$, it cannot be of the form in Eq.~\eqref{eq:sheaf_basis}. Consequently, we consider
	\begin{align}\label{eq:inequivalent_2}
		x = \sum_{x_{\mathfrak{E}}} \sum_{v_2,v_3,b_2,b_3} x_{\mathfrak{E}} \otimes [v_2;b_2] \otimes [v_3;b_3] \in \im \hat{H}_1 \otimes I_{m_2 V} \otimes I_{m_3 V},
	\end{align}
	where we suppress the local coefficients at the vertices for simplicity. The subscript $\mathfrak{E}$ of $x_{\mathfrak{E}} \in \im \hat{H}_1$ depends on $m_2 V \times m_3 V$ variables, i.e., the coordinates of vertices $[v_i;b_i]$ with $m_i$ standard basis vectors $\{\zeta_{i,j}\}_{j=1}^{m_i}$ of the local coefficients in $\mathbb{F}_q^{m_i}$ for $i = 2,3$. The $\mathbb{F}_q$-scalar factors are all moved to $x_{\mathfrak{E}}$ in the above tensor product. We are proving that even such linear combinations cannot be in the form of Eq.~\eqref{eq:sheaf_basis}. 
	
	We track the edge component $x_{\mathfrak{E}}$ and the vertex components separately. With scalar coefficients and without lifts, $\ker H_1^T$ is spanned by closed cycles in the graph, so elements of $\im H_1$ cannot be one-hot vectors. With nontrivial local coefficients and still without lifts, ruling out one-hot vectors in $\im H_1$ is generally difficult. Fortunately, Theorem~\ref{thm:cycle} shows that when the operator is defined on the double cover of the base graph, any local codeword in $\ker h_1$ can be extended to an element of $\ker H_1^T$. If a one-hot vector lay in $\im H_1$, then every element of $\ker H_1^T$ would have to vanish in that component, and hence every local codeword in $\ker h_1$ would have to vanish in the corresponding bit. This would give a one-hot dual codeword, which is impossible when the local distance $> 1$.
	
	The difficult case is when the lift is nontrivial. Then it is possible for $x_{\mathfrak{E}}$ to be defined exclusively on the fiber of some $[v_1;a_1]$, and the coboundary $x$ could deform one polarized logical representative to another one. Fortunately, Artin's conjecture prevents this from occurring. With sufficiently small $\nu'$ in Eq.~\eqref{eq:local_rank_3D}, $m_1<n/2$ and $\dim\ker\hat H_1^T$ grows with $l$. We must be able to find vectors in $\ker \hat{H}_1^T$ which are not all-ones lifts (see Definition \ref{def:all_ones}), i.e., their restrictions to the fiber of the base $1$-cubes are not all equal to the all-ones vector. Applying a counting argument analogous to that in Claim~\ref{claim_1} to any basis $\{v_i\}$ of $\ker\hat H_1^T$, there are at least $\left(\frac{1}{2}n - m_1 \right) V$ feasible $[v_1;a_1]$ such that the restrictions of the basis vectors to each of them generate $R = \mathbb{F}_q[C_l]$  by the primitive root property of $q$ modulo $l$. Consequently, for each such $[v_1;a_1]$, there must exist vectors $x_T^{(i)}\in\ker\hat H_1^T$, $i=0,\ldots,l-1$, such that 
	\begin{align}
		x_T^{(i)} \vert_{[v_1;a_1]} = \x^i \in \mathbb{F}_q[C_l], \ i = 0,\ldots,l-1. 
	\end{align}
	For the same $[v_1;a_1]$, if there existed $x_{\mathfrak E}\in\im\hat H_1$ supported entirely on its fiber, then there must be some $i$ such that $\langle x_{\mathfrak{E}},  x_T^{(i)} \rangle \neq 0$, contradicting $x_T^{(i)} \in \ker \hat{H}_1^T$. 
\end{proof}

Note that we carry out the above proofs for three graph coboundary operators $\{\hat{H}_i\}_{i=1}^3$ in three dimensions and the method holds in general $t$ dimensions. For brevity, we explain the outcome in three dimensions: after carefully choosing $[v_1;a_1]$, we see that Claim \ref{claim_3} provides exactly $l$ independent cohomology classes with polarized logical representatives. However, from the previous example of HGP codes, one can imagine that we still need logical representatives polarized in different directions. In forming nontrivial cup products, the choice of $[v_1;a_1]$ has to be compatible with the construction of other logical representatives (see also Section \ref{sec:cup_sheaf} for more details). The most ideal case is when $[v_1;a_1]$ can be taken arbitrarily, but the proof of Corollary \ref{coro:inequivalent} only grants $\left(\frac{1}{2}n - m_1 \right) V$ feasible ones. We solve this issue in Section \ref{sec:edge_choice}. 

Additionally, there is another caveat, we cannot assume that $h_2 = h_3$ in general. Although the high-dimensional expander $X$ is constructed as a product of copies of the same graph, the local codes associated with $h_2$ and $h_3$ may be distinct, so $\hat{H}_2 \neq \hat{H}_3$. Indeed, apart from the conditions in Theorem~\ref{thm:cLDPC}, HGP codes impose no additional constraints on the local codes. However, to obtain (almost) good qLDPC codes, the local codes in Ref.~\cite{KP2025Extendable} are sampled randomly so that the required conditions hold with probability $1 - 2n^t 2^{n^t - \log q + 1}$. On the other hand, let $m_2 = m_3$ and let $k = n - m_2$. Then the total number of $k$-dimensional subspaces in $\mathbb{F}_q^n$ is the Gaussian binomial coefficient $\binom{n}{k}_q$. For $t = 3$, the probability that $h_2 = h_3$ is
\begin{align}
	1 \Big/ \binom{n}{k}_q \approx q^{-k(n-k)} = 2^{-k(n-k) \log q},
\end{align}
whereas the success probability of the random construction satisfies
\begin{align}
	1 - 2n^3 2^{n^3 - \log q + 1} < 1 - 2^{-k(n-k) \log q},
\end{align}
and both sides of this inequality tend to $1$ as $q$ grows. We explain how to work with random local codes in Section \ref{sec:cup_sheaf}.


\subsection{Arbitrary support for polarized logical representatives}
\label{sec:edge_choice}

We reinforce Corollary \ref{coro:inequivalent} by showing that the polarized logical representatives can be defined on any $[v_1;a_1]$. More precisely, given the choices of $v_2,v_3,b_2,b_3$ in Claim~\ref{claim_3}, the family $\{L_1^i\}_{i=1}^l$ can be supported on any $1$-cube $[v_1;a_1]$ and provides $l$ independent cohomology classes. In addition to the assumptions of Corollary \ref{coro:inequivalent}, we require the local dual code $\im h_1^T$ to satisfy a simple condition in \eqref{eq:edge_condition}. The expander graph lifted from $G_0$ also needs to be carefully chosen. We give the details below.

The proof idea is simple: by the argument in Corollary~\ref{coro:inequivalent} and the primitive root property of $q$ modulo $l$, it suffices to find, for every $[v_1;a_1]$, a vector $x_T\in\ker\hat H_1^T$ such that $x_T|_{[v_1;a_1]}$ is not proportional to the all-ones vector. More precisely, we prove this by showing that, when solving the linear system $\hat H_1^Tx_T=0$, at least two variables in the fiber of $[v_1;a_1]$ can be chosen freely. Assigning different values to these variables yields the required vector $x_T$. Although the rank-nullity argument only gives a lower bound $\left(\frac{1}{2}n - m_1 \right) V$ on the number of feasible 1-cubes, we will show that the desired property holds under suitable conditions. For simplicity, let us drop the subscripts on $\hat{H}_1^T$, $m_1$, and the other relevant objects in the following context within this subsection. Using an \emph{Euler tour}, one can orient any finite undirected graph so that every vertex $v$ is incident to at least $\lfloor d(v)/2 \rfloor$ outgoing edges \cite{BangJensen2009Digraphs,Diestel2025Graph}. Here $d(v)$ is the node degree of $v$. In our case, let $G_0$ be the base Cayley graph, and let $e_0$ be the base edge corresponding to the target $[v_1;a_1]$ in the double cover. We then apply the preceding result to the subgraph $G_0 \setminus \{e_0\}$. For any vertex $v \in V(G_0 \setminus \{e_0\}) = V(G_0)$, there are at least $\lfloor (n-1)/2 \rfloor$ outgoing edges. Let $E_v$ denote this set, so $e_0 \notin E_v$. Recall that $A = \{a^{(1)},\ldots,a^{(n)}\}$ determines the local ordering of the edges in $G_0$ and corresponds to the $n$ columns of the local parity-check matrix $h$ (see Section \ref{sec:SS_code}). Since $m<n/2$, we have $m\leq\lfloor(n-1)/2\rfloor$. Choose $J_v\subseteq A$ with $|J_v|=m$ such that the edges labeled by elements of $J_v$ belong to $E_v$.

For conciseness, we also omit the superscript on $a^{(i)}$. For any $J \subseteq A$, let $J^{-1}\coloneq\{a^{-1}:a\in J\}$, and let $h_J$ denote the submatrix of $h$ formed by the columns indexed by $J$. To prove the result, we impose the following condition on the local parity-check matrix: for every subset $S \subseteq A$ with $|S|=\lfloor(n-1)/2\rfloor$, there exists a subset $ J\subseteq S$ with $|J|=m$ such that
\begin{align}\label{eq:edge_condition}
	\det h_J\neq0,\qquad\det h_{J^{-1}}\neq0.
\end{align}
Clearly, the parity-check matrix of every MDS code satisfies this condition. In particular, the RS parity-check matrices used below satisfy this condition.

We will show at the end of this subsection that most random codes also satisfy this condition. In our case, given any $E_v$, we set $J_v \subseteq E_v$ such that $\det h_{J_v}, \quad \det h_{J_v^{-1}} \neq 0$. Before taking lifts, we consider the Tanner parity-check matrix of the double cover $G_0 \times K_2$ with the following 1-cubes:
\begin{enumerate}
	\item For each $[v;0]$, we take the $m$ 1-cubes $[v;a]$ corresponding to edges in $E_v$, with $a \in J_v$. This gives $m V_0$ 1-cubes.
	
	\item Each such $[v;a]$ has $[v;0]$ and $[a \cdot v; 1]$ as endpoints. We then take $[a \cdot v; a^{-1}]$. 
\end{enumerate}
In total, we have taken exactly $2m V_0$ 1-cubes and the target $[v_1;a_1]$ is not included. We will prove that, for a suitable cyclic lift, $\hat{H}^T$ restricts to an operator that can be represented by an $(2m (l-1) \cdot V_0) \times (2m (l-1) \cdot V_0)$ matrix. As a reminder, we continue to use $V_0 = V/2$ in this subsection due to the necessity to deal with the base graph. It follows that $\hat{H}^T$ has $2m (l-1) \cdot V_0$ pivot columns, none of which lies in the fiber of $[v_1;a_1]$. Since $\hat{H}^T$ is an $(2m l \cdot V_0) \times (n l \cdot V_0)$ matrix, $\rk \hat{H}^T \leq 2m l \cdot V_0$. When $l > 2m \cdot V_0 + 2$, the fiber of $[v_1;a_1]$ contains at least two free variables. 

Suppose we have found the lift, i.e., the assignment $\gamma: E_0 \rightarrow C_l$ (see Section \ref{sec:cubical}). As explained before, $\hat{H}^T$ is simply defined by replacing nonzero entries of $H^T$ by permutation matrices. These permutation matrices are the $l \times l$ representations of the elements $\gamma(u,v)$. 
Let $W\subseteq \mathbb{F}_q^l$ be the subspace of the augmentation ideal
\begin{align}
	W \coloneq \left\{(z_0,\ldots,z_{l-1}) \in \mathbb{F}_q^l:
	\sum_{i=0}^{l-1} z_i = 0 \right\}.
\end{align}
Since $l$ is odd and $\operatorname{char}\mathbb{F}_q = 2$, the all-ones vector $\boldsymbol{1}$ does not belong to $W$, and hence $\mathbb{F}_q^l = W \oplus \operatorname{span}\{ \boldsymbol{1} \}$. The cyclic shift action of $C_l$ preserves $W$. For the chosen prime $l$ for which $q$ is a primitive root modulo $l$, the cyclotomic polynomial
\begin{align}
	\Phi_l(\x) = 1+\x +\cdots+ \x^{l-1}
\end{align}
is irreducible over $\mathbb{F}_q$. Therefore, the minimal polynomial of $\x|_W$ is just $\Phi_l$. The minimal polynomial then implies that
\begin{align}\label{eq:nonzero_cyclic}
	P(\x)\big|_W \neq 0 \ \Longleftrightarrow \  P(\x)\big|_W \text{ is invertible on }W,
\end{align}
We will also use the following result to determine the invertibility of the relevant matrices later.

\begin{lemma}\label{lemma:gamma_def}
	Let $M_0,M_1,M_2$ be generated by $\x$. If at least one of them is nonzero on $W$, then $M_0 + M_1 \x^j + M_2 \x^{2j} = 0$ on $W$ for at most two values $j \in \{0,\ldots,l-1\}$.
\end{lemma}
\begin{proof}
	Suppose, to the contrary, that the equality holds for three distinct values $j_1,j_2,j_3$. Then
	\begin{align}
		& ( M_0 + M_1 \x^{j_1} + M_2 \x^{2j_1} ) - ( M_0 + M_1 \x^{j_2} + M_2 \x^{2j_2} ) = ( \x^{j_1} - \x^{j_2} ) ( M_1 + M_2 ( \x^{j_1} + \x^{j_2} ) ) = 0 \\
		& ( M_0 + M_1 \x^{j_3} + M_2 \x^{2j_3} ) - ( M_0 + M_1 \x^{j_2} + M_2 \x^{2j_2} ) = ( \x^{j_3} - \x^{j_2} ) ( M_1 + M_2 ( \x^{j_3} + \x^{j_2} ) ) = 0.
	\end{align}
	Since $j_1 \neq j_2$, $\x^{j_1} - \x^{j_2} \neq 0$ and is therefore invertible. Hence $( M_1 + M_2 ( \x^{j_1} + \x^{j_2} ) ) = 0$. Similarly, the second identity gives
	\begin{align}
		( M_1 + M_2 ( \x^{j_3} + \x^{j_2} ) ) = 0 = ( M_1 + M_2 ( \x^{j_1} + \x^{j_2} ) ) \implies  M_2 ( \x^{j_3} - \x^{j_1} ) = 0 \implies M_2 = 0.
	\end{align}
	This forces $M_1 = M_0 = 0$ and leads to a contradiction.
\end{proof}


Consider $\hat{H}^T: \mathbb{F}_q^{n l \cdot V_0} = \mathbb{F}_q^{l \cdot E} \to \mathbb{F}_q^{2m l \cdot V_0} = \mathbb{F}_q^{m l \cdot V}$. Since each subblock $\hat{H}$ is of the form $c \x^i$ for some $c \in \mathbb{F}_q$, the boundary operator restricts to
\begin{align}
	\hat{R}^T: W^{n V_0} \to W^{2m V_0} 
\end{align}
by restricting each $c \x^i$ to $W$. After choosing any basis of $W$, $\hat{R}^T$ is represented by a $2m (l-1) \cdot V_0 \times n (l-1) \cdot V_0$ matrix. For the preceding choice of $2m V_0$ 1-cubes, which does not contain $[v_1;a_1]$, we select the corresponding $2m (l-1) \cdot V_0$ columns of $\hat{R}^T$ and denote the resulting submatrix by $B$.

\begin{claim}\label{claim_4}
	There exists a cyclic lift $\gamma:E_0\to C_l$ such that the lifted graph is a good expander and, under condition~\eqref{eq:edge_condition}, the corresponding matrix $B$ is invertible for every choice of the target base edge $e_0$. Consequently, the independent cohomology classes with polarized logical representatives $L_1^i$ from Claim \ref{claim_3} can be found with respect to any 1-cube $[v_1;a_1]$.
\end{claim}
\begin{proof}
	We continue to denote the subblocks of $\hat{R}^T$ by $c \x^i$, with $\x^i$ understood as its restriction to $W$. Since each subblock commutes, the determinant of $B$ can be found as if dealing with $c \x^i$ as ``scalars''. More precisely,
	\begin{align}\label{eq:block_det}
		\det_{\x} (B) = \sum_{\alpha \in S_{2m V_0}} \mathrm{sgn}(\alpha) \cdot B_{1,\alpha(1)} \cdots B_{2m V_0,\alpha(2m V_0)}
	\end{align}  
	where $B_{i,j}$ is a block of $B$. Note that $\det_{\x} (B)$ itself is an $(l-1) \times (l-1)$ matrix formed by linear combinations of $\x^i$. Note that $\mathrm{sgn}(\alpha) \equiv 1$ for characteristic $2$. By \eqref{eq:nonzero_cyclic}, $\det_{\x} (B) \neq 0$ implies that it is invertible, and basic linear algebra shows the block matrix $B$ is invertible.
	
	We now compute $\det_{\x} (B)$ given the restricted Tanner parity-check matrix $\hat{R}^T$. Recall that each column of $\hat{H}^T$ is obtained from a column $h(\text{-}, a)$ of the local parity-check matrix together with the permutation $\x^i$ determined by the voltage assignment $\gamma$. As explained before, each column of $\hat{R}^T$ also has the form $h(\text{-}, a) \otimes \x^i$, with $\x^i = \x^i \vert_W$. Viewed as a block matrix, each subblock is of the form $c \x^i$, where $c = 0$ or $c$ is an entry of $h(\text{-}, a)$. As introduced in Sections \ref{sec:covering}, \ref{sec:SS_code}, and \ref{sec:DLV_code}, the matrix locally has the form
	\begin{equation}\label{eq:block_local}
		\left(\begin{array}{cc|ccccc}
			& [v;0]: & h(\text{-}, a) \otimes \mathrm{id} & 0 & \cdots & h(\text{-}, a'^{-1}) \otimes \mathrm{id} & \cdots \\
			& [a \cdot v; 1]: & h(\text{-}, a) \otimes \x^i & 0 & \cdots & 0 & \cdots \\
			& [a \cdot v; 0]: & 0 & h(\text{-}, a^{-1}) \otimes \mathrm{id} & \cdots & 0 & \cdots \\
			& [v;1]: & 0 & h(\text{-}, a^{-1}) \otimes \x^{-i} & \cdots & 0 & \cdots 
		\end{array} \right),
	\end{equation} 
	where, in the first row, the column that contains $h(\text{-}, a'^{-1}) \otimes \text{id}$ corresponds to a 1-cube containing $[v;0]$ and some $[w;1]$ with $w \neq a \cdot v$. Its directed base edge starts at $w$ and ends at $v$. Explicit examples of Tanner parity-check matrices also appear in Section \ref{sec:cycles}. To obtain a term in Eq.~\eqref{eq:block_det}, we proceed through the rows of $B$ from top to bottom and multiply the selected entries. For the rows indexed by $[v;0]$, one may select nonzero entries from $h(\text{-}, a) \otimes \mathrm{id}$. Since each column can be selected only once in the determinant expansion, we cannot simultaneously use $h(\text{-}, a) \otimes \x^i$ from the rows indexed by $[a \cdot v; 1]$. Accordingly, the local block products in \eqref{eq:block_local} can be classified into four patterns (with scalars being hidden):
	\begin{align}
		\x^{-i} = \mathrm{id} \cdot \x^{-i}, \quad
		\mathrm{id} = \mathrm{id} \cdot \mathrm{id} = \x^{i} \cdot \x^{-i}, \quad
		\x^{i} = \x^{i} \cdot \mathrm{id}.
	\end{align}
	Let $\x^{j_1},\ldots,\x^{j_{m V_0}} \in C_l$ be the voltage assignments of $\gamma$ to the chosen base edges. If every local product follows the first or the fourth pattern, the complete product in Eq.~\eqref{eq:block_det} is
	\begin{align}\label{eq:block_product}
		\prod_{v \in V_0} \det h_{J_v} \cdot \det h_{J_v^{-1}} \x^{c_1 j_1 + \cdots c_{m V_0} j_{m V_0}}, 
	\end{align}
	where $c_1,\ldots,c_{m V_0} \in \{\pm 1\}$. This is because the local parity-check matrix restricts to $h_{J_v}$ on the rows indexed by $[v;0]$ and to $h_{J_v^{-1}}$ on the rows indexed by $[v; 1]$, as specified below Condition~\eqref{eq:edge_condition}. Similarly, it restricts to $h_{J_{a \cdot v}}$ on the rows indexed by $[a \cdot v;0]$ and to $h_{J_{a \cdot v}^{-1}}$ on the rows indexed by $[a \cdot v;1]$. In the remaining cases, e.g., when entries are selected first from $h(\text{-}, a) \otimes \x^i$ and then from $h(\text{-}, a^{-1}) \otimes \x^{-i}$, the submatrix $h_{J_v}$ is not completely preserved in the rows indexed by $[a \cdot v; 1]$. Consequently, the complete product may be either zero or a more complicated value without a compact form such as Eq.~\eqref{eq:block_product}. Consequently, the determinant has the form
	\begin{align}\label{eq:block_det_2}
		\det_{\x} (B) = \prod_{v \in V_0} \sum_{c_i = \{-1,1\}} \det h_{J_v} \cdot \det h_{J_v^{-1}} \x^{c_1 j_1 + \cdots c_{m V_0} j_{m V_0}} 
		+ \sum_{\text{some } c_i = 0} C_{j_1,\ldots,j_{m V_0} } \x^{c_1 j_1 + \cdots c_{m V_0} j_{m V_0}}.  
	\end{align}  
	where $C_{j_1,\ldots,j_{m V_0} }$ denotes the potentially complicated values.
	
	By Condition~\eqref{eq:edge_condition}, $\det_{\x} (B)$ has nonzero coefficients $\det h_{J_v} \cdot \det h_{J_v^{-1}}$. It remains to show that there exists a voltage assignment $\gamma$ that both yields a good expander and makes Eq.~\eqref{eq:block_det_2} invertible. Given $m V_0$ base edges, there are $l^{m V_0}$ possible assignments $\gamma$. We claim that at most $2m V_0 l^{m V_0 - 1}$ of them make $B$ singular. 
	
	If $m V_0 = 1$, Lemma \ref{lemma:gamma_def} shows that at most two values in $\{0,1,\ldots,l-1\}$ make $\det_{\x} (B) = 0$. For the general case, we argue by induction:
	\begin{align}
		\det_{\x} (B) = \left[ M_0(j_1,\ldots,j_{m V_0 - 1}) + M_1(j_1,\ldots,j_{m V_0 - 1}) \x^{j_{m V_0}} + M_2(j_1,\ldots,j_{m V_0 - 1}) \x^{2j_{m V_0}} \right] \x^{-j_1 - \cdots - j_{m V_0}}
	\end{align}
	where $M_i(j_1,\ldots,j_{m V_0 - 1})$ are obtained from Eq.~\eqref{eq:block_det_2} by collecting the terms involving the powers $\x^{j_1},\ldots,\x^{j_{m V_0 - 1}}$, and the factor $\x^{-j_1 - \cdots - j_{m V_0}}$ shifts each exponent from $\{-1,0,1\}$ to $\{0,1,2\}$ as in Lemma \ref{lemma:gamma_def}. By definition, $M_0(j_1,\ldots,j_{m V_0 - 1})$ and $M_2(j_1,\ldots,j_{m V_0 - 1})$ contain powers $0$ and $2$ of $\x^{j_1},\ldots,\x^{j_{m V_0 - 1}}$ with nonzero scalar coefficients. The induction hypothesis shows that there are at most $2(m V_0 - 1) l^{m V_0 - 2}$ choices for which the relevant $M_i(j_1,\ldots,j_{m V_0 - 1})$ vanish. For these choices, $\det_{\x} (B) = 0$ for every value of $j_{m V_0}$. For all remaining choices, Lemma \ref{lemma:gamma_def} gives at most two undesirable values of $j_{m V_0}$, so the total number of undesirable assignments is
	\begin{align}
		2(m V_0 - 1) l^{m V_0 - 2} \cdot l + l^{m V_0 - 1} \cdot 2 = 2m V_0 l^{m V_0 - 1}.
	\end{align}
	Therefore, for a fixed base edge $e_0$, the probability that a uniformly random $\gamma$ gives a singular matrix $B$ is at most $2m V_0 / l$.   
    Taking a union bound over all edges in $E_0$ gives $(nV_0/2)\cdot(2mV_0/l)=nmV_0^2/l$. On the other hand, the probability that a cyclic lift fails to be a good expander is at most $l\exp(-cV_0/n^2)$~\cite{Agarwal2016}. The total failure probability is therefore bounded above by
    \begin{align}
	\frac{n m V_0^2}{l} + l\exp\left(-c\frac{V_0}{n^2}\right).
	\end{align}
	This quantity is less than $1$ if 
	\begin{align}
	2n m V_0^2 < l < \frac{1}{2} \exp\left(c\frac{V_0}{n^2}\right).
	\end{align}
    For a sufficiently large prime $l$ for which $q$ is a primitive root modulo $l$, choose $V_0 = Cn^2 \log l$ with a sufficiently large constant $C>1/c$. For fixed $n$ and $m$, both inequalities then hold for sufficiently large $l$, which proves the statement.
\end{proof}


We now verify the existence of local codes satisfying Condition~\eqref{eq:edge_condition}. Given a local parity-check matrix $h$, suppose that it does not satisfy Condition~\eqref{eq:edge_condition}. Then there exists a subset $S \subseteq A$ with $|S| = \lfloor (n-1)/2 \rfloor$ for which no suitable $J$ can be found. Let $I \subseteq S$ be a maximal subset such that both
\begin{align}
	\{ h(\text{-},a): a \in I \}, \quad \{ h(\text{-},a^{-1}): a \in I \} 
\end{align}
are independent sets. Let $V_1 \coloneq \mathrm{span}\{ h(\text{-},a): a \in I \}$ and $V_2 \coloneq \mathrm{span}\{ h(\text{-},a^{-1}): a \in I \}$. Then
\begin{align}
	\dim V_1 = \dim V_2 = m' = \vert I \vert < m.
\end{align}
For any $a' \in S \setminus I$, maximality implies that either $h(\text{-},a') \in V_1$ or $h(\text{-},a'^{-1}) \in V_2$. Let
\begin{align}
	I_1 \coloneq \{a' \in S \setminus I: h(\text{-},a') \in V_1 \}, \quad
	I_2 \coloneq \{a' \in S \setminus I: h(\text{-},a'^{-1}) \in V_2 \}.
\end{align}
By definition, $I_1 \cup I_2 = S \setminus I$. Moreover,
\begin{align}
	\vert S \setminus I \vert = \lfloor (n-1)/2 \rfloor - m' \implies \max\{ |I_1|, |I_2| \} \geq \frac{\lfloor (n-1)/2 \rfloor - m' }{2}.
\end{align}
Suppose $|I_1| \geq \frac{\lfloor (n-1)/2 \rfloor - m' }{2}$. Then $|I_1 \cup I| \geq \frac{\lfloor (n-1)/2 \rfloor + m' }{2}$ and $\dim \im h \vert_{I_1 \cup I} = m'$. An analogous conclusion holds for $(I_2\cup I)^{-1}$ if
$|I_2|\geq\frac{\lfloor(n-1)/2\rfloor-m'}{2}$. In summary, if $h$ fails Condition~\eqref{eq:edge_condition}, then there exist some $0 \leq m' < m$ and a subset $\mathcal{I}$ such that
\begin{align}\label{eq:derived_failure}
	\dim \im h \vert_{\mathcal{I}} = m', \quad  |\mathcal{I}| \geq \frac{\lfloor (n-1)/2 \rfloor + m' }{2}.
\end{align}
Let $r_{m'}\coloneq \left\lceil \frac{\lfloor(n-1)/2\rfloor+m'}{2} \right\rceil$. After passing to a subset containing a basis if necessary, we may assume that $|\mathcal{I}|=r_{m'}$.

We now estimate the probability that a random code satisfies \eqref{eq:derived_failure}. For any fixed subset $\mathcal{I}$ with $|\mathcal{I}|=r_{m'}$, choose the corresponding columns independently at random. Then \eqref{eq:derived_failure} can hold only if there is some $I \subseteq \mathcal{I}$ with $\vert I \vert = m'$ and
\begin{align}
	\mathrm{span}\{ h(\text{-},a): a \in \mathcal{I} \setminus I \} \subseteq \mathrm{span}\{ h(\text{-},a): a \in I \}.
\end{align}
Since $\vert I \vert = m' < m$, $\mathrm{span}\{ h(\text{-},a): a \in I \}$ is a proper subspace of $\mathbb{F}_q^m$. The maximum probability that a uniformly random vector belongs to a proper subspace is $1/q$ attained by subspaces of codimension $1$. For this fixed $\mathcal{I}$, the probability that
\eqref{eq:derived_failure} holds is thus bounded by
\begin{align}
	\binom{r_{m'}}{m'} \left(\frac{1}{q}\right)^{r_{m'}-m'}.
\end{align}
Taking a union bound over all $0\leq m'<m$ and all possible subsets $\mathcal{I}\subseteq A$ with $|\mathcal{I}|=r_{m'}$, the failure probability is bounded by
\begin{align}\label{eq:edge_condition_prob}
	\sum_{m'=0}^{m-1} \binom{n}{r_{m'}} \binom{r_{m'}}{m'} \left(\frac{1}{q}\right)^{r_{m'}-m'}.
\end{align}
For fixed $n$ and $m$, this bound tends to zero.


\subsection{Hypergraph Tur\'an problem and local codes}\label{sec:cup_sheaf}

Combining the results of Secs.~\ref{sec:group} and~\ref{sec:module} and~\ref{sec:edge_choice} shows that there are $l$ inequivalent logical representatives $L_{1}^i$ such that the induced vectors
\begin{align}\label{eq:local_vector}
	h_2^T(a_2, \text{-}) \otimes h_3^T(a_3, \text{-}) \Big( L_{1}^i ( [v_1;a_1] \otimes [(h,v_2,v_3); b_2,b_3] ) \Big)
\end{align} 
form the standard basis of $\mathbb{F}_q[C_l] \cong \mathbb{F}_q^l$. Recall that the binary variables $b_2,b_3$ may be chosen arbitrarily by Eqs.~\eqref{eq:delta_condition1}--\eqref{eq:delta_condition4}. However, when forming cup products, if $(b_2,b_3) = (1,1)$, for example, we must replace $v_2$ and $v_3$ by $a_2^{-1}\cdot v_2$ and $a_3^{-1}\cdot v_3$, respectively. For simplicity, throughout the following discussion we assume that $(b_2,b_3)=(0,0)$. Again, for ease of exposition, we work with three-dimensional sheaf codes, and all the results extend to general dimensions.

Lemma~\ref{lemma:rank} applies only to the restricted coboundary operator associated with $h_2$ and $h_3$. We now write $h_1'$ for the parity-check matrix of the first local code, where $m_1 = \nu'n$. Since $h_1'$ has substantially smaller rank than $h_2$ and $h_3$, the same rank estimate does not apply to the other restrictions. To define $L_2^i$ and $L_3^i$, we propose a different strategy employing distinct cellular sheaves. Let $\F_1 = \F$ denote the cellular sheaf generated by $(h_1',h_2,h_3)$. We construct two additional qLDPC codes on the same cubical complex $X$, with local parity-check matrices arranged as
\begin{align}\label{eq:local_code_Mantel}
	(h_1, h_2', h_3), \quad (h_1, h_2, h_3'),
\end{align}
where 
\begin{align}
	& \rk h_1' = \rk h_2' = \rk h_3' = m_1 = \nu' n = n - k_1', \label{eq:local_code_Mantel_rank_1} \\
	& \rk h_1 = \rk h_2 = \rk h_3 = m_2 = m_3 = (1-\nu) n = n - k_1. \label{eq:local_code_Mantel_rank_2}
\end{align}
Without imposing any distance requirement on the sheaf codes, such local codes are easy to find. To obtain the desired code distances, however, the notion of product expansion~\cite{PK2023RobustlyTestable,DHLV2022,KP2025Extendable} is crucial. We now prove the existence of these local codes with Condition~\eqref{eq:edge_condition} and product expansion. A preferred triple $(h_1',h_2,h_3)$, or more precisely, a triple of subspaces $(\ker h_1',\ker h_2,\ker h_3)$, can be sampled from the product of Grassmannians $\text{Gr}_q(n,k_1') \times \text{Gr}_q(n,k_1) \times \text{Gr}_q(n,k_1)$ with success probability $1- 2n^3 2^{n^3-\log q+1}$. Viewing these subspaces as points and averaging over the first coordinate, we obtain a choice of $\ker h_1'$ for which at least a $1-2n^3 2^{n^3-\log q+1}$ fraction of the pairs in $\text{Gr}_q(n,k_1)\times\text{Gr}_q(n,k_1)$ are usable \cite{KP2025Extendable}. Even with Condition~\eqref{eq:edge_condition}, the failure probability \eqref{eq:edge_condition_prob} tends to $0$ for sufficiently large $q$. Hence, the averaging may be restricted to primed codes satisfying Condition~\eqref{eq:edge_condition}, with an $o_q(1)$ loss.
We then use the following result from extremal graph theory.
\begin{theorem}[Mantel~\cite{Mantel1907}]
	If a graph contains no triangles, then $\vert E(G) \vert \leq \lfloor \frac{1}{4} \vert V(G) \vert^2 \rfloor$.
\end{theorem} 

Here, the graph is defined by points in $\text{Gr}_q(n,k_1)$, with two vertices adjacent whenever the corresponding pair admits such a primed parity-check matrix. The probability estimate above shows that this graph is asymptotically complete as $q$ grows. Therefore, for sufficiently large $q$, Mantel's theorem guarantees the existence of a triangle. Its three edges correspond to the pairs $(h_2,h_3)$, $(h_1,h_2)$, and $(h_1,h_3)$. After adjoining the appropriate primed code to each pair, these edges give the triples
\begin{align}
	(h_1',h_2,h_3), \quad (h_3',h_1,h_2), \quad (h_2',h_1,h_3).
\end{align}
Permuting the coordinates yields the desired local codes and sheaves. The quantum codes obtained from
\begin{align}
	C^\bullet(X,\F_1), \quad C^\bullet(X,\F_2), \quad C^\bullet(X,\F_3)	
\end{align}
are not identical, but cup products with coefficients in $\F_1\otimes\F_2\otimes\F_3$ are well defined, as explained in Sec.~\ref{sec:cup}. The representatives $L_2^i$ and $L_3^i$ lie in $C^\bullet(X,\F_2)$ and $C^\bullet(X,\F_3)$, respectively. Therefore, by Lemma~\ref{lemma:basis} and Eqs.~\eqref{eq:delta_condition1}--\eqref{eq:delta_condition4}, for $i_1,i_2,i_3=1,\ldots,l$,
\begin{align}\label{eq:sheaf_cup}
L_1^{i_1} \smile_c L_2^{i_2} \smile_c L_3^{i_3} = \sum_{i=1}^l \delta_{i_1,i_2,i_3,i} [(\x^i,v_1,v_2,v_3);a_1,a_2,a_3].
\end{align}
where $\x$ is a generator of $C_l$ and $\x^i$ ranges over all elements of the group.

Generalizing Eq.~\eqref{eq:sheaf_cup} to higher dimensional sheaf codes is straightforward. The only problem is finding suitable local codes. The hypergraph Tur\'{a}n problem provides the higher-dimensional analogue of Mantel's theorem. We illustrate this with the case $t=4$, for which we need 
\begin{align}
	(h_1',h_2,h_3,h_4), \quad (h_1,h_2',h_3,h_4), \quad (h_1,h_2,h_3',h_4), \quad (h_1,h_2,h_3,h_4').
\end{align}
After masking the primed matrices, these tuples form the four edges of the $3$-uniform hypergraph $K_4^3$ on four vertices. In combinatorial terms, we consider $3$-uniform hypergraphs on $\vert\text{Gr}_q(n,k)\vert$ vertices. Each hyperedge contains three distinct vertices, and the problem is to determine the minimum number of hyperedges that forces a copy of $K_4^3$. The ordinary graph case is given by Mantel's theorem~\cite{Mantel1907} and is generalized to other forbidden subgraphs by Tur\'an's theorem~\cite{Turan1941} and the Erd\H{o}s--Stone theorem~\cite{Erdos1946}. Sharp bounds remain unknown even for $3$-uniform hypergraphs, but explicit estimates are available.

Formally, an $r$-uniform hypergraph $G$ is $\mathcal{H}$-free if $\mathcal{H} \nsubseteq G$. The \emph{extremal number} $\text{ex}(\vert V(G)\vert,\mathcal{H})$ is the maximum number of edges in an $\mathcal{H}$-free hypergraph on $\vert V(G)\vert$ vertices. The results of Refs.~\cite{kostochka1982,Chung1999,Razborov2010,Keevash2011} imply that
\begin{align}\label{eq:extremal_number}
	\Big( 1 - \Big( \frac{t-2}{t-1} \Big)^{t-2} - o(1) \Big) \binom{\vert \text{Gr}_q(n,k_1) \vert}{t-1} \leq \text{ex}(\vert \text{Gr}_q(n,k_1) \vert,K_t^{t-1}) \leq \Big( 1 - \frac{1}{\binom{t-1}{t-2}} + o(1) \Big) \binom{\vert \text{Gr}_q(n,k_1) \vert}{t-1}.
\end{align}
Since $t$ is a constant, we may take $q=2^s$ with $s\geq(n+1)^t$ so that
$1-2n^t2^{n^t-s+1}$ exceeds the upper-bound density in
Eq.~\eqref{eq:extremal_number}. This establishes the $t$-dimensional case. Combining this with the main results of Ref.~\cite{KP2025Extendable}, we obtain the following theorem.

\begin{theorem}\label{thm:local_codes}
	For any fixed $t$ and code rates $k_1'/n,k_1/n\in(0,1)$, there exists $\rho>0$ such that, for all sufficiently large $n$ and $q$,  there exist tuples of $\mathbb{F}_q$ local codes
	\begin{align}\label{eq:local_code_Turan}
		(\C_1',\C_2, \ldots,\C_t), \quad (\C_1,\C_2', \ldots,\C_t), \ \cdots, (\C_1,\C_2, \ldots,\C_t'),
	\end{align}
	such that every primed code $\C_i'$ has dimension $k_1'$, every unprimed code has dimension $k_1$, and each tuple is two-way $\rho$-product-expanding.
\end{theorem}

To summarize, let $\{X\}$ be a family of high-dimensional expanders obtained as lifts of products $\{X'\}$ of graphs, and let $\F_1,\ldots,\F_t$ be defined by the two-way product-expanding local codes supplied by Ref.~\cite{KP2025Extendable} and Theorem~\ref{thm:local_codes}. In this case, suppose the dual codes have distance $\geq d_0$. For every $i\in[t]$, $\{C^0(X,\F_i) \to C^1(X,\F_i) \to C^2(X,\F_i)\}$ defines a family of qLDPC codes with nearly optimal parameters
\begin{align}
	[\![ N, \Theta(N), \Theta(N/\mathrm{polylog}\, N)  ]\!].
\end{align}
Moreover, the following theorem holds.

\begin{theorem}\label{thm:cup_sheaf_1}
    Assume GRH and that $q=2^s$ with $s$ odd. Hooley's theorem on Artin's conjecture on primitive roots~\cite{Hooley1967} implies that there are infinitely many primes $l$ for which $q$ is a primitive root modulo $l$. For each such $l$, choose compatible base factors of size $V_0=\Theta(\log l)$ satisfying the hypotheses of Claim~\ref{claim_4}. Then there exists a suitable $l$-lift $X$. With the above conditions, there are $l^t$ tuples of inequivalent logical representatives $(L_1^{i_1},\ldots,L_t^{i_t})$ from the nearly optimal codes defined by $C^\bullet(X,\F_1),\ldots,C^\bullet(X,\F_t)$ satisfying
	\begin{align}\label{eq:sheaf_cup_thm}
		L_1^{i_1}\smile_c\cdots\smile_cL_t^{i_t}=\sum_{i=1}^l\delta_{i_1,\ldots,i_t,i}[(\x^i,v_1,\ldots,v_t);a_1,\ldots,a_t].
	\end{align}
	for a constant fraction of the vertices $v_j$ in the base graph and at least $d_0$ choices of each $a_j$.
\end{theorem}

Note that product expansion \cite{PK2023RobustlyTestable,DHLV2022,KP2025Extendable} is not required for this conclusion. As long as the local codes satisfy the rank conditions of Lemma~\ref{lemma:rank} and are arranged as in Eq.~\eqref{eq:local_code_Turan}, we have $k = \Theta(N)$ and Eq.~\eqref{eq:sheaf_cup_thm} remains valid. When Claim \ref{claim_4} holds, the involved polarized logical representatives correspond to independent cohomology classes. As a caveat, suppose that another collection of logical representatives satisfies
\begin{align}
L_1'^{i_1}\smile_c\cdots\smile_cL_t'^{i_t}=\sum_{i=1}^l\delta_{i_1,\ldots,i_t,i}[(\x^i,v_1',\ldots,v_t');a_1',\ldots,a_t'].
\end{align}
We cannot conclude that the corresponding cohomology classes of both collections are mutually independent. This ambiguity arises because their construction uses the global cyclic actions of Lemma~\ref{lemma:group_action} to transform the induced local vectors into the standard basis of $\mathbb{F}_q^l$. When constructing logical representatives for cup products on another $t$-cube, the new representatives may therefore be cyclic translates of the original ones. 
Since a constant fraction of these choices are admissible, it is
natural to seek a cycle having a nonzero component on one of the corresponding $t$-cubes. Suppose that such a cycle
$\xi'\in C_t(X',\F_1 \otimes \cdots \otimes \F_t)$ exists.
Its chain map extension $\xi=P_t\xi'$ where $P_t=(P^t)^T$ as in Theorem~\ref{thm:induct_cup} then defines a tensor $T_\xi$ satisfying
\begin{align}\label{eq:sheaf_invariant}
	T_\xi ( L_{1}^{i_1}, \cdots, L_{t}^{i_t} ) = \delta_{i_1, \ldots,i_t},
\end{align}
whose subrank is at least $l=\Omega(N/\mathrm{polylog}\,N)$ and is therefore nearly linear in $k$, up to a polylogarithmic factor. Sec.~\ref{sec:cycles} gives an example of such a cycle $\xi$ using Reed--Solomon code families, but their product expansion behavior remains unclear. We leave as future work the problem of finding suitable local codes for which both Theorem~\ref{thm:cup_sheaf_1} and the existence of $\xi$ hold.


We conclude this subsection by discussing subrank estimates without assuming GRH or Artin's conjecture. In this case, the induced vectors may generate proper ideals $I_i\subset\mathbb{F}_q[C_l]$. Suppose that a cycle $\xi$ has been found as in Eq.~\eqref{eq:sheaf_invariant}, and assume for simplicity that $I_1=\cdots=I_t=I=\langle f(\x)\rangle$. Then by Eq.~\eqref{eq:sheaf_cup}, evaluating
$T_\xi(L_1^{i_1},\ldots,L_t^{i_t})$ amounts to computing the standard multiple dot products within the ideal $I$. More precisely, write $f(\x)=\sum_{r=0}^{l-1}c_r\x^r$. The induced vectors associated with the logical representatives are
$\x^{i_1}f(\x),\ldots,\x^{i_t}f(\x)$, and hence
\begin{align}\label{eq:dot_product}
	T_\xi(L_1^{i_1},\ldots,L_t^{i_t}) = \sum_{r=0}^{l-1} c_{r-i_1} \cdots c_{r-i_t},
\end{align}
where all coefficient subscripts are understood modulo $l$.
For $t=2$, this is the standard bilinear form on $\mathbb{F}_q[C_l]$ restricted to $I$, whose rank is determined by the radical $I\cap I^\perp$ where
\begin{align}
	I^\perp = \left\{ u\in\mathbb{F}_q[C_l]: \langle u,v\rangle=0 \text{ for every }v\in I \right\}.
\end{align}
By elementary linear algebra, $\dim I+\dim I^\perp=l$. Moreover, since
$\langle\x u,v\rangle=\langle u,\x^{l-1}v\rangle$,
$I^\perp$ is also an ideal. In the language of classical coding theory, $I$ defines a cyclic code and $I^\perp$ is its dual code. We therefore use the following standard result from the theory of cyclic codes~\cite{HuffmanPless2003}.

\begin{proposition}
	Let $g(\x) \in \mathbb{F}_q[\x]$ be a monic divisor of $1+\x^l$, and regard $\langle g(\x) \rangle$ as a cyclic code in $\mathbb{F}_q[C_l] \cong \mathbb{F}_q[\x]/(1+\x^l)$. Let $h(\x) = (1 + \x^l)/g(\x)$. Then $h^\ast(\x)/h(0)$ generates the dual code of $\langle g(\x) \rangle$, where $h^\ast(\x) \coloneq \x^{\deg h}h(1/\x)$ is the reciprocal polynomial of $h(\x)$ and $h(0) \neq 0$ is its constant term.
\end{proposition}

In our case, we further simplify the problem by assuming that $f(\x)$ is an irreducible factor of $1+\x+\cdots+\x^{l-1}$. Let
\begin{align}
	g(\x)=\frac{1+\x^l}{f(\x)}.
\end{align}
Then $g(\x)$ generates a minimal ideal. Applying the proposition to $I=\langle f(\x)\rangle$ gives $I^\perp = \langle g^\ast(\x) \rangle$, up to an irrelevant nonzero normalization factor. Since $g(\x)$ is the product of all irreducible factors of $1+\x^l$ other than $f(\x)$, if $f(\x)$ is \emph{self-reciprocal}, then $g^\ast(\x)$ is associated with $g(\x)$ and
\begin{align}
	\langle g(\x) \rangle \cap \langle f(\x) \rangle = \{0\}.
\end{align}
The bilinear form restricted to $I$ is therefore nondegenerate, with rank $\dim I=l-\deg f$, which is $\Theta(l)$ whenever $l-\deg f =\Theta(l)$.

If $f(\x)$ is not self-reciprocal, then $f^\ast(\x)$ is a distinct irreducible factor of $1 + \x^l$, and $g^\ast(\x)$ is divisible by $f(\x)$. Consequently,
\begin{align}
	I^\perp = \langle g^\ast(\x) \rangle  \subset I.
\end{align}
The bilinear form restricted to $I$ is therefore degenerate and has rank $\dim I-\dim I^\perp=l-2\deg f$. When $l$ is an odd prime and $q$ is even, the Chinese remainder theorem, or equivalently the decomposition of $\mathbb{F}_q[C_l]$ into irreducible representations over $\mathbb{F}_q$, gives
\begin{align}\label{eq:CRT_irrep}
	\mathbb{F}_q[C_l] \cong \mathbb{F}_q[\x]/(1+\x^l) \cong \mathbb{F}_q[\x]/(1+\x) \times \prod_i \mathbb{F}_q[\x]/(f_i(\x)) \cong \mathbb{F}_q \times \prod_i \mathbb{F}_{q^{\deg f_i}},
\end{align} 
where the $f_i(\x)$ are the irreducible factors of the cyclotomic polynomial $1+\x+\cdots+\x^{l-1}$. Each $f_i(\x)$ has degree $\operatorname{ord}_l(q)$. In particular, if $f=f_i$, then $\deg f=\operatorname{ord}_l(q)\mid(l-1)$. If $\deg f < l-1$, then $\deg f$ is a proper divisor of $l-1$, so $2\deg f\leq l-1 < l$. In the non-self-reciprocal case, depending on the scaling of $\deg f$, the rank $l-2\deg f$ may be a constant.

We now consider $t=3$ corresponding to the case of $\CCZ$ gates. Bounding the subrank of Eq.~\eqref{eq:dot_product} becomes substantially more difficult. A direct lower bound on the subrank is given by the maximum number $m_{\mathrm{disj}}$ of pairwise disjoint cyclic translates of $\operatorname{supp} f$ if $\sum_{r=0}^{l-1}c_r^3\neq0$. Over $\mathbb{F}_2$, this condition holds because
\begin{align}
	\sum_{r=0}^{l-1}c_r^3 = \sum_{r=0}^{l-1}c_r = f(1) = 1,
\end{align}
where the final equality follows because $f(1)=1$ for the irreducible factors of $1+\x + \cdots + \x^{l-1}$ considered here. 
Over a general field $\mathbb{F}_q$, the triple dot product of
$f(\x)$, $\x^i f(\x)$, and $\x^i f(\x)$ is nonzero for some $i$,
since $\sum_{i,r}c_r c_{r-i}^2=f(1)^3\neq0$. Restricting to such shifts may further reduce the lower bound. Other approaches are also possible, such as applying row reduction to the circulant matrix generated by the vectors $\x^{i_1} f(\x)$. After a suitable permutation of columns, this produces a systematic matrix of the form $(I_{l-\deg f}\mid M)$. However, the generally nonzero block $M$ makes it more difficult to extract a large diagonal subtensor and hence to bound the subrank. Even in this simplified setting where $I_i \equiv I$ is a maximal ideal, the precise behavior of $T_\xi$ as a $t$-tensor remains elusive. We leave the general case for future work.

It is worth noting that this multilinear structure is reminiscent of multi-orthogonal codes~\cite{Bravyi_2012,Paetznick_2013,Krishna_Tillich2019,Wills2024magic}. Nevertheless, the two settings are conceptually distinct: here the multilinear form in Eq.~\eqref{eq:dot_product} is required to be nonzero on selected tuples of codewords.


\subsection{Solving the cycle condition}\label{sec:cycles}

We now discuss when the boundary operator from $C_3(X,\F_1 \otimes \F_2 \otimes \F_3)$ to $C_2(X,\F_1 \otimes \F_2 \otimes \F_3)$ has a nontrivial kernel. That is, we ask when there exists a nonzero cycle $\xi$. Suppose $\F_1,\F_2,\F_3$ are generated by
\begin{align}
	(h_1,h_2,h_3), \quad (h_1',h_2',h_3'), \quad (h_1'',h_2'',h_3'').
\end{align}
As described at the end of Sec.~\ref{sec:cup}, the boundary operator on the tensor product cellular sheaf $\F_1 \otimes \F_2 \otimes \F_3$ has blocks of the following forms:
\begin{align}
	& \Big[ \cdots h_1(\text{-},a_1) \otimes h_1'(\text{-},a_1) \otimes	h_1''(\text{-},a_1) \cdots \Big], \\
	& \Big[ \cdots h_2(\text{-},a_2) \otimes h_2'(\text{-},a_2) \otimes	h_2''(\text{-},a_2) \cdots \Big], \\
	& \Big[ \cdots h_3(\text{-},a_3) \otimes h_3'(\text{-},a_3) \otimes	h_3''(\text{-},a_3) \cdots \Big].
\end{align}

\begin{theorem}\label{thm:cycle}
    The kernel of the above boundary operator is nontrivial if and only if all columnwise tensor product parity-check matrices have nontrivial kernels.
\end{theorem}
\begin{proof}
	The ``only if'' direction follows by restricting a nonzero global codeword to a vertex. For the ``if'' direction, any nonzero local kernel vector can be extended globally with the double cover.
\end{proof}

We provide further details in the one-dimensional classical setting by considering when the corresponding Tanner code is nontrivial. Recall that, in our construction, each underlying graph is replaced by its double cover to construct the cubical complexes (see Sec.~\ref{sec:cubical}) and the local codes are imposed as in Sec.~\ref{sec:DLV_code}. This structure makes it easier to write down Tanner codewords explicitly, as illustrated by the following example.

Let $r=(1\,2\,3)$ and $s=(1\,2)$, where permutations are composed from right to left. Then
\begin{align}
	r^3=s^2=e, \ sr=r^{-1}s=r^2s.
\end{align}
We consider the Cayley graph $\operatorname{Cay}(S_3,A)$ with the symmetric generating set $A=\{s,r,r^{-1}\}$. Label the vertices $S_3=\{v_1,\ldots,v_6\}$ by
\begin{align}
	(v_1,\ldots,v_6) = (e,r,r^2,s,sr,sr^2) = \bigl(e,(1\,2\,3),(1\,3\,2),(1\,2),(2\,3),(1\,3)\bigr).
\end{align}
Although the smaller symmetric generating set $\{(1\,2),(2\,3)\}$ also generates $S_3$, the present choice better illustrates the role of the double cover because $r\neq r^{-1}$.

From the generator action, there are nine edges labeled as follows:
\begin{align}
	\begin{aligned}
		& e_1 = [v_1,v_2], \quad e_2 = [v_1,v_3], \quad e_3 = [v_1,v_4], \quad e_4 = [v_2,v_3], \quad e_5 = [v_2,v_6], \\
		& e_6 = [v_3,v_5], \quad e_7 = [v_4,v_5], \quad e_8 = [v_4,v_6], \quad e_9 = [v_5,v_6].
	\end{aligned}
\end{align}

We order the generators as $\{a_1,a_2,a_3\} = \{s,r,r^{-1}\}$ and let $h = (h^{(1)},h^{(2)},h^{(3)})$ be the local parity-check matrix with three columns. Then the boundary operator and its associated Tanner code are given by
\begin{align*}
	\partial = 
	\begin{bmatrix}
		1 & 1 & 1 & 0 & 0 & 0 & 0 & 0 & 0 \\
		1 & 0 & 0 & 1 & 1 & 0 & 0 & 0 & 0 \\
		0 & 1 & 0 & 1 & 0 & 1 & 0 & 0 & 0 \\
		0 & 0 & 1 & 0 & 0 & 0 & 1 & 1 & 0 \\
		0 & 0 & 0 & 0 & 0 & 1 & 1 & 0 & 1 \\
		0 & 0 & 0 & 0 & 1 & 0 & 0 & 1 & 1 \\
	\end{bmatrix}, \quad
	H = 
	\begin{bmatrix}
		h^{(2)} & h^{(3)} & h^{(1)} & 0 & 0 & 0 & 0 & 0 & 0 \\
		h^{(3)} & 0 & 0 & h^{(2)} & h^{(1)} & 0 & 0 & 0 & 0 \\
		0 & h^{(2)} & 0 & h^{(3)} & 0 & h^{(1)} & 0 & 0 & 0 \\
		0 & 0 & h^{(1)} & 0 & 0 & 0 & h^{(2)} & h^{(3)} & 0 \\
		0 & 0 & 0 & 0 & 0 & h^{(1)} & h^{(3)} & 0 & h^{(2)} \\
		0 & 0 & 0 & 0 & h^{(1)} & 0 & 0 & h^{(2)} & h^{(3)}
	\end{bmatrix}.
\end{align*}
Generally, even if $\ker h \neq \{0\}$, it does not follow that $\ker H$ is nontrivial because different parts of $h$ may be inserted into one column of $\partial$. However, the construction in Secs.~\ref{sec:covering} and~\ref{sec:DLV_code}
uses the double cover of the Cayley graph
\begin{align}
		\left(
		\begin{array}{cc|cc|cc|cc|cc|cc|cc|cc|cc}
			1&0&1&0&1&0&0&0&0&0&0&0&0&0&0&0&0&0 \\
			0&1&0&1&0&1&0&0&0&0&0&0&0&0&0&0&0&0 \\
			\hline
			0&1&0&0&0&0&1&0&1&0&0&0&0&0&0&0&0&0 \\
			1&0&0&0&0&0&0&1&0&1&0&0&0&0&0&0&0&0 \\
			\hline
			0&0&0&1&0&0&0&1&0&0&1&0&0&0&0&0&0&0 \\
			0&0&1&0&0&0&1&0&0&0&0&1&0&0&0&0&0&0 \\
			\hline
			0&0&0&0&0&1&0&0&0&0&0&0&1&0&1&0&0&0 \\
			0&0&0&0&1&0&0&0&0&0&0&0&0&1&0&1&0&0 \\
			\hline
			0&0&0&0&0&0&0&0&0&0&0&1&0&1&0&0&1&0 \\
			0&0&0&0&0&0&0&0&0&0&1&0&1&0&0&0&0&1 \\
			\hline
			0&0&0&0&0&0&0&0&0&1&0&0&0&0&0&1&0&1 \\
			0&0&0&0&0&0&0&0&1&0&0&0&0&0&1&0&1&0 \\
		\end{array}
		\right).
\end{align}
Accordingly, the classical code is defined by
\begin{align}\label{eq:Tanner_cycle}
		\left(
		\begin{array}{cc|cc|cc|cc|cc|cc|cc|cc|cc}
			h^{(2)}&0&h^{(3)}&0&h^{(1)}&0&0&0&0&0&0&0&0&0&0&0&0&0 \\
			0&h^{(3)}&0&h^{(2)}&0&h^{(1)}&0&0&0&0&0&0&0&0&0&0&0&0 \\
			\hline
			0&h^{(3)}&0&0&0&0&h^{(2)}&0&h^{(1)}&0&0&0&0&0&0&0&0&0 \\
			h^{(2)}&0&0&0&0&0&0&h^{(3)}&0&h^{(1)}&0&0&0&0&0&0&0&0 \\
			\hline
			0&0&0&h^{(2)}&0&0&0&h^{(3)}&0&0&h^{(1)}&0&0&0&0&0&0&0 \\
			0&0&h^{(3)}&0&0&0&h^{(2)}&0&0&0&0&h^{(1)}&0&0&0&0&0&0 \\
			\hline
			0&0&0&0&0&h^{(1)}&0&0&0&0&0&0&h^{(2)}&0&h^{(3)}&0&0&0 \\
			0&0&0&0&h^{(1)}&0&0&0&0&0&0&0&0&h^{(3)}&0&h^{(2)}&0&0 \\
			\hline
			0&0&0&0&0&0&0&0&0&0&0&h^{(1)}&0&h^{(3)}&0&0&h^{(2)}&0 \\
			0&0&0&0&0&0&0&0&0&0&h^{(1)}&0&h^{(2)}&0&0&0&0&h^{(3)} \\
			\hline
			0&0&0&0&0&0&0&0&0&h^{(1)}&0&0&0&0&0&h^{(2)}&0&h^{(3)} \\
			0&0&0&0&0&0&0&0&h^{(1)}&0&0&0&0&0&h^{(3)}&0&h^{(2)}&0 
		\end{array}
		\right).
\end{align}
Now given any solution to $hz = 0$, copying the components of $z$ to suitable positions of the ambient code space yields a solution of $Hx = 0$. 

For lifted graphs and the corresponding cubical complexes, we use the chain map $P_3=(P^3)^T$ from Theorem~\ref{thm:induct_cup} to extend the solution from the base graph to the lifted graph and thereby obtain $\xi$.


\medskip 
We now give an explicit construction of a nonzero cycle $\xi$. Let $\mathbb{F}_q$ be a finite field with $q=2^s>n-1$, and choose distinct elements $\beta_1,\ldots,\beta_n\in\mathbb{F}_q$. For an integer $m$ satisfying $1\leq m\leq n-1$, define the $m\times n$ Vandermonde matrix as follows:
\begin{align}
	V_m(\beta_1,\ldots,\beta_n) \coloneq
	\begin{pmatrix}
		1 & 1 & \cdots & 1 \\
		\beta_1 & \beta_2 & \cdots & \beta_n \\
		\beta_1^2 & \beta_2^2 & \cdots & \beta_n^2 \\
		\vdots & \vdots & \ddots & \vdots \\
		\beta_1^{m-1} & \beta_2^{m-1} & \cdots & \beta_n^{m-1}
	\end{pmatrix}.
\end{align}
Viewed as a generator matrix, $V_m$ defines an $[n,m,n-m+1]$ Reed--Solomon (RS) code. Viewed as a parity-check matrix, its kernel is a generalized RS code with parameters $[n,n-m,m+1]$.
Assume that $n$ is even and that $\lceil2\sqrt{n-1}\rceil$ is odd.
To satisfy the conditions in Theorem~\ref{thm:cLDPC}, let
\begin{align}
	m_1 = \left\lceil 2\sqrt{n-1} \right\rceil, \quad
    m_2=m_3=\frac{n+1-m_1}{2}.
\end{align}
Let $h'$ and $h$ denote the above Vandermonde matrices with $m=m_1$ and $m=m_2$, respectively, and set $h_1=h'$ and $h_2=h_3=h$. Their three-fold columnwise tensor product contains, with repetitions, all Vandermonde rows of degrees from $0$ through $n-2$. It therefore has the same row space as $V_{n-1}$. Consequently, the tensor product parity-check matrix has rank $n-1$ and defines an $[n,1,n]$ code. The codewords in its one-dimensional kernel are precisely the scalar multiples of $x=(x_1,\ldots,x_n)$, where
\begin{align}
	x_i= \frac{1}{\displaystyle \prod_{j\neq i}(\beta_i-\beta_j)}, \ i=1,\ldots,n.
\end{align}
Theorem~\ref{thm:cycle} then yields the desired cycle $\xi$.

As in Sec.~\ref{sec:cup_sheaf}, we let $\F_1,\F_2,\F_3$ be generated by the respective triples of local codes
\begin{align}
	(h',h,h), \qquad (h,h',h), \qquad (h,h,h').
\end{align}
Let $\{G_V\}$ be a family of Lubotzky--Phillips--Sarnak (LPS) graphs of fixed degree $n$, with $\lvert V(G_V)\rvert=V/2$~\cite{LPS1988}. On the double cover of $G_V$, let $H'^T$ and $H^T$ denote the Tanner parity-check matrices associated with $h'$ and $h$, respectively. For $\F_1$, we have
$(H_1^T,H_2^T,H_3^T)=(H'^T,H^T,H^T)$, while the corresponding triples for $\F_2$ and $\F_3$ are obtained by cyclic permutation. For each $i\in\{1,2,3\}$, the hypergraph product cochain complex
\begin{align}
	C^0(X,\F_i) \rightarrow C^1(X,\F_i) \rightarrow C^2(X,\F_i) \rightarrow C^3(X,\F_i)
\end{align}
gives one of the three code families.

Recall from Sec.~\ref{sec:HGP_local} that
\begin{align}
	\xi \in \ker  
	\begin{pmatrix}
		\mathring{H}_1^T \otimes I_E \otimes I_E \\
		I_E \otimes \mathring{H}_2^T \otimes I_E \\
		I_E \otimes I_E \otimes \mathring{H}_3^T
	\end{pmatrix} = \ker \mathring{H}_1^T \otimes \ker \mathring{H}_2^T \otimes \ker \mathring{H}_3^T,
\end{align}
where $\mathring{H}_i^T$ is induced by the tensor product sheaf. For the Vandermonde matrices above, Theorem~\ref{thm:cycle} yields a nonzero cycle $\xi$. The construction in its proof, illustrated by Eq.~\eqref{eq:Tanner_cycle}, shows that $\xi$ can be chosen to be nowhere vanishing.


The next step is to find suitable logical representatives. In the three-dimensional construction above, $m_2=m_3<n/2$, so the available bounds do not guarantee that $\dim\ker H$ grows with $V$. Nevertheless, the preceding construction extends to the following $t$-dimensional setting without additional assumptions. Indeed, if $\sum_i m_i = n+1$, the matrix generated by elementwise tensor products of rows from the $t$ Vandermonde matrices is a parity-check matrix for an $[n,t-2,n-t+3]$ generalized RS code. In particular, this code contains full-support codewords, so the construction of the nowhere-vanishing cycle $\xi$ extends to general $t$.
\begin{corollary}
	For fixed $t\geq3$, consider $t$-dimensional HGP codes over a finite
field $\mathbb{F}_q$ with $q=2^s>n-1$, based on LPS graph families of
sufficiently large fixed degree $n$ and equipped with RS local codes
defined by Vandermonde matrices of ranks
\begin{align}\label{eq:HGP_RS_m}
	m_1=\left\lceil2\sqrt{n-1}\right\rceil,
	\qquad
	m_2=\cdots=m_t=\frac{n+1-m_1}{t-1},
\end{align}
where $n$ is chosen so that the
second expression is an integer. Then $\sum_{i=1}^t m_i=n+1$, and the resulting codes have parameters
	\begin{align}
			[\![N, k=\Omega(N^{1/t}), d = \Theta(N^{1/t})]\!].
	\end{align}
		They support nontrivial constant-depth cup product gates.
\end{corollary}

As mentioned in Remark~\ref{remark:HGP_2}, the above distance is optimal within the HGP code framework, but a tight bound on $k$ is not currently available. We make the following conjecture.

\begin{conjecture}
	For fixed $t$ and sufficiently large $n$ with $m_1,\ldots,m_t$ defined in Eq.~\eqref{eq:HGP_RS_m}, $\dim \ker H = \Theta(V)$, implying $k = \Theta(V^t) = \Theta(N)$.
\end{conjecture} 

If this conjecture holds, the techniques for handling local coefficients in Sec.~\ref{sec:module} can be applied here. In particular, we can construct $\Theta(V^t)$ logical representatives for each code family such that 
\begin{align}
L_{1}^{j_1} \smile_c \cdots \smile_c L_{t}^{j_t}
=
\sum_j \delta_{j_1,\ldots,j_t,j}
[(h_j,v_1,\ldots,v_t);a_1,\ldots,a_t].
\end{align}
These logical representatives, together with the nowhere-vanishing cycle $\xi$, would yield an invariant form $T_\xi$ with subrank $\Theta(V^t) = \Theta(N)$. 
Consequently, the validity of this conjecture would directly imply constant-overhead magic state distillation with qLDPC codes.


\end{document}